\documentclass[onecolumn,superscript address]{revtex4-2}
\usepackage{amsmath,amssymb,amsfonts}
\usepackage{float}
\usepackage{graphicx}
\usepackage{physics}
\usepackage{amsmath, mathrsfs}
\usepackage{amsthm}
\usepackage{amssymb}
\usepackage{enumitem} 
\usepackage{natbib} 
\usepackage{bm}
\usepackage{mathtools}
\usepackage{thmtools} 
\usepackage{thm-restate}
\usepackage{csquotes}
\usepackage{tikz}
\usepackage{soul}
\usepackage{xcolor}
\usepackage{array}
\usepackage{tabularx}
\usepackage[
    colorlinks=true,     
    linkcolor=blue,      
    citecolor=blue,      
    urlcolor=blue        
]{hyperref}
\usepackage{cleveref}
\definecolor{darkgreen}{rgb}{0.0,0.5,0.0}
\usetikzlibrary{arrows.meta, positioning}
\begin{document}
\renewcommand{\log}{\ln}
\newcommand{\pn}{\text{M}}
\newcommand{\pnm}{m}
\newcommand{\bd}[1]{\bm{#1}}
\newcommand{\pxa}{p_X^A}
\newcommand{\pxb}{s}
\newcommand{\pza}{p_Z^A}
\newcommand{\pzb}{1-s}
\newcommand{\etadeti}{\eta_{\text{det,i}}}
\newcommand{\nx}{n_X}
\newcommand{\nk}{n_K}
\newcommand{\errxobs}{e^{\text{obs}}_X}
\newcommand{\errzobs}{e^{\text{obs}}_{Z}}
\newcommand{\errzobsb}{\bd{e}^{\textbf{obs}}_\textbf{Z}}
\newcommand{\errkobs}{{e}^{\text{key}}_{{X}}}

\newcommand{\errx}{\bd{e}^{\textbf{obs}}_{\textbf{X}}}
\newcommand{\errk}{\bd{e}^{\textbf{key}}_{\textbf{X}}}
\newcommand{\gserf}[1]{\gamma^{#1}_{serf}}
\newcommand{\epat}{\epsilon_{\text{AT}}}
\newcommand{\epatsingle}{\epsilon_{\text{AT-single}}}
\newcommand{\epata}{\epsilon_{\text{AT-a}}}
\newcommand{\epev}{\epsilon_{\text{EV}}}
\newcommand{\eppa}{\epsilon_{\text{PA}}}

\newcommand{\nkb}{\bd{\nk}}
\newcommand{\nxb}{\bd{\nx}}

\newcommand{\nmulti}{\bd{n_{>1}}}
\newcommand{\mcrateb}{\bd{e}_{\textbf{mc}}}
\newcommand{\mcrate}{{e}_{\text{mc}}}
\newcommand{\lec}{\lambda_{\text{EC}}}

\newcommand{\errks}{\bd{{e}_{{{(\bm{X},\textbf{\pn})}}}^{\textbf{key}}}}
\newcommand{\eppnt}{\epsilon_{\text{PNE}}}
\newcommand{\eppntb}{\epsilon_{\text{PNE}}}
\newcommand{\lm}{\lambda_{\min}}
\newcommand{\ntkb}{\boldsymbol{\tilde{n}_\textbf{K,1}}}
\newcommand{\Nxobs}{{N}_{X}^{\text{obs}}}
\newcommand{\Nxobsb}{\bd{N}_{\bm{X}}^{\textbf{obs}}}
\newcommand{\xnoeobssb}{\bd{{N}_{{X,1}}}^{\textbf{obs}}}
\newcommand{\xnoeobsmb}{\bd{N}_{{X},\bd{>1}}^{\textbf{obs}}}
\newcommand{\xnoeobs}{{N}_{X}^{\text{obs}}}
\newcommand{\bimp}{\mathcal{B}_{\text{imp}(a,\delta,q_Z)}}
\newcommand{\K}{\mathcal{B}_{\text{K}}}
\newcommand{\gmzero}{\gamma_0}
\newcommand{\epazua}{\epsilon_{\text{Az-a}}}
\newcommand{\epazub}{\epsilon_{\text{Az-b}}}
\newcommand{\epzero}{\epsilon_{0}}
\newcommand{\epatd}{\epsilon_{\text{AT-d}}}
\newcommand{\epats}{\epsilon_{\text{AT-s}}}
\newcommand{\epatt}{\epsilon_{\text{AT-t}}}

\newcommand{\etadet}{\eta_{\text{det}}}
\newcommand{\deltaeta}{\Delta_{{\eta}}}
\newcommand{\ddet}{d_{\text{det}}}
\newcommand{\deltadc}{\Delta_{\text{dc}}}

\newcommand{\nt}{\tilde{n}}
\newcommand{\idd}{\mathbb{I}}
\newcommand{\POVMZeq}{\Gamma_{(Z,=)}}
\newcommand{\POVMZneq}{\Gamma_{(Z,\neq)}}
\newcommand{\POVMZeqM}{\Gamma^{(\pn)}_{(Z,=)}}
\newcommand{\POVMZneqM}{\Gamma^{(\pn)}_{(Z,\neq)}}
\newcommand{\GZeq}{G^{(\pn)}_{(Z,=)}}
\newcommand{\GZneq}{G^{(\pn)}_{(Z,\neq)}}
\newcommand{\POVMXeq}{\Gamma_{(X,=)}}
\newcommand{\POVMXneq}{\Gamma_{(X,\neq)}}
\newcommand{\POVMXeqM}{\Gamma^{(\pn)}_{(X,=)}}
\newcommand{\POVMXneqM}{\Gamma^{(\pn)}_{(X,\neq)}}
\newcommand{\GXeq}{G^{(\pn)}_{(X,=)}}
\newcommand{\GXneq}{G^{(\pn)}_{(X,\neq)}}
\newcommand{\fx}{F_X}
\newcommand{\fz}{F_Z}
\newcommand{\fxone}{F_X^{(1)}}
\newcommand{\fzone}{F_Z^{(1)}}
\newcommand{\ftx}{{\tilde{F}^{(\pn)}_X}}
\newcommand{\ftz}{{\tilde{F}^{(\pn)}_Z}}
\newcommand{\ftxone}{{\tilde{F}^{(1)}_X}}
\newcommand{\ftzone}{{\tilde{F}^{(1)}_Z}}
\newcommand{\GtZeq}{\tilde{G}^{(\pn)}_{(Z,=)}}
\newcommand{\GtZneq}{\tilde{G}^{(\pn)}_{(Z,\neq)}}
\newcommand{\GtXeq}{\tilde{G}^{(\pn)}_{(X,=)}}
\newcommand{\GtXneq}{\tilde{G}^{(\pn)}_{(X,\neq)}}
\newcommand{\GtZeqone}{\tilde{G}^{(1)}_{(Z,=)}}
\newcommand{\GtZneqone}{\tilde{G}^{(1)}_{(Z,\neq)}}
\newcommand{\GtXeqone}{\tilde{G}^{(1)}_{(X,=)}}
\newcommand{\GtXneqone}{\tilde{G}^{(1)}_{(X,\neq)}}

\newcommand{\kb}[1]{\ketbra{#1}}
\newcommand{\POVMBXone}{\Gamma_{(X,1)}^{\text{B}}}
\newcommand{\POVMBXzero}{\Gamma_{(X,0)}^{\text{B}}}
\newcommand{\POVMBZone}{\Gamma_{(Z,1)}^{\text{B}}}
\newcommand{\POVMBZzero}{\Gamma_{(Z,0)}^{\text{B}}}
\newcommand{\POVMBXoneM}{\Gamma_{(X,1)}^{(\text{B},\pn)}}
\newcommand{\POVMBXzeroM}{\Gamma_{(X,0)}^{(\text{B},\pn)}}
\newcommand{\POVMBZoneM}{\Gamma_{(Z,1)}^{(\text{B},\pn)}}
\newcommand{\POVMBZzeroM}{\Gamma_{(Z,0)}^{(\text{B},\pn)}}
\newcommand{\POVMBXdc}{\Gamma_{(X,dc)}^{\text{B}}}
\newcommand{\POVMBZdc}{\Gamma_{(Z,dc)}^{\text{B}}}
\newcommand{\POVMBcc}{\Gamma_{\mathrm{cc}}^{\text{B}}}
\newcommand{\POVMBvac}{\Gamma_{\mathrm{nc}}^{\text{B}}}
\newcommand{\POVMBmc}{\Gamma_{\mathrm{mc}}^{\text{B}}}

\newcommand{\POVMAXone}{\Gamma_{(X,1)}^{\text{A}}}
\newcommand{\POVMAXzero}{\Gamma_{(X,0)}^{\text{A}}}
\newcommand{\POVMAZone}{\Gamma_{(Z,1)}^{\text{A}}}
\newcommand{\POVMAZzero}{\Gamma_{(Z,0)}^{\text{A}}}

\newcommand{\half}{\frac{1}{2}}
\newcommand{\event}[1]{\Omega(#1)}

\newcommand{\condistate}[1]{\rho_{|\event{#1}}}
\newcommand{\X}{\bm{X}}
\newcommand{\Z}{\bm{Z}}
\newcommand{\fserf}{f_\text{serf}}
\newcommand{\PiN}{\Pi_{\leq \pn}}
\newcommand{\Pign}{\Pi_{>\pn}}
\newcommand{\ngN}{\tilde{n}_{(>\pn)}}
\newcommand{\ngNb}{\bd{\tilde{n}_{(>\textbf{\pn})}}}
\newcommand{\nlambb}{\bd{\tilde{n}_{(\lm\Pi_{>\textbf{\pn}})}}}
\newcommand{\neb}{\bd{n_{mc}}}
\newcommand{\netb}{\bd{\tilde{n}_{\mathrm{mc}}}}
\newcommand{\num}{\bm{N}}
\newcommand{\condi}{|\Omega}
\newcommand{\nmc}{{n_{\mathrm{mc}}}}
\newcommand{\nmcb}{\bd{n_{\mathrm{mc}}}}
\newcommand{\nmct}{{\tilde{n}_{mc,>\text{\pn}}}}
\newcommand{\nmctb}{\bd{\tilde{n}_{(\mathrm{mc},>\textbf{\pn})}}}

\newcommand{\errxs}{\boldsymbol{{e}_{{\textbf{X,1}}}}}
\newcommand{\POVMBXones}{\Gamma_{(X,1)}^{\text{B},\leq1}}
\newcommand{\POVMBXzeros}{\Gamma_{(X,0)}^{\text{B},\leq1}}
\newcommand{\POVMBZones}{\Gamma_{(Z,1)}^{\text{B},\leq1}}
\newcommand{\POVMBZzeros}{\Gamma_{(Z,0)}^{\text{B},\leq1}}

\newcommand{\nks}{\tilde{n}_{{(K,1)}}}
\newcommand{\nxs}{\tilde{n}_{{(X,1)}}}
\newcommand{\nksb}{\bd{\tilde{n}}_{\bm{(K,1)}}}
\newcommand{\nxsb}{\bd{\tilde{n}}_{\bm{(X,1)}}}

\newcommand{\ngone}{n_{>1}}
\newcommand{\ngoneb}{\bd{\tilde{n}_{(>1)}}}
\newcommand{\V}{\bm{V}}

\newcommand{\Pigone}{\Pi_{>1}}
\newcommand{\binfunction}[3]{F(#1,#2,#3)}

\newcommand{\Nxpred}{\bar{N}_X}
\newcommand{\rvx}{{\textbf{X}}_i}
\newcommand{\rvy}{{\textbf{Y}}_i}
\newcommand{\rvyt}{{\tilde{\textbf{Y}}}_i}
\newcommand{\rvm}{\textbf{M}_i}
\newcommand{\rvmt}{\tilde{\textbf{M}}_i}
\newcommand{\rvmn}{\textbf{M}_{{n}}}
\newcommand{\rvmtn}{\tilde{\textbf{M}}_{{n}}}
\newcommand{\rvE}[1]{\textbf{E}(#1)}
\newcommand{\Mhis}[1]{\textbf{M}_0^{#1}}
\newcommand{\mhis}[1]{m_0^{#1}}
\newcommand{\Mhist}[1]{\tilde{\textbf{M}}_0^{#1}}
\newcommand{\Xhis}[1]{\textbf{X}_1^{#1}}
\newcommand{\xhis}[1]{x_1^{#1}}
\newcommand{\tX}{\textbf{X}}

\newcommand{\Fonesc}{F^{(\pn)}_{\mathrm{sc}}}
\newcommand{\Fonescone}{F^{(1)}_{\mathrm{sc}}}
\newcommand{\Gonemc}{\Gamma^{(\pn)}_{\mathrm{mc}}}
\newcommand{\Gonenc}{\Gamma^{(\pn)}_{\mathrm{nc}}}
\newcommand{\nkzero}{\boldsymbol{\tilde{n}_{\bm{(K,0)}}}}
\newcommand{\Nkoneb}{\bd{\tilde{N}^{\textbf{key}}_{{(\bm{X},\textbf{\pn})}}}}
\newcommand{\Nkone}{{N}^{\text{key}}_{{(X,\pn)}}}
\newcommand{\Nxoneb}{\bd{\tilde{N}_{(\bm{X},\textbf{\pn})}}}
\newcommand{\Nxone}{{\tilde{N}}_{{(X,\pn)}}}
\newcommand{\nxone}{\tilde{n}_{{(X,\pn)}}}
\newcommand{\nxoneb}{\bd{\tilde{n}_{{(\bm{X},\textbf{\pn})}}}}
\newcommand{\nkone}{\tilde{n}_{{(K,\pn)}}}
\newcommand{\nkoneb}{\bd{\tilde{n}_{{(\bm{K},\textbf{\pn})}}}}
\newcommand{\Gzeromc}{\Gamma^{(0)}_{\mathrm{mc}}}
\newcommand{\Gzeronc}{\Gamma^{(0)}_{\mathrm{nc}}}
\newcommand{\GzeroXeq}{\Gamma^{(0)}_{({X,=})}}
\newcommand{\GzeroXneq}{\Gamma^{(0)}_{({X},{\neq})}}
\newcommand{\GzeroZeq}{\Gamma^{(0)}_{({Z,=})}}
\newcommand{\GzeroZneq}{\Gamma^{(0)}_{({Z},{\neq})}}
\newcommand{\GzeroZ}{\Gamma^{(0)}_{{Z}}}
\newcommand{\nonet}{\tilde{n}_{(\mathrm{sc},\pn)}}
\newcommand{\nonetb}{\bd{\tilde{n}_{(\mathrm{sc},\textbf{\pn})}}}
\newcommand{\nkgone}{\bd{\tilde{n}_{({K},>\textbf{\pn})}}}
\newcommand{\GoneXeq}{\Gamma^{(1)}_{(X,=)}}
\newcommand{\GoneXneq}{\Gamma^{(1)}_{({X},{\neq})}}
\newcommand{\GoneZeq}{\Gamma^{(1)}_{(Z,=)}}
\newcommand{\GoneZneq}{\Gamma^{(1)}_{({Z},{\neq})}}

\newcommand{\GMXeq}{\Gamma^{(\pn)}_{({X,=})}}
\newcommand{\GMXneq}{\Gamma^{(\pn)}_{({X},{\neq})}}
\newcommand{\GMZeq}{\Gamma^{(\pn)}_{({Z,=})}}
\newcommand{\GMZneq}{\Gamma^{(\pn)}_{({Z},{\neq})}}
\newcommand{\etaxzero}{\eta_{(X,0)}}
\newcommand{\etaxone}{\eta_{(X,1)}}
\newcommand{\etaxmin}{\eta_{(X,\text{{min}})}}
\newcommand{\etaxmax}{\eta_{(X,\text{{max}})}}
\newcommand{\etazzero}{\eta_{(Z,0)}}
\newcommand{\etazone}{\eta_{(Z,1)}}
\newcommand{\etazmin}{\eta_{(Z,\text{{min}})}}
\newcommand{\etazmax}{\eta_{(Z,\text{{max}})}}
\newcommand{\etamax}{\eta_{\text{max}}}
\newcommand{\etamin}{\eta_{\text{min}}}

\newcommand{\dxzero}{d_{(X,0)}}
\newcommand{\dxone}{d_{(X,1)}}
\newcommand{\dxmin}{d_{(X,\text{min})}}
\newcommand{\dxmax}{d_{(X,\text{max})}}
\newcommand{\dzzero}{d_{(Z,0)}}
\newcommand{\dzone}{d_{(Z,1)}}
\newcommand{\dzmin}{d_{Z,\text{{min}}}}
\newcommand{\dzmax}{d_{Z,\text{{max}}}}

\newcommand{\dmax}{d_{\text{max}}}
\newcommand{\dmin}{d_{\text{min}}}

\newcommand{\etat}{\tilde{p}}
\newcommand{\infnorm}[1]{\norm{#1}_{\infty}}
\newcommand{\infnormlong}[1]{{\norm{#1}_{\infty}}}
\newcommand{\inorm}[1]{\norm{#1}_{\text{ind},\infty}}
\newcommand{\onenorm}[1]{\norm{#1}_{\text{ind},1}}
\newcommand{\sone}{s_1}
\newcommand{\stwo}{s_2}
\newcommand{\sthree}{s_2}

\newcommand{\POVMBXonehat}{\hat{\Gamma}_{(X,1)}^{\text{B}}}
\newcommand{\POVMBXzerohat}{\hat{\Gamma}_{(X,0)}^{\text{B}}}
\newcommand{\POVMBZonehat}{\hat{\Gamma}_{(Z,1)}^{\text{B}}}
\newcommand{\POVMBZzerohat}{\hat{\Gamma}_{(Z,0)}^{\text{B}}}

\newcommand{\POVMABXeq}{\Gamma_{X,=}^{\text{AB}}}
\newcommand{\POVMABXneq}{\Gamma_{X,\neq}^{\text{AB}}}
\newcommand{\POVMABXcon}{F_X}

\newcommand{\POVMABZeq}{\Gamma_{Z,=}^{\text{AB}}}
\newcommand{\POVMABZneq}{\Gamma_{Z,\neq}^{\text{AB}}}
\newcommand{\POVMABZcon}{F_Z}
\newcommand{\POVMABcon}{F^{(1)}_{sc}}
\newcommand{\POVMBX}{\Gamma_{(X,\text{con})}^{{(B,1)}}}
\newcommand{\POVMBZ}{\Gamma_{(Z,\text{con})}^{{(B,1)}}}
\newcommand{\POVMBcon}{\Gamma_{\text{con}}^{{(B,1)}}}
\newcommand{\POVMBconM}{\Gamma_{\text{con}}^{{(B,\pn)}}}
\newcommand{\POVMBXM}{\Gamma_{(X,\text{con})}^{{(B,\pn)}}}
\newcommand{\POVMBZM}{\Gamma_{(Z,\text{con})}^{{(B,\pn)}}}
\newcommand{\gmxe}{\Gamma_{X,=}}
\newcommand{\gmxne}{\Gamma_{X,\neq}}
\newcommand{\gmze}{\Gamma_{(Z,=)}}
\newcommand{\gmzne}{\Gamma_{(Z,\neq)}}
\newcommand{\gmx}{\Gamma_{X}}
\newcommand{\gmz}{\Gamma_{Z}}
\newcommand{\Gmc}{\Gamma_{\mathrm{mc}}}

\newcommand{\etacom}{\eta_{\text{min}}}
\newcommand{\Pimc}{\Pi_{\mathrm{mc}}}
\newcommand{\Pinc}{\Pi_{\mathrm{nc}}}
\newcommand{\Pisc}[1]{\Pi_{\mathrm{sc},#1}}

\newcommand{\rhoN}{\rho^{(N)}}
\newcommand{\muk}{\mu_k}
\newcommand{\mukvec}{\mu_{\vec{k}}}
\newcommand{\errkss}{\boldsymbol{{e}_{(K,\textbf{\pn,single})}}^{\textbf{key}}}
\newcommand{\nxbs}{\bd{n_{(\bd{X},\textbf{single})}}}
\newcommand{\nkbs}{\bd{n_{(\bd{K},\textbf{single})}}}
\newcommand{\nkbss}{\bd{\tilde{n}}_{\textbf{(K,\pn,{single})}}}

\newcommand{\xnoeobss}{\bd{N_{(X,\textbf{,single})}^{\textbf{obs}}}}
\newcommand{\bmaxone}{\mathcal{B}^{\text{decoy}}_{\text{max-1}}}
\newcommand{\bminone}{\mathcal{B}^{\text{decoy}}_{\text{min-1}}}
\newcommand{\bminzero}{\mathcal{B}^{\text{decoy}}_{\text{min-0}}}
\newcommand{\nxmuvec}{\bd{n_{X,\mukvec}}}
\newcommand{\nxneqmuvec}{\bd{n_{X\neq,\mukvec}}}
\newcommand{\nkmuvec}{\bd{n_{K,\mukvec}}}
\newcommand{\nomuvec}{\bd{n_{O,\mukvec}}}

\newcommand{\nxmu}[1]{\bd{n_{X,\mu_{#1}}}}
\newcommand{\nxneqmu}[1]{\bd{n_{X\neq,\mu_{#1}}}}
\newcommand{\nkmu}[1]{\bd{n_{K,\mu_{#1}}}}
\newcommand{\nomu}[1]{\bd{n_{O,\mu_{#1}}}}
\newcommand{\nomuplus}[1]{\bd{n_{O,\mu_{#1}}^+}}
\newcommand{\nomuminus}[1]{\bd{n_{O,\mu_{#1}}^-}}
\newcommand{\no}{\bd{n_O}}

\newcommand{\be}{\mathcal{B}_e}
\newcommand{\bone}{\mathcal{B}_1}

\newcommand{\nxmuvecobs}{{n_{X,\mukvec}}}
\newcommand{\nxneqmuvecobs}{{n_{X\neq,\mukvec}}}
\newcommand{\nkmuvecobs}{{n_{K,\mukvec}}}
\newcommand{\nomuvecobs}{{n_{O,\mukvec}}}

\newcommand{\nkbssobs}{{\tilde{n}}_{({{K},\pn},\text{single})}}
\newcommand{\Evente}{\Omega_{(\nxmuvecobs,\nkmuvecobs,\nxneqmuvecobs,e_{Z}^{\text{obs}},\nmc,\nkbssobs)}}
\newcommand{\Se}{S_{(\nxmuvecobs,\nkmuvecobs,\nxneqmuvecobs,e_{Z}^{\text{obs}},\nmc)}}
\newcommand{\Hmindecoy}[1]{H_{\text{min}}^{\sqrt{\kappa{(\bar{\Omega})}}}(Z^{#1}|E^nC)_{\rho_{|\bar{\Omega}}}}
\newcommand{\Hmaxdecoy}[1]{H_{\text{max}}^{\sqrt{\kappa{(\bar{\Omega})}}}(X^{\nkbssobs}|{#1})_{\rho_{|\bar{\Omega}}}}
\newcommand{\red}[1]{\textcolor{red}{#1}}
\newcommand{\blue}[1]{\textcolor{blue}{#1}}
\newcommand{\brown}[1]{\textcolor{brown}{#1}}
\newcommand{\black}[1]{\textcolor{black}{#1}}
\newcommand{\green}[1]{\textcolor{green}{#1}}

\newcommand{\Kcor}{\mathcal{B}_{\text{rounds}}^{\text{cor}}}
\newcommand{\Fc}{F_{\text{click}}}
\newcommand{\Fcgm}{F^{(>1)}_{\text{click}}}
\newcommand{\Fcm}{F^{(1)}_{\text{click}}}
\newcommand{\FonescCor}{F^{(1,\text{keep})}_{\mathrm{sc}}}
\newcommand{\FonescCoruncor}{F^{(1)}_{\mathrm{sc}}}

\newcommand{\nmccor}{n_{\mathrm{mc}}}
\newcommand{\nxcor}{n_{X}}
\newcommand{\nkcor}{n_{K}}
\newcommand{\Nxobscor}{N_{X}^{\text{obs}}}
\newcommand{\errzobscor}{e^{\text{obs}}_{Z}}
\newcommand{\ncor}{n}
\newcommand{\ngNcor}{\tilde{n}_{(>1,\text{keep})}}
\newcommand{\ngNcorb}{\bm{\tilde{n}_{(>1,\textbf{keep})}}}

\newcommand{\nmcbcor}{\bm{n_{\mathrm{mc}}}}
\newcommand{\nxbcor}{\bm{n}_{\bm{X}}}
\newcommand{\nkbcor}{\bm{n}_{\bm{K}}}
\newcommand{\Nxobsbcor}{\bm{N_{X}^{\textbf{obs}}}}
\newcommand{\errzobsbcor}{e^{\text{obs}}_{Z}}

\newcommand{\ngNbcor}{\bd{\tilde{n}_{(>1,\textbf{keep,click})}}}
\newcommand{\ngNbcorobs}{{\tilde{n}_{(>1,\text{keep,click})}}}

\newcommand{\nkzerocor}{\bd{\tilde{n}}_{\bm{(K,0)}}}
\newcommand{\nkzerocorobs}{{\tilde{n}}_{{(K,0)}}}
\newcommand{\nxzerocorobs}{{\tilde{n}}_{{X},\text{0,keep}}}
\newcommand{\nzeromccorobs}{{\tilde{n}}_{{mc},\text{0,keep}}}
\newcommand{\nonecor}{\bd{\tilde{n}_{1,\textbf{keep,click}}}}
\newcommand{\nonecorobs}{{\tilde{n}_{1,\text{keep,click}}}}
\newcommand{\nmctbcor}{\bd{\tilde{n}_{(\mathrm{mc},>1)}}}
\newcommand{\nonesccor}{\tilde{n}_{(\mathrm{sc},1)}}
\newcommand{\ftxcor}{{\tilde{F}^{(1,\text{keep})}_X}}
\newcommand{\ftzcor}{{\tilde{F}^{(1,\text{keep})}_Z}}

\newcommand{\Nxonebcor}{\bd{\tilde{N}_{(\bm{X}\textbf{,1})}}}
\newcommand{\Nkonebcor}{\bd{\tilde{N}^{\textbf{key}}_{(\bm{X}\textbf{,1})}}}
\newcommand{\nkonebcor}{\bd{\tilde{n}}_{\bm{(K,1)}}}
\newcommand{\nkonecor}{{\tilde{n}}_{{(K,1)}}}
\newcommand{\nxonebcor}{\bd{\tilde{n}}_{\bm{(X,1)}}}

\newcommand{\netbcor}{\bm{\tilde{n}_{(\mathrm{mc},>1)}}}
\newcommand{\bimpcor}{\mathcal{B}^{\text{cor}}_{\text{imp}(a,\delta,q_Z)}}
\newcommand{\errkscor}{\bd{{e}_{{{(X,1)}}}^{\textbf{key}}}}

\newcommand{\sminentropy}[2]{H^{\sqrt{\kappa}}_{\text{min}}(#1|#2)}
\newcommand{\smaxentropy}[2]{H^{\sqrt{\kappa}}_{\text{max}}(#1|#2)}
\newcommand{\sminentropynew}[3]{H^{\sqrt{\kappa(#3)}}_{\text{min}}(#1|#2)}
\newcommand{\smaxentropynew}[3]{H^{\sqrt{\kappa(#3)}}_{\text{max}}(#1|#2)}
\newcommand{\filteredstate}{{\rho_{A^{n_{(\text{keep,click})}}B^{n_{(\text{keep,click})}}E^n}}_{|\Omega(n_{(\text{keep,click})})}}
\newcommand{\nxmuvecobscor}{{n_{X,\mukvec}}}
\newcommand{\nxneqmuvecobscor}{{n_{X\neq,\mukvec}}}
\newcommand{\nkmuvecobscor}{{n_{K,\mukvec}}}
\newcommand{\nxmuveccor}{\bm{n_{X,\mukvec}}}
\newcommand{\nxneqmuveccor}{\bm{n_{X\neq,\mukvec}}}
\newcommand{\nkmuveccor}{\bm{n_{K,\mukvec}}}
\newcommand{\fcx}{\mathcal{F}_X}
\newcommand{\fcz}{\mathcal{F}_Z}
\newcommand{\Sindep}{S_{0,0}}
\newcommand{\Prindep}{\Pr_{\Sindep}}
\newcommand{\epssrc}{\varepsilon_\mathrm{ind}}\newcommand{\Eindep}{\mathcal{E}_{0,0}}
\newcommand{\epsATb}{\varepsilon_\mathrm{dep-1}}
\newcommand{\epsATc}{\varepsilon_\mathrm{dep-2}}
\newcommand{\Sdep}{S_{\delta_1,\delta_2}}
\newcommand{\Prdep}{\Pr_{\Sdep}}
\newcommand{\Edep}{\mathcal{E}_{\delta_1,\delta_2}}
\newcommand{\GtZeqcor}{\tilde{G}^{(1,\text{keep})}_{(Z,=)}}
\newcommand{\GtZneqcor}{\tilde{G}^{(1,\text{keep})}_{(Z,\neq)}}
\newcommand{\GtXeqcor}{\tilde{G}^{(1,\text{keep})}_{(X,=)}}
\newcommand{\GtXneqcor}{\tilde{G}^{(1,\text{keep})}_{(X,\neq)}}
\newcommand{\circletext}[1]{\textcircled{\raisebox{-0.9pt}{#1}}}
\newcommand{\diag}{\text{diag}}
\theoremstyle{definition} 
\newtheorem{theorem}{Theorem}
\newtheorem{definition}[theorem]{Definition}
\newtheorem{lemma}{Lemma}
\newtheorem{corollary}{Corollary}[lemma]
\newtheorem{remark}{Remark}
\newtheorem{prot}{Protocol}
\crefname{section}{Sec.}{Sec.}
\crefname{remark}{Remark}{Remark}
\crefname{equation}{Eq.}{Eqs.}
\crefname{figure}{Fig.}{Figs.}
\crefname{lemma}{Lemma}{Lemmas}
\crefname{corollary}{Corollary}{Corollaries}
\crefname{appendix}{Appendix}{Appendix}
\crefname{theorem}{Theorem}{Theorems}
\crefname{enumi}{Step}{Step}

\title{Phase error estimation for passive detection setups with imperfections and memory effects}
\author{Zhiyao Wang}
\email{z2425wan@uwaterloo.ca}
\author{Devashish Tupkary}
\email{djtupkary@uwaterloo.ca}
\author{Shlok Nahar}
\email{sanahar@uwaterloo.ca}
\affiliation{Institute for Quantum Computing and Department of Physics and Astronomy,\\
University of Waterloo, Waterloo, Ontario, Canada, N2L 3G1}

\begin{abstract}
We develop a generic framework to bound the phase error rate for quantum key distribution (QKD) protocols using passive detection setups with imperfections and memory effects. This framework can be used in proof techniques based on the entropic uncertainty relation (EUR) or phase error correction, to prove security in the finite-size regime against coherent attacks. Our framework can incorporate on-the-fly announcements of click/no-click outcomes on Bob’s side. In the case of imperfections without memory effects, it can be combined with proofs addressing source imperfections in a modular manner. We apply our framework to compute key rates for the decoy-state BB84 protocol, when the beam splitting ratio, the detection efficiency, and dark counts of the detectors are only known to be within some ranges. We also compute key rates in the presence of memory effects in the detectors. In this case, our results allow for protocols to be run at higher repetition rates, resulting in a significant improvement in the secure key generation rate. 
\end{abstract}
\maketitle

\section{Introduction}
Prepare‑and‑measure quantum key distribution (QKD) systems are often preferred over measurement\--device\--independent (MDI) schemes \cite{Lo_2012} because they are simpler and easier to implement. However, they rely on detection modules that inevitably deviate from the ideal models assumed in security proofs. In particular, threshold detectors can exhibit memory effects such as afterpulsing and detector dead times, which introduce correlations between successive protocol rounds and invalidate the usual assumption of independent measurements across the rounds. We refer to these kinds of detectors as correlated detectors. Addressing such memory effects in any security proof is a crucial step towards enabling the protocol to run at high repetition rates.
Further, passive detection schemes are attractive for their simplicity and stability at high protocol repetition rates. However, certain proof techniques, such as those based on the entropic uncertainty relation (EUR) and phase error correction (PEC), are unable to prove security for passive protocols~\cite{bdr2} \footnote{These proof techniques can be used for idealised passive protocols where the basis-choice beam splitting ratio is set so that the protocol can be considered to be equivalent to an active protocol \cite{Fung_2011,Gittsovich_2014}. However, they are not robust to even infinitesimal deviations from this idealised setting.}. While other proof techniques based on the entropy accumulation theorems \cite{MEAT,MEATapplication,eat} and the postselection technique \cite{postselection,larsPS} can directly prove security for passive detection setups, none of them can currently address memory effects in detectors. Our primary contributions to QKD security proofs are two-fold:
\begin{enumerate}
    \item We incorporate certain types of detector memory effects into phase error rate based proof techniques. Prior to our work, this could not be done in \emph{any} proof technique \footnote{Note that \cite[Section V. B.]{shlok} describes a proof sketch to address detector memory effects via entropy accumulation-based proof techniques.}, and is a foundational advancement to the field.
    \item We develop a generic framework to prove security for passive protocols within the EUR-based and PEC-based proof techniques. This was already possible in postselection~\cite{postselection,larsPS} and entropy accumulation theorems (EAT)-based~\cite{MEAT,MEATapplication,eat} proof techniques, but constitutes a major advancement in the commonly used EUR/PEC-based proof techniques. 
    Indeed, this analysis forms the basis for our treatment of detector memory effects.
    
\end{enumerate}
In addition to the primary contributions of our work, we have a number of other important contributions:
\begin{enumerate}
    \item  Our analysis is applicable to imperfections that are only partially characterised. In order to obtain this result, we extend the method in Ref.~\cite{lars} to bound the probability that a state input into Bob's detection setup has few photons \footnote{In our work, we more accurately bound the number of multi-photon rounds input into Bob's detection setup.}. This is an important step to address imperfect detectors in other proof techniques as well \cite{shlok}.
    \item Our analysis is compatible with on-the-fly announcements of click/no-click outcomes~\cite{iterativesifting,Pfister_2016}, a common feature in practical implementations aimed at reducing hardware memory requirements. 
    \item Our analysis for memoryless imperfect detectors is compatible with imperfect sources via the method in Ref.~\cite{deviceimp}. However, the analysis for detectors with memory effects is not currently compatible in the same way, as Ref.~\cite{deviceimp} requires additional assumptions on the detector side. That said, our analysis is modular—if these additional assumptions could be removed or relaxed, our results would remain applicable.
    \item Our analysis is performed in the single-mode setting, but can be easily extended to a simple multi-mode detector model.
\end{enumerate}

Recently, Ref.~\cite{shun} independently proposed a similar approach to this work, for phase error estimation for passive detection setups. Ref.~\cite{shun} focuses on the asymptotic regime and imposes several simplifying assumptions: Bob’s detectors are assumed to have identical efficiencies within each basis; dark count rates are taken to be the same across all four detectors; and the case of partially characterized detector imperfections is not addressed. In contrast, this work allows for distinct and partially characterized imperfections in each of the four detectors, is carried out in the finite-size regime, and supports detectors with memory effects as well as protocols involving on-the-fly announcements.

This work is organized as follows. In~\cref{sec:active vs passive}, we describe the major differences and challenges between active and passive detection setups when estimating phase error rate. In \cref{sec:basicprotdescription} we describe the QKD protocol we analyze in this work, which is suitably modified to incorporate decoy-states and memory effects in later sections. In~\cref{sec:subspace}, we describe an extension to compute the minimum multi-click probability for multi-photon input signals on Bob's side.
In~\cref{sec:phase error estimation}, we present a framework to bound the phase error rate in the decoy-state BB84 protocol with memoryless detector imperfections. In~\cref{sec:correlated detectors}, we extend this framework to the case of correlated imperfect detectors.
In~\cref{sec:results}, we present key rate calculations for both analyses. In~\cref{sec:extension}, we extend the memoryless detector analysis to incorporate imperfect sources. We also generalize the correlated detector analysis to include on-the-fly announcements and provide additional comments on the multi-mode detector model and side-channel attacks in~\cref{sec:extension}.

 \section{phase error Estimation in Active vs Passive Detection Setups}
\label{sec:active vs passive}

We will first highlight the difference between active and passive setups when using EUR-based or PEC-based proofs, which will help focus the rest of our discussion. These proof techniques proceed by reducing the task of proving security to that of upper bounding the ``phase error rate''. This can intuitively be understood as follows (see \cref{appendix:security} for a more formal discussion in the EUR framework): Consider a BB84 protocol where Alice generates her key string from rounds in which she measures in the $Z$-basis. Then, the phase errors are the errors between her bit string if she had instead measured her key generation rounds in the $X$-basis, and the string obtained by Bob performing \emph{some} measurement on his quantum systems. Importantly, Bob has some freedom in deciding this measurement and how exactly to map his measurement outcomes to a bit string that approximates Alice's $X$-basis string. A lower phase error rate leads to higher secrecy of the $Z$-basis string.

For active protocols, Bob actively chooses between the $X$ and $Z$ bases. Thus, it is natural for Bob to use his $X$-basis measurements outcomes \footnote{Here, there is some ambiguity in how to map double-clicks to bits. A common choice is to assign them randomly to $0$ or $1$.} to estimate Alice's (virtual) $X$-basis string. For an active detection setup, this estimation follows from standard sampling arguments \footnote{This analysis becomes slightly more involved in the presence of imperfect detectors \cite{devEUR}, but the underlying intuition remains the same.} such as Serfling's inequality \cite{serfling_probability_1974}, which relies on the fact that the choice of whether a round is used for testing ($X$-basis) or key generation ($Z$-basis) is made randomly and is trusted.

On the other hand, in passive protocols, Bob \emph{passively} chooses between the $X$ and $Z$ bases using a beam splitter (see \cref{diag:passivedetection}). Thus, the probability that a round is used for testing depends on the number of photons entering Bob's measurement setup. For example, in a perfect passive detection setup with beam splitting ratio $s$ to the $X$-basis and a signal containing $\pnm$ photons, the probability of obtaining a test ($X$-basis) round, conditioned on single-click and double-click outcomes, is  
    $p = \frac{s^{\pnm}}{s^{\pnm} + (1 - s)^{\pnm}}$.  
This means that an adversary can increase or decrease the testing fraction (depending on whether $s > \frac{1}{2}$ or not), and can vary it across different rounds. This makes it difficult to use the same sampling arguments as in the active case for estimating the phase error rate, since the probability of testing depends on the photon number (which is under Eve's control). This is the primary difficulty in bounding the phase error rate in passive detection setups, and the bulk of this work is devoted to addressing this sampling issue.

To circumvent this problem, we borrow intuition from the flag-state squasher technique \cite{fss} used in other 
security proofs~\cite{nicky,lars}, which uses multiple detector click 
outcomes to upper bound the probability of a multi-photon signal entering Bob's lab. Our approach first uses the cross-clicks and double-clicks to estimate photon numbers of pulses entering Bob's detection setup. Therefore, cross-click outcomes cannot be discarded in our approach. This, in turn, allows us to bound the phase error rate within the single-photon rounds, where the analysis becomes more tractable. (Note that in this work, we use the phrase `single-photon' or `multi-photon' pulse to refer to the number of photons entering Bob's detection setup, and not the number of photons leaving Alice's lab. When we wish to discuss the latter, we explicitly mention Alice.) 

Using standard properties of smooth min-entropy, the entropy of all key rounds can be lower bounded by the entropy of single-photon key rounds. Thus, we can apply the EUR statement along with the upper bound on the phase error rate in the single-photon rounds to prove security (see~\cref{appendix:security}). Thus,  we require the following two bounds to prove security:
\begin{enumerate}
    \item An upper bound on the phase error rate in the single-photon rounds;
    \item A lower bound on the number of key rounds in the single-photon rounds.
\end{enumerate}

We emphasize that although our key rate computations are carried out within the EUR framework, our sampling arguments can readily be applied to the PEC approach as well by assuming that all key generation rounds outside the single-photon rounds result in phase errors. 
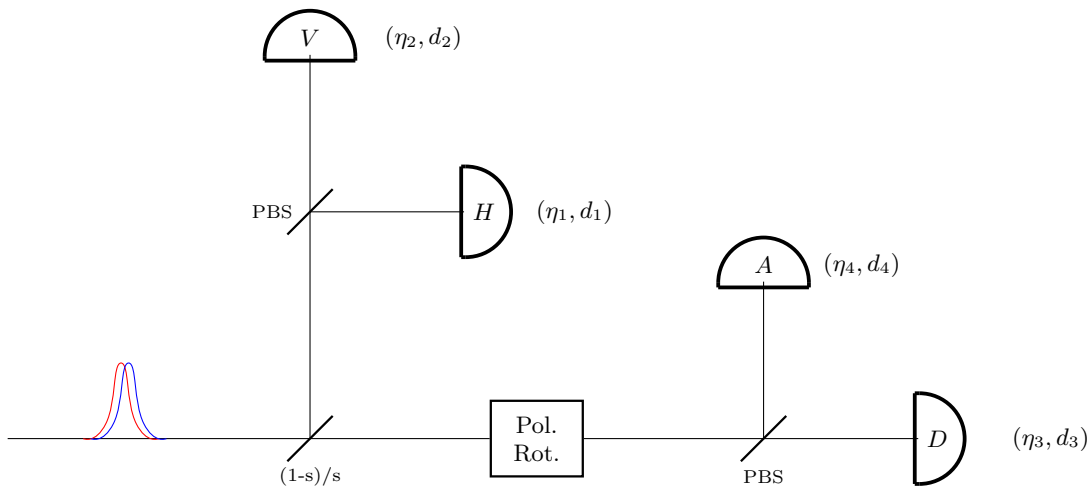
\begin{figure}
    \centering
    \scalebox{1}{\begin{tikzpicture}
\tikzset{
  pics/detector/.style args={#1,#2,#3}{
    code={
      \draw[line width=0.5mm] (0,-1) -- (0,1);
      \draw[line width=0.5mm] (0,1) -- ++(0.1,0);
      \draw[line width=0.5mm] (0,-1) -- ++(0.1,0);
      \draw[line width=0.5mm] (0.1,1) arc (90:-90:1);
      \node[#3, align=center] (#1) at (0.5,0) {#2};
    }
  }
}
\tikzset{
  pics/beamsplitter/.style args={#1,#2,#3}{
    code={
    \coordinate (A) at (0.3,0.3);
    \coordinate (B) at (-0.3,-0.3);
    \coordinate (#1) at (0,0);
    \draw[line width=0.3mm] (A) -- (B);
    \node[align=center,font=\fontsize{7}{7}\selectfont] at #3 {#2};
    }
  }
}
\tikzset{
  pics/pulse/.style args = {#1}{
    code ={
    \draw[#1] plot[smooth,tension=1] coordinates {(0,0) (0.3,0.2) (0.5,1) (0.7,0.2) (1,0)};
    }
  }
}
\tikzstyle{process} = [rectangle, line width=0.3mm, minimum width=1cm, minimum height=1cm, text centered, text width=1cm, draw=black]

\pic{beamsplitter={5050,(1-s)/s,(0,-0.5)}};
\pic[above of = 5050, yshift = 2cm]{beamsplitter={HV,PBS,(-0.5,0)}};
\pic [right of = HV,xshift = 1cm,scale = 0.6]{detector={H, $H$, black}};
\pic [above of = HV,yshift = 1cm,scale = 0.6,rotate=90]{detector={V, $V$, black}};
\draw (-4,0) -- (5050);
\draw (5050) -- (HV);
\draw (HV) -- ([xshift = -5cm]H);
\draw (HV) -- ([yshift = -5cm]V);

\node[process,right of = 5050,xshift = 2cm](PolRot){Pol.\\Rot.};
\pic[right of = PolRot, xshift = 2cm]{beamsplitter={AD,PBS,(0,-0.5)}};
\pic[above of = AD, yshift = 1cm, ,scale = 0.6, rotate = 90]{detector = {A,$A$, black}};
\pic[right of = AD, xshift = 1cm, scale = 0.6]{detector = {D,$D$, black}};
\draw (5050) -- (PolRot);
\draw (PolRot) -- (AD);
\draw (AD) -- ([xshift = 0cm]A);
\draw (AD) -- ([yshift = 0cm]D);

\pic at (-3,0) {pulse = {red}};
\pic  at (-2.9,0) {pulse = {blue}};

\node[right of = H,xshift=0.2cm]{\((\eta_1,d_1)\)};
\node[right of = V,xshift=0.5cm]{\((\eta_2,d_2)\)};
\node[right of = D,xshift=0.5cm]{\((\eta_3,d_3)\)};
\node[right of = A,xshift=0.3cm]{\((\eta_4,d_4)\)};

\end{tikzpicture}}
    \caption{Passive detection setups on Bob's side. The beam splitter has a splitting ratio $s$. Each of the four detectors has a different efficiency $\eta_i$ and dark count rate $d_i$. The parameters $s$, $\eta_i$, and $d_i$ are only partially characterized: specifically, $\eta_i \in [\eta^l_i, \eta^u_i]$, $d_i \in [d^l_i, d^u_i]$, and the beam splitting ratio, denoted $\pxb$, lies within the interval $[s^* - \theta, s^* + \theta]$. Loss can be modeled as a beam splitter sitting in front of a perfect threshold detector. Dark counts can be modeled as a post-processing map applied to the outcomes obtained without dark counts.}
    \label{diag:passivedetection}
\end{figure}
\section{Subspace estimation for imperfect passive detectors}\label{sec:subspace}
A crucial step in the approach discussed in \cref{sec:active vs passive} is to connect multi-click outcomes (including cross-clicks and double clicks) to the number of multi-photon signals. Similar to the analysis required for the application of the flag-state squasher technique, this proceeds by lower bounding the minimum eigenvalue of Bob's multi-click POVM element in the multi-photon subspace. Following past work \cite{lars}, we refer to this as “subspace estimation”. Ref.~\cite{lars} provides a method to compute this minimum eigenvalue for generic passive detection setups with loss, when the loss values are known \emph{exactly}.

In particular, the analysis \cite[Theorem 3]{lars} first provides a subspace estimation bound for lossless detection setups. Ref.~\cite{lars} extends this result to the case of lossy setups is addressed as follows.
First, an equivalent passive detection setup is constructed, which can be described as a lossy linear optical setup, followed by a lossless passive detection setup. The lossy components are then assumed to be under Eve's control, and the subspace estimation is carried out on the equivalent lossless detection setup.
However, note that the security must also be proved for this equivalent lossless passive detection setup.

Two difficulties might arise while applying this method:
\begin{enumerate}
    \item In general, the new passive detection setup produced by this construction is different from the original one. As a result, we may lose useful structural properties of the original setup, which can cause difficulties if the security proof relies on them. For example, the setup in \cref{diag:passivedetection} has the desirable property that when the input contains exactly one photon, the measurement behaves like an active detection setup. The equivalent setup obtained via the method of Ref.~\cite{lars} typically does not preserve this property.
    \item The method is not straightforwardly applicable when the detector parameters are only partially known, which is the practically relevant case.
\end{enumerate}
To address the above-mentioned issues, we extend the method in Ref.~\cite{lars} to directly perform subspace estimation for lossy passive detection setups. Our method, which applies even in the case when the loss and dark count rates are only partially known, is critical for practical applications. Moreover, our results are not limited to the subspace estimation methods in Ref.~\cite{lars}, and can straightforwardly be used to extend other (protocol-specific) subspace estimation methods \cite{nicky,nahar_imperfect_2023} that were previously applicable only to lossless detection setups.

We now give a brief description of our subspace estimation technique, and refer to \cref{appendix:lmcalculation} for the formal details. We model a generic passive detection setup as some linear optics, followed by threshold detectors with losses and dark counts. \cref{diag:passivedetection} is an example of such a setup. We now introduce some notation to better describe the main steps used in our subspace estimation technique:
\begin{itemize}
    \item \( \Gamma_{\mathrm{mc}}^{(\vec{\eta},\vec{d})} \):  
    Bob's multi-click POVM element corresponding to the original setup with both losses and dark counts \footnote{As we see in \cref{sec:multibound}, we want to compute the minimum eigenvalue for the joint multi-click POVM element. But since we use all multi-click outcomes for the subspace estimation, the corresponding multi-click POVM element has an identity on Alice's system. Therefore, we can focus on Bob's multi-click POVM only. We drop the superscript $B$ for simplicity in this section}.  
    Here \( \vec{\eta} = [\eta_1, \ldots, \eta_k] \) and \( \vec{d} = [d_1, \ldots, d_k] \) are the detector efficiencies and dark count rates of the various threshold detectors used in the detection setup.
    
    \item \( \Gamma_{\mathrm{mc}}^{\vec{\eta}} \):  
    The multi-click POVM element corresponding to the setup with losses only (i.e., without dark counts).
    
    \item \( \Gamma_{\mathrm{mc}}^{\etacom} \):  
    The multi-click POVM element corresponding to the setup without dark counts in which all detectors have the same efficiency  
    \( \etacom := \min\{\eta_1, \ldots, \eta_k\} \).  
    That is, we increase the loss of each detector until they match the most lossy one.
    
    \item \( \Gamma_{\mathrm{mc}}^{\text{perfect}} \):  
    The multi-click POVM element corresponding to the setup with perfect threshold detectors, which have no loss or dark counts.
\end{itemize}
To compute the minimum eigenvalue $\lm(\Pign\Gamma_{\mathrm{mc}}^{(\vec{\eta},\vec{d})}\Pign)$, we follow the following steps:
\begin{enumerate}
    \item We model dark counts as a classical post-processing performed on the outcomes obtained from the setup without dark counts. Specifically, the classical post-processing maps no-click and single-click outcomes to multi-click outcomes with non-zero probability, but it never maps multi-click outcomes to no-click or single-click outcomes. Thus, we have 
    \begin{align}
        \Gamma_{\mathrm{mc}}^{\vec{\eta}}\leq \Gamma_{\mathrm{mc}}^{(\vec{\eta},\vec{d})}.
    \end{align} 
    This is formally proved via \cref{lemma:mcdarkcounts} and \cref{cor:darkcounts}.
    \item Next, we show in \cref{lemma:mcmonoton} and \cref{cor:lmmono} that the multi-click POVM element is non-increasing as a function of the loss in each detector. 
    \begin{align}
        \Gamma_{\mathrm{mc}}^{{\etacom}}\leq\Gamma_{\mathrm{mc}}^{\vec{\eta}}.
    \end{align}
    \item Lastly, we compute the minimum eigenvalue $\lm(\Pign\Gamma_{\mathrm{mc}}^{\etacom}\Pign)$ by expressing it as a function of the common loss $\etacom$ and the minimum eigenvalue $\lm(\Pign\Gamma_{\mathrm{mc}}^{\text{perfect}}\Pign)$. The bound on the minimum eigenvalue of the perfect multi-click POVM element can be computed using Ref.~\cite[Theorem 1]{lars} for generic passive detection setups, or via protocol-dependent methods as in Refs.~\cite{nicky,nahar_imperfect_2023}.
\end{enumerate}
For the case of interest, namely the passive detection setup in \cref{diag:passivedetection} with imperfection parameters satisfying \( \eta_i \in [\eta_i^l,\eta_i^u] \) and \( s \in [s^{*}-\theta,\, s^{*}+\theta] \), the estimation of the multi-photon subspace admits the following closed-form expression:
\begin{align}
\lm(\Pigone\Gamma_{\mathrm{mc}}^{(\vec{\eta},\vec{d})}\Pigone)
    \;\geq\; \etamin^{2}\, 2\bar{s}(1-\bar{s}),
\end{align}
where \( \eta_{\min} \) is the minimum efficiency across all detectors, and 
\( \bar{s} = \tfrac{1}{2} + \max\!\left\{ \lvert s^{*} - \theta - \tfrac{1}{2} \rvert,\; \lvert s^{*} + \theta - \tfrac{1}{2} \rvert \right\} \), 
which corresponds to the value of \( s \) in the allowed range that is farthest from \( \tfrac{1}{2} \).
The proofs of the lemmas and corollaries, together with the detailed calculations, are given in \cref{appendix:lmcalculation}. We make use of this minimum eigenvalue in the proof in \cref{sec:multibound}.
\begin{remark}
    Although the above description mirrors Ref.~\cite{lars} in our choice of multi-click as the outcome used for subspace estimation, our results (\cref{lemma:mcdarkcounts,lemma:mcmonoton}) can be directly applied to other choices of outcomes used for subspace estimation. The only restrictions on the choice of outcomes required to apply the results described in \cref{appendix:lmcalculation} is as follows. Let \( \{m_1,\ldots,m_k\} \) represent a click pattern, where \( m_l = 1 \) if the \( l^{\text{th}} \) detector clicks and \( m_l = 0 \) otherwise.  
    We partition all click patterns into two sets, \( \mathcal{A} \) and \( \mathcal{E} \), and we use the coarse-grain POVM element corresponding to the set $\mathcal{E}$ to perform subspace estimation. For example, $\mathcal{E}$ could be chosen to be the set of cross-click outcomes of detection setup \cref{diag:passivedetection}. The sets $\mathcal{A}$ and $\mathcal{E}$ should satisfy the following equivalent conditions:
\begin{enumerate}
\item For any click pattern \( c \in \mathcal{E} \), changing any 0 (no-click) entry in \( c \) to 1 (click) must produce another pattern \( c' \) that also lies in \( \mathcal{E} \).  
For example, for any cross-click pattern, adding additional clicks leaves the pattern within the set of cross-clicks.

\item For any click pattern \( c \in \mathcal{A} \), changing any 1 (click) entry in \( c \) to 0 (no-click) must produce another pattern \( c' \) that also lies in \( \mathcal{A} \).  
For example, for the set containing no-click, single-click, and double-click patterns (which is the set of all outcomes that are not cross-clicks), removing clicks keeps the pattern within the same set.

\end{enumerate}
Additionally, the application of our results to a full security proof requires the additional computation of the minimum eigenvalue of the chosen outcome consisting of click patterns in $\mathcal{E}$ for the detection setup with no loss or dark counts.

\end{remark}
This extension allows one to perform subspace estimation for generic passive detection setups without requiring any change in the POVM considered in the rest of the security proof. This is a crucial step in the security analysis of protocols with imperfect detectors, not only for the phase error estimation approach to prove security, but also for other proof techniques. In particular, it is critically required to apply Ref.~\cite[Theorem~2]{shlok} with the postselection technique~\cite{postselection}, or with EAT-based proof techniques~\cite{eat}.

\section{Protocol Description} \label{sec:basicprotdescription}
In this section, we describe the main QKD protocol whose phase error rate we bound in this work. We consider the simpler setting where Alice uses a perfect single-photon source and Bob uses a passive detection setup. The decoy-state version of this protocol is a straightforward extension described in \cref{subsec:partialimperfections}. The protocol modifications we require in order to handle memory effects in the detectors are described in \cref{sec:correlated protocol description}.

\begin{description}
    \item [State Preparation]Alice sends a perfect single-photon state in $X$ ($Z$) basis with probability $\pxa$ $(\pza)$. If she chooses $Z$ basis, she sends the state in $\{\ket{0},\ket{1}\}$ with equal probability. If she chooses $X$ basis, she sends the state in $\{\ket{+},\ket{-}\}$ with equal probability. Using the source replacement scheme \cite{sourcerep1,sourcerep2}, Alice's state preparation is equivalent to preparing $\ket{\psi}_{AA'} = \frac{1}{\sqrt{2}}(\ket{00}_{AA'}+\ket{11}_{AA'})$, sending $A^{\prime}$ system to Bob through Eve, and measuring system $A$ using the following POVM:
\begin{align}\label{eq:definePOVMA}
    \begin{split}
    &\POVMAZzero = \pza\kb{0}, \qquad \POVMAZone = \pza\kb{1}, \\
    &\POVMAXzero = \pxa\kb{+}, \qquad \POVMAXone = \pxa\kb{-}.
    \end{split}
\end{align}
    \item [Measurement]Bob performs his measurement using POVM $\{\vec{\Gamma}^B\}$ on the state and records the click pattern of detectors.
    Since Bob uses threshold detectors, Bob's POVM elements are block-diagonal in photon number $\pnm$. Here, we label Bob's POVM as 
\begin{align}
    \begin{split}
    \text{single-click outcome: }&\POVMBZzero,\;\POVMBZone,\;\POVMBXzero,\;\POVMBXone,\\
    \text{multi-click outcome: }&\POVMBmc \text{ (more than 1 detector clicks)},\\
    \text{no-click outcome: }&\POVMBvac.
    \end{split}
\end{align}

    In this work, we use the notation $\{\vec{\Gamma}\} = \{\Gamma_1, \Gamma_2, ...\}$ to denote a POVM. Note that 
   Bob has access to more fine-grained click patterns, such as cross-clicks and double-clicks in each basis. However, these are not used in our proof.

    \item [Repetition] Alice and Bob repeat this process for $n$ rounds.
    \item [Post-processing on Bob's side] Bob applies a post-processing map to the click patterns (each click pattern can take 1 out of 16 values). Specifically, he discards no-click outcomes. Whenever Bob obtains a single-click outcome in the $X$ ($Z$) basis, he labels it as an $X$ ($Z$) basis round. If more than one detector clicks in a round, he records it as a multi-click outcome, and he records the total number of multi-click outcomes $\nmc$. We use the number of multi-click outcomes to bound the number of multi-photon rounds. This part of the analysis is referred to as photon-number estimation (PNE).

     We can now define the following joint POVM elements for Alice and Bob corresponding to measurement outcomes in the $X$ ($Z$) basis, where Alice’s and Bob’s bits either agree ($=$) or disagree ($\neq$):
\begin{align}\label{eq:definejoint}
    \begin{split}
    &\POVMXeq = \POVMAXzero \otimes \POVMBXzero + \POVMAXone \otimes \POVMBXone\\
    &\POVMXneq = \POVMAXzero \otimes \POVMBXone + \POVMAXone \otimes \POVMBXzero\\
    &\POVMZeq = \POVMAZzero \otimes \POVMBZzero + \POVMAZone \otimes \POVMBZone\\
    &\POVMZneq = \POVMAZzero \otimes \POVMBZone + \POVMAZone \otimes \POVMBZzero\\
    &\Gamma_{\mathrm{mc}} = \idd_A\otimes\POVMBmc\\
    &\Gamma_{\mathrm{nc}} = \idd_A\otimes\POVMBvac\\
    &\Gamma_{\text{basis mismatch}}\text{ (basis-mistached POVM element, irrelavant to our proof)}.
    \end{split}
\end{align}
We denote Alice's (Bob's) POVM elements with an appropriate superscript. For the joint POVM elements, we omit the superscript. Note that Alice and Bob's joint POVM elements inherit a block-diagonal structure in photon number $\pnm$, due to the fact that Bob’s POVM is block-diagonal in $\pnm$.
    \item[Classical Announcement and Sifting] Alice and Bob announce the bases they used and perform sifting. They discard basis-mismatched rounds. All $X$-basis rounds are used for testing. They denote the number of $X$-basis rounds as $\nx$ and the number of erroneous outcomes in $X$-basis rounds as $\Nxobs$. A small portion of $Z$-basis rounds is revealed for error correction. The error rate in these rounds is denoted as $\errzobs$. The rest of the $Z$-basis rounds are used for key generation. They denote the number of rounds used in key generation as $\nk$.
    \item[Variable-length Decision] When event $\Omega{(\errzobs,\Nxobs,\nx,\nk,\nmc)}$ occurs, Alice and Bob compute \\$\lec(\errzobs,\Nxobs,\nx,\nk,\nmc)$ (number of bits to be used for one-way error-correction) and decide final key length $l(\errzobs,\Nxobs,\nx,\nk,\nmc)$ (the final key length to be produced). The exact formula for $l(\errzobs,\Nxobs,\nx,\nk,\nmc)$ is in \cref{appendix:security}.
    \item[Error Correction and Error Verification] Alice and Bob implement error-correction using a one-way error-correction protocol with $\lec(\errzobs,\Nxobs,\nx,\nk,\nmc)$ bits of information. They implement error-verification by implementing a common two-universal hash function to log(2/$\epev$) bits.
    \item[Privacy Amplification] Alice and Bob implement a common two-universal hash function to produce a secret key with length $l(\errzobs,\Nxobs,\nx,\nk,\nmc)$.
\end{description}
We use bold symbols $(\errzobsb,\Nxobsb,\nxb,\nkb,\nmcb)$ to distinguish random variables from the actual values \\ $(\errzobs,\Nxobs,\nx,\nk,\nmc)$ they take. 
We use the notation $\Omega{(\nx...)}$ to denote the event where the random variables $(\nxb...)$ take specific values $(\nx...)$.

\section{Phase error estimation with memoryless detector imperfections}\label{sec:phase error estimation}

In this section, we bound the phase error rate for the weak coherent pulse (WCP) \footnote{The source is still assumed to be a perfect phase-randomized WCP source. For compatibility with source imperfections, see~\cref{sec:sourceimp}.
} decoy-state BB84 protocol~\cite{hwang2003,lo2005,XBWang,ma2005,hayashi2014,curty2014,charles}, assuming Bob uses memoryless imperfect detectors with partially characterised imperfections (see \cref{diag:passivedetection}). We start by analyzing the simpler protocol described in \cref{sec:basicprotdescription}. The general proof idea is outlined in~\cref{sec:proofsketch}. The complete proof for the single-photon source protocol is provided in~\cref{sec:reformulatingmeasurements,sec:NumberofErrorBound,sec:multibound,sec:zerobound,sec:combine bounds}. We then extend this analysis to the more practical WCP decoy-state protocol, incorporating partially characterized detector imperfections in~\cref{subsec:partialimperfections}.
\begin{remark}
   While our framework is more general, we consider an explicit model of detectors to compute key rates,  where the efficiency and the dark count rate of each detector are only known in some ranges ($\eta_i \in [\eta^l_i, \eta^u_i]$, $d_i \in [d^l_i, d^u_i]$). For the explicit model, we use the polarization-encoding BB84 protocol as an example. As we will see later in this section, our analysis relies only on the fact that Bob’s POVM has a block-diagonal structure with respect to photon number. Therefore, our analysis can be applied to more general QKD protocols, as long as Bob’s positive operator-valued measure (POVM) is block-diagonal and the relevant parameters (see later in this section) defined by the POVM elements can be computed.
\end{remark}
 
\subsection{Proof Idea}\label{sec:proofsketch}
\begin{figure}
    \centering
    \scalebox{1}{\begin{tikzpicture} [
    dashedbox/.style={draw, rectangle, minimum width=2.5cm, minimum height=1cm, dashed, align=center},
    box/.style={draw, rectangle, minimum width=2cm, minimum height=1cm, align=center},
    arrow/.style={-Stealth},
    discard/.style={->, thick, dashed},
    every node/.style={font=\small}
]
\node(top) {\( \rho_{A^{n}B^{n}E^{n}} \)};
\node[box, below=1cm of top] (qnd) {QND};
\node [ below left=1cm and 0.8cm of qnd](0part){\( \rho_{A^{{\tilde{n}_{(0)}}}B^{\tilde{n}_{(0)}}} \)};
\node [ below  = 1cm of qnd](1part){\( \rho_{A^{{\tilde{n}_{(1)}}}B^{\tilde{n}_{(1)}}E^nC^{\tilde{n}_{(0)}}C^{\tilde{n}_{(>1)}}} \)};
\node [ below right=1cm and 0.8cm of qnd](mpart){\( \rho_{A^{{\tilde{n}_{(>1)}}}B^{\tilde{n}_{(>1)}}} \)};
\node [box, below  = 1cm of 1part](bc){coarse-grained \\measurement\\ for basis choice};
\node [ below =1cm of bc,xshift=-1.5cm](x){\( \rho_{A^{{\tilde{n}_{({X},1)}}}B^{\tilde{n}_{({X},1)}}} \)};
\node [ below =1cm of bc,xshift=1.5cm](k){\( \rho_{A^{{\tilde{n}_{({K},1)}}}B^{\tilde{n}_{({K},1)}}E^nC^n} \)};
\node [ box,below =1cm of k,xshift=1.5cm](key){key};
\node [ box,below =1cm of k,xshift=-1.5cm](virtual){phase-errors};
\node [ box,below =1cm of x,xshift=-1.5cm](test){testing};
\node [ align=center,right = 1cm of bc](discard){discard no-click and\\basis-mismatched rounds};
\node [ align=center, below = 2.3cm of k](EUR){\blue{EUR}\\(\cref{appendix:security})};
\node [align=center, below = 2.3cm of x](sam){\blue{Sampling}\\ (\cref{sec:NumberofErrorBound})};
\node [ align=center, above right= 0.1cm and 0.1cm of mpart](mc){\blue{Estimate using multi-click events} \\ (\cref{sec:multibound})};
\node [ align=center,above left= 0.1cm and 0.1cm of 0part](dc){\blue{Estimate using dark counts}\\(\cref{sec:zerobound})};

\draw[arrow] (top) -- (qnd);
\draw[arrow] (qnd) --  node[right] {\( \black{{\tilde{n}}_{(0)}} \)}(0part);
\draw[arrow] (qnd) -- node[right] {\( \black{{\tilde{n}}_{(1)}} \)}(1part);
\draw[arrow] (qnd) -- node[right] {\( \black{{\tilde{n}}_{(>1)}} \)}(mpart);
\draw[arrow] (1part) -- (bc);

\draw[arrow] (bc) -- node[right] {\( \black{{\tilde{n}}_{({X},1)}} \)}(x);
\draw[arrow] (bc) -- node[right] {\( \black{{\tilde{n}}_{({K},1)}} \)}(k);

\draw[arrow] (k) -- (key);
\draw[dashed,arrow] (k) -- (virtual);
\draw[arrow] (x) -- (test);
\draw[arrow] (bc) -- (discard);
\draw[dashed,blue] (key) -- (EUR);
\draw[dashed,blue] (virtual) -- (EUR);
\draw[dashed,blue] (virtual) -- (sam);
\draw[dashed,blue] (test) -- (sam);

\draw[dashed,blue] (dc) -- (0part);
\draw[dashed,blue] (mpart) -- (mc);

\end{tikzpicture}}
    \caption{{Proof structure:} The tilde notation ($\tilde{n}$) on values or random variables such as $\tilde{n}_{{(X,1)}}$ indicates that these quantities are not directly accessible to Alice and Bob in the actual protocol.}
    \label{diag:proofstructure}
\end{figure}
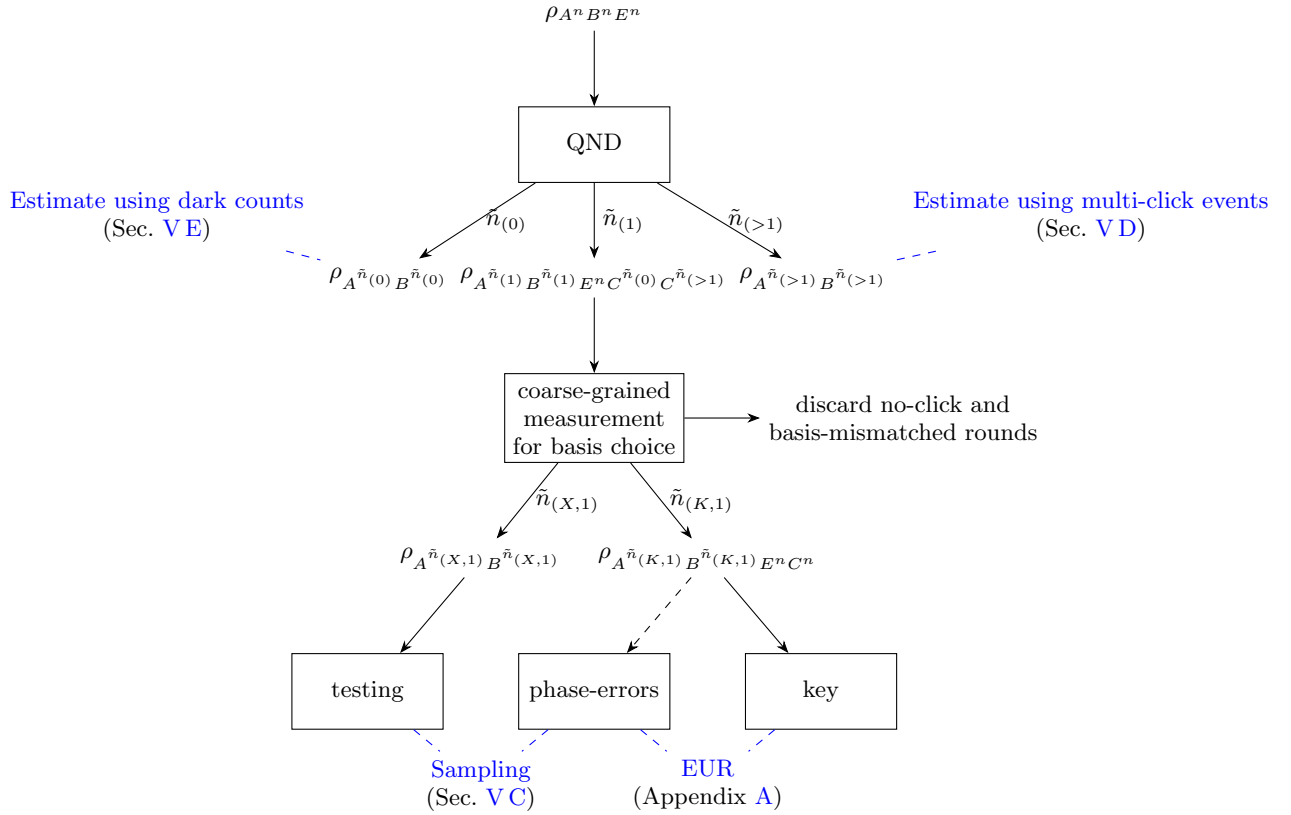

As we mentioned in \cref{sec:active vs passive}, single-photon signals behave ``nicely'' in passive detection setups. 
In security proofs based on the EUR statement together with leftover hashing lemmas or the phase error correction approach, estimating the phase error rate is the main task for computing key rates. In this work, we choose to prove security via the leftover hashing lemma~\cite[Proposition 9]{tomamichel2017}, which requires a lower bound on the smooth min-entropy of $Z$ measurements on the key generation rounds (as these form the raw key that undergoes privacy amplification). Our approach proceeds in the following way:

\begin{enumerate}
    \item Since Bob’s POVM elements are block-diagonal in the photon number \( \pnm \), his measurement is equivalent to first performing a quantum non-demolition (QND) measurement to determine whether the incoming signal is  0-, 1-, or multi-photon round, followed by the original measurement within each subspace (see \cref{diag:proofstructure}). Conditioning on the QND measurement outcome, the key string from all key rounds is partitioned into three substrings, corresponding to different photon-number.

    \item By exploiting the properties of the smooth min-entropy~\cite[Lemma 6.7]{tomamichelbook}, we can reduce the problem of bounding the entropy of the entire key string to that of the substring corresponding to the single-photon. 

    \item We then apply the EUR statement to this substring, reducing the task of lower bounding the smooth min-entropy to that of \begin{itemize}
        \item lower bounding the number of single-photon key generation rounds, and
        \item upper bounding the phase error rate in the single-photon key generation rounds.
    \end{itemize} 
\end{enumerate}

The mathematical details of this procedure are provided in~\cref{appendix:security}. The remainder of this section is devoted to obtaining the two key ingredients for bounding the smooth min-entropy: an upper bound on the phase error rate in the single-photon rounds and a lower bound on the number of single-photon key rounds.

To bound the phase error rate, let us first develop some intuition from the case of single-photon signals arriving at Bob's passive detection setups. Given a beam splitter with splitting ratio \( s \), and an incoming signal containing exactly one photon, the probability that the photon is directed to the \( X \)- or \( Z \)-basis measurement setup is precisely \( s \) or \( (1-s) \), respectively. Thus, conditioned on the input being a single-photon signal, the beam splitter effectively performs an active basis choice for Bob. We use this intuition to obtain the bound on the phase error rate in the single-photon rounds. (Note that while this intuition is helpful, our proof does not rely on any formal reduction to the active basis scenario, since we wish to also handle basis-efficiency mismatch and imperfections in the detectors.)

Following this intuition, we can decompose the sampling to the following (see \cref{diag:proofstructure}): we first apply a QND measurement to separate the incoming signals by photon number: 0, 1, and \( > \)1. For the single-photon signals, the beam splitter is equivalent to an active basis choice, and Bob's measurement can be interpreted as first determining the basis, and then measuring the outcome in that basis. For the rounds in which both Alice and Bob choose the \( Z \)-basis (denoted by the state \( \rho_{A^{\tilde{n}_{(K,1)}} B^{\tilde{n}_{(K,1)}} C^nE^n} \)), Alice performs a \( Z \)-basis measurement to generate part of her raw key in the single-photon subspace. The virtual \( X \)-basis measurements by both Alice and Bob on this state  \( \rho_{A^{\tilde{n}_{(K,1)}} B^{\tilde{n}_{(K,1)}} C^nE^n} \) yield the phase error rate we aim to bound.

To lower bound the number of single-photon key rounds, we equivalently upper bound the number of key rounds arising from 0-photon and multi-photon signals. In particular, for input signals containing 0 photons, any detection outcomes must arise from dark counts. We use this information to lower bound the number of key rounds from 0-photon signals. For multi-photon signals, we upper bound the number of these rounds using multi-click outcomes~\cite{fss,lars,nicky}.
\begin{remark}
    We remove key rounds from 0-photon signals when bounding the smooth min-entropy for better key rate performance. Note that one can instead keep key rounds from 0-photon signals; In this case, one has to upper bound the phase error rate in the $(\leq1)$-photon rounds.
\end{remark}

For our formal proof, we generalize the single-photon subspace to ($1\leq\pnm\leq\pn$)-photon subspace, where $\pn$ is the photon-number cut-off. Therefore, the goal of this section is as follows:  
we aim to upper bound $\errks$, the phase error rate within the $(1 \leq \pnm \leq \pn)$-photon rounds, in terms of observable quantities, and to lower bound the number of key rounds from  $(1 \leq \pnm \leq \pn)$-photon signals $\nkoneb$:
\begin{align}\label{eq:phaseerroridea}
    &\Pr\biggl(\errks \geq \bimp(\nxb,\nkb,\Nxobsb,\nmcb) \ \lor\  \nkoneb \leq \K(\nkb,\nmcb)\biggr) \leq \epsilon^2,
\end{align}
where $\errks := \frac{\Nkoneb}{\nkoneb}\notag$, and 
where $a$, $\delta$, and $q_Z$ are parameters determined by the detector imperfections (explained later).  Thus, the rest of this subsection (and a bulk of this work), is devoted to obtaining the functions $\bimp$ and $\K$ that satisfy \cref{eq:phaseerroridea}.

We achieve this goal by establishing the following three bounds:
\begin{enumerate}
    \item An upper bound on the number of phase errors, $\Nkoneb$;
    \item An upper bound on the number of key rounds from 0-photon signals;
    \item An upper bound on the number of key rounds from $(>\pn)$-photon signals. 
\end{enumerate}
The second and the third bounds will give a lower bound on the number of key rounds from $(1 \leq \pnm \leq \pn)$-photon signals, $\nkoneb$.

\subsection{Constructing the phase error estimation protocol}\label{sec:reformulatingmeasurements}

We now formally construct a phase error rate estimation protocol that is equivalent to the original protocol. We do so by reformulating the measurements undertaken by Alice and Bob as a multi-step process. The equivalent measurement process is described in \cref{diag:phaseerror}.

\subsubsection{QND measurements}
Since the joint POVM elements are block-diagonal in photon number $\pnm$, we can construct an equivalent protocol where Bob first performs a QND measurement $\{\Pi_0,\Pi_{\pn},\Pign\}$ to determine the photon number on each round. These projectors correspond to 0, $(1\leq\pnm\leq\pn)$, ($>\pn$) photon subspaces, respectively. After the QND measurement, the state is separated into three subspaces. Alice and Bob then measure by $\{\vec{\Gamma}^{(0)}\}$, $\{\vec{\Gamma}^{(\pn)}\}$ or $\{\vec{\Gamma}^{(> {\pn})}\}$, depending on the outcomes of QND measurement, where each $\{\vec{\Gamma}^{(i)}\}$ acts within the corresponding subspace. (For simplicity, we use ($\pn$) denotes the subspace $1\leq \pnm \leq \pn$, and use $\{\vec{\Gamma}^{(\pn)}\}$ to represent $\{\vec{\Gamma}^{(1\leq \pnm \leq \pn)}\}$.)
\begin{remark}
Since the measurements \( \{ \vec{\Gamma}^{(0)} \} \), \( \{ \vec{\Gamma}^{(\pn)} \} \), and \( \{ \vec{\Gamma}^{(> \pn)} \} \) are applied to states with photon number 0, \( 1 \leq \pnm \leq \pn \), and \( > \pn \), respectively, we can, without loss of generality,  truncate the zeros in the subspace they do not live in and treat the corresponding POVM elements as follows: 
\( \Gamma^{(0)} \) can be treated as a finite-dimensional POVM element acting entirely within the 0-photon subspace; 
\( \Gamma^{(\pn)} \) can be treated as a finite-dimensional POVM element acting within the (\( 1 \leq \pnm \leq \pn \))-photon subspace; 
and \( \Gamma^{(> \pn)} \) is an infinite-dimensional POVM element that  lives in the $(>\pn)$-photon subspace.
\end{remark}
\subsubsection{\texorpdfstring{Measurements in the $1\leq m < M$ rounds}{}}
Next, we need the following two-step lemma~\cite[Lemma 1]{devEUR} to help us break one measurement down to multiple steps of measurements.
\begin{restatable}{lemma}{twosteplemma}\cite[Lemma 1]{devEUR} \label{lemma:twostep}
				Let $\{\Gamma_{k} \in \text{Pos}(Q) | k \in \mathcal{A}\}$ be a POVM, and let $\{ \mathcal{A}_i\}_{i \in \mathcal{P}_\mathcal{A}}$ be a partition of $\mathcal{A}$, and let  $\rho \in S_\bullet(Q)$ be a state. The classical register storing the measurement outcomes when $\rho$ is measured using $\{\Gamma_k\}_{k \in \mathcal{A}}$ is given by
		\begin{equation}
			\rho_\text{final} \coloneq	\sum_{k \in \mathcal{A}} \Tr(\Gamma_{k} \rho)  \ketbra{k}.
		\end{equation}
		This measurement procedure is equivalent (in the sense of being the same quantum to classical channel) to the following two-step measurement procedure: First, perform a coarse-grained measurement of $i$, using POVM $\{ \tilde{F}_i \}_{ i \in \tilde{\mathcal{A}}}$, where
		\begin{equation}
			\begin{aligned}
				\tilde{F}_i  &\coloneq \sum_{j \in \mathcal{A}_i} \Gamma_{j}, \quad \quad \text{leading to the post-measurement state} \\
				\rho^\prime_\text{intermediate} &= \sum_{i \in \mathcal{P}_{\mathcal{A}}}  \sqrt{\tilde{F}_i} \rho \sqrt{\tilde{F}_i}^\dagger  \otimes \ketbra{i}.
			\end{aligned}
		\end{equation}
		Upon obtaining outcome $i$ in the first step, measuring using  POVM $\{ G_{k} \}_{ k \in \mathcal{A}_i}$, where
		\begin{equation} 
			\begin{aligned}
				G_{k} &\coloneq \sqrt{\tilde{F}}^+_i \Gamma_{k} \sqrt{\tilde{F}}^+_i + P_{k} \quad \quad \text{leading to the post-measurement classical state} \\
				\rho^\prime_\text{final} &= \sum_{i \in \mathcal{P}_{\mathcal{A}}} \sum_{k \in \mathcal{A}_i}  \Tr( G_{k }\sqrt{\tilde{F}_i} \rho \sqrt{\tilde{F}_i})  \ketbra{k}.
			\end{aligned}
		\end{equation}
		$F^+$ denotes the pseudo-inverse of $F$, and $P_{k}$ are any positive operators satisfying $\sum_{k \in \mathcal{A}_i} P_k = \idd - \Pi_{\tilde{F}_i}$, where $\Pi_{\tilde{F}_i}$ denotes the projector onto the support of $\tilde{F}_i$.
\end{restatable}

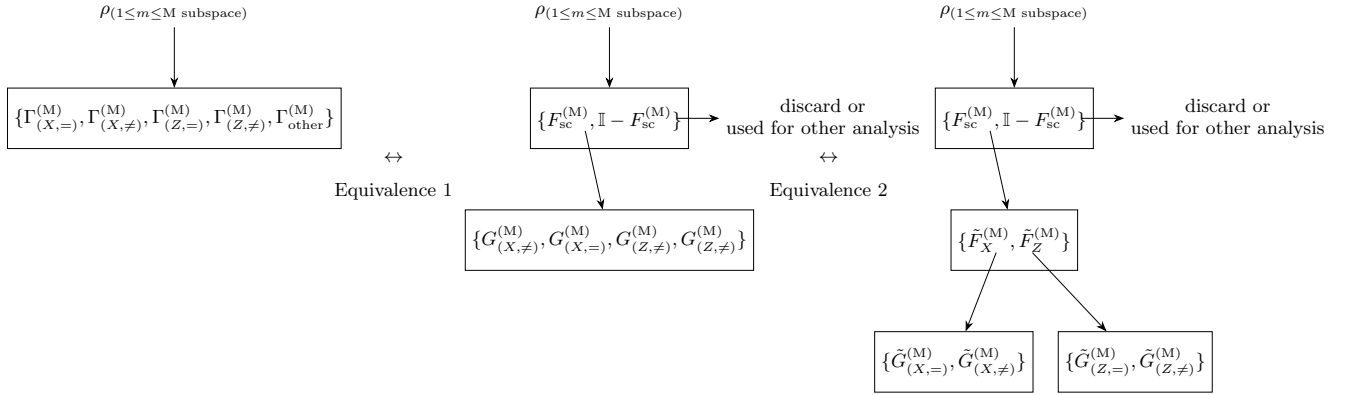
\begin{figure}[h]
    \centering
    \scalebox{0.8}{
\begin{tikzpicture} [
    dashedbox/.style={draw, rectangle, minimum width=2.5cm, minimum height=1cm, dashed, align=center},
    box/.style={draw, rectangle, minimum width=2cm, minimum height=1cm, align=center},
    arrow/.style={-Stealth},
    discard/.style={->, thick, dashed},
    every node/.style={font=\small}
]
\node(top) {\( \rho_{(1\leq \pnm \leq \pn \text{ subspace})} \)};
\node[box, below=1cm of top] (original) {\(\{\POVMXeqM,\POVMXneqM,\POVMZeqM,\POVMZneqM,\Gamma^{(\pn)}_{\text{other}}\}\)};
\node[below right = 2cm and 2cm of top](equal){\(\leftrightarrow\)};
\node[below = 0.1cm of equal]{Equivalence 1};

\node[right=4.5cm of top](rightrho){\( \rho_{(1\leq \pnm \leq \pn\text{ subspace})} \)};
\node[box, below=1cm of rightrho] (F) {\(\{\Fonesc,\idd-\Fonesc\}\)};
\node[right=0.5cm of F,align=center](discard){discard or \\used for other analysis };
\node[box, below=1cm of F](ftx){\(\{\GXneq,\GXeq,\GZneq,\GZneq\}\)};
\draw[arrow] (top) -- (original);
\draw[arrow] (rightrho) --(F);
\draw[arrow] ([xshift = -0.4cm,yshift = 0.3cm]F.south) --(ftx);
\draw[arrow] ([xshift = -0.2cm]F.east) -- (discard);
\node[below right = 2cm and 2cm of rightrho](rightequal){\(\leftrightarrow\)};
\node[below = 0.1cm of rightequal]{Equivalence 2};
\node[right=4cm of rightrho](rightrightrho){\( \rho_{(1\leq \pnm \leq \pn\text{ subspace})} \)};
\node[box, below=1cm of rightrightrho] (rightF) {\(\{\Fonesc,\idd-\Fonesc\}\)};
\node[right=0.5cm of rightF,align=center](rightdiscard){discard or \\used for other analysis };
\node[box, below=1cm of rightF](rightftx){\(\{\ftx,\ftz\}\)};
\node[box, below =1cm of rightftx,xshift=-1cm](rightgx){\(\{\GtXeq,\GtXneq\}\)};
\node[box, below =1cm of rightftx,xshift=2cm](rightgz){\(\{\GtZeq,\GtZneq\}\)};
\node[ below =0.5cm of rightftx,xshift=2.5cm](no){\(\)};

\draw[arrow] (rightrightrho) -- (rightF);
\draw[arrow] ([xshift = -0.4cm,yshift = 0.3cm] rightF.south) --(rightftx);
\draw[arrow] ([xshift = -0.3cm,yshift=0.3cm]rightftx.south) -- (rightgx);
\draw[arrow] ([xshift = 0.3cm,yshift=0.3cm]rightftx.south) --  (rightgz) ;
\draw[arrow] ([xshift = -0.2cm]rightF.east) -- (rightdiscard);

\end{tikzpicture}}
    \caption{Applying \cref{lemma:twostep} to break  measurements into a multi-step process, for the rounds where Bob receives a photon number $m$ such that $1\leq m <M$.}
    \label{diag:phaseerror}
\end{figure}
Let us focus on the measurement $\vec{\Gamma}^{(\pn)}$ in the ($1\leq\pnm\leq\pn$)-photon subspace after the QND measurement.  
 \cref{lemma:twostep} allows us to break down Alice and Bob’s measurement as shown in \cref{diag:phaseerror}.
\begin{description}
    \item[Equivalence 1] In ($1\leq\pnm\leq\pn$)-photon subspace. We use \cref{lemma:twostep} to group the single-click, basis-matched POVM elements together. Thus, Alice and Bob first perform a coarse-grained measurement \( \{ \Fonesc, \idd - \Fonesc \} \) to determine whether the rounds correspond to single-click, basis-matched outcomes.  The operator \( \Fonesc \) is defined as the sum of the single-click, basis-matched POVM elements:
    \begin{align}\label{eq:Fonesc}
        \begin{split}
        &\fx:=\POVMXeqM + \POVMXneqM\\
        &\fz:=\POVMZeqM + \POVMZneqM\\
        &\Fonesc := \fx+\fz\\&\;\phantom{\Fonesc}=\POVMXeqM + \POVMXneqM + \POVMZeqM + \POVMZneqM.
        \end{split}
    \end{align}
If the outcome corresponds to \( \idd - \Fonesc \), they complete the measurement using some POVM.  If the outcome corresponds to \( \Fonesc \), they proceed with a refined measurement given by the POVM \( \{ \GXneq, \GXeq, \GZneq, \GZeq \} \), where
\begin{align}
    \begin{split}
    &G^{(\pn)}_{(X,=)} = \sqrt{\Fonesc}^+{\Gamma}^{(\pn)}_{(X,=)}\sqrt{\Fonesc}^++\frac{a}{a+1}(\idd^{(\pn)}-\Pi_{\Fonesc})\\
     &G^{(\pn)}_{(X,\neq)} = \sqrt{\Fonesc}^+{\Gamma}^{(\pn)}_{(X,\neq)}\sqrt{\Fonesc}^+\\
    &G^{(\pn)}_{(Z,=)} = \sqrt{\Fonesc}^+{\Gamma}^{(\pn)}_{(Z,=)}\sqrt{\Fonesc}^++\frac{1}{a+1}(\idd^{(\pn)}-\Pi_{\Fonesc})\\
     &G^{(\pn)}_{(Z,\neq)} = \sqrt{\Fonesc}^+{\Gamma}^{(\pn)}_{(Z,\neq)}\sqrt{\Fonesc}^+.
    \end{split}
\end{align}
$\Pi_{\Fonesc}$ denotes the projector onto the support of ${\Fonesc}$. $\idd^{(\pn)}$ is the identity on $(1 \leq \pnm \leq \pn)$-photon subspace. $a$ is a non-negative real number we can pick for our convenience.
\begin{remark}
In most of the cases we consider, \( \Fonesc \) has full support on the \( (1 \leq \pnm \leq \pn) \)-photon subspace, and \( \idd^{(\pn)} - \Pi_{\Fonesc} \) is the zero matrix. We keep these terms here for technical completeness, and we choose the prefactors \( \frac{a}{a+1} \) and \( \frac{1}{a+1} \) so that they cancel out later when computing parameters related to these POVM elements. For these reasons, these terms can be ignored in the entirety of this work.
\end{remark}
    \item[Equivalence 2] We apply \cref{lemma:twostep} again to the POVM \( \{ \GXneq, \GXeq, \GZneq, \GZeq \} \). The first step is a coarse-grained measurement \( \{ \ftx, \ftz \} \), where Alice and Bob determine whether the round is measured in the $X$-basis or $Z$-basis:
\begin{align}\label{eq:ftx}
    \begin{split}
\ftx :& = \GXneq + \GXeq\\
& = \sqrt{\Fonesc}^+\left({\Gamma}^{(\pn)}_{(X,=)}+{\Gamma}^{(\pn)}_{(X,\neq)}\right)\sqrt{\Fonesc}^++\frac{a}{a+1}(\idd^{(\pn)}-\Pi_{\Fonesc})\\
 &=  \sqrt{\Fonesc}^+\fx\sqrt{\Fonesc}^++\frac{a}{a+1}(\idd^{(\pn)}-\Pi_{\Fonesc}).
    \end{split}
\end{align}
Similarly, 
\begin{align}\label{eq:ftz}
    \begin{split}
\ftz :=  \sqrt{\Fonesc}^+\fz\sqrt{\Fonesc}^++\frac{1}{a+1}(\idd^{(\pn)}-\Pi_{\Fonesc}).
    \end{split}
\end{align}

Upon obtaining the outcome \( \ftx \), they complete the measurement with the POVM \( \{ \GtXeq, \GtXneq \} \). Similarly, upon obtaining the outcome \( \ftz \), they complete the measurement with the POVM \( \{ \GtZeq, \GtZneq \} \). From \cref{lemma:twostep}, we have
\begin{align}\label{eq:Gtilde}
    \begin{split}
         &\tilde{G}^{(\pn)}_{(X,=)} = \sqrt{\ftx}^+G^{(\pn)}_{(X,=)}\sqrt{\ftx}^++\idd^{(\pn)}-\Pi_{\ftx}\\
          &\tilde{G}^{(\pn)}_{(X,\neq)} = \sqrt{\ftx}^+G^{(\pn)}_{(X,\neq)}\sqrt{\ftx}^+\\
         &\tilde{G}^{(\pn)}_{(Z,=)} = \sqrt{\ftz}^+G^{(\pn)}_{(Z,=)}\sqrt{\ftz}^++\idd^{(\pn)}-\Pi_{\ftz}\\
           &\tilde{G}^{(\pn)}_{(Z,\neq)} = \sqrt{\ftz}^+G^{(\pn)}_{(Z,\neq)}\sqrt{\ftz}^+,
    \end{split}
\end{align}
where $\Pi_{\ftx}$ and $\Pi_{\ftz}$ denote the projectors onto the supports of $\ftx$ and $\ftz$ respectively.
\end{description}

\subsubsection{The full phase error estimation protocol} \label{subsubsec:fullphaseprotocol}
\begin{figure}
    \centering
    \scalebox{0.9}{
\begin{tikzpicture}
    [
    dashedbox/.style={draw, rectangle, minimum width=2.5cm, minimum height=1cm, dashed, align=center},
    box/.style={draw, rectangle, minimum width=2cm, minimum height=1cm, align=center},
    arrow/.style={-Stealth},
    discard/.style={->, thick, dashed},
    every node/.style={font=\small}
]
\node(top) {\( \rho_{A^{{n}}B^{{n}}E} \)};
\node[box, below=1cm of top] (qnd) {\(\{\red{\Pi_0},\blue{\Pi_{\pn}},\brown{\Pi_{>\pn}}\}\)};
\node[box, below left=3cm and 2cm of qnd] (0part){\(\{\vec{\Gamma}^{(0)}\} \)};
\node[left =0.2cm of 0part](n0part){\(\nkzero\)};
\node[box, below=3cm of qnd] (discard) {\(\{\blue{F^{(\pn)}_{\mathrm{sc}}},\Gamma^{(\pn)}_{\mathrm{mc}},\Gamma^{(\pn)}_{\mathrm{nc}},\Gamma^{(\pn)}_{\text{basis mismatch}}\} \)};
\node[box, below right=1cm and 1cm of qnd] (mc) {\(\{\Gamma^{(>\pn)}_{\mathrm{mc}},\red{\mathbb{I}-\Gamma^{(>\pn)}_{\mathrm{mc}}}\}\)};
\node[above =0.15cm of mc,xshift=0.5cm](nmc){\(\nmctb\)};
\node[box, right=1cm of mc](mccomplete){\(\{\vec{G}^{(>\pn)}\}\)};
\node[box, below=2cm of discard, xshift=-0.7cm] (f) {\( \{\red{\ftx}, \blue{\ftz} \}\)};
\node[box, below left=2cm and 0.5cm of f] (ca1) {\(\{ \GtXeq, \GtXneq\} \)};
\node[dashedbox, below=4cm of f] (ca2) {\(\{ \GtXeq, \GtXneq\} \)};
\node[box, below right=4cm and 2cm of f] (ca3){\(\{ \GtZeq, \GtZneq\} \)};
\node[below right=3cm and 2cm of f] (middle) {\(\)};
\node[below right=0.8cm and 0.2cm of f] (nk1) {\(\;\;\nkoneb\)};
\node[below left=0.8cm and 0.2cm of f] (nx1) {\(\nxoneb\;\;\)};
\node[below =0.8cm of discard,xshift=-0.4cm] (nx1) {\(\;\;\;\;\nonet\)};
\node[below =0.4cm of ca2](phase){\(\Nkoneb\)};
\node[left =0.4cm of ca1](test){\(\Nxoneb\)};

\node at (6.5,-0.8)(eqmulti){\(\text{Pr}\biggl(\ngNb > V(\nmctb,\eppntb,\lm)\biggr)\leq \eppntb^2\)};
\node at (-5,-7.5)(eqzero){\(\text{Pr}\biggl(\nkzero \geq n(q_Z+\gmzero(n,\epzero))\biggr) \leq \epzero^2\)};
\node at (1,-16)(eqerror){\(\text{Pr}\biggl(\Nkoneb \geq \frac{\Nxoneb+\sqrt{-2\text{log}(\epazua^2)(\nxoneb+\nkoneb)}+\delta(\nxoneb+\nkoneb)}{a}+\sqrt{-2\text{log}(\epazub^2)(\nxoneb+\nkoneb)}\biggr)\leq \epazua^2 + \epazub^2\)};

\node[below =0.1cm of eqmulti]{(upper bound rounds in ($>\pn$)-subspace)};
\node[below =0.1cm of eqzero]{(upper bound key rounds in 0-subspace)};
\node[below =0.1cm of eqerror]{(upperbound number of phase-errors)};

\node at (8,-2.2){\circletext{1}};
\node at (8,-4.8){\circletext{2}};
\node at (8,-6.8){\circletext{3}};
\node at (8,-12.1){\circletext{4}};
\node at (8,-15.3){\circletext{5}};
\node[ above=0.5cm of ca3,xshift=0.5cm,yshift=-0.2cm](EUR){\textcolor{blue}{EUR}(\cref{appendix:security})};

\draw[arrow] (top) -- (qnd);
\draw[arrow,red] (qnd) --  node[left] {\( \black{\tilde{n}_{(0)}} \)} (0part);
\draw[arrow,blue] (qnd) -- node[right] {\( \black{\tilde{n}_{(\pn)}} \)} (discard);
\draw[arrow,brown] (qnd) -- node[right] {\( \black{\ngN} \)} (mc);
\draw[arrow,blue] ([xshift=-1.7cm,yshift=0.3cm] discard.south) -- (f.north);
\draw[arrow,red] ([xshift=-0.2cm]mc.east) -- (mccomplete);
\draw[arrow] ([xshift=-0.7cm,yshift=-0.3cm] mc.north) -- (nmc);
\draw[arrow] ([xshift=0.3cm] 0part.west) -- (n0part);
\draw[arrow,red] (f) -- (ca1);
\draw[dashed,->] (middle) -- (ca2);
\draw[dashed,->] (ca2) -- (phase);
\draw[arrow] (ca1) -- (test);
\draw[dashed] (-5,-2.5) -- (8,-2.5);
\draw[dashed] (-5,-5) -- (8,-5);
\draw[dashed] (-5,-7) -- (8,-7);
\draw[dashed] (-5,-12.3) -- (8,-12.3);
\draw[dashed] (-5,-15.5) -- (8,-15.5);

\draw[arrow,blue] (f) -- (ca3);
\draw[dashed,darkgreen] (phase.west) to[out=180,in=-90] node[left]{\cref{lemma:azuma}}(test.south);
\draw[dashed,darkgreen] (n0part.west) to[out=180,in=180] node{\cref{lemma:smallPOVM}}(eqzero.west);
\draw[dashed,darkgreen] (nmc.west) to[out=190,in=180] node[left,xshift=0.2cm]{\cref{lemma:orderingonPOVMs}}(eqmulti.west);

\end{tikzpicture}}
    \caption{The phase error rate estimation protocol described in \cref{subsubsec:fullphaseprotocol}: The colors indicate which POVM is used to complete the measurement given the outcome of the previous step. A tilde \( \nt \) denotes variables that cannot be directly observed. The superscript \( (0, \pn, >\pn) \) indicates that the POVM element belongs to the corresponding photon number subspace.}
    \label{diag:equivalent}
\end{figure}
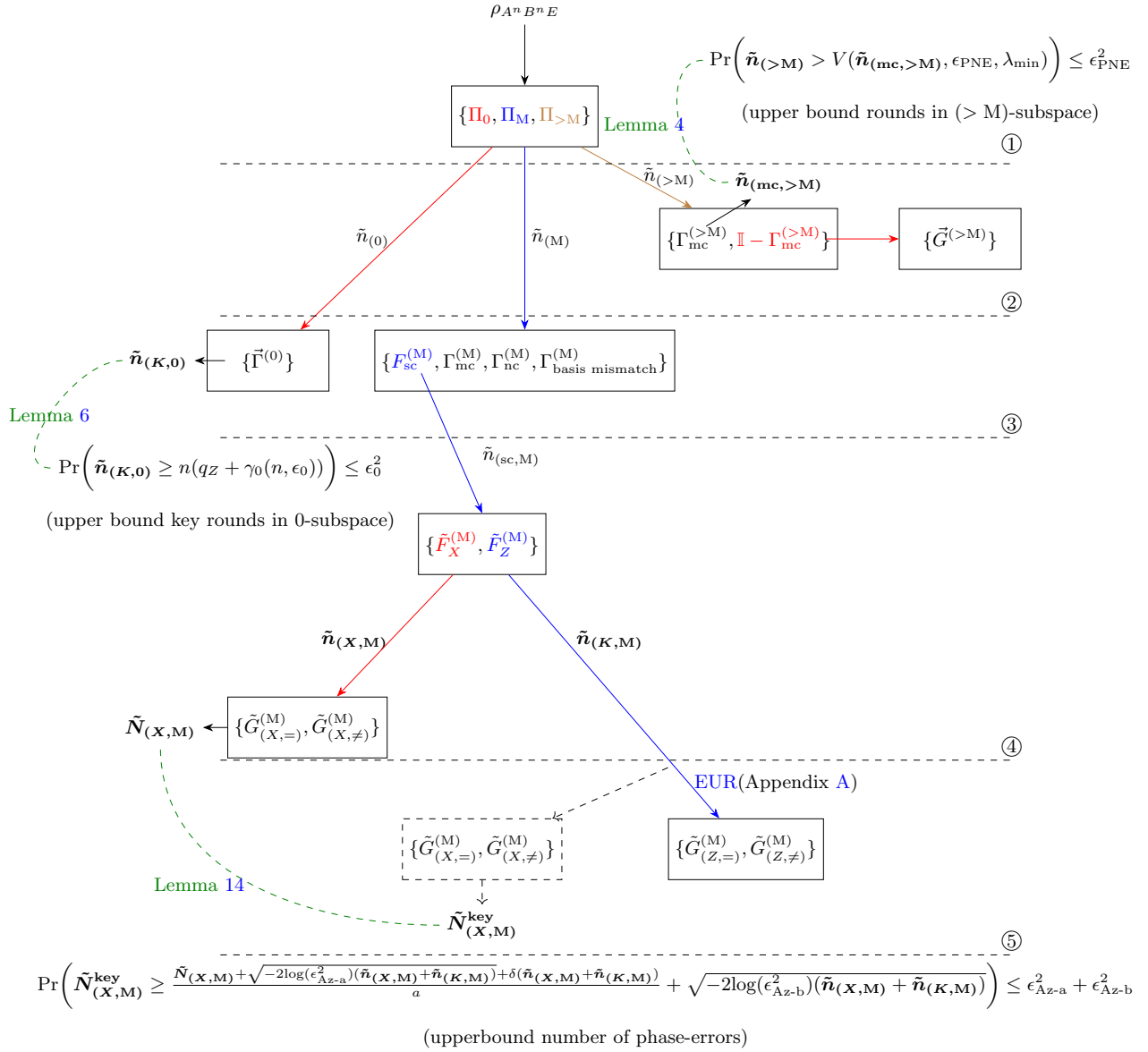
We can now use the analysis in the previous subsections to obtain the full phase error estimation protocol. Alice and Bob perform measurements in the following steps as described in \cref{diag:equivalent}. We state the sampling statements required by our proof at each step and prove them later in \cref{sec:NumberofErrorBound,sec:zerobound,sec:multibound}.

\begin{enumerate}[label=\circletext{\arabic*}, leftmargin=*]
    \item They first perform QND measurement with $\{\Pi_0,\Pi_{\pn},\Pign\}$. $\Pi_{\pn}$ is the projector onto ($1\leq\pnm\leq\pn$)-photon subspace. $\Pi_{0}$ is the projector on $0$-subspace. $\Pi_{>\pn}$ is the projector on $(>\pn)$-subspace.
    \item Upon the result $\Pign$ at stage \circletext{1} in \cref{diag:equivalent}, they complete the multi-click measurement $\{\Gamma^{(>\pn)}_{\mathrm{mc}},\idd-\Gamma^{(>\pn)}_{\mathrm{mc}}\}$. Denote the number of $\Gamma^{(>\pn)}_{\mathrm{mc}}$ outcomes as $\nmctb$, the number of multi-click outcomes from $(>\pn)$-photon signals. This gives an upper bound on  $\ngNb$ (number of multi-photon signals) in terms of $\nmctb$ (number of multi-click outcomes from $(>\pn)$-photon signals):
    \begin{align}\label{step2multibound}
        &\text{Pr}\biggl(\ngNb >V(\nmctb,\eppntb,\lm)\biggr)\leq \eppntb^2,
    \end{align}
    where 
    \begin{align}\label{step2multiboundsub}
        V(\nmctb,\eppntb,\lm) := \frac{-2\log(\eppntb)}{4{\lm}^2}+\frac{\nmctb}{\lm}+\frac{\sqrt{(-\log(\eppntb))(-\log(\eppntb)+4\lm\nmctb)}}{2{\lm}^2}.
    \end{align}
    The proof of this bound and the definition of $\lm$ are given in \cref{sec:multibound}.
    They then complete all other measurements on $(>\pn)$-photon signals using some POVM $\{\vec{G}^{(>\pn)}\}$. The specific structure of $\{\vec{G}^{(>\pn)}\}$ is not relevant to the present proof.

    \begin{remark}
         In this proof, we use a tilde mark ($\bd{\tilde{n}}$) to denote random variables that \textbf{cannot} be directly observed by Alice and Bob during the protocol. For example, $\nmctb$ is the number of multi-click outcomes from $(>\pn)$-photon signals, which cannot be observed directly in our protocol. What Bob can actually measure in the protocol is $\nmcb$,  the number of multi-click outcomes from all signals.
    \end{remark}
    \item (a) They complete the measurement for 0-photon signals with $\{\vec{\Gamma}^{(0)}\}$. We denote $\nkzero$ as the number of $Z$-basis clicks (key rounds) from 0-photon signals. $\nkzero$ is upper bounded through dark count rates:
    \begin{align}\label{step3zerobound}
        \text{Pr}\biggl(\nkzero \geq n\biggl(q_Z+\gmzero(n,\epzero)\biggr)\biggr) \leq \epzero^2,
    \end{align}
    where $q_Z$ is a fixed number calculated from dark count rates. $\gmzero(n,\epzero)$ is a finite-size correction term from Hoeffding's inequality. The proof of this bound is given in \cref{sec:zerobound}.
    
    (b) They also perform a coarse-grained measurement
    $\{\Fonesc,\Gonemc,\Gonenc,\Gamma^{(\pn)}_{\text{basis mismatch}}\}$ on ($1\leq\pnm\leq\pn$)-photon subspace. We recall from \cref{eq:Fonesc} that
    \begin{align}
        \begin{split}\label{eq:definefx}
        &\fx^{(\pn)} = \GMXeq + \GMXneq; \;\; \fz^{(\pn)} = \GMZeq + \GMZneq;\\
        & \Fonesc = \fx^{(\pn)} + \fz^{(\pn)} = \GMXeq + \GMXneq+\GMZeq + \GMZneq.
        \end{split}
    \end{align}
    $\Fonesc$ is the POVM element corresponding to single-click, basis-matched outcomes.
    \item They complete the measurements for \( (1 \leq \pnm \leq \pn) \)-photon signals. 
    In particular, they first perform another coarse-grained basis choice measurement \( \{ \ftx, \ftz \} \). Recall from \cref{eq:ftx} and \cref{eq:ftz} that
\begin{align}\label{eq:defineftx}
    \ftx& = {\sqrt{\Fonesc}}^+\fx^{(\pn)}{\sqrt{\Fonesc}}^++\frac{a}{a+1}(\idd^{(\pn)}-\Pi_{\Fonesc}),   \\
    \ftz& = {\sqrt{\Fonesc}}^+\fz^{(\pn)}{\sqrt{\Fonesc}}^++\frac{1}{a+1}(\idd^{(\pn)}-\Pi_{\Fonesc}).
\end{align}
    Then, they complete the $X$-basis measurement with \( \{ \GtXeq, \GtXneq \} \) upon obtaining the result \( \ftx \). We recall from \cref{eq:Gtilde}  that \( \{ \GtXeq, \GtXneq \} \) is the POVM followed from \cref{lemma:twostep} corresponding to $\GMXeq$ and $\GMXneq$ in the original POVM $\{\vec{\Gamma}^{(\pn)}\}$.
    At this point, there are \( \Nxoneb \) number of erroneous outcomes corresponding to \( \GtXneq \) 
    from \( X \)-basis measurements on \( (1 \leq \pnm \leq \pn) \)-photon signals.

    \begin{remark}
         For proving security using EUR statement, we apply EUR on the state with the result $\ftz$ after this stage conditioned on event $\Omega(\errzobs,\Nxobs,\nx,\nk,\nmc,\nkone)$, namely \\${\rho_{A^{\nkone}B^{\nkone}Z_A^{\tilde{n}_{{(K,0)}}}Z_A^{\ngN}E^nC}}_{\condi(\errzobs,\Nxobs,\nx,\nk,\nmc,\nkone)}$. Then, we use the bound on phase error rate in  \( (1 \leq \pnm \leq \pn) \)-photon rounds to bound the smooth min-entropy of key string for these rounds, as shown in \cref{appendix:security}. This argument of reduction on smooth min-entropy is essentially identical to the decoy-state argument in \cite{devEUR,charles}.
    \end{remark}

    \begin{remark}
    Alice's original $X$ and $Z$ measurements have $c_q = 1$, where $c_q$ quantifies the overlap between the measurements in the EUR statement~\cite{EUR1}. However, after applying \cref{lemma:twostep} to reformulate the measurements, it is not immediately obvious that the $c_q$ factor remains 1. Intuitively, since we only reformulate Bob's measurement while keeping Alice's measurements unchanged, the $c_q$ factor should still be 1. We explicitly verify this in \cref{appendix:cq}.
\end{remark}
   
    \item They complete the $Z$-basis measurement for ($1 \leq \pnm \leq \pn$)-photon signals using $\{\GtZeq,\GtZneq\}$ (recall \cref{eq:Gtilde}) on the state ${\rho_{A^{\nkone}B^{\nkone}Z_A^{\tilde{n}_{(K,0)}}Z_A^{\ngN}E^nC}}_{\condi(\errzobs,\Nxobs,\nx,\nk,\nmc,\nkone)}$. The number of phase errors $\Nkoneb$ is defined as the number of $\GtXneq$ outcomes if they complete the measurement using $\{\GtXeq,\GtXneq\}$ on the same state instead. We aim to bound the number of phase errors:
    \begin{align}
        \begin{split}
        &\text{Pr}\biggl(\Nkoneb \geq \frac{\Nxoneb+\sqrt{-2\ln(\epazua^2)(\nxoneb+\nkoneb)}+\delta\times(\nxoneb+\nkoneb)}{a}\\&\;\;\;\;+\sqrt{-2\ln(\epazub^2)(\nxoneb+\nkoneb)}\biggr)\\
        &\leq \epazua^2 + \epazub^2, \label{NumberofErrorBound}
        \end{split}
    \end{align}
    where $\nxoneb(\nkoneb)$ is the number of signals measured by $X$ ($Z$)-basis measurement for ($1 \leq \pnm \leq \pn$)-photon signals. $a$ is a parameter one can pick freely. $\delta$ is a parameter computed from $\ftx,\ftz$ and $a$, which characterises the efficiency mismatch in detectors. The proof is given in \cref{sec:NumberofErrorBound}.
\end{enumerate}
\begin{remark}
In the estimation protocol, without loss of generality, we can view that Alice and Bob complete the measurements at each stage for the relevant rounds before proceeding to the next stage. For example, at stage \circletext{2} in~\cref{diag:equivalent}, they complete the measurements on the rounds with more than $\pn$ photons before moving to stage \circletext{3}. 
Although the measurement outcome \( \nmctb \) is not directly accessible to Alice and Bob, it allows us to discuss the outcomes of the measurement 
\(
\{ \Gamma^{(>\pn)}_{\mathrm{mc}}, \idd - \Gamma^{(>\pn)}_{\mathrm{mc}} \}
\) 
on a subset of \( \ngN \) rounds. This is important because it allows us to perform our statistical analysis on a specific subset of rounds, as will be shown later.
\end{remark}
We restate our goal here: we want to upper bound the phase error rate in the ($1\leq\pnm\leq\pn$)-photon rounds, $\errks={\Nkoneb}/{\nkoneb}$, and lower bound the number of rounds we bound phase errors on, $\nkoneb$.
\cref{NumberofErrorBound} provides an upper bound on $\Nkoneb$, the number of phase errors. \cref{step2multibound} and \cref{step3zerobound} together provide a lower bound on $\nkoneb$, the number of rounds on which phase errors are bounded. The upper bound on $\Nkoneb$ \cref{NumberofErrorBound} and the lower bound on $\nkoneb$ together provide a lower bound on $\errks$. Later in \cref{sec:combine bounds}, we replace varibles $\nkoneb,\; \nxoneb,$ etc. with observed quantities $\nkb,\; \nxb,$ etc..
\subsection{\texorpdfstring{Upper bound on the number of phase errors~(\cref{NumberofErrorBound})}{}}\label{sec:NumberofErrorBound} 
Let us first focus on the middle branch in \cref{diag:equivalent} for $(1 \leq \pnm \leq \pn)$-photon rounds.
We use Azuma's inequality~\cite{azuma} to bound the number of phase errors in the $(1 \leq \pnm \leq \pn)$-photon rounds. We begin by stating the lemma that will be used in the sampling argument.
\begin{restatable}{lemma}{operatorazuma}[Sampling with Azuma's inequality]\label{lemma:operatorazuma}
    Let $\rho_{Q^n} \in S_{\circ}(Q^{\otimes n})$ be an arbitrary state. Let $\{\Gamma_{\text{a}},\Gamma_{\text{b}},...\}$ be a $k$-element POVM. Let $\bd{n_{\textbf{a}}}$($\bd{n_{\textbf{b}}}$) be the number of ${\text{a}}({\text{b}})$-outcomes of the measurement when $\rho_{Q^n}$ is measured by $\{\Gamma_{\text{a}},\Gamma_{\text{b}},...\}$. Let $a,\epazua,\epazub\geq0$. If the following bound holds:
    \begin{align}
        \infnorm{\Gamma_{\text{a}} - a\Gamma_{\text{b}}} \leq \delta,
    \end{align}
    then
    \begin{align}
  \text{Pr}\biggl(\bd{n_{\textbf{b}}} \geq \frac{\bd{n_{\textbf{a}}}+\sqrt{-2\ln(\epazua^2)n}+\delta n}{a}+\sqrt{-2\ln(\epazub^2)n}\biggr)\leq \epazua^2 + \epazub^2  .
    \end{align}
\end{restatable}
The proof of this lemma is in \cref{technical} and constitutes two uses of Azuma's inequality together with a union bound.

To use this lemma, we consider the following state: Let $\nonetb$ be the random variable recording the number of outcomes corresponding to $\Fonesc$ in the measurement $\{\Fonesc,\Gonemc,\;\Gonenc,\;$$\Gamma^{(\pn)}_{\text{basis mismatch}}\}$ on the ($1 \leq \pnm \leq \pn$)-photon rounds, and consider the state conditioned on the event $\Omega(\nonetb=\nonet)$. Alice and Bob will perform $\{\ftx,\ftz\}$ followed by $\{\GtXeq\;,\;\GtXneq\}$ on this state. This can be written as a single POVM measurement on the state at stage \circletext{3} in \cref{diag:equivalent} given by:
\begin{align}
    \{\sqrt{\ftx} \GtXeq \sqrt{\ftx},\;\;\;\sqrt{\ftx} \GtXneq \sqrt{\ftx},\;\;\;\sqrt{\ftz} \GtXeq \sqrt{\ftz},\;\;\;\sqrt{\ftz} \GtXneq \sqrt{\ftz}\}
\end{align}
We are interested in $\Nxoneb$, the number of outcomes corresponding to $\sqrt{\ftx} \GtXneq \sqrt{\ftx}$ (the number of errors in $X$-basis measurement from ($1\leq\pnm\leq\pn$)-photon signals) and $\Nkoneb$, the number of outcomes corresponding to $\sqrt{\ftz} \GtXneq \sqrt{\ftz}$ (phase errors).
In the case where the detectors are fully characterised and all detector parameters ($s,\;\eta_i,\;d_i$) are known exactly (implying that the POVM elements are precisely known), one can explicitly compute \( \delta \) via
\begin{align}\label{deltabound}
    \delta:=\infnormlong{\sqrt{\ftx} \GtXneq \sqrt{\ftx} - a\sqrt{\ftz} \GtXneq \sqrt{\ftz}} ,
\end{align}
where $a$ is a free parameter.  The parameter $\delta$ quantifies the basis-efficiency mismatch between the two bases. (The parameter $a$ can be picked to minimize the value of $\delta$.) When \( \pn = 1 \) and there is no basis-efficiency mismatch, the optimal choice of \( a \) is the ratio between Alice and Bob's effective \( X \)- and \( Z \)-basis choice probabilities, i.e., \( a = \frac{\pxa \pxb}{\pza (\pzb)} \), and in this case \( \delta = 0 \). In this work, we choose the same value \( a = \frac{\pxa \pxb}{\pza (\pzb)} \) even in the presence of detector imperfections.

\begin{remark}
We remind the reader that since the measurements are applied to states with photon number \( 1 \leq \pnm \leq \pn \), we can, without loss of generality, treat the POVM element \( \Gamma^{(\pn)} \) as a finite-dimensional POVM element acting within the (\( 1 \leq \pnm \leq \pn \))-photon subspace. Consequently, the POVM elements \( \ftx \), \( \ftz \), and \( \GtXneq \) live entirely within the (\( 1 \leq \pnm \leq \pn \))-photon subspace, which simplifies the computation of $\delta$. Further, we provide a way to compute $\delta$ for the case $\pn=1$ in \cref{appendix:adelta}. In this method we do not need to use $\GtXneq$: knowledge of \( \ftx \), \( \ftz \) suffices. 
\end{remark}

Applying \cref{lemma:operatorazuma} on these $\nonet$ rounds, and using the POVM elements and $\delta$ from \cref{deltabound}, we obtain:
\begin{align}
   & \text{Pr}\biggl(\Nkoneb \geq \frac{\Nxoneb+\sqrt{-2\ln(\epazua^2)\nonet}+\delta\nonet}{a}+\sqrt{-2\ln(\epazub^2)\nonet}\biggr)_{\condi{(\nonet)}}\leq \epazua^2 + \epazub^2  
\end{align}
We sum over all possible $\nonet$ and replace $ \nxoneb+\nkoneb$ with $\nonetb $ (since they are equal) to obtain:
\begin{align}
    \begin{split}
    &\text{Pr}\biggl(\Nkoneb \geq \frac{\Nxoneb+\sqrt{-2\ln(\epazua^2)(\nxoneb+\nkoneb)}+\delta(\nxoneb+\nkoneb)}{a}\\&\;\;\;+\sqrt{-2\ln(\epazub^2)(\nxoneb+\nkoneb)}\biggr)\\[15pt]
    & = \text{Pr}\biggl(\Nkoneb \geq \frac{\Nxoneb+\sqrt{-2\ln(\epazua^2)\nonetb}+\delta\nonetb}{a}+\sqrt{-2\ln(\epazub^2)\nonetb}\biggr)\\[15pt]
    & = \sum_{\nonet}\text{Pr}\biggl(\Omega({\nonetb=\nonet})\biggr)\\
    &\times \text{Pr}\biggl(\Nkoneb \geq \frac{\Nxoneb+\sqrt{-2\ln(\epazua^2){\nonet}}+\delta{\nonet}}{a}+\sqrt{-2\ln(\epazub^2){\nonet}}\biggr)_{\condi{(\nonet)}}\\[15pt]
    &\leq\sum_{\nonet}\text{Pr}\biggl(\Omega({\nonetb=\nonet})\biggr) \biggl(\epazua^2 + \epazub^2 \biggr) \\
    &=\epazua^2 + \epazub^2 . 
    \end{split}
\end{align}
\begin{remark} \label{remark:partialcharacterizeone}
When imperfections are only partially characterized and lie within known ranges, we can identify the worst-case scenario to obtain an upper bound on \( \delta \) and a lower bound on \( a \), and proceed with the same proof using these worst-case values. This approach is valid because the bound we derived above is monotonically increasing with respect to \( \delta \) and decreasing with respect to \( a \). The details of how \( \delta \) is computed and how \( a \) is chosen in the presence of partially known imperfections are provided in~\cref{subsec:partialimperfections}.
\end{remark}
\begin{remark}
For the case $\pn=1$, the beam splitter is effectively doing an active basis choice. Therefore, one can sample the phase error rate using the Serfling bound \cite{serfling_probability_1974} in the same way as in \cite[Section V]{devEUR} for imperfect active setups. This idea is also used when we consider source imperfections in \cref{sec:sourceimp}. However, the analysis using the Serfling bound does extend straightforwardly to the scenario where we consider memory effects (see \cref{sec:correlated detectors}). In addition, Serfling's bound requires the beam splitting ratio \( s \) to remain fixed across all \( n \) rounds, which may not hold for practical beam splitters. In contrast, Azuma's inequality allows \( s \) to vary across different rounds in the protocol. Therefore, we use the more general Azuma's inequality in our analysis to ensure robustness of the proof and to accommodate larger photon-number cutoffs.
\end{remark}
\begin{remark}
One could use Kato’s inequality~\cite{kato} instead of Azuma’s inequality to prove a similar lemma as \cref{lemma:operatorazuma} and bound the number of phase errors. Kato's inequality takes the following form:
Given $\alpha \in \mathbb{R}, \beta>\abs{\alpha}$,
\begin{align}
    &\text{Pr}\biggl(\sum_{i=1}^n\rvE{\rvy|\Xhis{i-1}} - \sum_{i=1}^n\rvy \geq \biggl(\beta+\alpha(\frac{2\sum_{i=1}^n\rvy}{n}-1)\biggr)\sqrt{n}\biggr)\leq \exp(-\frac{2(\beta^2-\alpha^2)}{(1+\frac{4\alpha}{3\sqrt{n}})^2})\;.
\end{align}
Kato's inequality can outperform Azuma's inequality when a prediction for \( \sum_{i=1}^n \rvy \) (typically the observed number of \( X \)-basis errors) is available, and \( \sum_{i=1}^n \rvy \) is expected to be much smaller than the total number of random variables $n$. Given such a prediction, one can choose parameters \( \alpha \) and \( \beta \) to obtain a tighter bound than Azuma's inequality (see~\cite[Eq.~32]{curras2021tight} for the optimal choice of \( \alpha \) and \( \beta \)). 
However, we do not use Kato's inequality in this work for pedagogical clarity, as it requires an additional lower bound on \( \nxoneb + \nkoneb \).

\end{remark}

We summarize this subsection in the following lemma.
\begin{lemma}\label{lemma:NumberofErrorBound}
 Let $\rho_{A^nB^n} \in S_{\circ}(A^nB^n)$ be arbitrary state representing $n$ rounds of a generic QKD protocol. Suppose the joint POVM for Alice and Bob can be written as follows:
    \begin{align}
        \{\vec{\Gamma}\} = \{\Gamma_{(\text{test},=)},\;\Gamma_{(\text{test},\neq)},\;\Gamma_{(\text{key},=)},\;\Gamma_{(\text{key},\neq)},\;\Gamma_e,\;\Gamma_{\text{other}}\}
    \end{align}
    and the joint POVM inherits a block-diagonal structure from the fact that Bob's POVM has a block-diagonal structure:
    \begin{align}
        \Gamma_i = \bigoplus_{\pnm=0}^{\infty} \Gamma_i^{(\pnm)} .
    \end{align}
       Then, Alice and Bob's measurement can be viewed as the following:
     \begin{enumerate}[label=\circletext{\arabic*}, leftmargin=*]
    \item They first perform QND measurement with projectors on the subspaces $\{\Pi_0,\Pi_{\pn},\Pign\}$. They complete measurements on the 0 and $(>\pn)$ subspaces when convenient.
    \item They then perform a coarse-grained measurement
    $\{F^{(\pn)},1-F^{(\pn)}\}$ on ($1\leq\pnm\leq\pn$) subspace, where
    \begin{align}
        F^{(\pn)}:= \Gamma^{(\pn)}_{(\text{test},=)}+\Gamma^{(\pn)}_{(\text{test},\neq)}+\Gamma^{(\pn)}_{(\text{key},=)}+\Gamma^{(\pn)}_{(\text{key},\neq)}
    \end{align}
    \item They complete the measurements for ($1\leq\pnm\leq\pn$) subspace. They first perform another coarse-grained basis choice measurement $\{\tilde{F}^{(\pn)}_{\text{test}},\tilde{F}^{(\pn)}_{\text{key}}\}$. We denote the number of ($1\leq\pnm\leq\pn$)-photon test rounds (corresponding to $\tilde{F}^{(\pn)}_{\text{test}}$) as $\bd{\tilde{n}_{\textbf{(test,\pn)}}}$. Then, they complete the measurement with $\{\tilde{G}^{(\pn)}_{(\text{test},=)},\tilde{G}^{(\pn)}_{(\text{test},\neq)}\}$ upon the result $\tilde{F}^{(\pn)}_{\text{test}}$. We denote the number of erroneous outcomes corresponding to $\tilde{G}^{(\pn)}_{(\text{test},\neq)}$ as $\bd{\tilde{N}_{\textbf{(test,\pn)}}}$
    \begin{align}
    \tilde{F}^{(\pn)}_{\text{test}} = {\sqrt{F^{(\pn)}}}^+(\Gamma^{(\pn)}_{(\text{test},=)}+\Gamma^{(\pn)}_{(\text{test},\neq)}){\sqrt{F^{(\pn)}}}^+;\;\;\;    \tilde{F}^{(\pn)}_{\text{key}} = {\sqrt{F^{(\pn)}}}^+(\Gamma^{(\pn)}_{(\text{key},=)}+\Gamma^{(\pn)}_{(\text{key},\neq)}){\sqrt{F^{(\pn)}}}^+
    \end{align}
    \item They complete the key measurement for ($1\leq\pnm\leq\pn$) subspace using $\{\tilde{G}^{(\pn)}_{(\text{key},=)},\tilde{G}^{(\pn)}_{(\text{key},\neq)}\}$ upon the result $\tilde{F}^{(\pn)}_{\text{key}}$. We denote the number of ($1\leq\pnm\leq\pn$)-photon key rounds (corresponding to $\tilde{F}^{(\pn)}_{\text{key}}$) as $\bd{\tilde{n}_{\textbf{(key,\pn)}}}$. The number of phase error $\bd{\tilde{N}^{\textbf{ph}}_{\textbf{\pn}}}$ is defined as the number of $\tilde{G}_{(\text{test},\neq)}$ outcomes if they complete the measurement using $\{\tilde{G}^{(\pn)}_{(\text{test},=)},\tilde{G}^{(\pn)}_{(\text{test},\neq)}\}$ instead.
  \begin{align}
            \tilde{G}^{(\pn)}_{(\text{test},\neq)} = \sqrt{\tilde{F}^{(\pn)}_{\text{test}}}^+\sqrt{F_{\mathrm{sc}}^{(\pn)}}^+\Gamma^{(\pn)}_{(\text{test},\neq)}\sqrt{F_{\mathrm{sc}}^{(\pn)}}^+\sqrt{\tilde{F}^{(\pn)}_{\text{test}}}^+
        \end{align}
\end{enumerate}
   For any \( a > 0 \), let \( \delta \) be defined as the following infinity norm (or an upper bound thereof):
\begin{equation}
    \delta \geq \left\lVert 
        \sqrt{\tilde{F}^{(\pn)}_{\text{test}}} \, \tilde{G}^{(\pn)}_{(\text{test},\neq)} \, \sqrt{\tilde{F}^{(\pn)}_{\text{test}}} 
        - a \sqrt{\tilde{F}^{(\pn)}_{\text{key}}} \, \tilde{G}^{(\pn)}_{(\text{test},\neq)} \, \sqrt{\tilde{F}^{(\pn)}_{\text{key}}}
    \right\rVert_{\infty}
\end{equation}
    Given $\epazua>0$ and $\epazub>0$ , the following relation holds:
       \begin{align}
        \begin{split}
        &\text{Pr}\biggl(\bd{\tilde{N}^{\textbf{ph}}_{\textbf{\pn}}} \geq \frac{\bd{\tilde{N}_{\textbf{(test,\pn)}}}+\sqrt{-2\ln(\epazua^2)(\bd{\tilde{n}_{\textbf{(test,\pn)}}}+\bd{\tilde{n}_{\textbf{(key,\pn)}}})}+\delta(\bd{\tilde{n}_{\textbf{(test,\pn)}}}+\bd{\tilde{n}_{\textbf{(key,\pn)}}})}{a}\\&\;\;\;\;+\sqrt{-2\ln(\epazub^2)(\bd{\tilde{n}_{\textbf{(test,\pn)}}}+\bd{\tilde{n}_{\textbf{(key,\pn)}}})}\biggr)\\
        &\leq \epazua^2 + \epazub^2.
        \end{split}
    \end{align}
\begin{proof}
    The proof is described in this subsection by identifying the following
 \begin{align}
    \begin{split}
           &\Gamma_{(X,=)}=\Gamma_{(\text{test},=)};\;\;\Gamma_{(X,\neq)}=\Gamma_{(\text{test},\neq)}\\
           &\Gamma_{(Z,=)}=\Gamma_{(\text{key},=)};\;\;\;\Gamma_{(Z,\neq)}=\Gamma_{(\text{key},\neq)}\\
           &\nxoneb=\bd{\tilde{n}_{\textbf{(test,\pn)}}};\;\; \nkoneb=\bd{\tilde{n}_{\textbf{(key,\pn)}}}\\
           &\Nxoneb = \bd{\tilde{N}_{\textbf{(test,\pn)}}};\;\Nkoneb = \bd{\tilde{N}^{\textbf{ph}}_{\textbf{\pn}}}.
    \end{split}
 \end{align}
\end{proof}
\end{lemma}
\subsection{\texorpdfstring{Upper bound on the number of $(>\pn)$-photon signals (\cref{step2multibound})}{} }\label{sec:multibound}
Next, we bound the number of $(>\pn)$-photon signals, $\ngNb$, on the right branch of \cref{diag:equivalent}.
The main idea is taken from other security proofs~\cite{nicky,var,lars} where the multi-click outcomes are used to bound the number of $(>\pn)$-photon rounds. In our proof, we use the minimum eigenvalue of the multi-click POVM element. 
If we can find the minimum eigenvalue of \( \Pign \Gamma_{\mathrm{mc}} \Pign \) (which we denote using \( \lm(\Pign \Gamma_{\mathrm{mc}} \Pign) \)), then we have:
\begin{align} \label{lm_relation}
    \Pign\Gamma_{\mathrm{mc}}\Pign =\Gamma_{\mathrm{mc}}^{(>\pn)} \geq \lm(\Pign\Gamma_{\mathrm{mc}}\Pign)\;\Pign.
\end{align}

We denote \( \lm(\Pign \Gamma_{\mathrm{mc}} \Pign) \) as \( \lm \) for simplicity. Intuitively, $\lambda_{\min}$ is the minimum probability that Bob obtains a multi-click outcome when the incoming signal contains more than $\pn$ photons.
For a passive detection setup with no loss, 
\( \lm = 1 - (\pxb)^{(\pn+1)} - (\pzb)^{(\pn+1)} \), 
where \( \pxb \) is the splitting ratio of Bob's beam splitter~\cite[Eq.~(54)]{lars}. We provide a method to compute this minimum eigenvalue as we discuss in \cref{sec:subspace} and \cref{appendix:lmcalculation}.
We aim to use \( \lm \) to bound \( \ngNb \), 
the number of signals arriving at Bob with \( (>\pn) \)-photon.
 
Let us consider the state \( \rho_{|\event{\tilde{n}_{(0)},\tilde{n}_{(\pn)},\ngN}} \), 
the state after performing the QND measurement \( \{ \Pi_0, \Pi_{\pn}, \Pign \} \) and condition on the results of QND measurement. In particular, there are $n_{(>\pn)}$ signals with more than $\pn$ photons.
At stage \circletext{2} in \cref{diag:equivalent}, we focus on completing the measurement for those signals with outcome \( \Pign \).
For the rounds with outcome \( \Pign \), we complete the measurement with 
\( \{ \Gamma_{\mathrm{mc}}^{(>\pn)}, \idd - \Gamma_{\mathrm{mc}}^{(>\pn)} \} \), which is in turn followed
followed by some measurement upon obtaining the outcome \( \idd - \Gamma_{\mathrm{mc}}^{(>\pn)} \). 
The \( \{ \Gamma_{\mathrm{mc}}^{(>\pn)}, \idd - \Gamma_{\mathrm{mc}}^{(>\pn)} \} \) measurement produces the number of multi-click outcomes 
from \( (>\pn) \)-photon signals \( \nmctb \).

We use the following lemma to relate the random variable \( \nmctb \), 
obtained by measuring those $\ngN$ rounds  with \( \{ \Gamma_{\mathrm{mc}}^{(>\pn)}, \idd - \Gamma_{\mathrm{mc}}^{(>\pn)} \} \), 
to the number of outcomes corresponding to $\lm\Pign$ \( (\nlambb) \), obtained by measuring the same rounds with 
\( \{ \lm \Pign, \idd - \lm \Pign \} \). The lemma turns an operator inequality on the POVM elements into a statistical claim on the number of corresponding outcomes obtained on measuring a state using the POVM. 

\begin{lemma}\cite[Lemma 8]{devEUR} \label{lemma:orderingonPOVMs}
    Let $\rho_{Q^n} \in S_\circ(Q^{\otimes n})$ be an arbitrary state. Let $\{P,I - P\}$ and $\{P^\prime,I-P^\prime\}$ be two sets of POVM elements such that $P \leq P^\prime $. Then, for any $e$, it is the case that
    \begin{equation}
        \Pr({\num_{P}} \geq e ) \leq \Pr({\num_{P^\prime}} \geq e), \notag
    \end{equation}
    where $\num_{P}(\num_{P^\prime})$ is the number of outcomes corresponding to $P(P^{\prime})$
\end{lemma}

Using \cref{lm_relation} with \cref{lemma:orderingonPOVMs} directly gives:
\begin{align}\label{lemmaresult}
    &\text{Pr}\biggl(\nmctb \geq \ngN\biggl(\lm + \gamma_{\eppntb}(\ngN)\biggr)\biggr)_{\condi{({\tilde{n}}_{(0)},{\tilde{n}}_{(\pn)},\ngN)}} \notag\\\geq&
    \text{Pr}\biggl(\nlambb \geq \ngN\biggl(\lm + \gamma_{\eppntb}(\ngN)\biggr)\biggr)_{\condi{({\tilde{n}}_{(0)},{\tilde{n}}_{(\pn)},\ngN)}} ,
\end{align}
where $\gamma_{\eppntb}(\ngN) = \sqrt{\frac{-\log{(\eppntb)}}{ \ngN}}$. The quantity \( \ngN(\lm + \gamma_{\eppntb}(\ngN)) \) is not immediately useful on its own, since we do not have access to $(\ngN)$ from our measurements.

In order to make it useful, we will first relate \( \nlambb \) and \( \ngN \). 
Note that \( \nlambb \) is the number of rounds corresponding to \( \lm \Pign \) 
if Alice and Bob apply the POVM \( \{ \lm \Pign, \idd - \lm \Pign \} \) 
to the rounds with more than $\pn$ photons. 
Since these rounds have photon-number \( (>\pn) \), measuring them with \( \{ \lm \Pign, \idd -\lm \Pign \} \) 
is equivalent to measuring with \( \{ \lm \idd, (1 -\lm) \idd \} \).
But the distribution of the outcome of $\{\lm\idd,(1-\lm)\idd\}$ is equivalent to Bernoulli sampling.
Therefore,
\begin{align}
    \begin{split}
    \text{Pr}\left(\nlambb <\ngN(\lm - c)\right)_{\condi({\tilde{n}}_{(0)},{\tilde{n}}_{(\pn)},\ngN)}
    \leq& \sum_{i=0}^{\lfloor\ngN(\lm - c)\rfloor}{\ngN\choose i}\lm^i(1-\lm)^{\ngN-i}\label{eq:multiloose} \\
    \leq& \text{ exp}\left(-2\ngN\biggl(\lm-\lm+c\biggr)^2\right)\\
    =&\text{ exp}(-2\ngN\; c^2),
    \end{split}
\end{align}
 where we use Hoeffding's inequality to find the tail bound for the binomial distribution for the second inequality with $c\geq0$. We pick $c = \sqrt{\frac{-\ln{(\eppntb)}}{\ngN}}$ such that $\text{exp}(-2\ngN\; c^2) = \eppntb^2$.
Then,
 \begin{align}
    \begin{split}
    \text{Pr}\biggl(\nlambb <\ngN\biggl(\lm - \sqrt{\frac{-\ln{(\eppntb)}}{\ngN}}\biggr)\biggr)_{\condi({\tilde{n}}_{(0)},{\tilde{n}}_{(\pn)},\ngN)}\leq\eppntb^2.
    \end{split}
\end{align}
\begin{remark}\label{remark:tightbound}
    Using Hoeffding's inequality introduces looseness in the bound. 
    One can set \cref{eq:multiloose} equal to \( \eppntb^2 \) 
    and solve for \( c \) in terms of \( \lm \) and \( \ngN \) 
    to obtain a tighter bound. The tight bound is not analytical, however, and we opt for simplicity over achieving the tightest possible analysis.
\end{remark}

Following from the bound we obtained above, we have
\begin{align}
    &\text{Pr}\biggl(\nlambb\ \geq \lm\ngN + \sqrt{-\log(\eppnt)}\sqrt{\ngN}\biggr)_{\condi({\tilde{n}}_{(0)},{\tilde{n}}_{(\pn)},\ngN)} > 1-\eppntb^2\;.
\end{align}
Combine with \cref{lemmaresult} and sum all possible    $\tilde{n}_{(0)}$, $\tilde{n}_{(\pn)}$ and $\ngN$  we have the following bound:
\begin{align}
    \begin{split}
   &\text{Pr}\biggl(\nmctb \geq \lm\ngNb + \sqrt{-\log(\eppnt)}\sqrt{\ngNb}\biggr) > 1-\eppntb^2\\
   \iff &\text{Pr}\biggl(\nmctb < \lm\ngNb + \sqrt{-\log(\eppnt)}\sqrt{\ngNb}\biggr) \leq \eppntb^2\\
   \implies &\text{Pr}\biggl(\nmctb < \lm\ngNb + \sqrt{-\log(\eppnt)}\sqrt{\ngNb}\;\;\land\;\; \ngNb \geq 0\biggr) \leq \eppntb^2\\
   \iff & \text{Pr}\biggl(\lm(\sqrt{\ngNb})^2-\sqrt{\log(\eppntb)}\sqrt{\ngNb}-\nmctb>0\;\;\land\;\; \ngNb \geq 0\biggr) \leq \eppntb^2\\
   \iff & \text{Pr}\biggl(\ngNb > \frac{-2\log(\eppntb)}{4{\lm}^2}+\frac{\nmctb}{\lm}+\frac{\sqrt{(-\log(\eppntb))(-\log(\eppntb)+4\lm\nmctb)}}{2{\lm}^2}\biggr) \leq \eppntb^2\;.
    \end{split}
\end{align}
The first step comes from changing the direction of the inequality.
The second step uses the fact that \( \Pr(A \land B) \leq \Pr(A) \). 
The third step is simply rearranging terms in the inequality. 
In the final step, we solve for \( \ngNb \). This is the bound $V$ we define in \cref{step2multiboundsub}.
Since \( \nmcb \geq \nmctb \) and the right-hand side of the inequality is monotonically increasing with respect to \( \nmctb \), we can replace \( \nmctb \) by \( \nmcb \). Finally we obtain
\begin{align}\label{step3multibound}
    \text{Pr}\biggl(\ngNb > \frac{\nmcb}{\lm}+\frac{\sqrt{(-\log(\eppntb))(-\log(\eppntb)+4\lm\nmcb)}}{2{\lm}^2}+\frac{-2\log(\eppntb)}{4{\lm}^2}\biggr) \leq \eppntb^2.
\end{align}
\begin{remark} \label{remark:partialcharacterizetwo}
     Note that we can replace $\lm$ with a lower bound on it by the monotonicity of the bound we find. This is necessary when we do not know imperfection precisely. We extend the method of finding $\lm$ in Ref.~\cite{lars} to the partially characterised case. We show the details in \cref{subsec:partialimperfections}.
\end{remark}
We summarize this subsection in the following lemma.
\begin{lemma}\label{lemma:multibound}
 Let $\rho_{A^nB^n} \in S_{\circ}(A^nB^n)$ be arbitrary state representing $n$ rounds of a generic QKD protocol. Suppose the joint POVM for Alice and Bob can be written as follows:
    \begin{align}
        \{\Gamma_e,\;\Gamma_{\text{other}},...\},
    \end{align}
    The joint POVM inherits a block-diagonal structure from the fact that Bob's POVM has a block-diagonal structure:
    \begin{align}
        \Gamma_i = \bigoplus_{\pnm=0}^{\infty} \Gamma_i^{(\pnm)} .
    \end{align}
       Then, Alice and Bob's measurement can be viewed as the following:
     \begin{enumerate}[label=\circletext{\arabic*}, leftmargin=*]
    \item They first perform QND measurement with projectors on the subspaces $\{\Pi_0,\Pi_{\pn},\Pign\}$. They complete measurements on the 0 and $(1\leq\pnm\leq\pn)$ subspaces when convenient. We denote the number of rounds corresponding to $\Pign$ as $\ngNb$.
    \item  Upon the result $\Pign$ at stage \circletext{1}, they then complete the measurement $\{\Gamma^{(>\pn)}_{e},\idd-\Gamma^{(>\pn)}_{e}\}$. Denote the number of $\Gamma^{(>\pn)}_{e}$ outcomes as $\bd{\tilde{n}_e}$, the number of $e$ outcomes from $(>\pn)$ subspace.  
    They then complete all other measurements on $(>\pn)$ subspace using some POVM $\{\vec{G}^{(>\pn)}\}$. We also denote the number of $\Gamma_{e}$ outcomes as  $\bd{{n}_e}$
\end{enumerate}
Let \( \lambda_{\min} \) be the minimum eigenvalue of \( \Gamma_e \) (or a lower bound thereof) in the \( (>\pn) \) subspace:
\begin{equation}
    \lm\leq\lm(\Pign\Gamma_e\Pign).
\end{equation}
Then, the following relation holds:
\begin{align}
    \text{Pr}\biggl(\ngNb > \frac{\bd{n_e}}{\lm}+\frac{\sqrt{(-\log(\eppntb))(-\log(\eppntb)+4\lm\bd{n_e})}}{2{\lm}^2}+\frac{-2\log(\eppntb)}{4{\lm}^2}\biggr) \leq \eppntb^2
\end{align}
for any $\eppntb>0$.
\end{lemma}
\begin{proof}
    The proof is described in this subsection by identifying $\Gamma_{\mathrm{mc}}=\Gamma_e$ and $\nmcb=\bd{n_e}$.
\end{proof}
\subsection{\texorpdfstring{Upper bound on the number of 0-photon signals in key generation rounds~(\cref{step3zerobound})}{}}\label{sec:zerobound}
Finally, we bound the number of 0-photon key rounds, $\nkzero$, on the left branch of \cref{diag:equivalent}.
After stage \circletext{2} in \cref{diag:equivalent}, Alice and Bob have a post-measurement state with classical registers 
\( {\tilde{n}}_{(0)},{\tilde{n}}_{(\pn)},\ngN \) recording the number of rounds corresponding to different photon numbers.
We condition on the event \( \event{{\tilde{n}}_{(0)},{\tilde{n}}_{(\pn)},\ngN} \).
Specifically, there are \(\tilde{n}_{(0)}  \) signals with 0 photons.

At stage \circletext{3} in \cref{diag:equivalent}, Alice and Bob complete the measurement on these \(\tilde{n}_{(0)}  \) rounds with the POVM $\{\vec{\Gamma}^{(0)}\}=$
\( \{ \GzeroZeq,\; \GzeroZneq,\; \GzeroXeq, \;\GzeroXneq, \;\Gzeromc,\; \Gzeronc \} \). 
Define \( \GzeroZ = \GzeroZeq + \GzeroZneq \). 
Using \cref{lemma:twostep} again, we can view this as first measuring with 
\( \{ \GzeroZ, \idd - \GzeroZ \} \) to obtain \( \nkzero \), 
and then completing the measurement with some other POVM.

We use the following lemma, which states that a small POVM element leads to fewer observed outcomes.
\begin{lemma}\cite[Lemma 4]{devEUR}
    \label{lemma:smallPOVM}
    Let $\rho_{Q^n} \in S_\circ(Q^{\otimes n})$ be an arbitrary state. Let $\{P,I - P\}$ be a POVM such that $\norm{P}_\infty \leq \delta$. Then
    \begin{equation*}
        \Pr(\frac{\num_P}{n} \geq \delta+c) \le \sum_{i = n (\delta+c)}^{n} {n \choose i} \delta^ i (1-\delta)^{n-i} ,
    \end{equation*}
    where $\num_P$ is the number of $P$-outcomes when each subsystem of $\rho_{Q^n}$ is measured using POVM $\{P,\idd - P\}$.
\end{lemma} 
Bob's POVM element in the 0-photon subspace is a one-dimensional subspace associated with dark count rates. Thus, for the joint POVM for Alice and Bob, one can compute 
\begin{align}
    \begin{split}\label{eq:defineqz}
   & \infnorm{\POVMZeq^{(0)} + \POVMZneq^{(0)}}=: q_Z,
  \end{split}
\end{align}
where $q_Z$ depends on dark count rates $d_i$. We explicitly compute $q_Z$ in \cref{appendix:qz}. Then, we use \cref{lemma:smallPOVM} to bound the number of key rounds from 0-photon signals.
\begin{align}
    \begin{split}
    \text{Pr}\biggl(\nkzero \geq \tilde{n}_{(0)}  (q_Z+c)\biggr)_{\condi({\tilde{n}}_{(0)},{\tilde{n}}_{(\pn)},\ngN)} &\leq \sum_{i = \tilde{n}_{(0)}  (q_Z+c)}^{\tilde{n}_{(0)}  }{{\tilde{n}_{(0)}}\choose i}{q_Z}^i(1-q_Z)^{\tilde{n}_{(0)}-i}\\
    & \leq \text{exp}(-2n_0\Biggl(1-q_Z-\biggl(1-\frac{\tilde{n}_{(0)}  }{\tilde{n}_{(0)}  }(q_Z+c)\biggr)\Biggr)^2)\\
    & =\text{exp}(-2\tilde{n}_{(0)} c^2)\;.
    \end{split}
\end{align}
For the second inequality, we use Hoeffding's inequality to find the tail bound for the binomial distribution (see \cref{remark:tightbound} for a tighter bound).
Let $c^2 = \frac{n}{\tilde{n}_{(0)} }{\big(\gmzero(n,\epzero)\big)}^2$, where $\gmzero(n,\epzero) = \sqrt{\frac{-\ln(\epzero)}{n}}$. We have

\begin{align}
    \text{Pr}\biggl(\nkzero \geq \tilde{n}_{(0)} q_Z+\sqrt{\tilde{n}_{(0)}  }\sqrt{n}\gmzero(n,\epzero)\biggr)_{\condi({\tilde{n}}_{(0)},{\tilde{n}}_{(\pn)},\ngN)} \leq \text{exp}(-2n{\gmzero(n,\epzero)}^2) = \epzero^2
\end{align}
Using the monotonicity of the bound, we replace $\tilde{n}_{(0)}$ with $n$, since $n \geq \tilde{n}_{(0)}$ ($n$ is the total number of rounds and $\tilde{n}_{(0)}$ is the number of rounds with 0 photons):
\begin{align}
    \Pr\biggl(\nkzero \geq \tilde{n}_{(0)} q_Z + \sqrt{\tilde{n}_{(0)}} \sqrt{n} \, \gmzero(n,\epzero) \biggr)_{\condi(\tilde{n}_{(0)},\tilde{n}_{(\pn)},\ngN)} \leq \epzero^2 .
\end{align}
Summing over all possible values of ${\tilde{n}}_{(0)}$, ${\tilde{n}}_{(\pn)}$, and $\ngN$ to remove the conditioning, we obtain:
\begin{align}
    \Pr\biggl( \nkzero \geq n \bigl( q_Z + \gmzero(n,\epzero) \bigr) \biggr) \leq \epzero^2 .
\end{align}
\begin{remark} \label{remark:partialcharacterizethree}
    Observe that we can replace $q_Z$ with an upper bound on it. This is required when we extend the proof to the case of partially characterized imperfections in \cref{subsec:partialimperfections}.
\end{remark}
\begin{remark}
In the analysis of the active BB84 in~\cite{devEUR}, the key rate becomes zero if the detectors have a minimum dark count rate of zero. This oddity is avoided in our analysis, as we upper bound the number of 0-photon key rounds, which is not done in that analysis.
\end{remark}
We summarize this subsection in the following lemma.
\begin{lemma}\label{lemma:zerobound}
Let $\rho_{A^nB^n} \in S_{\circ}(A^nB^n)$ be arbitrary state representing $n$ rounds of a generic QKD protocol. Suppose the joint POVM for Alice and Bob can be written as follows:
    \begin{align}
        \{\vec{\Gamma}\} = \{\Gamma_{(\text{test},=)},\;\Gamma_{(\text{test},\neq)},\;\Gamma_{(\text{key},=)},\;\Gamma_{(\text{key},\neq)},\;\Gamma_{\text{other}},\;...\},
    \end{align}
    and the joint POVM inherits a block-diagonal structure from the fact that Bob's POVM has a block-diagonal structure:
    \begin{align}
        \Gamma_i = \bigoplus_{\pnm=0}^{\infty} \Gamma_i^{(\pnm)} .
    \end{align}
    Then, Alice and Bob's measurement can be viewed as the following:
     \begin{enumerate}[label=\circletext{\arabic*}, leftmargin=*]
    \item They first perform QND measurement with projectors on the subspaces $\{\Pi_0,\Pi_{\pn},\Pign\}$. They complete measurements on the $(1\leq\pnm\leq\pn)$ and $(>\pn)$ subspaces when convenient.
    \item  Upon the result $\Pi_0$ at stage \circletext{1}, They complete the measurement for 0 subspace with $\{\vec{\Gamma}^{(0)}\}$. Denote the number of 0-photon key rounds (corresponding to $\Gamma_{(\text{key},=)}^{(0)}$ and $\Gamma_{(\text{key},\neq)}^{(0)}$) as $\bd{\tilde{n}_{(\textbf{key,0})}}$.
\end{enumerate}
Let \( q_Z \) be the infinity norm of the POVM elements (or an upper bound thereof) in the \( 0 \)-photon subspace:
\begin{equation}
    q_Z \geq \infnorm{\Gamma_{(\text{key},=)}^{(0)}+\Gamma_{(\text{key},\neq)}^{(0)}}.
\end{equation}
Then, the following relation holds 
\begin{align}
    \text{Pr}\biggl(\bd{\tilde{n}}_{\textbf{(key,0)}} \geq nq_Z+\sqrt{{(-\ln(\epzero))}{n}}\biggr) \leq \epzero^2
\end{align}
\begin{proof}
    The proof is described above in this subsection by identifying $\Gamma_{(Z,=)}=\Gamma_{(\text{key},=)}$, $\Gamma_{(Z,\neq)}=\Gamma_{(\text{key},\neq)}$ and $\nkzero=\bd{\tilde{n}_{(\textbf{key,0})}}$.
\end{proof}
\end{lemma}
\subsection{Combining bounds}\label{sec:combine bounds}
In this subsection, we combine the bounds we find above to obtain our main goal \cref{eq:phaseerroridea}. That is, we want an upper bound on the phase error rate in ($1\leq\pnm\leq\pn$)-photon rounds and a lower bound on the number of ($1\leq\pnm\leq\pn$)-photon key rounds in terms of observed quantities.
From \cref{NumberofErrorBound}, the phase error rate is bounded by:
\begin{align}\label{eq:memorylessbounds}
    \begin{split}
    &\text{Pr}\biggl(\frac{\Nkoneb}{\nkoneb} \geq \frac{\Nxoneb+\sqrt{-2\ln(\epazua^2)(\nxoneb+\nkoneb)}+\delta(\nxoneb+\nkoneb)}{a\nkoneb}\\
    &\phantom{\text{Pr}\biggl(\frac{\Nkoneb}{\nkoneb} \geq}+\frac{\sqrt{-2\ln(\epazub^2)(\nxoneb+\nkoneb)}}{\nkoneb}\biggr)\\
    &=\text{Pr}\biggl(\errks \geq \frac{\Nxoneb}{a\nkoneb}+\sqrt{\frac{-2\ln(\epazua^2)}{a}\biggl(\frac{\nxoneb}{\nkoneb^2}+\frac{1}{\nkoneb}\biggr)}\\
    &\phantom{=\text{Pr}\biggl(\errks \geq}+\frac{\delta}{a}\biggl(\frac{\nxoneb}{\nkoneb}+1\biggr)+\sqrt{{-2\ln(\epazub^2)}\biggl(\frac{\nxoneb}{\nkoneb^2}+\frac{1}{\nkoneb}\biggr)}\biggr)\\
    &\leq \epazua^2 + \epazub^2 .
    \end{split}
\end{align}
We can replace $\Nxoneb$ by $\Nxobsb$, $\nxoneb$ by $\nxb$ since $\Nxoneb\leq\Nxobsb$, $\nxoneb\leq\nxb$, and the bound on phase error rate is monotonically increasing in these variables. Doing so, we obtain:
\begin{align}\label{eq:numberoferror}
    \begin{split}
    &\text{Pr}\Biggl(\errks \geq \frac{\Nxobsb}{a\nkoneb}+\sqrt{\frac{-2\ln(\epazua^2)}{a}\biggl(\frac{\nxb}{\nkoneb^2}+\frac{1}{\nkoneb}\biggr)}\\&\phantom{\text{Pr}\biggl(\errks\geq}+\frac{\delta}{a}\biggl(\frac{\nxb}{\nkoneb}+1\biggr)+\sqrt{{-2\ln(\epazub^2)}\biggl(\frac{\nxb}{\nkoneb^2}+\frac{1}{\nkoneb}\biggr)}\Biggr)\\
    &\leq \epazua^2 + \epazub^2.
\end{split}
\end{align}

To obtain an estimate on $\nkoneb$, we need to use the fact that $\nkoneb = \nkb -\nkzero-\nkgone$ and $\ngNb \geq \nkgone$. Then,
\begin{align}\label{eq:uncorrelated nk}
    \nkoneb \geq \nkb - \nkzero-\ngNb\;.
\end{align}
Combining~\cref{step3zerobound} and~\cref{step3multibound} using the union bound (see~\cref{appendix:combine bounds}), we obtain a lower bound on the number of rounds used to bound the phase errors \( \nkoneb \):
\begin{align}\label{eq:K}
    &\text{Pr}\biggl(\nkoneb \leq \K(\nkb,\nmcb)\biggr) \leq \epzero^2+\eppntb^2,
\end{align}
where
\begin{align}
    &\K(\nkb,\nmcb) = \nkb - n q_Z-\frac{\nmcb}{\lm}-\sqrt{{-\ln(\epzero)}n}-\frac{\sqrt{(-\log(\eppntb))(-\log(\eppntb)+4\lm\nmcb)}}{2{\lm}^2}-\biggl(\frac{-2\log(\eppntb)}{4{\lm}^2}\biggr)\;.
\end{align}
Then together with \cref{eq:numberoferror} and \cref{eq:K},  the required statement follows:
\begin{align}\label{eq:mainresult}
    \text{Pr}&\biggl(\errks \geq  \bimp(\nxb,\nkb,\Nxobsb,\nmcb)\;\lor\;\nkoneb \leq \K(\nkb,\nmcb) \biggr)\notag\\
    &\leq \epazua^2 + \epazub^2 +\eppntb^2+\epzero^2,
\end{align}
where 
\begin{align}
    \begin{split}
    &\bimp(\nxb,\nkb,\Nxobsb,\nmcb)=\\
    &\frac{\Nxobsb}{a\K(\nkb,\nmcb)}+\sqrt{\frac{-2\ln(\epazua^2)}{a}\biggl(\frac{\nxb}{\K(\nkb,\nmcb)^2}+\frac{1}{\K(\nkb,\nmcb)}\biggr)}+\frac{\delta}{a}\biggl(\frac{\nxb}{\K(\nkb,\nmcb)}+1\biggr)\\\
    &+\sqrt{{-2\ln(\epazub^2)}\biggl(\frac{\nxb}{\K(\nkb,\nmcb)^2}+\frac{1}{\K(\nkb,\nmcb)}\biggr)}.
    \end{split}
\end{align}
Then, the key rate formula for the single-photon source BB84 protocol is the following:
\begin{align}
    \begin{split}
    &l(\errzobs,\Nxobs,\nx,\nk,\nmc)\\
    &:= \text{max}\{\K(\nk,\nmc)\biggl(1-h(\bimp(\nx,\nk,\Nxobs,\nmc))\biggr)\\&\;\;\;\;\;\;-\lambda_{\text{EC}}(\errzobs,\Nxobs,\nx,\nk,\nmc)-2\text{log}(1/2\eppa)-\text{log}(2/\epev),0\}\\
    &  \epat^2:=\epazua^2 + \epazub^2 +\eppntb^2+\epzero^2 
     \end{split}
\end{align}
and the protocol is $(2\epat+\eppa+\epev)$-secure, as shown in \cref{appendix:security}.
\begin{remark}
    Although the proof itself is generalized to accommodate any value of \( \pn \), 
the resulting bound on the phase error rate may not be useful when \( \pn > 1 \). 
The reason is that, for the bound \cref{eq:mainresult} to be useful, we must find a value of \( a \) such that \( \delta \) is small (We define them in \cref{deltabound}).  
However, since all double-click outcomes are announced, \( \ftx \) and \( \ftz \) 
do not include the POVM elements corresponding to double-click outcomes.
To gain some intuition, consider the case of perfect detectors: 
for \( \pn = 1 \), we expect no double-click outcomes, and hence \( \delta = 0 \). 
For \( \pn = 2 \), however, when we look at the 2-photon subspace, $\delta=0$ only when we also include the double-click POVM elements in \( \ftx \) and \( \ftz \). The absence of the double-click terms 
can lead to a significant increase in \( \delta \), 
potentially making the phase error rate bound ineffective. 
This is also why, in a standard EUR proof, Bob randomly maps them to one of the outcomes.

One possible solution is for Bob to map them so that only cross-click outcomes are announced. 
For estimation in the \( (>\pn) \) subspace, this would require determining 
the minimum eigenvalue of the cross-click POVM element \( \Gamma_{\mathrm{cc}} \) 
within that subspace; i.e., computing \( \lm(\Pign \Gamma_{\mathrm{cc}} \Pign) \). 
This introduces additional challenges, as the method in Ref.~\cite{lars} 
does not directly apply to the cross-click POVM element \( \Gamma_{\mathrm{cc}} \). However, these challenges are not fundamental. 
\end{remark}
\begin{remark}
The proof relies on the fact that Bob’s POVM has a block-diagonal structure, which allows us to decompose the original joint measurement \( \{ \vec{\Gamma} \} \) into three branches of sequential measurements, as illustrated in~\cref{diag:equivalent}. For the proofs in the branches, they require only that one can compute (or choose) the parameters \( a \), \( \delta \), \( \lambda_{\min} \), and \( q_Z \), all of which are defined by the POVM elements in the respective branches.
We will later assume polarization encoding for the calculation of these parameters, and we also assume a model for loss and dark counts for detectors in \cref{diag:passivedetection}. However, the proof does not rely on any specific property of polarization encoding. For a generic QKD protocol, as long as the POVM has the block-diagonal structure and the required parameters can be obtained from the POVM elements, our proof can be applied.
\end{remark}
We summarize the results in the following theorem.
\begin{theorem}
    Let $\rho_{A^nB^n} \in S_{\circ}(A^nB^n)$ be arbitrary state representing $n$ rounds of a generic QKD protocol. Suppose the joint POVM for Alice and Bob can be written as follows:
    \begin{align}
        \{\vec{\Gamma}\} = \{\Gamma_{(\text{test},=)},\;\Gamma_{(\text{test},\neq)},\;\Gamma_{(\text{key},=)},\;\Gamma_{(\text{key},\neq)},\;\Gamma_e,\;\Gamma_{\text{other}}\}
    \end{align}
    and the joint POVM inherits a block-diagonal structure from the fact that Bob's POVM has a block-diagonal structure:
    \begin{align}
        \Gamma_i = \bigoplus_{\pnm=0}^{\infty} \Gamma_i^{(\pnm)} .
    \end{align}
    Then, Alice and Bob's measurement can be viewed as the following:
\begin{enumerate}[label=\circletext{\arabic*}, leftmargin=*]
    \item They first perform QND measurement with projectors on the subspaces $\{\Pi_0,\Pi_{\pn},\Pign\}$. 
    \item  Upon the result $\Pign$ at stage \circletext{1}, they then complete the measurement $\{\Gamma^{(>\pn)}_{e},\idd-\Gamma^{(>\pn)}_{e}\}$.
    They then complete all other measurements on $(>\pn)$ subspace using some POVM $\{\vec{G}^{(>\pn)}\}$.
    \item (a) They complete all the measurements for 0 subspace.

    (b) They then perform a coarse-grained measurement
    $\{F^{(\pn)},1-F^{(\pn)}\}$ on ($1\leq\pnm\leq\pn$) subspace, where
    \begin{align}
        F^{(\pn)}:= \Gamma^{(\pn)}_{(\text{test},=)}+\Gamma^{(\pn)}_{(\text{test},\neq)}+\Gamma^{(\pn)}_{(\text{key},=)}+\Gamma^{(\pn)}_{(\text{key},\neq)}
    \end{align}
    \item They start to complete the measurements for ($1\leq\pnm\leq\pn$) subspace. They first perform another coarse-grained basis choice measurement $\{\tilde{F}^{(\pn)}_{\text{test}},\tilde{F}^{(\pn)}_{\text{key}}\}$. Then, they complete the measurement with $\{\tilde{G}^{(\pn)}_{(\text{test},=)},\tilde{G}^{(\pn)}_{(\text{test},\neq)}\}$ upon the result $\tilde{F}^{(\pn)}_{\text{test}}$. 
    \begin{align}
    \tilde{F}^{(\pn)}_{\text{test}} = {\sqrt{F^{(\pn)}}}^+(\Gamma^{(\pn)}_{(\text{test},=)}+\Gamma^{(\pn)}_{(\text{test},\neq)}){\sqrt{F^{(\pn)}}}^+;\;\;\;    \tilde{F}^{(\pn)}_{\text{key}} = {\sqrt{F^{(\pn)}}}^+(\Gamma^{(\pn)}_{(\text{key},=)}+\Gamma^{(\pn)}_{(\text{key},\neq)}){\sqrt{F^{(\pn)}}}^+
    \end{align}
    \item They complete the key measurement for ($1\leq\pnm\leq\pn$) subspace using $\{\tilde{G}^{(\pn)}_{(\text{key},=)},\tilde{G}^{(\pn)}_{(\text{key},\neq)}\}$ upon the result $\tilde{F}^{(\pn)}_{\text{key}}$. We denote the number of ($1\leq\pnm\leq\pn$) key rounds as $\bd{\tilde{n}_{(\textbf{key,\pn})}}$. The number of phase error $\bd{\tilde{N}^{\textbf{ph}}_{\textbf{\pn}}}$ is defined as the number of $\tilde{G}^{(\pn)}_{(\text{test},\neq)}$ outcomes if they complete the measurement using $\{\tilde{G}^{(\pn)}_{(\text{test},=)},\tilde{G}^{(\pn)}_{(\text{test},\neq)}\}$ instead.
  \begin{align}
            \tilde{G}^{(\pn)}_{(\text{test},\neq)} = \sqrt{\tilde{F}^{(\pn)}_{\text{test}}}^+\sqrt{F_{\mathrm{sc}}^{(\pn)}}^+\Gamma^{(\pn)}_{(\text{test},\neq)}\sqrt{F_{\mathrm{sc}}^{(\pn)}}^+\sqrt{\tilde{F}^{(\pn)}_{\text{test}}}^+
        \end{align}
\end{enumerate}
Let us define the following random variables:
\begin{itemize}
    \item $\bd{n_{\textbf{test}}}$: The number of rounds corresponding to $\Gamma_{(\text{test},=)}$ and $\Gamma_{(\text{test},\neq)}$.
    \item $\bd{n_{\textbf{key}}}$: The number of rounds corresponding to $\Gamma_{(\text{key},=)}$ and $\Gamma_{(\text{key},\neq)}$.
    \item $\bd{N^{\textbf{obs}}_{\textbf{test}}}$: The number of rounds corresponding to $\Gamma_{(\text{test},\neq)}$.
    \item $\bd{n_e}$: The number of rounds corresponding to $\Gamma_e$.
\end{itemize}
Define phase error rate in ($1\leq\pnm\leq\pn$) rounds
\begin{align}
    \bd{e}^{\textbf{ph}}_{\textbf{\pn}}:= \frac{\bd{\tilde{N}^{\textbf{ph}}_{\textbf{\pn}}}}{\bd{\tilde{n}_{\textbf{(key,\pn)}}}}\;.
\end{align}
For any $a>0$, if one can compute the following parameters from the POVM elements,
\begin{align}
    \begin{split}
    &\lm \leq \lm(\Pign\Gamma_e\Pign)\\
    & q_Z \geq \infnorm{\Gamma_{(\text{key},=)}^{(0)}+\Gamma_{(\text{key},\neq)}^{(0)}}\\
    & \delta \geq \infnorm{\sqrt{\tilde{F}^{(\pn)}_{\text{test}}} \tilde{G}^{(\pn)}_{(\text{test},\neq)} \sqrt{\tilde{F}^{(\pn)}_{\text{test}}} - a\sqrt{\tilde{F}^{(\pn)}_{\text{key}}} \tilde{G}^{(\pn)}_{(\text{test},\neq)} \sqrt{\tilde{F}^{(\pn)}_{\text{key}}}}
     \end{split}
\end{align}
then the following bounds hold for any $\epzero$, $\eppntb$, $\epazua$, $\epazub>0$:
\begin{align}
    &\text{Pr}\biggl(\bd{e}^{\textbf{ph}}_{\textbf{\pn}} \geq  \bimp(\bd{n_{\textbf{test}}},\bd{n_{\textbf{key}}},\bd{N^{\textbf{obs}}_{\textbf{test}}},\bd{n_e}) \;\lor\; \bd{\tilde{n}_{\textbf{(key,\pn)}}} \leq \K(\bd{n_{\textbf{key}}},\bd{n_e})\biggr)\notag\\&\leq \epazua^2 + \epazub^2 +\eppntb^2+\epzero^2,
\end{align}
where 
\begin{align}
    &\K(\bd{n_{\textbf{key}}},\bd{n_e}) = \bd{n_{\textbf{key}}} - nq_Z-\frac{\bd{n_e}}{\lm}-\frac{\sqrt{(-\log(\eppntb))(-\log(\eppntb)+4\lm\bd{n_e})}}{2{\lm}^2}-\sqrt{{-\ln(\epzero)}n}-\biggl(\frac{-2\log(\eppntb)}{4{\lm}^2}\biggr)\\
    &\bimp(\bd{n_{\textbf{test}}},\bd{n_{\textbf{key}}},\bd{N^{\textbf{obs}}_{\textbf{test}}},\bd{n_e})=\notag\\
    &\frac{\bd{N^{\textbf{obs}}_{\textbf{test}}}}{a\K(\bd{n_{\textbf{key}}},\bd{n_e})}+\frac{\delta}{a}\biggl(\frac{\bd{n_{\textbf{test}}}}{\K(\bd{n_{\textbf{key}}},\bd{n_e})}+1\biggr)+\sqrt{\frac{-2\ln(\epazua^2)}{a}\biggl(\frac{\bd{n_{\textbf{test}}}}{\K(\bd{n_{\textbf{key}}},\bd{n_e})^2}+\frac{1}{\K(\bd{n_{\textbf{key}}},\bd{n_e})}\biggr)}\notag\\
    &+\sqrt{{-2\ln(\epazub^2)}\biggl(\frac{\bd{n_{\textbf{test}}}}{\K(\bd{n_{\textbf{key}}},\bd{n_e})^2}+\frac{1}{\K(\bd{n_{\textbf{key}}},\bd{n_e})}\biggr)}.
\end{align}
\end{theorem}
\begin{proof}
    The proof of this theorem follows from the analysis presented in \cref{sec:phase error estimation} by identifying
    \begin{align}
            \begin{split}
           &\Gamma_{(X,=)}=\Gamma_{(\text{test},=)};\;\;\Gamma_{(X,\neq)}=\Gamma_{(\text{test},\neq)};\;\;\Gamma_{(Z,=)}=\Gamma_{(\text{key},=)};\;\;\Gamma_{(Z,\neq)}=\Gamma_{(\text{key},\neq)};\;\;\Gamma_{\mathrm{mc}} = \Gamma_e\\
           &\nxb=\bd{{n}_{\textbf{test}}};\;\; \nkb=\bd{{n}_{\textbf{key}}};\;\;\nmcb=\bd{n_e};\;\;\nkoneb=\bd{\tilde{n}_{\textbf{(key,\pn)}}};\\
           &\Nxobsb = \bd{{N}^{\textbf{obs}}_{\textbf{test}}};\;\;\Nkoneb = \bd{\tilde{N}^{\textbf{ph}}_{\textbf{\pn}}};\;\;\errks=\bd{e}^{\textbf{ph}}_{\textbf{\pn}}.
    \end{split}
    \end{align}
\end{proof}
\subsection{Extensions to Imperfectly Characterised Detectors and the Decoy-state Protocol}\label{subsec:partialimperfections}
In the previous section, we presented the proof under the assumption that detector imperfections are fully characterized, 
meaning that if we take our model in \cref{diag:passivedetection}, the efficiency \( \eta_i \), the dark count rate \( d_i \), and the beam-splitting ratio \( \pxb \) are exactly known.
We now consider the case where these parameters lie within known ranges: 
\( \eta_i \in [\eta^l_i, \eta^u_i] \), \( d_i \in [d^l_i, d^u_i] \), and 
\( \pxb \in [s^* - \theta, s^* + \theta] \).
As we mentioned in \cref{sec:phase error estimation},
two challenges arise when the imperfections are not precisely known. 
The first missing ingredient is determining a lower bound on \( \lambda_{\min}(\Pign \Gamma_{\mathrm{mc}} \Pign) \) when \( \Gamma_{\mathrm{mc}} \) is not known exactly. This is addressed in~\cref{sec:subspace}.
Second, we need to determine the detector parameters \( a \), \( \delta \), and \( q_Z \), 
which characterize the basis-efficiency mismatches. 
We show in~\cref{appendix:adelta} that, for the case \( \pn = 1 \), these three parameters can be appropriately bounded given that the imperfections lie within known ranges. Once these bounds are obtained, one uses the monotonicity of the computed bounds to replace the relevant parameters with their maximum or minimum values, as described in \cref{remark:partialcharacterizeone,remark:partialcharacterizetwo,remark:partialcharacterizethree}. This addresses the extension to imperfectly characterized detectors. 

In the decoy-state protocol, when Alice prepares a quantum state, she additionally selects an intensity $\muk \in \{\mu_1, \mu_2, \mu_3\}$ (required for the decoy analysis we use) with corresponding probabilities $p_{\muk}$. We assume the intensities satisfy $\mu_1 > \mu_2 + \mu_3$ and $\mu_2 > \mu_3 \geq 0$. She then encodes the bit using a phase-randomized weak coherent pulse with the selected intensity. Because the pulses are phase-randomized, they can be viewed as a classical mixture of states with different photon numbers. In some rounds, Alice sends out single-photon signals, although this information is not directly accessible during the protocol. 
The key idea in the decoy-state protocol~\cite{hwang2003,lo2005,XBWang} is that, since Alice selects the pulse intensities randomly and reveals her choices only after all signals have been sent and measured by Bob, she can use this information—independent of Bob’s detection setup—to estimate the number of rounds in which she sent single-photon signals.

Therefore, the analysis for the decoy-state protocol is a straightforward extension of the single-photon source analysis. We first bound the phase error rate in terms of quantities associated with Alice sending single-photon signals (for example, the number of key rounds that Alice sends with single-photon signals), following the same analysis as in the single-photon source case. Then, the quantities corresponding to Alice sending single photons are estimated using decoy-state analysis based on actual observations in the protocol (for example, the number of key rounds sent by Alice with different intensities). We adopt the decoy-state method from Ref.~\cite[Appendix A]{charles}. For the detailed decoy-state analysis and the corresponding secret key rate formula, we refer the reader to~\cref{appendix:decoy}.

\section{phase error estimation with correlated detector imperfections}\label{sec:correlated detectors}
We now consider the case where Bob's detectors exhibit some amount of correlation, such as afterpulsing or detector dead time \cite{burenkov2010}, where a detection outcome may influence subsequent detections for some period of time. 
 First, we state the model for memory effects we study in this work, which is identical to the one from Ref.~\cite{devEUR}. 
Let \( k_i \) be the outcome of Bob's measurement in the \( i \)th round, 
and let \( k_i^j \) denote the sequence of measurement outcomes from round \( i \) to \( j \). 
Let \( l_c \) denote the correlation length, defined as the number of no-click rounds required for the detectors to return to uncorrelated behavior.
In the \( i \)th round, Alice and Bob measures using the POVM \( \{\vec{\Gamma}^{k_{i-l_c}^{i-1}}\} \), where ${k_{i-l_c}^{i-1}}$ are the measurement outcomes for previous $l_c$ rounds. Therefore, the POVM depends on the outcomes of the previous $l_c$ rounds.
If the previous \( l_c \) rounds are all no-click outcomes, 
then they measure with the uncorrelated POVM \( \{\vec{\Gamma}^{k_{i-l_c}^{i-1}}\} = \{\vec{\Gamma}\} \). We assume that Bob has characterised the imperfections for the uncorrelated POVM \( \{\vec{\Gamma}\} \). Our methods do not require any knowledge about the correlated POVMs.

The main problem when dealing with correlated detectors is that it's not clear which state we should apply the EUR statement to. Even if we use a coarse-grained measurement to decide the basis for each round and detect vs no-detect outcome for each round, \textit{this} POVM itself requires the previous round to be completely measured in order to be specified. However, if the previous round has been completely measured, then the EUR statement cannot be applied on it. This poses a fundamental problem when we attempt to establish secret key from all rounds, in the presence of memory effects, using the EUR proof technique.

This problem is addressed in Ref.~\cite[Lemma 5]{devEUR} by adding a click/no-click filtering step before doing any basis measurements. The main idea is to include a protocol modification that rejects rounds whose previous $l_c$ measurements contained a click. Since memory effects only last for $l_c$ rounds,  if all previous $l_c$ rounds are no-click outcomes, then one measures the current round with the uncorrelated POVM. Therefore, the extra click/no-click filtering step help us to filter out those rounds with memory effects. Once those rounds are rejected, we can apply the analysis from the memoryless case (\cref{sec:phase error estimation}) on the kept rounds.

\begin{lemma}\cite[Lemma 5]{devEUR}\label{lemma:correlated filter}
        Consider the state $\rho_{Q^n}$, and let the $i$th subsystem be measured using POVM \( \{\vec{\Gamma}^{k_{i-l_c}^{i-1}}\} \), where $k_i$ denotes the outcome of the $i$th measurement, and $k^{i-1}_{i-l_c}$ denotes the string of outcomes in the previous $l_c$ rounds. Suppose that $\{ \vec{\Gamma}\}$ (corresponding to $k^{i-1}_{i-l_c} = \text{ no-click})$ is such that it can be described by a two-step measurement, where the first measurement $\{\tilde{F},\idd-\tilde{F}\}$ determines the detect vs no-detect, followed by a second step measurement.
    Consider the state obtained after rejecting all rounds $i$ for which a detection occurred in the previous $l_c$ rounds. Then this state can be obtained via a procedure that only performs the $\{\tilde{F}, \idd -\tilde{F} \}$ measurements on the rounds that are not rejected.
\end{lemma}
This lemma states that, for each round using the uncorrelated POVM, 
Alice and Bob must apply a filter \( \{\Fc, \idd - \Fc\} \) 
to determine whether it is a click or no-click outcome \emph{before} completing the measurement. 
Thus, after applying \cite[Lemma 5]{devEUR}, we can interpret Alice and Bob's measurement procedure as follows:
\begin{enumerate}
    \item The first round is always kept. 
    The filter \( \{\Fc, \idd - \Fc\} \) is applied to determine whether it is a click or no-click outcome.
    
    \item At the \( i^{\text{th}} \) round, one of two things can happen. If any of the previous \( l_c \) rounds had a click, 
    Bob rejects this round, and the measurement can be completed using a correlated POVM. Since he has completed the measurement he knows whether it was a click or a no-click. On the other hand, if none of the previous $l_c$ rounds had a click, then he keeps the round.
    If he keeps the round, the filter \( \{\Fc, \idd - \Fc\} \) is applied. 
    In either case, whether the round results in a click or a no-click outcome is determined at this step. 
    This information is passed to the next round to decide whether to reject or keep it.
    
    \item After all \( n \) rounds have passed through this filtering process, 
    Alice and Bob complete the measurements on the kept rounds.
\end{enumerate}

We refer to this process as the \emph{filtering process} in this section, 
and we call \( \{\Fc, \idd - \Fc\} \) the \emph{filter}.

\begin{remark}
Note that Ref.~\cite{devEUR} uses \cref{lemma:correlated filter} on active detection setups and provides a proof sketch for estimating the phase error rate in the presence of correlated detectors. With perfect (same loss and dark count rates for each detector) active detectors, one can first perform a QND measurement as the filter to figure out whether it is a click or a no-click outcome. Later one can continue the proof on the rounds without memory effects. However,   for imperfect detectors, a QND measurement is not able to correctly implement the required filter. Moreover, the measurement that does implement the correct filter no longer leads to active basis choice, rendering the approach from Ref.~\cite{devEUR} to be inapplicable. 
(While Ref.~\cite{devEUR} proposes that one can lift the perfect case to the imperfect case by transferring imperfections to Eve through Ref.~\cite{shlok}, this is not straightforward.)

This problem is avoided in this work due to the use of Azuma's inequality, which allows us to construct any convenient 
\( \ftx \) and \( \ftz \), as long as we can bound their \( \infty \)-norm difference. Therefore, we can construct the filter as the click/no-click measurement at the begining of each uncorrelated round. We apply \cite[Lemma 5]{devEUR} to reduce the analysis to those kept rounds without memory effects.  We observe that our proofs \cref{lemma:NumberofErrorBound} and \cref{lemma:multibound} for memoryless imperfect passive detection setups only require a global state $\rho_{A^nB^n}$ and some POVM with block-diagonal structure. Thus, the analysis in the memoryless case still holds for the state \emph{after} applying the filter to all rounds conditioned on the filtering process outcomes. Therefore, the phase error rate can be bounded using the same method as in \cref{sec:phase error estimation} on the state after the filtering process.
\end{remark}

\subsubsection{Protocol modification for memory affects}\label{sec:correlated protocol description}
We note that we require the following protocol modification to incorporate memory effects. Alice and Bob's protocol is identical to the one described in \cref{sec:basicprotdescription}, except that they include the following additional post-processing after all $n$ rounds are sent and measured: Bob only keeps those rounds for which the previous $l_c$ rounds are all no-click outcomes. For the rest of the rounds that do not satisfy this condition, he rejects (discards) them. He announces this information to Alice so that Alice also rejects those correlated rounds. The protocol proceeds by only considering the observations and raw key generated from the kept rounds.

Note that due to this protocol modification, the observed quantities used in the key length decision used in this protocol are different from the quantities discussed in \cref{sec:phase error estimation}. However, we use the same notation to refer to them. That is: $\nmc,\nx,\nk,\Nxobs,\errzobs$ denote the various quantities within the kept rounds.

Note that on-the-fly announcements of click/no-click outcomes are naturally accommodated in the proof. We show this extension in \cref{sec:onthefly}
 
\subsection{Proof Idea}

In this subsection, we present the main idea behind our proof. The proof is formally undertaken in \cref{appendix:commute}.

We choose the filter as the click/no-click POVM $\{\Fc,\idd-\Fc\}$\footnote{The measurement in the filter can be more than click/no-click measurement. For instance, one can include the 0-photon subspace measurement inside the filter. }, where $\Fc = \idd-\Gamma_{\mathrm{nc}}$: if Alice and Bob keep a round, he applies the filter to that round. If he rejects, they complete the full measurement on this round. After applying this filtering process to all $n$ rounds, they share a filtered state (see~\cref{fig:correlated channel}): 
\begin{align}
    \begin{split}
        &\rho_{\text{filtered}} = {\filteredstate}
    \end{split}
\end{align}
where $\Omega$ denotes the classical outcomes of the filtering process.
\begin{figure}
    \centering
    \scalebox{0.9}{\begin{tikzpicture}
    [
    dashedbox/.style={draw, rectangle, minimum width=2.5cm, minimum height=1cm, dashed, align=center},
    box/.style={draw, rectangle, minimum width=2cm, minimum height=1cm, align=center},
    arrow/.style={-Stealth},
    discard/.style={->, thick, dashed},
    every node/.style={font=\small}
]
\node(rho0){\(\rho_{A^nB^nE^n}\)};
\node[draw,fill=gray!10, minimum width=2cm, minimum height=1.5cm, below= 2cm of rho0,xshift=-5cm] (outbox1) {\(\)};
\node[draw, fill=green!10, minimum width=1cm, minimum height=1cm] at (outbox1.center){Filter};

\node[draw,fill=gray!10, minimum width=4cm, minimum height=1.5cm, right = 2cm of outbox1] (outbox2) {\(\)};
\node[draw, fill=green!10, minimum width=1cm, minimum height=1cm] at ([xshift=-1.1cm]outbox2.center){Filter};
\node[draw, fill=red!10, minimum width=1cm, minimum height=1.2cm,align=center] at ([xshift=0.9cm]outbox2.center) {Full\\ Measurement};
\draw[thick] ([xshift=-0.2cm,yshift=-0.1cm]outbox2.north) -- ([xshift=-0.3cm,yshift=0.1cm]outbox2.south);

\node[right = 2cm of outbox2](processes){......};

\node[draw,fill=gray!10, minimum width=4cm, minimum height=1.5cm, right = 0.1cm of processes] (outboxn) {\(\)};
\node[draw, fill=green!10, minimum width=1cm, minimum height=1cm] at ([xshift=-1.1cm]outboxn.center){Filter};
\node[draw, fill=red!10, minimum width=1cm, minimum height=1.2cm,align=center] at ([xshift=0.9cm]outboxn.center) {Full\\ Measurement};
\draw[thick] ([xshift=-0.2cm,yshift=-0.1cm]outboxn.north) -- ([xshift=-0.3cm,yshift=0.1cm]outboxn.south);
\node[right = 1cm of outboxn](rhon){\(\rho_{\text{filtered}}\)};

\node[above = 0.2cm of outbox1,xshift=0.3cm]{\(A_1B_1\)};
\node[above = 0.2cm of outbox2,xshift=0.3cm]{\(A_2B_2\)};
\node[above = 0.2cm of outboxn,xshift=0.3cm]{\(A_nB_n\)};

\draw[arrow](rho0.south) to [out=-90,in=90] (outbox1.north);
\draw[arrow](rho0.south) to [out=-90,in=90] (outbox2.north);
\draw[arrow](rho0.south) to [out=-90,in=90] (outboxn.north);



\draw[dashed,arrow] (outbox1.east) -- node[draw,fill=gray!10]{\textcolor{darkgreen}{keep}/\red{rej}} (outbox2.west);

\draw[dashed,arrow] (outbox2.east) -- node[draw,fill=gray!10]{\textcolor{darkgreen}{keep}/\red{rej}} (processes.west);
\draw[arrow](outboxn) -- (rhon);


\end{tikzpicture}}
\caption{The filtering process. In each round, Bob performs either the filter or the full measurement depending on the click patterns observed in the previous $l_c$ rounds, to determine whether the current round results in a click or no-click. After $n$ rounds, the final shared state, conditioned on the filtering outcomes, is denoted by $\filteredstate$.}
    \label{fig:correlated channel}
\end{figure}
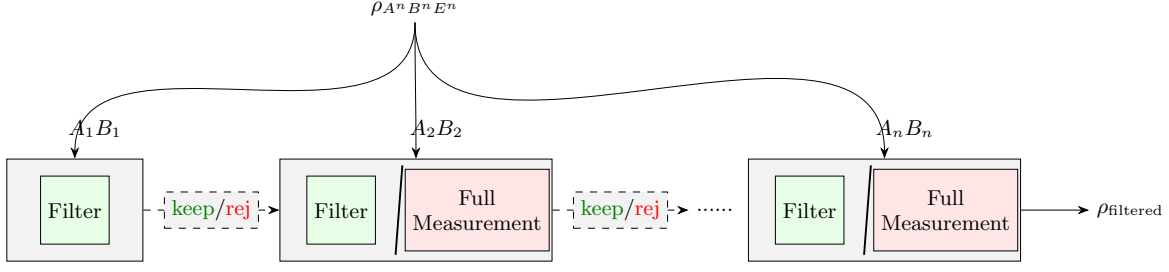 

Alice and Bob complete the measurement in a similar manner as \cref{diag:equivalent}, and we set $\pn=1$. After the filtering process, Alice and Bob measure $\filteredstate$ with POVM 
\begin{align}
    \begin{split}\label{eq:filteredPOVM}
    \{&\sqrt{\Fc}^+\Gamma_{(X,=)}\sqrt{\Fc}^+,\;\sqrt{\Fc}^+\Gamma_{(X,\neq)}\sqrt{\Fc}^+,\;\sqrt{\Fc}^+\Gamma_{(Z,=)}\sqrt{\Fc}^+,\;\sqrt{\Fc}^+\Gamma_{(Z,\neq)}\sqrt{\Fc}^+,\\&\sqrt{\Fc}^+\Gamma_{\mathrm{mc}}\sqrt{\Fc}^+,\;\sqrt{\Fc}^+\Gamma_{\text{other}}\sqrt{\Fc}^++\idd-\Pi_{\Fc}\},
     \end{split}
\end{align}
where $\Pi_{\Fc}$ is the projector onto the support of $\Fc$. This POVM has a block-diagonal structure on photon number $\pnm$.

Therefore, we can apply the same measurement break down as \cref{diag:equivalent} using QND measurement and \cref{lemma:twostep} on this POVM. That is, we perform the QND measurement, break down measurements on 0-,1- and ($>1$)-photon subspaces, and construct an estimation protocol as shown in \cref{fig:correlated_equivalent}.
\begin{figure}
    \centering
    \scalebox{0.85}{
\begin{tikzpicture}
    [
    dashedbox/.style={draw, rectangle, minimum width=2.5cm, minimum height=1cm, dashed, align=center},
    box/.style={draw, rectangle, minimum width=2cm, minimum height=1cm, align=center},
    arrow/.style={-Stealth},
    discard/.style={->, thick, dashed},
    every node/.style={font=\small}
]
\node(top) {\( i^{\text{th}}\text{ round}\)};
\node[draw,fill=gray!10, minimum width=5cm, minimum height=1.5cm, below=0.6cm of top] (outbox) {\(\)};
\node[box,below=0.6cm of outbox] (qnd) {QND};
\node[draw, fill=green!10, minimum width=2cm, minimum height=1.2cm] at ([xshift=-1.2cm]outbox.center) {Filtering};
\node[draw, fill=red!10, minimum width=2cm, minimum height=1.2cm,align=center] at ([xshift=1.2cm]outbox.center) (reject) {Full\\ Measurement};

\node[draw, fill=gray!10,minimum width=3cm, minimum height=1.5cm,left = 2.5cm of outbox](previous) {\( \text{previous }l_c \text{ rounds}\)};

\node[draw,fill=gray!10,right = 0.3cm of previous,yshift=0.4cm](keep){\textcolor{darkgreen}{keep}/\red{reject}};

\node[draw, fill=gray!10,minimum width=2cm, minimum height=1.5cm,right = 3cm of outbox](next) {\( \text{next rounds}\)};

\node[draw,fill=gray!10,left = 0.4cm of next,yshift=0.4cm](click){click/no click};

\node[box, left =2.5cm of qnd](n0){\(\{\vec{F}^{(0,\text{keep})}\}\)};
\node[left =0.5cm of n0](n0part){\(\nkzerocor\)};
\node[box, below=3cm of qnd] (discard) {\(\{\blue{F^{(1,\text{keep})}_{\mathrm{sc}}},F^{(1,\text{keep})}_{\mathrm{mc}},F^{(1,\text{keep})}_{\text{basis mismatch}}\} \)};
\node[box, below right=1cm and 1cm of qnd] (mc) {\(\{F^{(>1,\text{keep})}_{\mathrm{mc}},\red{\mathbb{I}-F^{(>1,\text{keep})}_{\mathrm{mc}}}\}\)};
\node[below =0.2cm of mc,xshift=-1cm](nmc){\(\netbcor\)};
\node[box, right=1cm of mc](mccomplete){\(\{\vec{G}^{(>1,\text{keep})}\}\)};
\node[box, below=2cm of discard, xshift=-0.7cm] (f) {\( \{\red{\ftxcor}, \blue{\ftzcor} \}\)};
\node[box, below left=1.5cm and 0.2cm of f] (ca1) {\(\{ {\GtXeqcor}, {\GtXneqcor}\} \)};
\node[dashedbox, below=4cm of f] (ca2) {\(\{ {\GtXeqcor}, {\GtXneqcor}\} \)};
\node[box, below right=4cm and 2cm of f] (ca3) {\(\{ {\GtZeqcor}, {\GtZneqcor}\} \)};
\node[below right=3cm and 2cm of f] (middle) {\(\)};
\node[below right=0.8cm and 0.2cm of f] (nk1) {\(\;\;\nkonebcor\)};
\node[below left=0.8cm and 0.2cm of f] (nx1) {\(\nxonebcor\;\;\)};
\node[below =0.8cm of discard,xshift=-0.4cm] (nx1) {\(\;\;\;\;\nonesccor\)};
\node[below =0.4cm of ca2](phase){\(\Nkonebcor\)};
\node[left =0.4cm of ca1](test){\(\Nxonebcor\)};

\node at (8,-4.2){\circletext{1}};
\node at (8,-6.6){\circletext{2}};
\node at (8,-8.3){\circletext{3}};
\node at (8,-13.6){\circletext{4}};
\node at (8,-16.8){\circletext{5}};
\node[above=0.5cm of ca3](EUR){\textcolor{blue}{EUR}};
\draw[dashed] (-5,-4.5) -- (8,-4.5);
\draw[dashed] (-5,-6.8) -- (8,-6.8);
\draw[dashed] (-5,-8.5) -- (8,-8.5);
\draw[dashed] (-5,-13.8) -- (8,-13.8);
\draw[dashed] (-5,-17) -- (8,-17);

\draw[arrow,blue] (qnd) -- node[right] {\( \black{\tilde{n}_{(1,\text{keep,click})}} \)} (discard);
\draw[arrow,brown] (qnd) -- node[right,yshift=0.2cm] {\( \black{\tilde{n}_{(>1,\text{keep,click})} }\)} (mc);
\draw[arrow] (qnd) -- node[below] {\( \black{\tilde{n}_{(0,\text{keep,click})} }\)} (n0);
\draw[arrow,blue] ([xshift=-1.7cm,yshift=0.3cm] discard.south) -- (f.north);
\draw[arrow,red] ([xshift=-0.2cm]mc.east) -- (mccomplete);
\draw[arrow] ([xshift=-0.7cm,yshift=0.3cm] mc.south) -- (nmc);
\draw[arrow] (n0) -- (n0part);
\draw[arrow,red] (f) -- (ca1);
\draw[dashed,->] (middle) -- (ca2);
\draw[dashed,->] (ca2) -- (phase);
\draw[arrow] (ca1) -- (test);

\draw[arrow]([xshift=-1cm,yshift=0.2cm]outbox.south) -- (qnd.north);
\draw[arrow,blue] (f) -- (ca3);
\draw[arrow](previous) --  (outbox);
\draw[arrow](outbox) -- (next);

\draw[thick] (0.05,-1.1) -- (-0.1,-2.15);

\end{tikzpicture}}
    \caption{
The estimation protocol for correlated detectors: 
The colors indicate which measurement is used to complete the measurement 
based on the outcomes from previous steps.}
    \label{fig:correlated_equivalent}
\end{figure}
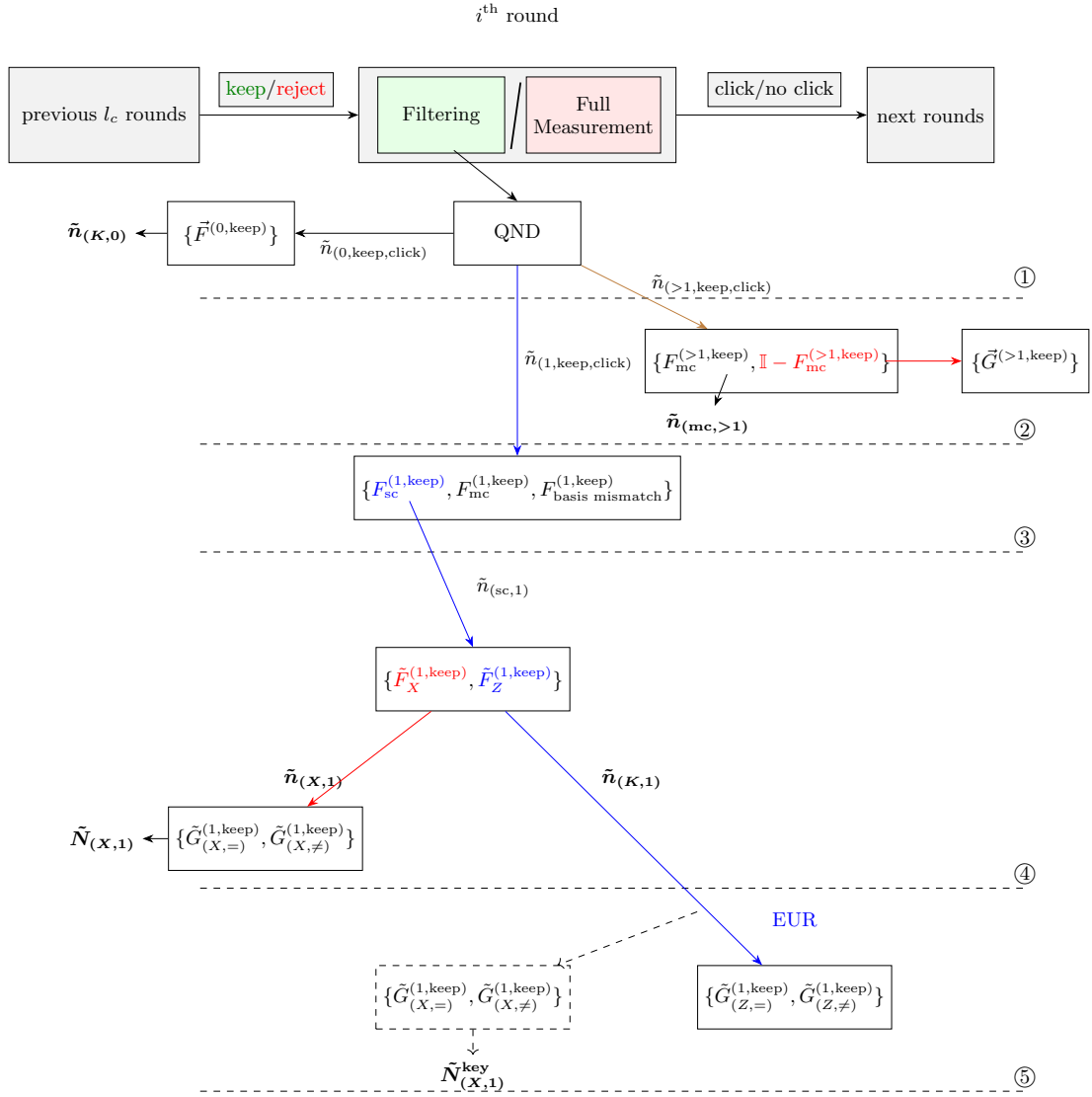 
\begin{remark}
    We add an extra superscript ``keep'' on the POVM elements to remind that these POVM elements are technically different from what we have in the memoryless case, since we construct them from the POVM \cref{eq:filteredPOVM} after the click/no-click filter.
\end{remark}

\begin{itemize}
    \item For single-photon rounds, one can bound the number of phase-error in single-photon rounds using \cref{lemma:NumberofErrorBound}. The only missing step is computing $\delta$ given some choice of $a$.
   We show in~\cref{appendix:commute} that \( \ftxcor ( \ftzcor )\) is equal to \( \ftxone ( \ftzone) \) in the memoryless case with $\pn=1$ for the model we consider. Therefore, we can use the same \( a \) and \( \delta \) computed for the memoryless case in~\cref{appendix:adelta}, as they depend only on \( \ftxone \) and \( \ftzone \). We use~\cref{lemma:NumberofErrorBound} and sum over all possible $n_{(\text{keep,click})}$. Then, the following relation holds:
    \begin{align}
    &\text{Pr}\biggl(
        \Nkonebcor \geq 
        \frac{
            \Nxonebcor 
            + \sqrt{-2 \log(\epazua^2) (\nxonebcor + \nkonebcor)} 
            + \delta (\nxonebcor + \nkonebcor)
        }{a} \notag \\
    &\hspace{30pt}
    + \sqrt{-2 \log(\epazub^2) (\nxonebcor + \nkonebcor)}
    \biggr) \notag \\
    &\leq \epazua^2 + \epazub^2. \label{NumberofErrorBound2}
\end{align}
    \item For multi-photon rounds,
    we show in \cref{appendix:commute} that $F_{\mathrm{mc}}^{(>1,\text{keep})}\geq \lm(\Pi_{(>1)}\Gamma_{\mathrm{mc}}\Pi_{(>1)})\Pi_{(>1)}$. Thus, we can use the same $\lm$ in the memoryless case to upper bound the minimum eigenvalue of $F_{\mathrm{mc}}^{(>1,\text{keep})}$.  We use \cref{lemma:multibound} and sum over all possible $n_{(\text{keep,click})}$. Then, the following relation holds:
        \begin{align}\label{step2multibound2}
        &\text{Pr}\biggl(\ngNbcor > \frac{\nmcbcor}{\lm}+\frac{\sqrt{(-\log(\eppntb))(-\log(\eppntb)+4\lm\nmcbcor)}}{2{\lm}^2}+\frac{(-2\log(\eppntb))}{4{\lm}^2}\biggr) \leq \eppntb^2\;.
    \end{align}
    \item For 0-photon rounds, we cannot use \cref{lemma:zerobound} since after the click/no-click filter, the POVM elements corresponding to 0-photon key rounds are no longer small. Instead, we apply Azuma's inequality to bound the number of key rounds from 0-photon signals $\nkzerocor$. We show that the following relation holds:
        \begin{align}\label{step3zerobound2}
        \text{Pr}\biggl(\nkzerocor \geq n\big(q_Z+2\gmzero(n,\epzero)\big)\biggr) \leq \epzero^2\;,
    \end{align}
    where $\gmzero(n,\epzero) = \sqrt{\frac{-\ln(\epzero)}{n}}$. We observe that we pick up a factor of 2 in front of the finite-size correction term compared to the memoryless case.
\end{itemize}
The detailed proofs of \cref{NumberofErrorBound2}, \cref{step2multibound2} and \cref{step3zerobound2} are in \cref{appendix:commute}.
Combining these three bounds together we obtain our goal:

\begin{align}
   & \text{Pr}\biggl(\errkscor \geq  \bimpcor (\nxbcor,\nkbcor,\Nxobsbcor,\nmcbcor) \lor \nkonebcor \leq {\Kcor}(\nkbcor,\nmcbcor)\biggr)\leq \epazua^2 + \epazub^2 +\eppntb^2+\epzero^2,
\end{align}
where 
\begin{align}
        &\Kcor(\nkbcor,\nmcbcor) = \notag\\ &\nkbcor - nq_Z -\frac{\nmcbcor}{\lm}-\frac{-2\log(\eppntb)}{4{\lm}^2}-2\sqrt{({-\ln(\epzero)})n}-\frac{\sqrt{(-\log(\eppntb))(-\log(\eppntb)+4\lm\nmcbcor)}}{2{\lm}^2}\\
    &\bimpcor(\nxbcor,\nkbcor,\Nxobsbcor,\nmcbcor)=\notag\\
    &\frac{\Nxobsbcor}{a\Kcor(\nkb,\nmcb)}+\sqrt{\frac{-2\ln(\epazua^2)}{a}\biggl(\frac{\nxbcor}{\Kcor(\nkbcor,\nmcbcor)^2}+\frac{1}{\Kcor(\nkbcor,\nmcbcor)}\biggr)}+\frac{\delta}{a}\biggl(\frac{\nxbcor}{\Kcor(\nkbcor,\nmcbcor)}+1\biggr)\notag\\&
    +\sqrt{{-2\ln(\epazub^2)}\biggl(\frac{\nxbcor}{\Kcor(\nkbcor,\nmcbcor)^2}+\frac{1}{\Kcor(\nkbcor,\nmcbcor)}\biggr)}\;.
\end{align}
The only difference compared to the memoryless case is the extra factor of 2 in front of the finite-size correction term in $\Kcor(\cdot)$. This arises due to the use of Azuma's inequality in this proof, as compared to the use of binomial tail bounds and Hoeffding's inequality in the memoryless case. 
For the decoy-state protocol, the decoy analysis~\cite{charles} we use in this work is independent of detection setups. Again, we can first bound the phase error rate and number of key rounds for those single-photon kept rounds in the in terms of quantities associated with Alice sending single-photon signals. Later, we estimate these single-photon quantities using decoy analysis.

We state the key rate formula for decoy-state BB84 with passive, imperfect, correlated detection setups without proof in \cref{eq:keyratecorrelated} in \cref{appendix:decoy}.

\section{Key Rates Simulation}\label{sec:results}
We compute key rates for the decoy-state protocol in the case where Alice uses a phase-randomized WCP source without imperfections and Bob uses passive detection setups with basis-efficiency mismatch. 
We refer the reader to~\cref{appendix:recipememoryless} and~\cref{appendix:recipecorrelated} for step-by-step recipes to compute the key rate in the memoryless and correlated cases, respectively. For plots, we perform the analysis in the case where the photon number cutoff $\pn=1$. The security parameters are set as follows: 
\( \epazua = \epazub = \eppntb = \epzero = \epatd = \epev = \eppa = 10^{-12} \), 
which leads to \( \epat = \sqrt{13} \times 10^{-12} \) and an overall security parameter of \( (2\sqrt{13} + 2) \times 10^{-12} \) from \cref{eq:epat}. In principle, our proof is valid when we choose completely different ranges of loss and dark count rates for each detector. However, for simplicity, we use the same range for all detectors and vary the width $\Delta$ of the range.
\begin{align}\label{eq:detectormodel}
    \begin{split}
    &\eta_i \in [0.7(1 - \Delta),0.7(1 + \Delta)],\\ 
    &d_i \in [10^{-6}(1 - \Delta),10^{-6}(1 + \Delta)],\\
    &s \in [0.3 (1- \Delta), 0.3 (1+ \Delta)] \quad \text{(Bob's beam splitter ratio to }X\text{ branch)},
    \end{split}
\end{align}
Alice chooses the $X$ basis with probability 0.3. The decoy intensities are selected with equal probabilities \( \frac{1}{3} \), optimizing over \( \mu_1 \) and \( \mu_2 \), while fixing \( \mu_3 = 0 \). The devices have a physical misalignment of \( 2^\circ \) \footnote{One can further optimize over parameters controlled by Bob, such as the beam splitter ratio \( s \), the decoy intensities \( \muk \), and the probabilities \( p_{\muk} \) associated with choosing each intensity. This is left for future work}. 

\subsection{Simulation with Imperfect Characterised Detectors}
\begin{figure}
    \centering
    \includegraphics[width=0.8\textwidth]{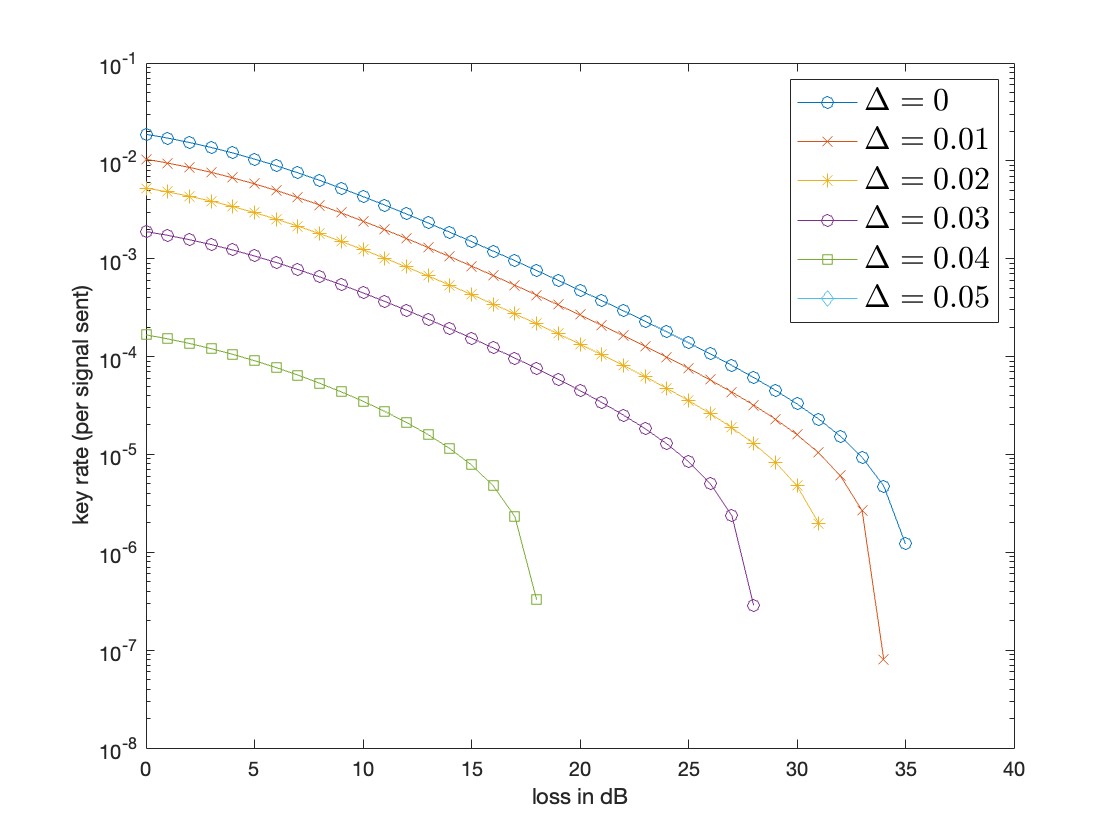}
    \caption{Finite-size key rates for decoy-state BB84 with memoryless passive detectors. The total number of signals sent is $10^{12}$. Alice chooses the $X$ basis with probability 0.3. The decoy intensities are selected with equal probabilities \( \frac{1}{3} \), optimizing over \( \mu_1 \) and \( \mu_2 \), while fixing \( \mu_3 = 0 \). Bob measures with partially characterised imperfect memoryless detectors as described in \cref{eq:detectormodel}. We increase $\Delta$ to increase the efficiency mismatch on the detectors.
    The devices have a physical misalignment of \( 2^\circ \). The $y$-axis is key rate (per signal sent). We do not have key rate when $\Delta=0.05$ because the efficiency mismatch is too large, causing the key rate to drop to zero.}
    \label{decoy}
\end{figure}
We present key rates for the decoy-state BB84 passive protocol from \cref{sec:basicprotdescription,subsec:partialimperfections} as a function of channel loss in \cref{decoy}. In particular, Bob uses memoryless passive detection setups with basis-efficiency mismatch.
We set $\Delta$ to be different values for simplicity in plotting. The total number of signals is \( n = 10^{12} \). 
We simulate the observations by computing the key rate for expected behavior (expected frequencies of outcomes) $\vec{F}^{\text{obs}}$ of the channel. We obtain the observations by multiplying the frequencies by \( n \). From~\cref{decoy}, key rates drop as the amount of basis-efficiency mismatch increases. We see that our analysis can tolerate a fair amount of basis-efficiency mismatch. 

\subsection{Including Correlations}
Next, we present key rates for the decoy protocol with correlated detectors. We set \( \Delta = 0.01 \).
\begin{figure}
    \centering
    \includegraphics[width=0.8\textwidth]{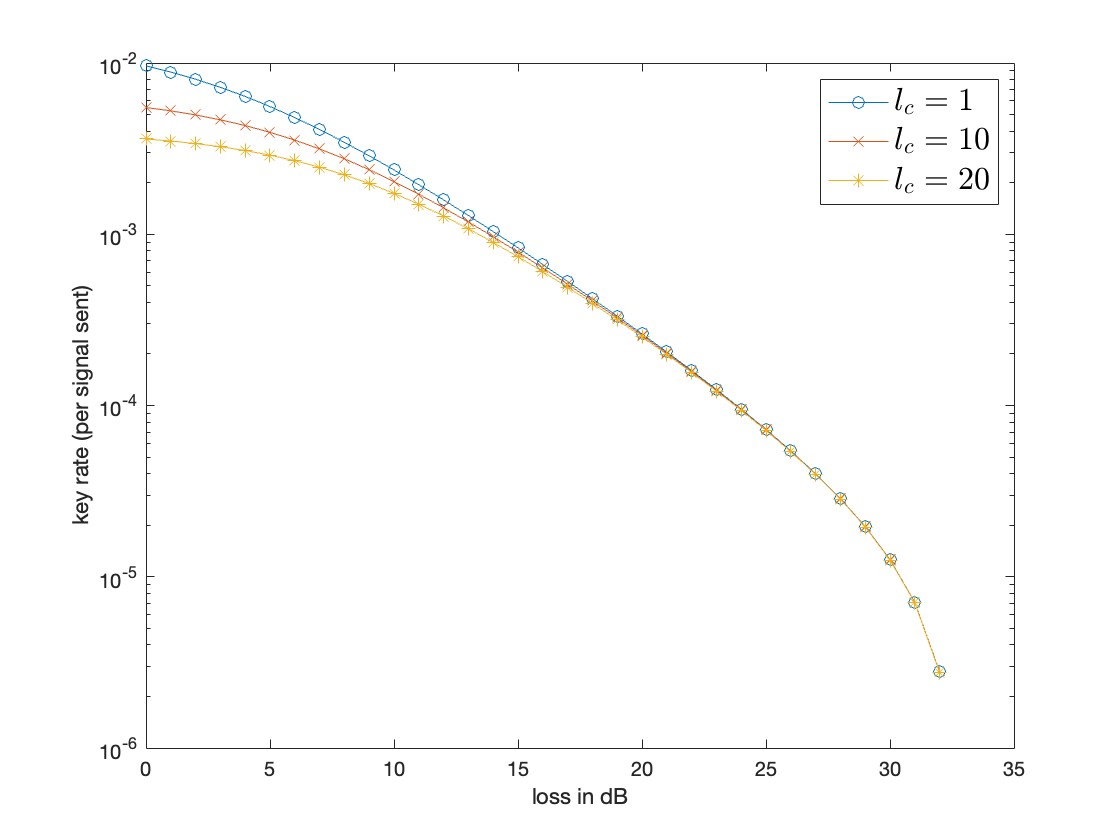}
    \caption{Finite-size key rates for decoy-state BB84 with passive correlated detectors. The total number of signals sent is $10^{12}$. Alice chooses the $X$ basis with probability 0.3. The decoy intensities are selected with equal probabilities \( \frac{1}{3} \), optimizing over \( \mu_1 \) and \( \mu_2 \), while fixing \( \mu_3 = 0 \). Bob measures with partially characterised imperfect correlated detectors as described in \cref{eq:detectormodel}. We fix $\Delta=0.01$ for efficiency mismatch on the detectors and vary the correlation length $l_c$.
    The devices have a physical misalignment of \( 2^\circ \). The $y$-axis is the key rate (per signal sent).}
    \label{fig:correlatedresult}
\end{figure}
For computing the observations from the protocol, recall that Bob obtains a kept outcome only when the previous $l_c$ rounds are no-click outcomes. 
For channel simulation, we assume that the probability of no-click $p_{\text{no click}}$ is independent of past round data. Thus, the probability that a given round is kept is $(p_{\text{no click}})^{l_c}$. We compute the observations by multiplying the frequencies \( \vec{F}^{\text{obs}} \) by \( n \times (p_{\text{no click}})^{l_c} \), where \( p_{\text{no click}} \) is the expected no-click frequency computed from \( \vec{F}^{\text{obs}} \), and $\vec{F}^{\text{obs}}$ is the expected frequencies obtained from the memoryless case. This is a heuristic model we use for simulating the channel and correlated detectors. A more thorough channel model that accurately captures the correlated effects on the statistics remains an open problem.
Note that this simplification is only undertaken for the channel simulation, and does not affect the rigor of our security proof in any way. 



We first compare key rates for different correlation lengths \( l_c \) in \cref{fig:correlatedresult} with $n={10}^{12}$. 
In the high-loss regime, most of the outcomes on Bob's side are no-click outcomes. 
As a result, Bob discards relatively few signals due to memory effects, and the key rates in this regime are therefore similar for different \( l_c \) values.

After looking at how \( l_c \) alone affects performance, 
we now examine the tradeoff between the number of rounds and the correlation length. 
Dead times and afterpulsing are common detector memory effects, 
and they last for a fixed amount of time. 
If the security proof does not handle memory effects, 
this time sets a limit on the repetition rate of the protocol. 
Our proof removes this restriction, 
allowing the protocol to run at higher repetition rates without compromising security.
To illustrate the potential impact on key generation, consider a toy protocol that runs for 100 seconds with a detector dead time of $10^{-9}$ seconds. Without a security proof that accounts for correlated detector effects, the maximum number of rounds that can be executed in this time is limited to no more than $100/10^{-9} = 10^{11}$ rounds at a repetition rate of $1$~GHz, as the detectors require time to reset between rounds. With our security proof, however, higher performance is possible, as shown in~\cref{fig:tradeoff}, where the protocol can run at a higher repetition rate but with a longer correlation length ($l_c=$ repetition rate $\times$ dead time). We see from~\cref{fig:tradeoff} that increasing the repetition rate can yield more key.
We leave a detailed analysis of this tradeoff, along with the development of a physical justification of the model, for future work.
\begin{figure} 
    \centering
    \includegraphics[width=0.8\textwidth]{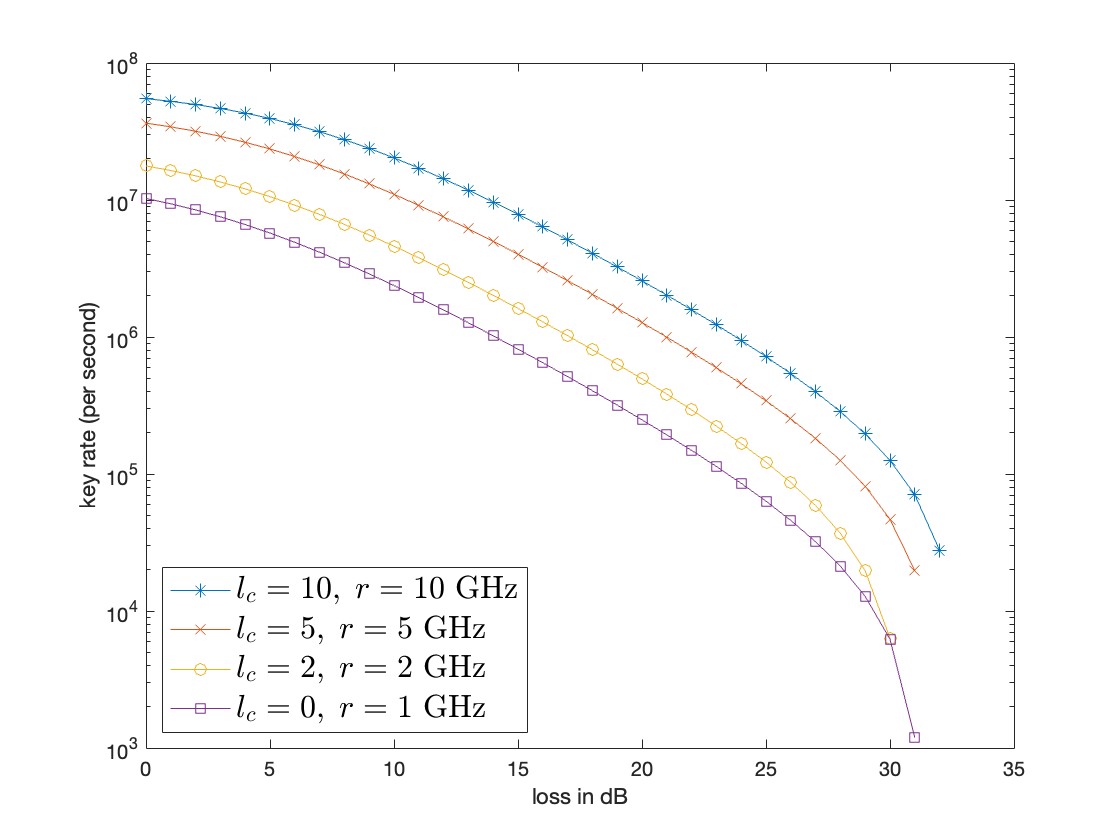}
\caption{Tradeoff between repetition rate $r$ and the correlation length for decoy-state BB84 with passive correlated detectors. We set $\Delta = 0.01$. The protocol runs for $100$ seconds, and the detector dead time is assumed to be $10^{-9}$ seconds. At a repetition rate of $1$~GHz, there are $10^{11}$ rounds and no correlation effects. Increasing the repetition rate $r$ to $2$~GHz, $5$~GHz, and $10$~GHz yields $2\times10^{11}$, $5\times10^{11}$, and $10^{12}$ rounds. These repetition rates result in correlation lengths of $2$, $5$, and $10$, respectively. The $y$-axis shows the key rate (per second).}
    \label{fig:tradeoff}
\end{figure}
\section{Other Extensions} \label{sec:extension}
\subsection{Compatibility with Source Imperfections}\label{sec:sourceimp}
So far, we have only discussed detector imperfections, assuming a perfect source. There exist several works addressing source imperfections within phase error estimation-based proofs~\cite{sourceimp1,sourceimp2,sourceimp3,sourceimp4,sourceimp5,sourceimp6,sourceimp7,sourceimp8}. The analysis for a single-photon source with memoryless imperfect detectors presented in~\cref{sec:phase error estimation} is compatible with source imperfections in the following way.
The key idea is that, if a bound on the phase error rate can be established for an imperfect single-photon source with perfect active detectors - such as the one provided in Ref.~\cite{sourceimp3} --- then, using~\cite[Theorem 1]{deviceimp}, this bound can be translated into a bound for the case of imperfect sources and imperfect active detectors.


In the memoryless case, when we restrict to \( \pn = 1 \), the photon-number estimation allows us to statistically argue that most of the signals received by Bob are single-photon signals. For these rounds, the beam splitter effectively performs an active basis choice on Bob's side. This property is essential for combining our results with source imperfection analysis via the method in Ref.~\cite{deviceimp}, and it allows us to directly apply results that were originally proved for active protocols.
However, it may be possible to modify our proof to more directly handle general passive protocols. We leave this modification for future work.

In the remainder of this section, we apply the method from Ref.~\cite{deviceimp}, originally developed for active detection setups, to the case of passive memoryless detection setups. We do it in 3 steps.

\begin{enumerate}
    \item We restrict our attention to the single-photon rounds using the analysis in~\cref{sec:phase error estimation} with $\pn=1$, where we argue that the measurements can be thought of as an active basis choice.
    \item Given a bound on the phase error rate in the single-photon rounds with an imperfect source and perfect passive detectors (in the sense that, for the single-photon rounds, perfect passive detectors can be viewed as perfect active detectors), we apply the result from Ref.~\cite{deviceimp} to obtain the bound on the phase error rate (see~\cref{eq2:objective}) in the single-photon rounds for the case of an imperfect source and imperfect detectors.
    \item Since we do not know the quantities explicitly in the single-photon rounds (e.g., $\bd{\tilde{n}_{(K,1)}}$), we perform a worst-case analysis on our estimates of them.
\end{enumerate}

We begin by arguing that measurements can be thought of as an active detection setups.
 When \( \pn = 1 \) and we examine the single-photon subspace after stage \circletext{2} 
in \cref{diag:equivalent}, the measurement used can be replaced 
with the one from Ref.~\cite[Theorem 1]{deviceimp}.
\begin{enumerate}
    \item In this case, the state first goes through a basis independent filter $\{\tilde{F},\idd-\tilde{F}\}$.
    \item For those rounds with outcome $\tilde{F}$, they perform a basis choice measurement $\{\ftxone,\ftzone,\idd-\ftxone-\ftzone\}$, where $\ftxone=p_X^{'}\idd$, $\ftzone=p_Z^{'}\idd$ and $p_X^{'}+p_Z^{'}\leq 1$.
    \item For those rounds with results $\ftxone$, Bob applies a $X$-basis dependent filter $\{\fcx,\idd-\fcx\}$. Then he measures with some POVM $\{\GtXeqone,\GtXneqone\}$.
    \item For those rounds with results $\ftzone$, Bob applies a $Z$-basis dependent filter $\{\fcz,\idd-\fcz\}$. Then he measures with some POVM $\{\GtXeqone,\GtXneqone\}$ to obtain phase errors.
\end{enumerate}
We consider the state \( \rho_{\condi{(\tilde{n}_{(1)})}} \), 
conditioned on the outcome of the QND measurement after stage \circletext{2} 
in the middle branch of \cref{diag:equivalent}. 
The new measurement mentioned above essentially corresponds to the active setup 
required by \cite[Theorem 1]{deviceimp}. 

For the second step,  
we consider the case in which a bound is provided on the phase error rate  
for the state \( \rho_{\condi{(\tilde{n}_{(0)},\tilde{n}_{(1)},\tilde{n}_{(>1)})}} \),  
assuming source imperfections and perfect active detection  
(perfect and active in the sense that \( \ftxone+\ftzone = \idd \) and \( \fcx = \fcz \),  
so that the basis-dependent filters \( \fcx \) and \( \fcz \) are no longer needed).
\begin{equation}
    \label{eq2:guarantee_BIDE}
    \Pr\biggl({\bd{{e}_{{\bm{(X,1)}}}^{\textbf{key}}}} > 	\Eindep(\bm{\vec{\tilde{n}}_{(X,1)}},\nksb)\biggr)_{\condi{(\tilde{n}_{(0)},\tilde{n}_{(1)},\tilde{n}_{(>1)})}} \leq \epssrc^2\;\;.
\end{equation}  
Here, $\bm{\vec{\tilde{n}}_{(X,1)}}$ is the vector containing all the statistics from single-photon, $X$-basis rounds, and  ${\bd{{e}_{{{(X,1)}}}^{\textbf{key}}}}:=\frac{\bd{N_{{(X,1)}}^{\textbf{key}}}}{\nksb}$.

Then, \cite[Theorem 1]{deviceimp} gives us an upper bound on phase error rate in single-photon rounds when there exist both source and detector imperfections:
\begin{equation}
    \label{eq2:objective}
    \Pr\left({\bd{{e}_{{{(X,1)}}}^{\textbf{key}}}} > \Edep\big(\bm{\vec{\tilde{n}}_{(X,1)}},\nksb\big)\right)_{\condi{(\tilde{n}_{(0)},\tilde{n}_{(1)},\tilde{n}_{(>1)})}} \leq \epssrc^2+\epsATb^2+\epsATc^2,
    \end{equation}
where
\begin{align}
    \begin{split}
&\Edep(\bm{\vec{\tilde{n}}_{(X,1)}},\nksb) = \frac{\max_{n \in \mathcal{W}_{\delta_2}(\nksb)}  \tilde{n}_1 \, \Eindep (\bm{\vec{\tilde{n}}_{(X,1)}},\tilde{n}_1)}{\nksb} + \frac{ \delta_1 + \gamma^{\epsATb}_{\mathrm{bin}}(\nksb, \delta_1)}{1-\delta_2-\gamma^{\epsATc}_{\mathrm{bin}}(\nksb, \delta_2)}\\
&\mathcal{W}_{\delta_2}(z) \coloneqq \mathcal{N}_{\leq} \left(\left\lfloor\frac{z}{1-\delta_2-\gamma^{\epsATc}_{\mathrm{bin}}(z, \delta_2)} \right \rfloor\right)\\
&\gamma^{\varepsilon}_{\mathrm{bin}}(n, \delta) := \min\left\{x \geq 0 : \sum_{i = \lfloor n(\delta+x) \rfloor}^n \binom{n}{i} \delta^i (1-\delta)^{n-i} \leq \varepsilon^2\right\}\\
& \delta_1:=\infnorm{(\idd_A\otimes\sqrt{\fcz})\GtXneqone(\idd_A\otimes\sqrt{\fcz}) - (\idd_A\otimes\sqrt{\fcx})\GtXneqone(\idd_A\otimes\sqrt{\fcx})}\\
&\delta_2:=\infnorm{\idd-\fcz}
  \end{split}
\end{align}
with $\mathcal{N}_{\leq}(z) \coloneqq \{0,1,...,z\}$ denoting the set of non-negative integers up to $z$; and $\gamma^{\varepsilon}_{\mathrm{bin}}$ is a function introducing finite-size deviation terms that approach zero as the number of key rounds $\nksb$ approaches infinity. One can sum over all possible ${\tilde{n}_{(0)},\tilde{n}_{(1)},\tilde{n}_{(>1)}}$ to remove the conditioning and obtain an upper bound on phase error rate for the single-photon rounds. 

Lastly, we need to find upper and lower bounds on $\bm{\vec{\tilde{n}}_{(X,1)}}$ and $\nksb$ in terms of observations. We already have both upper and lower bounds for $\nksb$. The upper bound is simply $\nkb>\nksb$. For the lower bound, the bounds on 0-photon key rounds \( \nkzero \) 
and the multi-photon rounds \( \ngoneb \) 
do not depend on the source. 
They are solely determined by the properties of Bob's POVM. 
Therefore, the bound on \( \nksb \) still holds with imperfect source:
\begin{equation}
    \text{Pr}\biggl(\nksb \leq \K(\nkb,\nmcb)\biggr) \leq \epzero^2+\eppntb^2
\end{equation}
For $\bm{\vec{\tilde{n}}_{\textbf(X,1)}}$, trivial upper bounds are given by $\bm{\vec{\tilde{n}}_{\textbf{X}}}$, the announced statistics from the testing rounds. Lower bounds can be obtained in a similar manner as in the memoryless case—by subtracting the estimated contributions from the 0-photon and multi-photon subspaces. 

Let us define $\mathcal{S}$ to be the set of values that satisfy the above bounds:
\begin{equation}
    \mathcal{S} (\vec{\tilde{n}}_{X},\nk):=\{(\vec{\tilde{n}}_{(X,1)},\nks)| \vec{\tilde{n}}_{(X,1)},\nks \text{ satisfy upper/lower bounds in terms of }\vec{\tilde{n}}_{X},\nk\}.
\end{equation}
Suppose we have the following:
\begin{align}
\text{Pr}\biggl((\bm{\vec{\tilde{n}}_{(X,1)}},\nksb)\notin  \mathcal{S}(\bm{\vec{\tilde{n}}_{{X}}},\nkb)\biggr) \leq \epsilon^2_{\text{bounds}}.
\end{align}
Then we can optimize over $\mathcal{S}$ to obtain a bound on the phase error rate in terms of observed quantities, as follows:
\begin{align}\label{eq:constraints}
\Pr&\left({\bd{{e}_{{{(X,1)}}}^{\textbf{key}}}} > \max\limits_{({\vec{\tilde{n}}_{(X,1)}},\nks)\in \mathcal{S}(\bm{\vec{\tilde{n}}_{{X}}},\nkb)}\;\Edep\big({\vec{\tilde{n}}_{(X,1)}},\nks\big)\lor \nksb \leq \K(\nkb,\nmcb)\right)\notag\\& \leq \epssrc^2+\epsATb^2+\epsATc^2 + \epsilon^2_{\text{bounds}},
\end{align}
where we combine \cref{eq2:objective} and \cref{eq:constraints} using the union bound (see \cref{appendix:combine bounds}). Then we obtain the bound on the phase error rate in the single-photon rounds and the number of single-photon key rounds \cref{eq:phaseerroridea}.

However, we are not able to easily make the proof for correlated detectors compatible with source imperfections using the approach from Ref.~\cite{deviceimp} in this work. In general, we can no longer replace the measurement in the single-photon subspace with an active-like one. Specifically, the active-like basis choice measurement $\{\ftxone, \ftzone, \idd - \ftxone - \ftzone\}$, where $\ftxone = p_X^{'} \idd$ and $\ftzone = p_Z^{'} \idd$ with $p_X^{'} + p_Z^{'} \leq 1$, requires that $\ftxone$ and $\ftzone$ include some no-click outcomes so that they are proportional to the identity operator. However, due to the additional click/no-click measurements introduced by the filter, no-click outcomes are eliminated from the filtered state. As a result, further work is required - such as proving a result analogous to Ref.~\cite{deviceimp} using Azuma's inequality instead of the Serfling bound - to handle this case rigorously.
\begin{remark}
To prove a similar statement about the phase error rate in single-photon rounds for the decoy-state protocol, one can use~\cite[Theorem 2]{deviceimp}. However, the decoy-state analysis used in that theorem differs from the one employed in this work. Therefore, for bounding the 0-photon and multi-photon subspaces, one may either use trivial upper bounds (as only upper bounds are required) or rederive a similar result using the decoy-state analysis presented in~\cite{deviceimp}.
\end{remark}
\subsection{On-the-fly Announcements}\label{sec:onthefly}
On-the-fly announcements \cite{iterativesifting,Pfister_2016} of click/no-click outcomes are naturally accommodated in the proof we used for correlated detectors. They can also be included in the analysis of memoryless detectors by simply setting $l_c=0$. In each round, Bob performs the filtering process and immediately announces whether the outcome is a click or a no-click. This process is illustrated in~\cref{fig:onthefly}. 
Since our analysis for the single-photon and multi-photon subspaces does not depend on how the global state \( \rho_{\text{filtered}} \) is generated, incorporating the announcement of click/no-click outcomes during the filtering process does not affect this part of the analysis. For the 0-photon key rounds, the analysis remains unchanged due to the nature of Azuma’s inequality, which conditions on previous outcomes—including the announcements. Therefore, the same analysis for correlated detectors can be applied directly to the filtered state \( \rho_{\text{filtered}} \).

However, we do not incorporate on-the-fly announcements of basis choices. Doing so would require moving the basis-choice measurement into the filtering process. As a result, the analysis on the filtered state would no longer be a straightforward extension of the correlated case. In particular, the multi-photon analysis — which connects multi-click outcomes to the number of multi-photon signals—becomes significantly more complicated. Additional work is needed to handle on-the-fly announcements of basis choices within this framework.
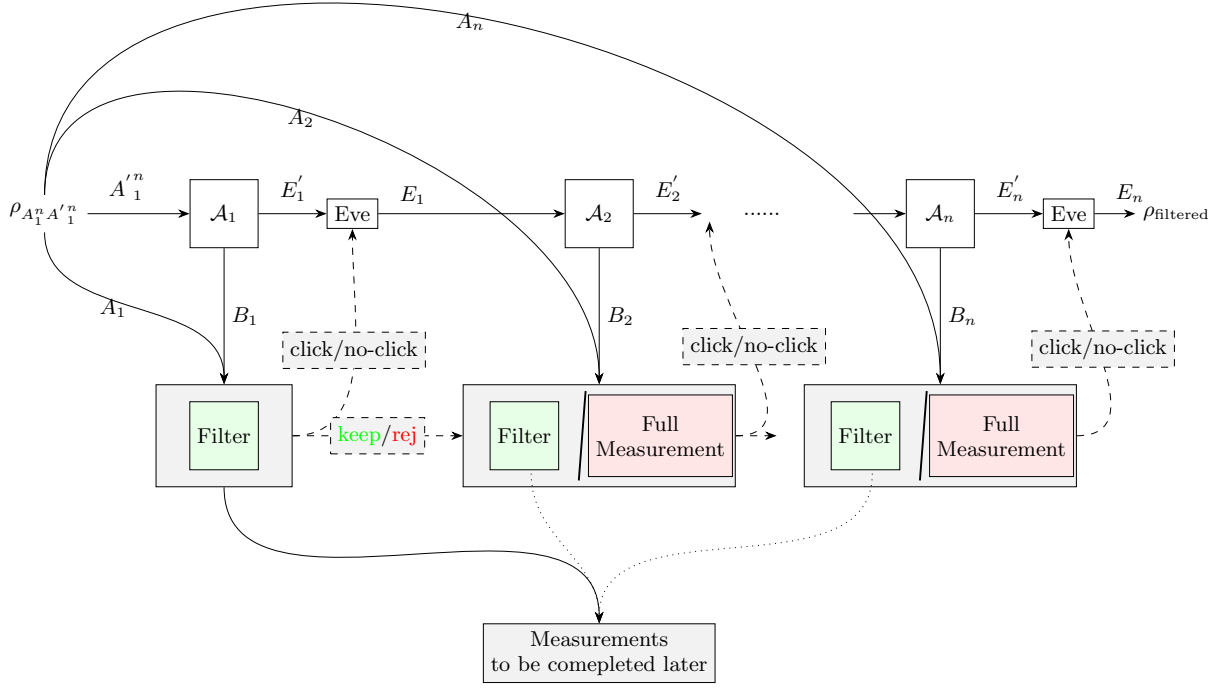
\begin{figure}[H]
    \centering
    \scalebox{0.9}{\begin{tikzpicture}
    [
    dashedbox/.style={draw, rectangle, minimum width=2.5cm, minimum height=1cm, dashed, align=center},
    box/.style={draw, rectangle, minimum width=2cm, minimum height=1cm, align=center},
    arrow/.style={-Stealth},
    discard/.style={->, thick, dashed},
    every node/.style={font=\small}
]
\node(rho0){\(\rho_{A_1^n {A^{'}}^n_1}\)};
\node[draw,minimum width=1cm, minimum height=1cm,right = 1.5cm of rho0](channel1){$\mathcal{A}_1$};
\node[draw,minimum width=1cm, minimum height=1cm,right = 4.5cm of channel1](channel2){$\mathcal{A}_2$};
\node[right = 1.5cm of channel2]{......};
\node[draw,minimum width=1cm, minimum height=1cm,right = 4cm of channel2](channeln){$\mathcal{A}_n$};
\node[draw,right = 1cm of channel1](Eve1){Eve};
\node[right = 1cm of channel2](Eve2){};
\node[draw,right = 1cm of channeln](Even){Eve};
\node[right = 0.6cm of Even](rhon){\(\rho_{\text{filtered}}\)};
\node[right = 0.4cm of Even,xshift=-0.2cm,yshift=0.3cm]{\(E_n\)};
\node[draw,fill=gray!10, minimum width=2cm, minimum height=1.5cm, below=2cm of channel1] (outbox1) {\(\)};
\node[draw, fill=green!10, minimum width=1cm, minimum height=1cm] at (outbox1.center){Filter};

\node[draw,fill=gray!10, minimum width=4cm, minimum height=1.5cm, below=2cm of channel2] (outbox2) {\(\)};
\node[draw, fill=green!10, minimum width=1cm, minimum height=1cm] at ([xshift=-1.1cm]outbox2.center){Filter};
\node[draw, fill=red!10, minimum width=1cm, minimum height=1.2cm,align=center] at ([xshift=0.9cm]outbox2.center) {Full\\ Measurement};
\draw[thick] ([xshift=-0.2cm,yshift=-0.1cm]outbox2.north) -- ([xshift=-0.3cm,yshift=0.1cm]outbox2.south);

\node[draw,fill=gray!10, minimum width=4cm, minimum height=1.5cm, below=2cm of channeln] (outboxn) {\(\)};
\node[draw, fill=green!10, minimum width=1cm, minimum height=1cm] at ([xshift=-1.1cm]outboxn.center){Filter};
\node[draw, fill=red!10, minimum width=1cm, minimum height=1.2cm,align=center] at ([xshift=0.9cm]outboxn.center) {Full\\ Measurement};
\draw[thick] ([xshift=-0.2cm,yshift=-0.1cm]outboxn.north) -- ([xshift=-0.3cm,yshift=0.1cm]outboxn.south);

\node[above right=0.01cm and 0.2cm of rho0,yshift=-0.2cm]{\({A^{'}}^{n}_1\)};
\draw[arrow] (rho0) -- (channel1);

\draw[arrow] (channel1) -- node[right]{\(B_1\)}(outbox1);
\node[above right=0.01cm and 0.2cm of channel1,yshift=-0.45cm]{\(E_{1}^{'}\)};

\draw[arrow](channel1) -- (Eve1);
\draw[dashed,arrow] (outbox1.east) to[out=0,in=-90] node[draw,fill=gray!10,yshift=0.2cm]{click/no-click}(Eve1.south);
\draw[arrow](Eve1) -- (channel2);
\node[above right=0.01cm and 0.2cm of Eve1,yshift=-0.2cm]{\(E_{1}\)};
\draw[dashed,arrow] (outbox1.east) -- node[draw,fill=gray!10]{\green{keep}/\red{rej}} (outbox2.west);

\draw[arrow] (channel2) -- node[right]{\(B_2\)}(outbox2);
\draw[arrow](channel2) -- (Eve2);
\node[above right=0.01cm and 0.2cm of channel2,yshift=-0.45cm]{\(E_{2}^{'}\)};
\draw[dashed,arrow] (outbox2.east) to[out=0,in=-90] node[draw,fill=gray!10,yshift=0.2cm]{click/no-click}(Eve2.south);
\draw[dashed,arrow] (outbox2.east) -- ([xshift=-0.4cm]outboxn.west);

\draw[arrow]([xshift=2cm]Eve2.east) -- (channeln);
\draw[arrow] (channeln) -- node[right]{\(B_n\)}(outboxn);
\draw[dashed,arrow] (outboxn.east) to[out=0,in=-90] node[draw,fill=gray!10,yshift=0.2cm]{click/no-click}(Even.south);
\draw[arrow](channeln) -- (Even);
\node[above right=0.01cm and 0.2cm of channeln,yshift=-0.45cm]{\(E_{n}^{'}\)};
\draw[arrow](Even) -- (rhon);

\node[draw,fill=gray!10,align=center,draw,below = 2cm of outbox2](later){Measurements\\to be comepleted later};

\draw[arrow](outbox1.south) to [out=-90,in=90](later.north);
\draw[dotted,arrow]([xshift=-1cm,yshift=0.2cm]outbox2.south) to [out=-90,in=90](later.north);
\draw[dotted,arrow]([xshift=-1cm,yshift=0.2cm]outboxn.south) to [out=-90,in=90](later.north);


\draw[arrow](rho0.south) to [out=-90,in=90]node[left]{\(A_1\)}(outbox1.north);

\draw[arrow](rho0.north) to [out=90,in=90]node[left]{\(A_2\)}(outbox2.north);

\draw[arrow](rho0.north) to [out=90,in=90]node[left]{\(A_n\)}(outboxn.north);
\end{tikzpicture}}
\caption{The filtering process with on-the-fly announcement of click/no-click outcomes. The map $\mathcal{A}_i$ represents the quantum channels that Eve forwards systems $B$ and $E$ to Bob and next Eve. System $A$ is in Alice's lab. In each round, Eve forwards $B$ to Bob, and Bob performs either the filter or the full measurement, depending on the click patterns observed in the previous $l_c$ rounds, to determine whether the current round results in a click or no-click. He then announces the click/no-click outcome. Eve may update her attack ($E' \rightarrow E$) based on this announcement.}
    \label{fig:onthefly}
\end{figure} 
\subsection{Side-channel Attacks and Multi-mode Detectors}
Regarding detector imperfections more generally, recall that our analysis remains valid when the efficiencies and dark count rates vary within known intervals (i.e., $\eta_i \in [\eta^l_i, \eta^u_i]$ and $d_i \in [d^l_i, d^u_i]$). Such variability may arise due to imperfect device characterization on Bob's side or as a result of active side-channel attacks launched by Eve \cite[Section VIII]{devEUR}. As long as these imperfections ultimately manifest in this bounded form, our analysis remains applicable.

Furthermore, the analysis can be readily extended to a simple multi-mode scenario, in which the POVM elements take the form:
\begin{align}
    \Gamma^{\text{multi}}_i = \bigoplus_{\textbf{d}} \Gamma_i(\eta_i(\textbf{d}), d_i(\textbf{d})),
\end{align}
where $\textbf{d}$ denotes a spatial-temporal mode, and $\Gamma_i(\eta_i(\textbf{d}), d_i(\textbf{d}))$ is the single-mode POVM element corresponding to that mode. This extension can be handled by computing the worst-case scenario over all modes. In particular, all we need to do is find the parameters \( a \), \( \delta \), \( q_Z \), and \( \lambda_{\min} \) from the POVMs \( \{ \vec{\Gamma}(\eta_i(\mathbf{d}), d_i(\mathbf{d})) \} \) for each mode \( \mathbf{d} \). These parameters can then be replaced by their worst-case values \footnote{For \( a \), both upper and lower bounds are required. For \( \delta \), \( q_Z \), and \( \lambda_{\min} \), only worst-case lower bounds are needed.}. This extension is the same as the one proposed in~\cite[Section VIII]{devEUR}.

\section{Conclusion}

In this work, we developed a framework to bound the phase error rate for QKD protocols with passive detection setups, in the presence of imperfections and memory effects.
Memory effects in detectors have so far not been addressed by \textit{any} proof technique prior to this work, while passive detection setups have not been addressed within phase error rate based techniques. Therefore, this work makes significant contributions towards developing security proofs for practical QKD implementations.
In particular, our results on memory effects have some interesting consequences for implementations. Prior to our work, the repetition rate of QKD protocols was limited to such that the detectors had no memory effects. However, our results show that the protocol can be run at higher repetition rates, leading to more secret key per unit time.

Our results allow for Bob’s detectors to be only partially characterized.
In particular, we extend the subspace estimation method from Ref.~\cite{lars} for computing $\lambda_{\min}(\Pign \Gamma_{\mathrm{mc}} \Pign)$ to the case of partially characterized detectors.
This extension is modular and can be applied to other proof techniques, such as EAT-based and postselection proofs. In fact, this result addresses a crucial missing step required to prove security for partially characterized detectors in other proof techniques \cite{shlok}.

Our results naturally extend to the convenient protocol variation which announces detection events during the quantum phase of the protocol. In practice, this enables real-time data processing and reduces the hardware storage requirements of QKD devices by allowing Alice and Bob to discard rounds that were not detected before the completion of the quantum phase of the protocol.

In addition, we extend the phase error estimation framework for memoryless detectors so that it is compatible with existing analyses of source imperfections. 
This allows works addressing source imperfections in the EUR/PEC framework to be combined with ours, enabling full security proofs for realistic QKD implementations.  
Taken together, these contributions significantly strengthen security proofs in both the EUR-based and PEC-based frameworks and improve the implementation security of QKD with trusted detectors.

There are a few points of interest for future work.
We consider one specific model for memory effects, in which whenever a detector clicks, the following \( l_c \) rounds exhibit memory effects. 
In this model, the correlation length \( l_c \) is finite, and the click outcome serves as an honest indicator of whether subsequent rounds are affected. 
Future work should examine whether this is the most appropriate model for memory effects, and explore alternative models informed by realistic detector behavior and experimental data. Further, our proof requires the protocol to discard rounds following a click outcome, independent of the strength of the correlations. It would be interesting to develop a more refined analysis that depends on the strength of the correlations, which could result in higher key rates for weak correlations. In addition, carrying out a complete experiment based on the proof developed in this work is an interesting open problem.

Moreover, our analysis on on-the-fly announcements can only accommodate detect/no-detect outcomes. It would be useful to extend this to account for basis choice announcements as well \cite{iterativesifting}.

We also provide one way to combine source imperfection analysis with our framework, using the method from Ref.~\cite{deviceimp}. 
However, Ref.~\cite{deviceimp} was developed for active protocols. 
Extending this approach to our correlated-detector analysis in passive setups is non-trivial, and could be addressed by developing an analogous result specifically for passive protocols.

Finally, we believe that a more careful analysis would result in tighter bounds using the same ideas presented in this work. Some examples of a more careful analysis would constitute using Kato's inequality instead of Azuma's inequality, bounding $\delta$ more directly instead of going through the simplifications described in \cref{appendix:adelta}, and using tighter tail bounds for the binomial distribution than Hoeffding's inequality.

\section{acknowledgements}
We thank Norbert Lütkenhaus for the discussion on this project.
We thank Masato Koashi for the useful discussion on the early stage of this project. We thank Lars Kamin for helpful discussions on subspace estimation.  
We also thank Aodhán Corrigan for comments on this manuscript, and Aodhán Corrigan and John Burniston for reviewing the code used to generate the plots for this project.
This work was funded by the NSERC Alliance QUINT and DND-MicroNet. It was conducted at the Institute for Quantum Computing, University of Waterloo, which is funded by the Government of Canada through ISED. DT is partially funded by the Mike and Ophelia Lazaridis Fellowship.
\bibliography{technical}

\begin{thebibliography}{59}%
\makeatletter
\providecommand \@ifxundefined [1]{%
 \@ifx{#1\undefined}
}%
\providecommand \@ifnum [1]{%
 \ifnum #1\expandafter \@firstoftwo
 \else \expandafter \@secondoftwo
 \fi
}%
\providecommand \@ifx [1]{%
 \ifx #1\expandafter \@firstoftwo
 \else \expandafter \@secondoftwo
 \fi
}%
\providecommand \natexlab [1]{#1}%
\providecommand \enquote  [1]{``#1''}%
\providecommand \bibnamefont  [1]{#1}%
\providecommand \bibfnamefont [1]{#1}%
\providecommand \citenamefont [1]{#1}%
\providecommand \href@noop [0]{\@secondoftwo}%
\providecommand \href [0]{\begingroup \@sanitize@url \@href}%
\providecommand \@href[1]{\@@startlink{#1}\@@href}%
\providecommand \@@href[1]{\endgroup#1\@@endlink}%
\providecommand \@sanitize@url [0]{\catcode `\\12\catcode `\$12\catcode `\&12\catcode `\#12\catcode `\^12\catcode `\_12\catcode `\%12\relax}%
\providecommand \@@startlink[1]{}%
\providecommand \@@endlink[0]{}%
\providecommand \url  [0]{\begingroup\@sanitize@url \@url }%
\providecommand \@url [1]{\endgroup\@href {#1}{\urlprefix }}%
\providecommand \urlprefix  [0]{URL }%
\providecommand \Eprint [0]{\href }%
\providecommand \doibase [0]{https://doi.org/}%
\providecommand \selectlanguage [0]{\@gobble}%
\providecommand \bibinfo  [0]{\@secondoftwo}%
\providecommand \bibfield  [0]{\@secondoftwo}%
\providecommand \translation [1]{[#1]}%
\providecommand \BibitemOpen [0]{}%
\providecommand \bibitemStop [0]{}%
\providecommand \bibitemNoStop [0]{.\EOS\space}%
\providecommand \EOS [0]{\spacefactor3000\relax}%
\providecommand \BibitemShut  [1]{\csname bibitem#1\endcsname}%
\let\auto@bib@innerbib\@empty
\bibitem [{\citenamefont {Lo}\ \emph {et~al.}(2012)\citenamefont {Lo}, \citenamefont {Curty},\ and\ \citenamefont {Qi}}]{Lo_2012}%
  \BibitemOpen
  \bibfield  {author} {\bibinfo {author} {\bibfnamefont {H.-K.}\ \bibnamefont {Lo}}, \bibinfo {author} {\bibfnamefont {M.}~\bibnamefont {Curty}},\ and\ \bibinfo {author} {\bibfnamefont {B.}~\bibnamefont {Qi}},\ }\bibfield  {journal} {\bibinfo  {journal} {Physical Review Letters}\ }\textbf {\bibinfo {volume} {108}},\ \href {https://doi.org/10.1103/physrevlett.108.130503} {10.1103/physrevlett.108.130503} (\bibinfo {year} {2012})\BibitemShut {NoStop}%
\bibitem [{\citenamefont {Tupkary}\ \emph {et~al.}(2025)\citenamefont {Tupkary}, \citenamefont {Tan}, \citenamefont {Nahar}, \citenamefont {Kamin},\ and\ \citenamefont {Lütkenhaus}}]{bdr2}%
  \BibitemOpen
  \bibfield  {author} {\bibinfo {author} {\bibfnamefont {D.}~\bibnamefont {Tupkary}}, \bibinfo {author} {\bibfnamefont {E.~Y.~Z.}\ \bibnamefont {Tan}}, \bibinfo {author} {\bibfnamefont {S.}~\bibnamefont {Nahar}}, \bibinfo {author} {\bibfnamefont {L.}~\bibnamefont {Kamin}},\ and\ \bibinfo {author} {\bibfnamefont {N.}~\bibnamefont {Lütkenhaus}},\ }\href {https://arxiv.org/abs/2502.10340} {\bibinfo {title} {Qkd security proofs for decoy-state bb84: protocol variations, proof techniques, gaps and limitations}} (\bibinfo {year} {2025}),\ \Eprint {https://arxiv.org/abs/2502.10340} {arXiv:2502.10340 [quant-ph]} \BibitemShut {NoStop}%
\bibitem [{Note1()}]{Note1}%
  \BibitemOpen
  \bibinfo {note} {These proof techniques can be used for idealised passive protocols where the basis-choice beam splitting ratio is set so that the protocol can be considered to be equivalent to an active protocol \cite {Fung_2011,Gittsovich_2014}. However, they are not robust to even infinitesimal deviations from this idealised setting.}\BibitemShut {Stop}%
\bibitem [{\citenamefont {Arqand}\ and\ \citenamefont {Tan}(2025)}]{MEAT}%
  \BibitemOpen
  \bibfield  {author} {\bibinfo {author} {\bibfnamefont {A.}~\bibnamefont {Arqand}}\ and\ \bibinfo {author} {\bibfnamefont {E.~Y.~Z.}\ \bibnamefont {Tan}},\ }\href {https://arxiv.org/abs/2502.02563} {\bibinfo {title} {Marginal-constrained entropy accumulation theorem}} (\bibinfo {year} {2025}),\ \Eprint {https://arxiv.org/abs/2502.02563} {arXiv:2502.02563 [quant-ph]} \BibitemShut {NoStop}%
\bibitem [{\citenamefont {Kamin}\ \emph {et~al.}(2024)\citenamefont {Kamin}, \citenamefont {Arqand}, \citenamefont {George}, \citenamefont {Lütkenhaus},\ and\ \citenamefont {Tan}}]{MEATapplication}%
  \BibitemOpen
  \bibfield  {author} {\bibinfo {author} {\bibfnamefont {L.}~\bibnamefont {Kamin}}, \bibinfo {author} {\bibfnamefont {A.}~\bibnamefont {Arqand}}, \bibinfo {author} {\bibfnamefont {I.}~\bibnamefont {George}}, \bibinfo {author} {\bibfnamefont {N.}~\bibnamefont {Lütkenhaus}},\ and\ \bibinfo {author} {\bibfnamefont {E.~Y.~Z.}\ \bibnamefont {Tan}},\ }\href {https://arxiv.org/abs/2406.10198} {\bibinfo {title} {Finite-size analysis of prepare-and-measure and decoy-state qkd via entropy accumulation}} (\bibinfo {year} {2024}),\ \Eprint {https://arxiv.org/abs/2406.10198} {arXiv:2406.10198 [quant-ph]} \BibitemShut {NoStop}%
\bibitem [{\citenamefont {Kamin}\ \emph {et~al.}(2025{\natexlab{a}})\citenamefont {Kamin}, \citenamefont {Burniston},\ and\ \citenamefont {Tan}}]{eat}%
  \BibitemOpen
  \bibfield  {author} {\bibinfo {author} {\bibfnamefont {L.}~\bibnamefont {Kamin}}, \bibinfo {author} {\bibfnamefont {J.}~\bibnamefont {Burniston}},\ and\ \bibinfo {author} {\bibfnamefont {E.~Y.~Z.}\ \bibnamefont {Tan}},\ }\href {https://arxiv.org/abs/2504.12248} {\bibinfo {title} {R\'enyi security framework against coherent attacks applied to decoy-state qkd}} (\bibinfo {year} {2025}{\natexlab{a}}),\ \Eprint {https://arxiv.org/abs/2504.12248} {arXiv:2504.12248 [quant-ph]} \BibitemShut {NoStop}%
\bibitem [{\citenamefont {Nahar}\ \emph {et~al.}(2024)\citenamefont {Nahar}, \citenamefont {Tupkary}, \citenamefont {Zhao}, \citenamefont {Lütkenhaus},\ and\ \citenamefont {Tan}}]{postselection}%
  \BibitemOpen
  \bibfield  {author} {\bibinfo {author} {\bibfnamefont {S.}~\bibnamefont {Nahar}}, \bibinfo {author} {\bibfnamefont {D.}~\bibnamefont {Tupkary}}, \bibinfo {author} {\bibfnamefont {Y.}~\bibnamefont {Zhao}}, \bibinfo {author} {\bibfnamefont {N.}~\bibnamefont {Lütkenhaus}},\ and\ \bibinfo {author} {\bibfnamefont {E.~Y.-Z.}\ \bibnamefont {Tan}},\ }\bibfield  {journal} {\bibinfo  {journal} {PRX Quantum}\ }\textbf {\bibinfo {volume} {5}},\ \href {https://doi.org/10.1103/prxquantum.5.040315} {10.1103/prxquantum.5.040315} (\bibinfo {year} {2024})\BibitemShut {NoStop}%
\bibitem [{\citenamefont {Kamin}\ \emph {et~al.}(2025{\natexlab{b}})\citenamefont {Kamin}, \citenamefont {Tupkary},\ and\ \citenamefont {Lütkenhaus}}]{larsPS}%
  \BibitemOpen
  \bibfield  {author} {\bibinfo {author} {\bibfnamefont {L.}~\bibnamefont {Kamin}}, \bibinfo {author} {\bibfnamefont {D.}~\bibnamefont {Tupkary}},\ and\ \bibinfo {author} {\bibfnamefont {N.}~\bibnamefont {Lütkenhaus}},\ }\href {https://arxiv.org/abs/2502.05382} {\bibinfo {title} {Improved finite-size effects in qkd protocols with applications to decoy-state qkd}} (\bibinfo {year} {2025}{\natexlab{b}}),\ \Eprint {https://arxiv.org/abs/2502.05382} {arXiv:2502.05382 [quant-ph]} \BibitemShut {NoStop}%
\bibitem [{Note2()}]{Note2}%
  \BibitemOpen
  \bibinfo {note} {Note that \cite [Section V. B.]{shlok} describes a proof sketch to address detector memory effects via entropy accumulation-based proof techniques.}\BibitemShut {Stop}%
\bibitem [{\citenamefont {Kamin}\ and\ \citenamefont {Lütkenhaus}(2024)}]{lars}%
  \BibitemOpen
  \bibfield  {author} {\bibinfo {author} {\bibfnamefont {L.}~\bibnamefont {Kamin}}\ and\ \bibinfo {author} {\bibfnamefont {N.}~\bibnamefont {Lütkenhaus}},\ }\bibfield  {journal} {\bibinfo  {journal} {Physical Review Research}\ }\textbf {\bibinfo {volume} {6}},\ \href {https://doi.org/10.1103/physrevresearch.6.043223} {10.1103/physrevresearch.6.043223} (\bibinfo {year} {2024})\BibitemShut {NoStop}%
\bibitem [{Note3()}]{Note3}%
  \BibitemOpen
  \bibinfo {note} {In our work, we more accurately bound the number of multi-photon rounds input into Bob's detection setup.}\BibitemShut {Stop}%
\bibitem [{\citenamefont {Nahar}\ \emph {et~al.}(2025)\citenamefont {Nahar}, \citenamefont {Tupkary},\ and\ \citenamefont {Lütkenhaus}}]{shlok}%
  \BibitemOpen
  \bibfield  {author} {\bibinfo {author} {\bibfnamefont {S.}~\bibnamefont {Nahar}}, \bibinfo {author} {\bibfnamefont {D.}~\bibnamefont {Tupkary}},\ and\ \bibinfo {author} {\bibfnamefont {N.}~\bibnamefont {Lütkenhaus}},\ }\href {https://arxiv.org/abs/2503.06328} {\bibinfo {title} {Imperfect detectors for adversarial tasks with applications to quantum key distribution}} (\bibinfo {year} {2025}),\ \Eprint {https://arxiv.org/abs/2503.06328} {arXiv:2503.06328 [quant-ph]} \BibitemShut {NoStop}%
\bibitem [{\citenamefont {Tamaki}\ \emph {et~al.}(2017)\citenamefont {Tamaki}, \citenamefont {Lo}, \citenamefont {Mizutani}, \citenamefont {Kato}, \citenamefont {Lim}, \citenamefont {Azuma},\ and\ \citenamefont {Curty}}]{iterativesifting}%
  \BibitemOpen
  \bibfield  {author} {\bibinfo {author} {\bibfnamefont {K.}~\bibnamefont {Tamaki}}, \bibinfo {author} {\bibfnamefont {H.-K.}\ \bibnamefont {Lo}}, \bibinfo {author} {\bibfnamefont {A.}~\bibnamefont {Mizutani}}, \bibinfo {author} {\bibfnamefont {G.}~\bibnamefont {Kato}}, \bibinfo {author} {\bibfnamefont {C.~C.~W.}\ \bibnamefont {Lim}}, \bibinfo {author} {\bibfnamefont {K.}~\bibnamefont {Azuma}},\ and\ \bibinfo {author} {\bibfnamefont {M.}~\bibnamefont {Curty}},\ }\href {https://doi.org/10.1088/2058-9565/aa89bd} {\bibfield  {journal} {\bibinfo  {journal} {Quantum Science and Technology}\ }\textbf {\bibinfo {volume} {3}},\ \bibinfo {pages} {014002} (\bibinfo {year} {2017})}\BibitemShut {NoStop}%
\bibitem [{\citenamefont {Pfister}\ \emph {et~al.}(2016)\citenamefont {Pfister}, \citenamefont {Lütkenhaus}, \citenamefont {Wehner},\ and\ \citenamefont {Coles}}]{Pfister_2016}%
  \BibitemOpen
  \bibfield  {author} {\bibinfo {author} {\bibfnamefont {C.}~\bibnamefont {Pfister}}, \bibinfo {author} {\bibfnamefont {N.}~\bibnamefont {Lütkenhaus}}, \bibinfo {author} {\bibfnamefont {S.}~\bibnamefont {Wehner}},\ and\ \bibinfo {author} {\bibfnamefont {P.~J.}\ \bibnamefont {Coles}},\ }\href {https://doi.org/10.1088/1367-2630/18/5/053001} {\bibfield  {journal} {\bibinfo  {journal} {New Journal of Physics}\ }\textbf {\bibinfo {volume} {18}},\ \bibinfo {pages} {053001} (\bibinfo {year} {2016})}\BibitemShut {NoStop}%
\bibitem [{\citenamefont {Currás-Lorenzo}\ \emph {et~al.}(2025{\natexlab{a}})\citenamefont {Currás-Lorenzo}, \citenamefont {Pereira}, \citenamefont {Nahar},\ and\ \citenamefont {Tupkary}}]{deviceimp}%
  \BibitemOpen
  \bibfield  {author} {\bibinfo {author} {\bibfnamefont {G.}~\bibnamefont {Currás-Lorenzo}}, \bibinfo {author} {\bibfnamefont {M.}~\bibnamefont {Pereira}}, \bibinfo {author} {\bibfnamefont {S.}~\bibnamefont {Nahar}},\ and\ \bibinfo {author} {\bibfnamefont {D.}~\bibnamefont {Tupkary}},\ }\href {https://arxiv.org/abs/2507.03549} {\bibinfo {title} {Security of quantum key distribution with source and detector imperfections through phase-error estimation}} (\bibinfo {year} {2025}{\natexlab{a}}),\ \Eprint {https://arxiv.org/abs/2507.03549} {arXiv:2507.03549 [quant-ph]} \BibitemShut {NoStop}%
\bibitem [{\citenamefont {Kawakami}\ \emph {et~al.}(2025)\citenamefont {Kawakami}, \citenamefont {Taniguchi}, \citenamefont {Tonomura}, \citenamefont {Takasugi},\ and\ \citenamefont {Azuma}}]{shun}%
  \BibitemOpen
  \bibfield  {author} {\bibinfo {author} {\bibfnamefont {S.}~\bibnamefont {Kawakami}}, \bibinfo {author} {\bibfnamefont {A.}~\bibnamefont {Taniguchi}}, \bibinfo {author} {\bibfnamefont {Y.}~\bibnamefont {Tonomura}}, \bibinfo {author} {\bibfnamefont {K.}~\bibnamefont {Takasugi}},\ and\ \bibinfo {author} {\bibfnamefont {K.}~\bibnamefont {Azuma}},\ }\href {https://arxiv.org/abs/2507.04248} {\bibinfo {title} {Security of the bb84 protocol with receiver's passive biased basis choice}} (\bibinfo {year} {2025}),\ \Eprint {https://arxiv.org/abs/2507.04248} {arXiv:2507.04248 [quant-ph]} \BibitemShut {NoStop}%
\bibitem [{Note4()}]{Note4}%
  \BibitemOpen
  \bibinfo {note} {Here, there is some ambiguity in how to map double-clicks to bits. A common choice is to assign them randomly to $0$ or $1$.}\BibitemShut {Stop}%
\bibitem [{Note5()}]{Note5}%
  \BibitemOpen
  \bibinfo {note} {This analysis becomes slightly more involved in the presence of imperfect detectors \cite {devEUR}, but the underlying intuition remains the same.}\BibitemShut {Stop}%
\bibitem [{\citenamefont {Serfling}(1974)}]{serfling_probability_1974}%
  \BibitemOpen
  \bibfield  {author} {\bibinfo {author} {\bibfnamefont {R.~J.}\ \bibnamefont {Serfling}},\ }\href {https://www.jstor.org/stable/2958379?seq=1} {\bibfield  {journal} {\bibinfo  {journal} {The Annals of Statistics}\ ,\ \bibinfo {pages} {39}} (\bibinfo {year} {1974})}\BibitemShut {NoStop}%
\bibitem [{\citenamefont {Zhang}\ \emph {et~al.}(2021)\citenamefont {Zhang}, \citenamefont {Coles}, \citenamefont {Winick}, \citenamefont {Lin},\ and\ \citenamefont {Lütkenhaus}}]{fss}%
  \BibitemOpen
  \bibfield  {author} {\bibinfo {author} {\bibfnamefont {Y.}~\bibnamefont {Zhang}}, \bibinfo {author} {\bibfnamefont {P.~J.}\ \bibnamefont {Coles}}, \bibinfo {author} {\bibfnamefont {A.}~\bibnamefont {Winick}}, \bibinfo {author} {\bibfnamefont {J.}~\bibnamefont {Lin}},\ and\ \bibinfo {author} {\bibfnamefont {N.}~\bibnamefont {Lütkenhaus}},\ }\bibfield  {journal} {\bibinfo  {journal} {Physical Review Research}\ }\textbf {\bibinfo {volume} {3}},\ \href {https://doi.org/10.1103/physrevresearch.3.013076} {10.1103/physrevresearch.3.013076} (\bibinfo {year} {2021})\BibitemShut {NoStop}%
\bibitem [{\citenamefont {Li}\ and\ \citenamefont {Lütkenhaus}(2020)}]{nicky}%
  \BibitemOpen
  \bibfield  {author} {\bibinfo {author} {\bibfnamefont {N.~K.~H.}\ \bibnamefont {Li}}\ and\ \bibinfo {author} {\bibfnamefont {N.}~\bibnamefont {Lütkenhaus}},\ }\bibfield  {journal} {\bibinfo  {journal} {Physical Review Research}\ }\textbf {\bibinfo {volume} {2}},\ \href {https://doi.org/10.1103/physrevresearch.2.043172} {10.1103/physrevresearch.2.043172} (\bibinfo {year} {2020})\BibitemShut {NoStop}%
\bibitem [{\citenamefont {Nahar}\ \emph {et~al.}(2023)\citenamefont {Nahar}, \citenamefont {Upadhyaya},\ and\ \citenamefont {Lütkenhaus}}]{nahar_imperfect_2023}%
  \BibitemOpen
  \bibfield  {author} {\bibinfo {author} {\bibfnamefont {S.}~\bibnamefont {Nahar}}, \bibinfo {author} {\bibfnamefont {T.}~\bibnamefont {Upadhyaya}},\ and\ \bibinfo {author} {\bibfnamefont {N.}~\bibnamefont {Lütkenhaus}},\ }\href {https://doi.org/10.1103/PhysRevApplied.20.064031} {\bibfield  {journal} {\bibinfo  {journal} {Physical Review Applied}\ }\textbf {\bibinfo {volume} {20}},\ \bibinfo {pages} {064031} (\bibinfo {year} {2023})},\ \bibinfo {note} {publisher: American Physical Society}\BibitemShut {NoStop}%
\bibitem [{Note6()}]{Note6}%
  \BibitemOpen
  \bibinfo {note} {As we see in \protect \cref {sec:multibound}, we want to compute the minimum eigenvalue for the joint multi-click POVM element. But since we use all multi-click outcomes for the subspace estimation, the corresponding multi-click POVM element has an identity on Alice's system. Therefore, we can focus on Bob's multi-click POVM only. We drop the superscript $B$ for simplicity in this section}\BibitemShut {NoStop}%
\bibitem [{\citenamefont {Curty}\ \emph {et~al.}(2004)\citenamefont {Curty}, \citenamefont {Lewenstein},\ and\ \citenamefont {Lütkenhaus}}]{sourcerep1}%
  \BibitemOpen
  \bibfield  {author} {\bibinfo {author} {\bibfnamefont {M.}~\bibnamefont {Curty}}, \bibinfo {author} {\bibfnamefont {M.}~\bibnamefont {Lewenstein}},\ and\ \bibinfo {author} {\bibfnamefont {N.}~\bibnamefont {Lütkenhaus}},\ }\bibfield  {journal} {\bibinfo  {journal} {Physical Review Letters}\ }\textbf {\bibinfo {volume} {92}},\ \href {https://doi.org/10.1103/physrevlett.92.217903} {10.1103/physrevlett.92.217903} (\bibinfo {year} {2004})\BibitemShut {NoStop}%
\bibitem [{\citenamefont {Ferenczi}\ and\ \citenamefont {Lütkenhaus}(2012)}]{sourcerep2}%
  \BibitemOpen
  \bibfield  {author} {\bibinfo {author} {\bibfnamefont {A.}~\bibnamefont {Ferenczi}}\ and\ \bibinfo {author} {\bibfnamefont {N.}~\bibnamefont {Lütkenhaus}},\ }\bibfield  {journal} {\bibinfo  {journal} {Physical Review A}\ }\textbf {\bibinfo {volume} {85}},\ \href {https://doi.org/10.1103/physreva.85.052310} {10.1103/physreva.85.052310} (\bibinfo {year} {2012})\BibitemShut {NoStop}%
\bibitem [{Note7()}]{Note7}%
  \BibitemOpen
  \bibinfo {note} {The source is still assumed to be a perfect phase-randomized WCP source. For compatibility with source imperfections, see~\protect \cref {sec:sourceimp}.}\BibitemShut {Stop}%
\bibitem [{\citenamefont {Hwang}(2003)}]{hwang2003}%
  \BibitemOpen
  \bibfield  {author} {\bibinfo {author} {\bibfnamefont {W.-Y.}\ \bibnamefont {Hwang}},\ }\bibfield  {journal} {\bibinfo  {journal} {Physical Review Letters}\ }\textbf {\bibinfo {volume} {91}},\ \href {https://doi.org/10.1103/physrevlett.91.057901} {10.1103/physrevlett.91.057901} (\bibinfo {year} {2003})\BibitemShut {NoStop}%
\bibitem [{\citenamefont {Lo}\ \emph {et~al.}(2005)\citenamefont {Lo}, \citenamefont {Ma},\ and\ \citenamefont {Chen}}]{lo2005}%
  \BibitemOpen
  \bibfield  {author} {\bibinfo {author} {\bibfnamefont {H.-K.}\ \bibnamefont {Lo}}, \bibinfo {author} {\bibfnamefont {X.}~\bibnamefont {Ma}},\ and\ \bibinfo {author} {\bibfnamefont {K.}~\bibnamefont {Chen}},\ }\bibfield  {journal} {\bibinfo  {journal} {Physical Review Letters}\ }\textbf {\bibinfo {volume} {94}},\ \href {https://doi.org/10.1103/physrevlett.94.230504} {10.1103/physrevlett.94.230504} (\bibinfo {year} {2005})\BibitemShut {NoStop}%
\bibitem [{\citenamefont {Wang}(2005)}]{XBWang}%
  \BibitemOpen
  \bibfield  {author} {\bibinfo {author} {\bibfnamefont {X.-B.}\ \bibnamefont {Wang}},\ }\href {https://doi.org/10.1103/PhysRevLett.94.230503} {\bibfield  {journal} {\bibinfo  {journal} {Phys. Rev. Lett.}\ }\textbf {\bibinfo {volume} {94}},\ \bibinfo {pages} {230503} (\bibinfo {year} {2005})}\BibitemShut {NoStop}%
\bibitem [{\citenamefont {Ma}\ \emph {et~al.}(2005)\citenamefont {Ma}, \citenamefont {Qi}, \citenamefont {Zhao},\ and\ \citenamefont {Lo}}]{ma2005}%
  \BibitemOpen
  \bibfield  {author} {\bibinfo {author} {\bibfnamefont {X.}~\bibnamefont {Ma}}, \bibinfo {author} {\bibfnamefont {B.}~\bibnamefont {Qi}}, \bibinfo {author} {\bibfnamefont {Y.}~\bibnamefont {Zhao}},\ and\ \bibinfo {author} {\bibfnamefont {H.-K.}\ \bibnamefont {Lo}},\ }\bibfield  {journal} {\bibinfo  {journal} {Physical Review A}\ }\textbf {\bibinfo {volume} {72}},\ \href {https://doi.org/10.1103/physreva.72.012326} {10.1103/physreva.72.012326} (\bibinfo {year} {2005})\BibitemShut {NoStop}%
\bibitem [{\citenamefont {Hayashi}\ and\ \citenamefont {Nakayama}(2014)}]{hayashi2014}%
  \BibitemOpen
  \bibfield  {author} {\bibinfo {author} {\bibfnamefont {M.}~\bibnamefont {Hayashi}}\ and\ \bibinfo {author} {\bibfnamefont {R.}~\bibnamefont {Nakayama}},\ }\href {https://doi.org/10.1088/1367-2630/16/6/063009} {\bibfield  {journal} {\bibinfo  {journal} {New Journal of Physics}\ }\textbf {\bibinfo {volume} {16}},\ \bibinfo {pages} {063009} (\bibinfo {year} {2014})}\BibitemShut {NoStop}%
\bibitem [{\citenamefont {Curty}\ \emph {et~al.}(2014)\citenamefont {Curty}, \citenamefont {Xu}, \citenamefont {Cui}, \citenamefont {Lim}, \citenamefont {Tamaki},\ and\ \citenamefont {Lo}}]{curty2014}%
  \BibitemOpen
  \bibfield  {author} {\bibinfo {author} {\bibfnamefont {M.}~\bibnamefont {Curty}}, \bibinfo {author} {\bibfnamefont {F.}~\bibnamefont {Xu}}, \bibinfo {author} {\bibfnamefont {W.}~\bibnamefont {Cui}}, \bibinfo {author} {\bibfnamefont {C.~C.~W.}\ \bibnamefont {Lim}}, \bibinfo {author} {\bibfnamefont {K.}~\bibnamefont {Tamaki}},\ and\ \bibinfo {author} {\bibfnamefont {H.-K.}\ \bibnamefont {Lo}},\ }\bibfield  {journal} {\bibinfo  {journal} {Nature Communications}\ }\textbf {\bibinfo {volume} {5}},\ \href {https://doi.org/10.1038/ncomms4732} {10.1038/ncomms4732} (\bibinfo {year} {2014})\BibitemShut {NoStop}%
\bibitem [{\citenamefont {Lim}\ \emph {et~al.}(2014)\citenamefont {Lim}, \citenamefont {Curty}, \citenamefont {Walenta}, \citenamefont {Xu},\ and\ \citenamefont {Zbinden}}]{charles}%
  \BibitemOpen
  \bibfield  {author} {\bibinfo {author} {\bibfnamefont {C.~C.~W.}\ \bibnamefont {Lim}}, \bibinfo {author} {\bibfnamefont {M.}~\bibnamefont {Curty}}, \bibinfo {author} {\bibfnamefont {N.}~\bibnamefont {Walenta}}, \bibinfo {author} {\bibfnamefont {F.}~\bibnamefont {Xu}},\ and\ \bibinfo {author} {\bibfnamefont {H.}~\bibnamefont {Zbinden}},\ }\bibfield  {journal} {\bibinfo  {journal} {Physical Review A}\ }\textbf {\bibinfo {volume} {89}},\ \href {https://doi.org/10.1103/physreva.89.022307} {10.1103/physreva.89.022307} (\bibinfo {year} {2014})\BibitemShut {NoStop}%
\bibitem [{\citenamefont {Tomamichel}\ and\ \citenamefont {Leverrier}(2017)}]{tomamichel2017}%
  \BibitemOpen
  \bibfield  {author} {\bibinfo {author} {\bibfnamefont {M.}~\bibnamefont {Tomamichel}}\ and\ \bibinfo {author} {\bibfnamefont {A.}~\bibnamefont {Leverrier}},\ }\href {https://doi.org/10.22331/q-2017-07-14-14} {\bibfield  {journal} {\bibinfo  {journal} {Quantum}\ }\textbf {\bibinfo {volume} {1}},\ \bibinfo {pages} {14} (\bibinfo {year} {2017})}\BibitemShut {NoStop}%
\bibitem [{\citenamefont {Tomamichel}(2016)}]{tomamichelbook}%
  \BibitemOpen
  \bibfield  {author} {\bibinfo {author} {\bibfnamefont {M.}~\bibnamefont {Tomamichel}},\ }\href {https://doi.org/10.1007/978-3-319-21891-5} {\emph {\bibinfo {title} {Quantum Information Processing with Finite Resources}}}\ (\bibinfo  {publisher} {Springer International Publishing},\ \bibinfo {year} {2016})\BibitemShut {NoStop}%
\bibitem [{\citenamefont {Tupkary}\ \emph {et~al.}(2024)\citenamefont {Tupkary}, \citenamefont {Nahar}, \citenamefont {Sinha},\ and\ \citenamefont {Lütkenhaus}}]{devEUR}%
  \BibitemOpen
  \bibfield  {author} {\bibinfo {author} {\bibfnamefont {D.}~\bibnamefont {Tupkary}}, \bibinfo {author} {\bibfnamefont {S.}~\bibnamefont {Nahar}}, \bibinfo {author} {\bibfnamefont {P.}~\bibnamefont {Sinha}},\ and\ \bibinfo {author} {\bibfnamefont {N.}~\bibnamefont {Lütkenhaus}},\ }\href {https://arxiv.org/abs/2408.17349} {\bibinfo {title} {Phase error rate estimation in qkd with imperfect detectors}} (\bibinfo {year} {2024}),\ \Eprint {https://arxiv.org/abs/2408.17349} {arXiv:2408.17349 [quant-ph]} \BibitemShut {NoStop}%
\bibitem [{\citenamefont {Tomamichel}\ and\ \citenamefont {Renner}(2011)}]{EUR1}%
  \BibitemOpen
  \bibfield  {author} {\bibinfo {author} {\bibfnamefont {M.}~\bibnamefont {Tomamichel}}\ and\ \bibinfo {author} {\bibfnamefont {R.}~\bibnamefont {Renner}},\ }\bibfield  {journal} {\bibinfo  {journal} {Physical Review Letters}\ }\textbf {\bibinfo {volume} {106}},\ \href {https://doi.org/10.1103/physrevlett.106.110506} {10.1103/physrevlett.106.110506} (\bibinfo {year} {2011})\BibitemShut {NoStop}%
\bibitem [{\citenamefont {Azuma}(1967)}]{azuma}%
  \BibitemOpen
  \bibfield  {author} {\bibinfo {author} {\bibfnamefont {K.}~\bibnamefont {Azuma}},\ }\href@noop {} {\bibfield  {journal} {\bibinfo  {journal} {Tohoku Mathematical Journal, Second Series}\ }\textbf {\bibinfo {volume} {19}},\ \bibinfo {pages} {357} (\bibinfo {year} {1967})}\BibitemShut {NoStop}%
\bibitem [{\citenamefont {Kato}(2020)}]{kato}%
  \BibitemOpen
  \bibfield  {author} {\bibinfo {author} {\bibfnamefont {G.}~\bibnamefont {Kato}},\ }\href {https://arxiv.org/abs/2002.04357} {\bibfield  {journal} {\bibinfo  {journal} {arXiv preprint arXiv:2002.04357}\ } (\bibinfo {year} {2020})}\BibitemShut {NoStop}%
\bibitem [{\citenamefont {Currás-Lorenzo}\ \emph {et~al.}(2021)\citenamefont {Currás-Lorenzo}, \citenamefont {Navarrete}, \citenamefont {Azuma}, \citenamefont {Kato}, \citenamefont {Curty},\ and\ \citenamefont {Razavi}}]{curras2021tight}%
  \BibitemOpen
  \bibfield  {author} {\bibinfo {author} {\bibfnamefont {G.}~\bibnamefont {Currás-Lorenzo}}, \bibinfo {author} {\bibfnamefont {{\'A}.}~\bibnamefont {Navarrete}}, \bibinfo {author} {\bibfnamefont {K.}~\bibnamefont {Azuma}}, \bibinfo {author} {\bibfnamefont {G.}~\bibnamefont {Kato}}, \bibinfo {author} {\bibfnamefont {M.}~\bibnamefont {Curty}},\ and\ \bibinfo {author} {\bibfnamefont {M.}~\bibnamefont {Razavi}},\ }\bibfield  {journal} {\bibinfo  {journal} {npj Quantum Information}\ }\textbf {\bibinfo {volume} {7}},\ \href {https://doi.org/10.1038/s41534-020-00345-3} {10.1038/s41534-020-00345-3} (\bibinfo {year} {2021})\BibitemShut {NoStop}%
\bibitem [{\citenamefont {Narasimhachar}(2011)}]{var}%
  \BibitemOpen
  \bibfield  {author} {\bibinfo {author} {\bibfnamefont {V.}~\bibnamefont {Narasimhachar}},\ }\emph {\bibinfo {title} {Study of realistic devices for quantum key-distribution}},\ \href@noop {} {Master's thesis},\ \bibinfo  {school} {University of Waterloo} (\bibinfo {year} {2011})\BibitemShut {NoStop}%
\bibitem [{\citenamefont {Burenkov}\ \emph {et~al.}(2010)\citenamefont {Burenkov}, \citenamefont {Qi}, \citenamefont {Fortescue},\ and\ \citenamefont {Lo}}]{burenkov2010}%
  \BibitemOpen
  \bibfield  {author} {\bibinfo {author} {\bibfnamefont {V.}~\bibnamefont {Burenkov}}, \bibinfo {author} {\bibfnamefont {B.}~\bibnamefont {Qi}}, \bibinfo {author} {\bibfnamefont {B.}~\bibnamefont {Fortescue}},\ and\ \bibinfo {author} {\bibfnamefont {H.-K.}\ \bibnamefont {Lo}},\ }\href {https://arxiv.org/abs/1005.0272} {\bibinfo {title} {Security of high speed quantum key distribution with finite detector dead time}} (\bibinfo {year} {2010}),\ \Eprint {https://arxiv.org/abs/1005.0272} {arXiv:1005.0272 [quant-ph]} \BibitemShut {NoStop}%
\bibitem [{Note8()}]{Note8}%
  \BibitemOpen
  \bibinfo {note} {The measurement in the filter can be more than click/no-click measurement. For instance, one can include the 0-photon subspace measurement inside the filter.}\BibitemShut {Stop}%
\bibitem [{Note9()}]{Note9}%
  \BibitemOpen
  \bibinfo {note} {One can further optimize over parameters controlled by Bob, such as the beam splitter ratio \( s \), the decoy intensities \( \mu _k\), and the probabilities \( p_{\mu _k} \) associated with choosing each intensity. This is left for future work}\BibitemShut {NoStop}%
\bibitem [{\citenamefont {Pereira}\ \emph {et~al.}(2022)\citenamefont {Pereira}, \citenamefont {Currás-Lorenzo}, \citenamefont {Álvaro Navarrete}, \citenamefont {Mizutani}, \citenamefont {Kato}, \citenamefont {Curty},\ and\ \citenamefont {Tamaki}}]{sourceimp1}%
  \BibitemOpen
  \bibfield  {author} {\bibinfo {author} {\bibfnamefont {M.}~\bibnamefont {Pereira}}, \bibinfo {author} {\bibfnamefont {G.}~\bibnamefont {Currás-Lorenzo}}, \bibinfo {author} {\bibnamefont {Álvaro Navarrete}}, \bibinfo {author} {\bibfnamefont {A.}~\bibnamefont {Mizutani}}, \bibinfo {author} {\bibfnamefont {G.}~\bibnamefont {Kato}}, \bibinfo {author} {\bibfnamefont {M.}~\bibnamefont {Curty}},\ and\ \bibinfo {author} {\bibfnamefont {K.}~\bibnamefont {Tamaki}},\ }\href {https://arxiv.org/abs/2210.11754} {\bibinfo {title} {Modified bb84 quantum key distribution protocol robust to source imperfections}} (\bibinfo {year} {2022}),\ \Eprint {https://arxiv.org/abs/2210.11754} {arXiv:2210.11754 [quant-ph]} \BibitemShut {NoStop}%
\bibitem [{\citenamefont {Tamaki}\ \emph {et~al.}(2014)\citenamefont {Tamaki}, \citenamefont {Curty}, \citenamefont {Kato}, \citenamefont {Lo},\ and\ \citenamefont {Azuma}}]{sourceimp2}%
  \BibitemOpen
  \bibfield  {author} {\bibinfo {author} {\bibfnamefont {K.}~\bibnamefont {Tamaki}}, \bibinfo {author} {\bibfnamefont {M.}~\bibnamefont {Curty}}, \bibinfo {author} {\bibfnamefont {G.}~\bibnamefont {Kato}}, \bibinfo {author} {\bibfnamefont {H.-K.}\ \bibnamefont {Lo}},\ and\ \bibinfo {author} {\bibfnamefont {K.}~\bibnamefont {Azuma}},\ }\bibfield  {journal} {\bibinfo  {journal} {Physical Review A}\ }\textbf {\bibinfo {volume} {90}},\ \href {https://doi.org/10.1103/physreva.90.052314} {10.1103/physreva.90.052314} (\bibinfo {year} {2014})\BibitemShut {NoStop}%
\bibitem [{\citenamefont {Currás-Lorenzo}\ \emph {et~al.}(2025{\natexlab{b}})\citenamefont {Currás-Lorenzo}, \citenamefont {Pereira}, \citenamefont {Kato}, \citenamefont {Curty},\ and\ \citenamefont {Tamaki}}]{sourceimp3}%
  \BibitemOpen
  \bibfield  {author} {\bibinfo {author} {\bibfnamefont {G.}~\bibnamefont {Currás-Lorenzo}}, \bibinfo {author} {\bibfnamefont {M.}~\bibnamefont {Pereira}}, \bibinfo {author} {\bibfnamefont {G.}~\bibnamefont {Kato}}, \bibinfo {author} {\bibfnamefont {M.}~\bibnamefont {Curty}},\ and\ \bibinfo {author} {\bibfnamefont {K.}~\bibnamefont {Tamaki}},\ }\href {https://arxiv.org/abs/2305.05930} {\bibinfo {title} {Security of high-speed quantum key distribution with imperfect sources}} (\bibinfo {year} {2025}{\natexlab{b}}),\ \Eprint {https://arxiv.org/abs/2305.05930} {arXiv:2305.05930 [quant-ph]} \BibitemShut {NoStop}%
\bibitem [{\citenamefont {Zapatero}\ \emph {et~al.}(2023)\citenamefont {Zapatero}, \citenamefont {Álvaro Navarrete},\ and\ \citenamefont {Curty}}]{sourceimp4}%
  \BibitemOpen
  \bibfield  {author} {\bibinfo {author} {\bibfnamefont {V.}~\bibnamefont {Zapatero}}, \bibinfo {author} {\bibnamefont {Álvaro Navarrete}},\ and\ \bibinfo {author} {\bibfnamefont {M.}~\bibnamefont {Curty}},\ }\href {https://arxiv.org/abs/2310.20377} {\bibinfo {title} {Implementation security in quantum key distribution}} (\bibinfo {year} {2023}),\ \Eprint {https://arxiv.org/abs/2310.20377} {arXiv:2310.20377 [quant-ph]} \BibitemShut {NoStop}%
\bibitem [{\citenamefont {Currás-Lorenzo}\ \emph {et~al.}(2023)\citenamefont {Currás-Lorenzo}, \citenamefont {Nahar}, \citenamefont {Lütkenhaus}, \citenamefont {Tamaki},\ and\ \citenamefont {Curty}}]{sourceimp5}%
  \BibitemOpen
  \bibfield  {author} {\bibinfo {author} {\bibfnamefont {G.}~\bibnamefont {Currás-Lorenzo}}, \bibinfo {author} {\bibfnamefont {S.}~\bibnamefont {Nahar}}, \bibinfo {author} {\bibfnamefont {N.}~\bibnamefont {Lütkenhaus}}, \bibinfo {author} {\bibfnamefont {K.}~\bibnamefont {Tamaki}},\ and\ \bibinfo {author} {\bibfnamefont {M.}~\bibnamefont {Curty}},\ }\href {https://doi.org/10.1088/2058-9565/ad141c} {\bibfield  {journal} {\bibinfo  {journal} {Quantum Science and Technology}\ }\textbf {\bibinfo {volume} {9}},\ \bibinfo {pages} {015025} (\bibinfo {year} {2023})}\BibitemShut {NoStop}%
\bibitem [{\citenamefont {Gottesman}\ \emph {et~al.}(2004)\citenamefont {Gottesman}, \citenamefont {Lo}, \citenamefont {Lütkenhaus},\ and\ \citenamefont {Preskill}}]{sourceimp6}%
  \BibitemOpen
  \bibfield  {author} {\bibinfo {author} {\bibfnamefont {D.}~\bibnamefont {Gottesman}}, \bibinfo {author} {\bibfnamefont {H.-K.}\ \bibnamefont {Lo}}, \bibinfo {author} {\bibfnamefont {N.}~\bibnamefont {Lütkenhaus}},\ and\ \bibinfo {author} {\bibfnamefont {J.}~\bibnamefont {Preskill}},\ }\href {https://arxiv.org/abs/quant-ph/0212066} {\bibinfo {title} {Security of quantum key distribution with imperfect devices}} (\bibinfo {year} {2004}),\ \Eprint {https://arxiv.org/abs/quant-ph/0212066} {arXiv:quant-ph/0212066 [quant-ph]} \BibitemShut {NoStop}%
\bibitem [{\citenamefont {Navarrete}\ and\ \citenamefont {Curty}(2022)}]{sourceimp7}%
  \BibitemOpen
  \bibfield  {author} {\bibinfo {author} {\bibfnamefont {{\'A}.}~\bibnamefont {Navarrete}}\ and\ \bibinfo {author} {\bibfnamefont {M.}~\bibnamefont {Curty}},\ }\href {https://doi.org/10.1088/2058-9565/ac74dc} {\bibfield  {journal} {\bibinfo  {journal} {Quantum Science and Technology}\ }\textbf {\bibinfo {volume} {7}},\ \bibinfo {pages} {035021} (\bibinfo {year} {2022})}\BibitemShut {NoStop}%
\bibitem [{\citenamefont {Zapatero}\ \emph {et~al.}(2021)\citenamefont {Zapatero}, \citenamefont {Navarrete}, \citenamefont {Tamaki},\ and\ \citenamefont {Curty}}]{sourceimp8}%
  \BibitemOpen
  \bibfield  {author} {\bibinfo {author} {\bibfnamefont {V.}~\bibnamefont {Zapatero}}, \bibinfo {author} {\bibfnamefont {{\'A}.}~\bibnamefont {Navarrete}}, \bibinfo {author} {\bibfnamefont {K.}~\bibnamefont {Tamaki}},\ and\ \bibinfo {author} {\bibfnamefont {M.}~\bibnamefont {Curty}},\ }\href {https://doi.org/10.22331/q-2021-12-07-602} {\bibfield  {journal} {\bibinfo  {journal} {Quantum}\ }\textbf {\bibinfo {volume} {5}},\ \bibinfo {pages} {602} (\bibinfo {year} {2021})}\BibitemShut {NoStop}%
\bibitem [{Note10()}]{Note10}%
  \BibitemOpen
  \bibinfo {note} {For \( a \), both upper and lower bounds are required. For \( \delta \), \( q_Z \), and \( \lambda _{\min } \), only worst-case lower bounds are needed.}\BibitemShut {Stop}%
\bibitem [{\citenamefont {Fung}\ \emph {et~al.}(2011)\citenamefont {Fung}, \citenamefont {Chau},\ and\ \citenamefont {Lo}}]{Fung_2011}%
  \BibitemOpen
  \bibfield  {author} {\bibinfo {author} {\bibfnamefont {C.-H.~F.}\ \bibnamefont {Fung}}, \bibinfo {author} {\bibfnamefont {H.~F.}\ \bibnamefont {Chau}},\ and\ \bibinfo {author} {\bibfnamefont {H.-K.}\ \bibnamefont {Lo}},\ }\bibfield  {journal} {\bibinfo  {journal} {Physical Review A}\ }\textbf {\bibinfo {volume} {84}},\ \href {https://doi.org/10.1103/physreva.84.020303} {10.1103/physreva.84.020303} (\bibinfo {year} {2011})\BibitemShut {NoStop}%
\bibitem [{\citenamefont {Gittsovich}\ \emph {et~al.}(2014)\citenamefont {Gittsovich}, \citenamefont {Beaudry}, \citenamefont {Narasimhachar}, \citenamefont {Alvarez}, \citenamefont {Moroder},\ and\ \citenamefont {Lütkenhaus}}]{Gittsovich_2014}%
  \BibitemOpen
  \bibfield  {author} {\bibinfo {author} {\bibfnamefont {O.}~\bibnamefont {Gittsovich}}, \bibinfo {author} {\bibfnamefont {N.~J.}\ \bibnamefont {Beaudry}}, \bibinfo {author} {\bibfnamefont {V.}~\bibnamefont {Narasimhachar}}, \bibinfo {author} {\bibfnamefont {R.~R.}\ \bibnamefont {Alvarez}}, \bibinfo {author} {\bibfnamefont {T.}~\bibnamefont {Moroder}},\ and\ \bibinfo {author} {\bibfnamefont {N.}~\bibnamefont {Lütkenhaus}},\ }\bibfield  {journal} {\bibinfo  {journal} {Physical Review A}\ }\textbf {\bibinfo {volume} {89}},\ \href {https://doi.org/10.1103/physreva.89.012325} {10.1103/physreva.89.012325} (\bibinfo {year} {2014})\BibitemShut {NoStop}%
\bibitem [{Note11()}]{Note11}%
  \BibitemOpen
  \bibinfo {note} {Technically $h(x)$ is the binary entropy function when $x \leq 1/2$, and $1$ otherwise}\BibitemShut {NoStop}%
\bibitem [{Note12()}]{Note12}%
  \BibitemOpen
  \bibinfo {note} {In \protect \cref {eq:D9} in \protect \cref {appendix:adelta}, we also denote this minimum efficiency as $\eta _{\protect \text {min}}$.}\BibitemShut {Stop}%
\bibitem [{Note13()}]{Note13}%
  \BibitemOpen
  \bibinfo {note} {We choose to further simplify the bound by removing the square root structure of ${\protect \tilde {F}^{(\protect \text {M})}_X}$ and ${\protect \tilde {F}^{(\protect \text {M})}_Z}$. One can explicit compute the bound \protect \cref {eq:D1} as it might be tighter than the current approach in this work.}\BibitemShut {Stop}%
\bibitem [{\citenamefont {Bhatia}(2013)}]{bhatia2013matrix}%
  \BibitemOpen
  \bibfield  {author} {\bibinfo {author} {\bibfnamefont {R.}~\bibnamefont {Bhatia}},\ }\href@noop {} {\emph {\bibinfo {title} {Matrix analysis}}},\ Vol.\ \bibinfo {volume} {169}\ (\bibinfo  {publisher} {Springer Science \& Business Media},\ \bibinfo {year} {2013})\BibitemShut {NoStop}%
\end{thebibliography}%

\appendix
\crefalias{section}{appendix}
\section{\texorpdfstring{Reduction to (1$\leq \pnm \leq \pn$)-photon subspace and variable length security}{}}\label{appendix:security}
In this appendix, we prove the 
 following theorem regarding the security of the protocol from \cref{sec:basicprotdescription} from the obtained bounds (\cref{eq:phaseerroridea}).

	\begin{theorem}\label{theorem:security}
Consider the protocol described in \cref{sec:basicprotdescription},  such that it produces a key of length 
\begin{equation}
\begin{aligned}
&l(\errzobs,\nx,\nk,\Nxobs,\nmc) \coloneqq \max\big\{\K(\nk,\nmc)\biggl(1-h(\bimp(\nx,\nk,\Nxobs,\nmc))\biggr)\\&\phantom{l(\nx,\nk,\Nxobs,\nmc) \coloneqq \max} - \lambda_{\text{EC}}(\errzobs,\nx,\nk,\Nxobs,\nmc) - 2\text{log}(1/2\eppa) - \text{log}(2/\eppa), 0\big\},
\end{aligned}
\end{equation}
upon the event $\event{\errzobs,\Nxobs,\nx,\nk,\nmc,\nkone} \wedge \Omega_\text{EV}$, where $\Omega_\text{EV}$ denotes the event that error-verification passes.   Here $h(\cdot)$ is the binary entropy function \footnote{Technically $h(x)$ is the binary entropy function when $x \leq 1/2$, and $1$ otherwise}, and  $\bimp(\cdot)$ and $\K(\cdot)$ are the bounds we obtained in \cref{sec:phase error estimation} which satisfy \cref{eq:memorylessbounds}. Then the QKD protocol is $(2\epat+ \eppa + \epev)$-secure. 
			\end{theorem}

\begin{proof}
 In order to prove security, we rely on 
 \cite[Theorem 4]{devEUR}.
 In order to use this theorem, we need to setup some notation. Let $\tilde{\Omega}=\event{\errzobs,\Nxobs,\nx,\nk,\nmc,\nkone}$ denote well-defined events that can take place in the protocol, where we fix the value  $\nkone$ in the event. Let $Z^{\nk}$ denote the raw key register, let ${E^n}$ denote Eve's quantum system, and let ${C}$ denote public announcements (excluding error-correction and error-verification). Let $S(\nx,\nk,\Nxobs,\nmc)$  denote the subset of possible values of the tuple $\nkone$ (which we will explicitly construct later). Suppose there exist values of $\kappa{\tilde{(\Omega)}}\geq 0$ such that 
\begin{equation} \label{eq:genericvariableone}
				\begin{aligned}
				&	 \sminentropynew{Z_A^{\nk}}{CE}{\tilde{\Omega}}_{\rho | \tilde{\Omega}} \geq \K(\nk,\nmc)\biggl(1-h(\bimp(\nx,\nk,\Nxobs,\nmc))\biggr)\quad \forall \nkone \in S(\nx,\nk,\Nxobs,\nmc)
					\end{aligned}
			\end{equation}
and 
\begin{equation} \label{eq:genericvariabletwo}
\sum_{\errzobs,\nx,\nk,\Nxobs,\nmc\;} \left( \sum_{\nkone \in S(\nx,\nk,\Nxobs,\nmc)} \Pr(\tilde{\Omega}) \kappa{(\tilde{\Omega})}  + \sum_{\nkone\notin S(\nx,\nk,\Nxobs,\nmc)} \Pr(\tilde{\Omega}) \right) \leq \epat^2 .
\end{equation}
Then,  \cite[Theorem 4]{devEUR} directly guarantees the security of the protocol. 

The idea here is that for while we do not know the values of $\nkone$, they either belong to $S(\nx,\nk,\Nxobs,\nmc)$ or they don't. If they belong to $S(\nx,\nk,\Nxobs,\nmc)$, then the smooth min-entropy of the raw key string can be suitably bounded with an appropriate choice of the smoothing parameter. The second inequality above ensures that these smoothing parameters (in an averaged sense) are not too large, and the probability of $\nkone$ not belonging to $S(\nx,\nk,\Nxobs,\nmc)$ is also not too large. 

To satisfy \cref{eq:genericvariableone,eq:genericvariabletwo} we choose $S(\nx,\nk,\Nxobs,\nmc)$ to be the set of values of $\nkone$ such that the computed bounds from \cref{eq:memorylessbounds} hold. That is, 
            \begin{align}
        S(\nx,\nk,\Nxobs,\nmc) = \{\nkone: \nkone > \K(\nk,\nmc)\}.
            \end{align}
We set $\kappa(\tilde{\Omega})$ to be
\begin{equation}
    \kappa(\tilde{\Omega}) = \Pr\biggl(\errks \geq \bimp(\nxb,\nkb,\Nxobsb,\nmcb) \ \lor\  \nkoneb \leq \K(\nkb,\nmcb)\biggr)_{| \event{\errzobs,\Nxobs,\nx,\nk,\nmc,\nkone}}
\end{equation}
We then show that, with this choice of $S(\nx,\nk,\Nxobs,\nmc)$ and $\kappa(\tilde{\Omega})$, \cref{eq:genericvariableone,eq:genericvariabletwo} are satisfied in \cref{appendix:cq,subsec:boundingprob} respectively. With this, we conclude the proof.  

\end{proof}

\subsection{Bounding the smooth min-entropy}\label{appendix:cq}
In this subsection, we prove \cref{eq:genericvariableone}. To bound the smooth min-entropy, we begin by showing that the $c_q$ factor in the EUR statement~\cite{EUR1} that we will use is 1.
To obtain Alice's POVM element, we need to break $\{\GtXeq,\GtXneq\}$ ($\{\GtZeq,\GtZneq\}$) down. Due to the tensor product structure between Alice and Bob, we can first let Alice perform her measurement with $\idd$ on Bob's side. We show the calculation for the $X$-basis measurement. $P$ are some positive semi-definite matrices that we pick, living on the orthogonal support of $\ftx$.
Note that in most of the cases we consider, $\ftx$ and $\ftz$ have full support on the ($1\leq\pnm\leq\pn$)-photon subspace. Thus, we ignore the term $(\idd^{(\pn)}-\Pi_{\ftx})$ from \cref{lemma:twostep}.
\begin{align}
    \begin{split}
    \GtXeq &= \sqrt{\ftx}^+\sqrt{\Fonesc}^+ \POVMXeqM\sqrt{\Fonesc}^+\sqrt{\ftx}^+\\
           &=\sqrt{\ftx}^+\sqrt{\Fonesc}^+ \biggl(\pxa\kb{+}\otimes \POVMBXzeroM + \pxa\kb{-}\otimes \POVMBXoneM\biggr)\sqrt{\Fonesc}^+\sqrt{\ftx}^+\\
    \GtXneq &= \sqrt{\ftx}^+\sqrt{\Fonesc}^+ \POVMXneqM\sqrt{\Fonesc}^+\sqrt{\ftx}^+\\
           &=\sqrt{\ftx}^+\sqrt{\Fonesc}^+ \biggl(\pxa\kb{+}\otimes \POVMBXoneM + \pxa\kb{-}\otimes \POVMBXzeroM\biggr)\sqrt{\Fonesc}^+\sqrt{\ftx}^+,
    \end{split}
\end{align}
where $\POVMBXzeroM$ and $\POVMBXoneM$ are Bob's POVM elements in ($1\leq\pnm\leq\pn$)-photon subspace.
We define:
\begin{align}
    \begin{split}
        &\POVMBXM:=\POVMBXzeroM+\POVMBXoneM\\
        &\POVMBZM:=\POVMBZzeroM+\POVMBZoneM\\
        &\POVMBconM:=\pxa\POVMBXM+\pza\POVMBZM
    \end{split}
\end{align}
Recall from \cref{eq:definePOVMA}, \cref{eq:definejoint}, \cref{eq:definefx} and \cref{eq:defineftx}, one can rewrite $\ftx$ and $\Fonesc$ as follows:
\begin{align}
    \begin{split}
    &\ftx = \pxa\idd_A \otimes \sqrt{\POVMBconM}^+ \POVMBXM \sqrt{ \POVMBconM}^+\\
    &\Fonesc = \idd_A \otimes \POVMBconM.
    \end{split}
\end{align}
Therefore, the POVM elements corresponding to Alice sending $\ket{+}$ are as follows:
\begin{align}
    \begin{split}
    \tilde{G}_{A,\ket{+}} &= \sqrt{\ftx}^+\sqrt{\Fonesc}^+ \biggl(\pxa\kb{+}\otimes (\POVMBXzeroM+\POVMBXoneM)\biggr)\sqrt{\ftx}^+\sqrt{\Fonesc}^+\\
    &= \sqrt{\ftx}^+\sqrt{\Fonesc}^+ \biggl(\pxa\kb{+}\otimes \POVMBXM\biggr)\sqrt{\ftx}^+\sqrt{\Fonesc}^+\\
    &=\kb{+}\otimes \sqrt{ \sqrt{\POVMBconM}^+ \POVMBXM \sqrt{ \POVMBconM}^+}^+  \sqrt{\POVMBconM}^+ \POVMBXM \sqrt{ \POVMBconM}^+ \sqrt{ \sqrt{\POVMBconM}^+ \POVMBXM \sqrt{ \POVMBconM}^+}^+\\
    & = \kb{+}\otimes \idd_B^{(\pn)}
    \end{split}
\end{align}
Similarly,
\begin{align}
    \begin{split}
    \tilde{G}_{A,\ket{-}}&= \kb{-}\otimes \idd_B^{(\pn)}\\
    \tilde{G}_{A,\ket{0}}&= \kb{0}\otimes \idd_B^{(\pn)}\\
    \tilde{G}_{A,\ket{1}}&= \kb{1}\otimes \idd_B^{(\pn)}
    \end{split}
\end{align}
Thus, we can compute $c_q$ for EUR statement:
\begin{align}
    c_q = \text{log}_2\frac{1}{c} =1,\;\;\;\;\text{where } c: = \max\limits_{x,z}\infnorm{\sqrt{M_x}\sqrt{N_z}}^2\\
    M_x \in \{\tilde{G}_{A,\ket{+}},\tilde{G}_{A,\ket{-}}\} ;\;\;\;  M_z \in \{\tilde{G}_{A,\ket{0}},\tilde{G}_{A,\ket{1}}\}\notag
\end{align}

Next, we bound the smooth min-entropy.
The state we are interested in is
\[
\rho_{|\tilde{\Omega}} = \rho_{A^{\nkone}B^{\nkone}Z_A^{\tilde{n}_{({K,0})}}Z_A^{\tilde{n}_{(K,>\pn)}}E^nC  |\tilde{\Omega}(\errzobs,\Nxobs,\nx,\nk,\nmc,\nkone)}.
\]
This is the state before Alice and Bob complete their $Z$-basis measurements for the $(1 \leq \pnm \leq \pn)$ rounds, conditioned on the event $\tilde{\Omega}(\errzobs,\Nxobs,\nx,\nk,\nmc,\nkone)$. The $Z$-basis measurements on the $0$-photon rounds and ($>M$)-photon rounds are completed.
 This state corresponds precisely to the state after step \circletext{4} in \cref{diag:equivalent}. Alice performs a $Z$-basis measurement on $\rho_{\condi}$ to obtain the actual key rounds $Z_A^{\nkone}$. On the other hand, in the virtual $X$-basis measurement, Alice and Bob perform $X$-basis measurements on $\rho_{|\tilde{\Omega}}$ to obtain $X_A^{\nkone}$ and $X_B^{\nkone}$.
Since we condition on a specific value of $\nkone$, we can split up $Z^{\nk}$ as $Z^{\nkone}Z^{\text{rest}}$. We can also split up $X_A^{\nkone}$ and $X_B^{\nkone}$ in the same way.

We aim to bound the following smooth min-entropy for all key rounds:
\begin{align*}
     \sminentropynew{Z_A^{\nk}}{CE^n}{\tilde{\Omega}}_{\rho_{|\tilde{\Omega}}} &=  \sminentropynew{Z^{\nkone}Z^{\text{rest}}}{CE^n}{\tilde{\Omega}}_{\rho_{|{\tilde{\Omega}}}} \\
    &\geq  \sminentropynew{Z_A^{\nkone}}{CE^n}{{\tilde{\Omega}}}_{\rho_{|{\tilde{\Omega}}}}\\
     &\geq \nkone - \smaxentropynew{X_A^{\nkone}}{B^{\nkone}}{{\tilde{\Omega}}}_{\rho_{|{\tilde{\Omega}}}}\\
             &\geq  \nkone - \smaxentropynew{X_A^{\nkone}}{X_B^{\nkone}}{{\tilde{\Omega}}}_{\rho_{|{\tilde{\Omega}}}}\\
             &\geq \nkone(1-h( \bimp(\nx,\nk,\Nxobs,\nmc)))\\
             &\geq \K(\nk,\nmc)\biggl(1-h(\bimp(\nx,\nk,\Nxobs,\nmc))\biggr)\;\;\;\forall \nkone \in S
\end{align*}
The first inequality follows from the chain rule \cite[Lemma 6.7]{tomamichel2017}. The second inequality follows from the EUR statement \cite{EUR1}. The third inequality follows from the data-processing inequality \cite[Theorem 6.2]{tomamichelbook}. The fourth inequality is due to \cite[Lemma 6]{devEUR}~\cite{tomamichel2017} and \cref{eq:mainresult}, and our choice of $\kappa(\tilde{\Omega})$. The final inequality follows from the fact that we consider $\nkone$ in the set $S(\nx,\nk,\Nxobs,\nmc)$. 
\subsection{Bounding the probabilites} \label{subsec:boundingprob}
In this subsection, we prove \cref{eq:genericvariabletwo} via straightforward algebra. We have 
\begin{align}
    \begin{split}
    &\sum_{\errzobs,\nx,\nk,\Nxobs,\nmc\;}\sum_{\nkone\in S}\text{Pr}({\tilde{\Omega}})\kappa{{(\tilde{\Omega})}} \\
    =&  \sum_{\errzobs,\nx,\nk,\Nxobs,\nmc\;}\sum_{\nkone\in S}\text{Pr}({\tilde{\Omega}})\text{Pr}\biggl(\errks\geq \bimp(\nx,\nk,\Nxobs,\nmc) \lor \nkone \leq  \K(\nk,\nmc)\biggr)_{|(\tilde{\Omega})}\\
    =&  \sum_{\errzobs,\nx,\nk,\Nxobs,\nmc\;}\sum_{\nkone\in S}\text{Pr}(\tilde{\Omega})\text{Pr}\biggl(\errks\geq \bimp(\nx,\nk,\Nxobs,\nmc)\land \nkone >  \K(\nk,\nmc)\biggr)_{|(\tilde{\Omega})}\\
    =& \text{Pr}(\errks\geq \bimp(\nxb,\nkb,\Nxobsb,\nmcb)\land \nkoneb > \K(\nkb,\nmcb)).
    \end{split}
\end{align}
Furthermore,
\begin{align}
    \begin{split}
    &\sum_{\errzobs,\nx,\nk,\Nxobs,\nmc\;}\sum_{j\notin S_i}\text{Pr}(\tilde{\Omega}) 
    =\text{Pr}(\nkoneb \leq \K(\nkb,\nmcb)).
    \end{split}
\end{align}
Therefore,
\begin{align}
    \begin{split}
    &\sum_{\errzobs,\nx,\nk,\Nxobs,\nmc\;}\sum_{\nkone\in S}\text{Pr}(\tilde{\Omega})\kappa{(\tilde{\Omega})} + \sum_{\errzobs,\nx,\nk,\Nxobs,\nmc\;}\sum_{\nkone\notin S}\text{Pr}(\tilde{\Omega})\\
    &= \text{Pr}\biggl(\errks\geq \bimp(\nxb,\nkb,\Nxobsb,\nmcb)\land \nkoneb > \K(\nkb,\nmcb)\biggr) + \text{Pr}\biggl(\nkoneb \leq K(\nkb,\nmcb)\biggr)\\
    &= \text{Pr}\biggl(\errks\geq \bimp(\nxb,\nkb,\Nxobsb,\nmcb)\lor \nkoneb \leq \K(\nkb,\nmcb)\biggr)\\
    &\leq \epatsingle^2\\
    & \epatsingle^2:= \epazua^2 + \epazub^2 +\eppntb^2+\epzero^2 
    \end{split}
\end{align}
We use the fact that $\Pr(\Omega_1 \land \Omega_2^c)+\Pr(\Omega_2) = \Pr(\Omega_1 \lor \Omega_2)$, where $\Omega^c $ denotes the complement of $\Omega$. The last inequality comes from our bound \cref{eq:mainresult} in the memoryless case.

\section{Decoy analysis and variable length security for decoy-state protocol}\label{appendix:decoy}
We demonstrate the decoy-state analysis for the decoy-state BB84 passive protocol~\cite{hwang2003,lo2005,ma2005,hayashi2014,curty2014,charles} with memoryless detectors. The security proof follows from a similar analysis as \cref{appendix:security}. The decoy analysis and the security proof for correlated detectors follow by performing the exact same method on the random variables obtained from the kept rounds.

We consider the case where Alice does not use a perfect single-photon source, but instead uses a phase-randomized WCP to encode bits. The idea is to first bound the phase error rate in terms of quantities corresponding to Alice sending single photons. This follows the same analysis as in the memoryless case. We then use decoy-state analysis to bound those single-photon (on Alice's side) quantities using observations.
Our decoy-state analysis follows the approach in \cite{charles}.

To implement the decoy-state protocol, when Alice prepares a state, she additionally selects an intensity $\muk \in \{\mu_1, \mu_2, \mu_3\}$ with corresponding probabilities $p_{\muk}$. We assume $\mu_1 > \mu_2 + \mu_3$ and $\mu_2 > \mu_3 \geq 0$. She then prepares a phase-randomized weak coherent pulse with the chosen intensity.

The decoy-state protocol description modifies the following steps of the protocol described in \cref{sec:basicprotdescription}.
\begin{description}
     \item [State Preparation]Alice sends a state in $X$ ($Z$) basis with probability $\pxa$ $(\pza)$. She additionally selects an intensity $\muk \in \{\mu_1, \mu_2, \mu_3\}$ with corresponding probabilities $p_{\muk}$. We assume $\mu_1 > \mu_2 + \mu_3$ and $\mu_2 > \mu_3 \geq 0$. She then prepares a phase-randomized weak coherent pulse with the chosen intensity.
    \item[Classical Announcement and Sifting] Alice and Bob announce the bases they used and perform sifting. They discard basis-mismatched rounds. All $X$-basis rounds are used for testing. They denote the number of $X$-basis rounds with intensity $\muk$ as $n_{X,\muk}$ and the number of erroneous outcomes in $X$-basis rounds with intensity $\muk$ as $n_{X\neq,\muk}$. A small portion of $Z$-basis rounds is revealed for error correction. The error rate in these rounds is denoted as $\errzobs$. The rest of the $Z$-basis rounds are used for key generation. They denote the number of rounds used in key generation with intensity $\muk$ as $n_{K,\muk}$. For brevity, let us us the notation $n_{O,\mukvec} = (n_{O,\mu_1},n_{O,\mu_2},n_{O,\mu_3})$ and and $n_{O}= n_{O,\mu_1}+n_{O,\mu_2}+n_{O,\mu_3}$, where $O\in\{X\neq,X,K\}$.
\end{description}
The remaining steps are similar to \cref{sec:basicprotdescription}. In particular, Alice and Bob perform one-way error correction with $\lec(\errzobs,n_{X\neq,\mukvec},\nxmuvecobs,\nkmuvecobs,\nmc)$ bits of information, followed by error-verification, and privacy amplification to produce a key of $l(\errzobs,n_{X\neq,\mukvec},\nxmuvecobs,\nkmuvecobs,\nmc)$. Note that Alice generates keys from all three intensities.\\

Let us first define some notations.
\begin{itemize}
    \item $\errkss$: phase error rate in the key set that Alice sends in 1 photon and Bob receives $1\leq\pnm\leq\pn$ photons.
    \item $\nxbs$: number of rounds in the test set that Alice sends in 1 photon.
    \item $\nkbs$: number of rounds in the key set that Alice sends in 1 photon.
    \item $\xnoeobss$: number of errors in the test set that Alice sends in 1 photon.
    \item $\nomuvec:=(\nomu{1},\nomu{2},\nomu{3})$: $O \in \{X\neq,X,K\}$  number of errors in test set ($X\neq$), number of rounds in test set ($X$), number of rounds in key set ($K$) at different intensities.
    \item $\no:=\nomu{1}+\nomu{2}+\nomu{3}$: number of rounds result in outcome $O$.
    \item $p_{m|\muk} := e^{-\muk}\frac{\muk^m}{m!}$: The probability Alice sends $\pnm$ photons given she chooses $\muk$. 
    \item $\tau_m := \sum_{\muk} p_{\muk}p_{m|\muk}$ The probability that Alice sends $\pnm$ photons.
\end{itemize}
The notation here is updated to reflect the fact that we consider two distinct notions of photon number. The symbol $\pn$ indicates that Bob receives a $(1 \leq \pnm \leq \pn)$-photon signal, while the term “single” refers to the case where Alice sends single-photon signals.
To prove the security of our protocol, we require the following two bounds:

\begin{align}
    \begin{split}
   &\text{Pr}\biggl(\errkss \geq  \bimp(\nxbs,\nkbs,\xnoeobss,\nmcb) \lor \nkbss \leq \K(\nkbs,\nmcb)\biggr)\leq \epats^2\\
    &\epats^2:=\epazua^2 + \epazub^2 +\eppntb^2+\epzero^2 
    \end{split}
\end{align}
\begin{remark}
    Technically, we should use the number of multi-click outcomes that Alice sends in single photons $\bd{n_{(\textbf{mc},\textbf{single})}}$. But here we use the monotonicity of the bound, and we replace it with the total number of multi-click outcomes $\nmcb$. 
\end{remark}

We follow the steps in \cite[Appendix A]{charles}. We present the results here.
\begin{align}
      &\text{Pr}\biggl(\xnoeobss \geq \bmaxone(\nxneqmuvec)\;\lor\;\nxbs \geq \bmaxone(\nxmuvec)\lor\;\nkbs \leq \bminone(\nkmuvec) 
      \biggr) \leq 9\epatd^2
\end{align}
where
\begin{align}\label{eq:BmaxBmin}
    \begin{split}
     \bd{n^{\pm}_{O,\muk}}& := \frac{e^{\muk}}{p_{\muk}}\Biggl(\nomu{k}\pm\sqrt{\frac{\no}{2}\ln\biggl(\frac{2}{\epatd^2}\biggr)}\Biggr)\\
    \bminzero(\nomuvec)& := \tau_0\frac{\mu_2\nomuminus{3}-\mu_3\nomuplus{2}}{\mu_2-\mu_3}\\
    \bminone(\nomuvec)& := \biggl(\frac{\mu_1\tau_1}{\mu_1(\mu_2-\mu_3)-\mu_2^2+\mu_3^2}\biggr) \times\\
    &\;\;\;\;\;\;\biggl(\nomuminus{2}-\nomuplus{3}-\frac{\mu_2^2-\mu_3^2}{\mu_1^2}\big(\nomuplus{1}-\frac{\bminzero(\nomuvec)}{\tau_0}\big)\biggr)\\
    \bmaxone(\nomuvec)& := \tau_1\frac{\nomuplus{2}-\nomuminus{3}}{\mu_2-\mu_3}.
    \end{split}
\end{align}
Combining all bounds together, we have the following bounds:
\begin{align}
    \begin{split}
   & \text{Pr}\biggl(\errkss \geq \be(\nxmuvec,\nkmuvec,\nxneqmuvec,\nmcb) \lor \bone(\nkmuvec,\nmcb)\biggr) \leq \epats^2+9\epatd^2 \\
    \end{split}
\end{align}
where
\begin{align}\label{eq:BeB1}
    \begin{split}
    &\bone(\nkmuvec,\nmcb) = \bminone(\nkmuvec) - nq_Z-\frac{\nmcb}{\lm}-\sqrt{{-\ln(\epzero)}n}\\&\phantom{\nkmuvec,\nmcb;;;}-\frac{\sqrt{(-\log(\eppntb))(-\log(\eppntb)+4\lm\nmcb)}}{2{\lm}^2}-\biggl(\frac{-2\log(\eppntb)}{4{\lm}^2}\biggr)\\[5pt]
    &\be(\nxmuvec,\nkmuvec,\nxneqmuvec,\nmcb) =\\ 
    &\frac{\bmaxone(\nxneqmuvec)}{a\bone(\nkmuvec,\nmcb)}+\frac{\delta}{a}\biggl(\frac{\bmaxone(\nxmuvec)}{\bone(\nkmuvec,\nmcb)}+1\biggr)+\sqrt{\frac{-2\ln(\epazua^2)}{a}\biggl(\frac{\bmaxone(\nxmuvec)}{\bone(\nkmuvec,\nmcb)^2}+\frac{1}{\bone(\nkmuvec,\nmcb)}\biggr)}
    \\
    &+\sqrt{{-2\ln(\epazub^2)}\biggl(\frac{\nxb}{\bone(\nkmuvec,\nmcb)^2}+\frac{1}{\bone(\nkmuvec,\nmcb)}\biggr)}
    \end{split}
\end{align}
$\bmaxone(\cdot)$ and $\bminone(\cdot)$ are define above in \cref{eq:BmaxBmin}. 
Here, we use the monotonicity of the bounds $\bimp(\cdot)$ and $\K(\cdot)$ to replace the quantities corresponding to Alice sending single photons with the bounds from decoy analysis.
To show security, one can follow a similar method as \cref{appendix:security}. We state the key rate formula without proof here.
We first define the security parameter for the decoy-state BB84 protocol with memoryless passive detectors: 
\begin{equation}\label{eq:epat}
    \epat^2 := \epats^2 +9\epatd^2 = \epazua^2 + \epazub^2 +\eppntb^2+\epzero^2 +9\epatd^2.
\end{equation}
For key length, we set:
\begin{align}\label{eq:keyratedecoy}
    \begin{split}
    &l(e_{Z}^{\text{obs}},\nxmuvecobs,\nkmuvecobs,\nxneqmuvecobs,\nmc)\\
    &:= \text{max}\{\bone(\nkmuvecobs,\nmc)\biggl(1-h(\be(\nxmuvecobs,\nkmuvecobs,\nxneqmuvecobs,\nmc))\biggr)\\&-\lambda_{\text{EC}}(e_{Z}^{\text{obs}},\nxmuvecobs,\nkmuvecobs,\nxneqmuvecobs,\nmc)-2 \text{log}(1/2\eppa)- \text{log}(2/\epev),0\},
    \end{split}
\end{align}
where $\bone(\cdot)$ and $\be(\cdot)$ are bounds we find in \cref{eq:BeB1}.
By a similar analysis as in \cite[Theorem 4]{devEUR}, our decoy-state BB84 protocol with memoryless detectors is $(2\epat+\eppa+\epev)$-secure.

For decoy-state BB84 protocol with correlated detectors, since the decoy analysis \cite{charles} we use in this work does not depend on the detectors, one can perform the same decoy analysis on the outcomes obtained from the rounds Bob does not reject. Then, one can prove security in the same way as the above analysis. We again state the key length without proof:

\begin{align}\label{eq:keyratecorrelated}
    \begin{split}
    &l(e_{Z}^{\text{obs}},\nxmuvecobscor,\nkmuvecobscor,\nxneqmuvecobscor,\nmccor)\\
    &:= \text{max}\{\bone^{cor}(\nkmuvecobscor,\nmccor)\biggl(1-h({\be^{cor}}(\nxmuvecobscor,\nkmuvecobscor,\nxneqmuvecobscor,\nmccor))\biggr)\\&-\lambda_{\text{EC}}(e_{Z}^{\text{obs}},\nxmuvecobscor,\nkmuvecobscor,\nxneqmuvecobscor,\nmccor)-2\text{log}(1/2\eppa)- \text{log}(2/\epev),0\}\\
    & \epat^2= \epazua^2 + \epazub^2 +\eppntb^2+\epzero^2 +9\epatd^2,
    \end{split}
\end{align}
where
\begin{align}
    \begin{split}
    &\bone^{cor}(\nkmuveccor,\nmcbcor) = \bminone(\nkmuveccor) - nq_Z-\frac{\nmcb}{\lm}\\&\phantom{(\nkmuveccor,\nmcbcor)11111}-2\sqrt{{-\ln(\epzero)}n}-\frac{-2\log(\eppntb)}{4{\lm}^2}-\frac{\sqrt{(-\log(\eppntb))(-\log(\eppntb)+4\lm\nmcb)}}{2{\lm}^2}\\[5pt]
    &\be^{cor}(\nxmuveccor,\nkmuveccor,\nxneqmuveccor,\nmcbcor) =\\ 
    &\frac{\bmaxone(\nxneqmuveccor)}{a\bone^{cor}(\nkmuveccor,\nmcbcor)}+\frac{\delta}{a}\biggl(\frac{\bmaxone(\nxmuveccor)}{\bone^{cor}(\nkmuveccor,\nmcbcor)}+1\biggr)\\&+\sqrt{\frac{-2\ln(\epazua^2)}{a}\biggl(\frac{\bmaxone(\nxmuveccor)}{\bone^{cor}(\nkmuveccor,\nmcbcor)^2}+\frac{1}{\bone^{cor}(\nkmuveccor,\nmcbcor)}\biggr)}\\
    &+\sqrt{{-2\ln(\epazub^2)}\biggl(\frac{\nxbcor}{\bone^{cor}(\nkmuveccor,\nmcbcor)^2}+\frac{1}{\bone^{cor}(\nkmuveccor,\nmcbcor)}\biggr)}
    \end{split}
\end{align}
$\bmaxone(\cdot)$, $\bminone(\cdot)$ are defined in \cref{eq:BmaxBmin} from decoy analysis. $\nkmuveccor$, $\nxmuveccor$, $\nxneqmuveccor$, and $\nmcbcor$ are obtained from rounds that Alice and Bob keep (see how to obtain kept rounds in \cref{sec:correlated protocol description}). The decoy-state passive BB84 protocol with correlated detectors is $(2\epat+\eppa+\epev)$-secure.

\section{\texorpdfstring{$\lm(\Pign\Gamma_{\mathrm{mc}}\Pign)$ calculation}{}}\label{appendix:lmcalculation}

In this appendix, we find a lower bound on $\lm(\Pign\Gamma_{\mathrm{mc}}\Pign)$ for the detection setup depicted in \cref{diag:passivedetection}, where the parameters describing the setup are the detector efficiencies $\{\eta_i\}_{i=1}^4$ and dark count rates $\{d_i\}_{i=1}^4$ of each detector, and the beam-splitting ratio $s$. Furthermore, we lower bound $\lm(\Pign\Gamma_{\mathrm{mc}}\Pign)$ when parameters $\{\eta_i\}_{i=1}^4$,$\{d_i\}_{i=1}^4$ and $s$ are only known to be in some ranges. Technically, $\Gmc$ is the joint POVM element for multi-click outcomes on Alice and Bob's systems. But it has an identity on Alice's system. Therefore, we only need to take care of Bob's multi-click POVM element.

We remind the reader of the notations:
\begin{itemize}
    \item \( \Gamma_{\mathrm{mc}}^{(\vec{\eta},\vec{d})} \):  
    Bob's multi-click POVM element corresponding to the original setup with both losses and dark counts.  
    We denote \(\Gamma_{\mathrm{mc}}^{(\vec{\eta},\vec{d})} = \Gmc \).  
    Here \( \vec{\eta} = [\eta_1, \ldots, \eta_k] \) and \( \vec{d} = [d_1, \ldots, d_k] \) are the detector efficiencies and dark count rates.
    
    \item \( \Gamma_{\mathrm{mc}}^{\vec{\eta}} \):  
    The multi-click POVM element corresponding to the setup with losses only (i.e., without dark counts).
    
    \item \( \Gamma_{\mathrm{mc}}^{\etacom} \):  
    The multi-click POVM element corresponding to the setup without dark counts in which all detectors are assigned the same efficiency  
    \( \etacom = \min\{\eta_1, \ldots, \eta_k\} \) \footnote{In \cref{eq:D9} in \cref{appendix:adelta}, we also denote this minimum efficiency as $\etamin$.}. 
    That is, we increase the loss of each detector until they match the most lossy one.
    
    \item \( \Gamma_{\mathrm{mc}}^{\text{perfect}} \):  
    The multi-click POVM element corresponding to the setup with perfect threshold detectors.
\end{itemize}
To compute the minimum eigenvalue $\lm(\Pign\Gamma_{\mathrm{mc}}\Pign) = \lm(\Pign\Gamma_{\mathrm{mc}}^{(\vec{\eta},\vec{d})}\Pign)$, we follow the following steps:
\begin{enumerate}
    \item We model dark counts as a post-processing map on the outcomes obtained by setups without dark counts. It maps no-click or single-click outcomes to multi-click outcomes, but it cannot map multi-click outcomes back to no-click or single-click outcomes. Thus, we have 
    \begin{align}
\Gamma_{\mathrm{mc}}^{\vec{\eta}}\leq \Gamma_{\mathrm{mc}}^{(\vec{\eta},\vec{d})}.
    \end{align} 
    This is formally proved via \cref{lemma:mcdarkcounts} and \cref{cor:darkcounts}. 
    \item Next, we show that the multi-click POVM element is not increasing as losses increase in \cref{lemma:mcmonoton} and \cref{cor:lmmono}. Thus, we can increase the loss of each detector to the most lossy one. That is, 
    \begin{align}
\Gamma_{\mathrm{mc}}^{{\etacom}}\leq\Gamma_{\mathrm{mc}}^{\vec{\eta}}.
    \end{align}
    \item Lastly, we compute the minimum eigenvalue $\lm(\Pign\Gamma_{\mathrm{mc}}^{\etacom}\Pign)$. In fact, this can be related to a function of the common loss $\etacom$ and the minimum eigenvalue $\lm(\Pign\Gamma_{\mathrm{mc}}^{\text{perfect}}\Pign)$.
\end{enumerate}

We prove the above statements through the following lemmas and corollaries.

\begin{lemma}\label{lemma:mcdarkcounts}
Let \( \{\vec{\Gamma}^{(\vec{\eta},\vec{d})}\} \) be the POVM of the original detection setup, and let  
\( \{\vec{\Gamma}^{\vec{\eta}}\} \) be the POVM with loss but without dark counts.  
Let \( \{m_1,\ldots,m_k\} \) represent a click pattern, where \( m_l = 1 \) if the \( l^{\text{th}} \) detector clicks and \( m_l = 0 \) otherwise.  
We partition all click patterns into two sets, \( \mathcal{A} \) and \( \mathcal{E} \).  
For any click pattern \( c \in \mathcal{E} \), if changing any 0 (no-click) element in \( c \) to 1 (click) yields another pattern \( c' \) that also lies in \( \mathcal{E} \), then the following relation holds:
    \begin{align}
\sum_{c\in\mathcal{E}}{\Gamma}_c^{\vec{\eta}}\leq \sum_{c\in\mathcal{E}}{\Gamma}_c^{(\vec{\eta},\vec{d})}.
    \end{align}
\end{lemma}

\begin{proof}
    We can express \( \sum_{c\in\mathcal{E}}{\Gamma}_c^{(\vec{\eta},\vec{d})} \) in terms of the POVM elements from the loss-only case together with probabilities describing the dark count post-processing map:
    \begin{align}
    \begin{split}
        \sum_{c\in\mathcal{E}}{\Gamma}_c^{(\vec{\eta},\vec{d})} &= \sum_{c\in\mathcal{E}}\left(\sum_{c'\in\mathcal{E}}\mathscr{P}_{c'\rightarrow c}{\Gamma}_{c'}^{\vec{\eta}}\right) + \sum_{c\in\mathcal{E}}\left(\sum_{c'\in\mathcal{A}}\mathscr{P}_{c'\rightarrow c}{\Gamma}_{c'}^{\vec{\eta}}\right)\\
        &\geq\sum_{c'\in\mathcal{E}}\left(\sum_{c\in\mathcal{E}}\mathscr{P}_{c'\rightarrow c}\right){\Gamma}_{c'}^{\vec{\eta}},
    \end{split}  
    \end{align}
  where \( \mathscr{P}_{c' \rightarrow c} \) denotes the probability that the outcome \( c' \) is mapped to the outcome \( c \) due to dark counts, and the inequality follows as $\sum_{c\in\mathcal{E}}\left(\sum_{c'\in\mathcal{A}}\mathscr{P}_{c'\rightarrow c}{\Gamma}_{c'}^{\vec{\eta}}\right)$ is a positive semidefinite operator.
  Now, first note that any classical post-processing can be modelled as a stochastic matrix, i.e., $\sum_{c}\mathscr{P}_{c'\rightarrow c} = 1$ for any click pattern $c'$ and where the sum ranges over all click patterns.
  Further, note that since dark counts only convert a no-click (0) in a given detector to a click (1), and never remove existing clicks, any dark count post-processing can only increase the number of clicks. This directly implies that $\mathscr{P}_{c'\rightarrow c} = 0$ for all $c\in\mathcal{E}$ and $c'\in\mathcal{A}$.
  Therefore, 
    \begin{align}
        1= \sum_{c}\mathscr{P}_{c'\rightarrow c} = \sum_{c\in\mathcal{E}}\mathscr{P}_{c'\rightarrow c},   
    \end{align}
    and we conclude our proof
     \begin{align}
\sum_{c\in\mathcal{E}}{\Gamma}_c^{\vec{\eta}}\leq \sum_{c\in\mathcal{E}}{\Gamma}_c^{(\vec{\eta},\vec{d})}.
    \end{align}
\end{proof}
\begin{corollary}\label{cor:darkcounts}
For any detection setup with multiple detectors, a multi-click outcome is defined as the outcome with more than one detector click.
     Let \( \Gamma^{(\vec{\eta},\vec{d})}_{\mathrm{mc}} \) be the multi-click POVM element of the original setup with losses and dark counts, and let  
    \( \Gamma^{\vec{\eta}}_{\mathrm{mc}} \) be the multi-click POVM element of the setup without dark counts. Then, we have the following relation:
\begin{align}
\Gamma^{\vec{\eta}}_{\mathrm{mc}}\leq\Gamma^{(\vec{\eta},\vec{d})}_{\mathrm{mc}}
\end{align}
\begin{proof}
    It is clear that any multi-click pattern \( c \) can only be mapped to another multi-click pattern \( c' \) by turning additional no-click elements in \( c \) into clicks. The result follows from \cref{lemma:mcdarkcounts}.
\end{proof}
\end{corollary}
\begin{figure}
    \centering
    \scalebox{1}{\begin{tikzpicture}
\tikzset{
  pics/detector/.style args={#1}{
    code={
      \draw[line width=0.5mm] (0,-1) -- (0,1);
      \draw[line width=0.5mm] (0,1) -- ++(0.1,0);
      \draw[line width=0.5mm] (0,-1) -- ++(0.1,0);
      \draw[line width=0.5mm] (0.1,1) arc (90:-90:1);
      \node[align=center] (#1) at (0,0) {};
    }
  }
}
\tikzset{
  pics/beamsplitter/.style args={#1,#2,#3}{
    code={
    \coordinate (A) at (0.3,0.3);
    \coordinate (B) at (-0.3,-0.3);
    \coordinate (#1) at (0,0);
    \draw[line width=0.3mm] (A) -- (B);
    \node[align=center,font=\fontsize{7}{7}\selectfont] at #3 {#2};
    }
  }
}
\tikzset{
  pics/pulse/.style args = {#1}{
    code ={
    \draw[#1] plot[smooth,tension=1] coordinates {(0,0) (0.3,0.2) (0.5,1) (0.7,0.2) (1,0)};
    }
  }
}
\tikzstyle{process} = [rectangle, line width=0.3mm, minimum width=2cm, minimum height=6cm, text centered, text width=1cm, draw=black]



\node[process,font=\large](LinOptics){\(\Phi_{\substack{\text{lossy}\\\text{optics}}}\)};
\node[left =2cm of LinOptics,font=\Large](rho){\(\rho\)};
\pic[right=4cm of LinOptics, yshift = 2.4cm, ,scale = 0.5]{detector={d1}};

\pic[right=4cm of LinOptics, yshift = 1.2cm, ,scale = 0.5]{detector={d2}};

\pic[right=4cm of LinOptics, yshift = 0cm, ,scale = 0.5]{detector={d3}};

\pic[right=4cm of LinOptics, yshift = -2.4cm, ,scale = 0.5]{detector={d4}};
\draw (rho) -- (LinOptics.west);
\draw ([yshift=2.4cm]LinOptics.east) -- (d1);
\draw ([yshift=1.2cm]LinOptics.east) -- (d2);
\draw ([yshift=0cm]LinOptics.east) -- (d3);
\draw ([yshift=-2.4cm]LinOptics.east) -- (d4);
\draw (3.6,2) --(4.2,2.8); 
\node at (3.7,2.6){\(\Delta\eta\)};
\draw[red] (-2,-3.8) -- (-2,3.8);
\draw[red] (6.5,-3.8) -- (6.5,3.8);
\draw[red] (-2,-3.8) -- (6.5,-3.8);
\draw[red] (-2,3.8) -- (6.5,3.8);
\node at (6.2,-4.2)[font=\Large]{\(\{\vec{\Gamma}^{\vec{\eta}\;'}\}\)};
\draw[blue] (4.5,-3.3) -- (4.5,3.3);
\draw[blue] (6,-3.3) -- (6,3.3);
\draw[blue] (6,-3.3) --(4.5,-3.3);
\draw[blue] (6,3.3) --(4.5,3.3);
\node at (4.2,-3.5)[font=\Large]{\(\{\vec{\Pi}\}\)};
\draw[dashed] (2.5,-3)--(2.5,3.3);
\node at (2.5,-3.5)[font=\Large]{\(\rho'\)};

\node at (1.8,-1)[font=\Large] {$\vdots$};
\node at (1.8,-1.4)[font=\Large] {$\vdots$};

\node at (1.8,2.6){\(b_1\)};
\node at (1.8,1.4){\(b_2\)};
\node at (1.8,0.2){\(b_3\)};
\node at (1.8,-2.2){\(b_k\)};

\end{tikzpicture}}
\caption{Definition of \( \rho' \): the state after the linear optics and detector losses, but before the additional loss. In this example, the extra loss is applied to mode \( b_1 \).}
    \label{fig:mc_mononicity}
\end{figure}
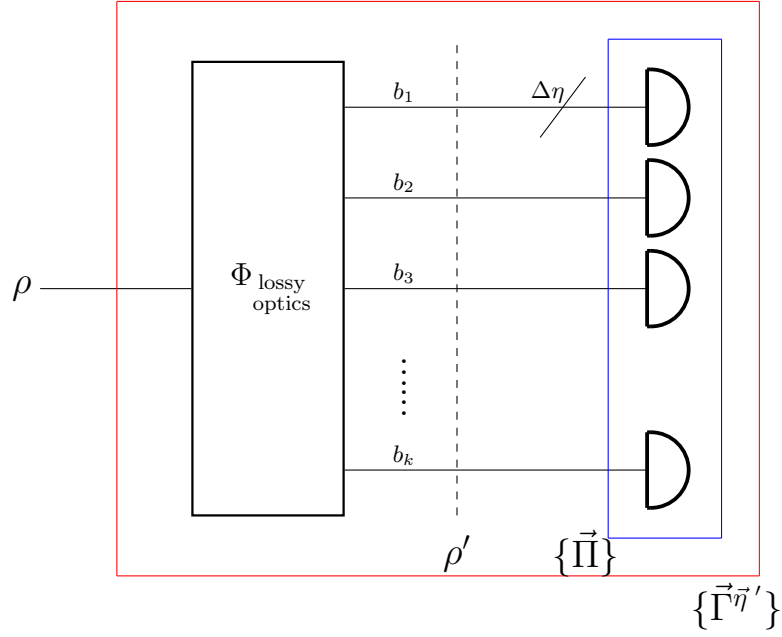
\begin{lemma}\label{lemma:mcmonoton}
    Let $\{\vec{\Gamma}^{\vec{\eta}}\}$ be Bob's POVM in the loss-only case, where \( \vec{\eta} = [\eta_1, \ldots, \eta_k] \) represents the efficiency of each detector. Let $\{\vec{\Gamma}^{\vec{\eta}'}\}$ be the POVM in the loss-only case, where $\vec{\eta}{\;'}=[\eta_1,\; ...,\; \eta_{i-1},\;\; \Delta\eta\times\eta_i, \;\;\eta_{i+1},\; ...,\; \eta_k]$. That is, we increase the loss at the $i^{\text{th}}$ detector. Let the set $\{m_1,...,m_k\}$ represent a click pattern, where $m_l =1$ if the $l^{\text{th}}$ detector clicks and $m_l =0$ otherwise. We separate all the click patterns into two sets $\mathcal{A}$ and $\mathcal{E}$ with the following properties. 
If $c \in \mathcal{A}$ and the $i$-th element of $c$ is 1, then the modified pattern $c'$ obtained by setting the $i$-th element of $c$ to 0 also belongs to $\mathcal{A}$. 

Then, we have the following relation:
\begin{align}\label{eq:C1}
\sum_{c\in\mathcal{E}}\Gamma_{c}^{\vec{\eta}}{'} \leq \sum_{c\in\mathcal{E}}\Gamma_{c}^{\vec{\eta}}.
\end{align}
\begin{proof}
     Consider the case where the detectors have efficiencies \( \vec{\eta} \). 
We model these losses as beam splitters with the corresponding splitting ratios placed in front of perfect threshold detectors. 
For any input state \( \rho \) entering the detection setup, it first passes through the linear optics and then through the beam splitters modeling the losses. Let \( \rho' \) denote the state at this stage, i.e., immediately before the extra loss ($\Delta\eta$) on the \( i^{\text{th}} \) mode and before the perfect threshold detectors. $\rho'$ is some state depending on the input state, the linear optical setups, and the losses. See~\cref{fig:mc_mononicity} for a diagrammatic illustration.
We then have the following expression in terms of projectors onto the output modes \( b_j \):
    \begin{align}
\Trace\left(\rho\sum_{c\in\mathcal{E}}\Gamma_{c}^{\vec{\eta}}\right) = \Trace\left(\rho{'}\sum_{c\in\mathcal{E}}\Pi_{c}\right) ,\quad \Trace\left(\rho\sum_{c\in\mathcal{E}}\Gamma_{c}^{\vec{\eta}'}\right) = \Trace\left(\Phi_{\Delta\eta}(\rho{'})\sum_{c\in\mathcal{E}}\Pi_{c}\right)\quad\forall \rho,
    \end{align}
    where
    \begin{align}\label{eq:C10}
        \begin{split}
        &\sum_{c\in\mathcal{E}}\Pi_{c} = \idd-\sum_{c\in\mathcal{A}}\Pi_{c}\\
        &\Pi_c = \bigotimes_{l=1}^kP_{m_l,b^{'}_l}\\
        &P_{m_l,b^{'}_l}=\begin{cases}
            \kb{0}_{b^{'}_l}  & \text{if } m_l=0, \\
             \idd_{b^{'}_l}-\kb{0}_{b^{'}_l}  & \text{if } m_l=1, \\
        \end{cases}
        \end{split}
    \end{align}
and \( \Phi_{\Delta\eta} \) is a lossy channel with transmittance \( \Delta\eta \) acting only on the \( i^{\text{th}} \) mode. For the modes $b^{'}_l$, $b^{'}_l=b_l$ if $l\neq i$.
Note that \( \Phi_{\Delta\eta} \) acts as the identity on all other modes. Instead of handling the lossy channel \( \Phi_{\Delta\eta} \), we look at its dual map on the projectors:
\begin{align}
     \Trace\left(\Phi_{\Delta\eta}(\rho{'})\sum_{c\in\mathcal{E}}\Pi_{c}\right) = \Trace\left(\rho{'}\Phi^{\dagger}_{\Delta\eta}\left(\sum_{c\in\mathcal{E}}\Pi_{c}\right)\right).
\end{align}
Thus, if we wish to prove the monotonic relation \cref{eq:C1}, we can turn our attention to proving the monotonic relation on the projectors:
\begin{align}
\sum_{c\in\mathcal{A}}\Pi_{c}
    \leq
    \Phi_{\Delta\eta}^{\dagger}\!\left(\sum_{c\in\mathcal{A}}\Pi_{c}\right)
&\iff
\sum_{c\in\mathcal{E}}\Pi_{c}
    \geq
    \Phi_{\Delta\eta}^{\dagger}\!\left(\sum_{c\in\mathcal{E}}\Pi_{c}\right)\label{eq:C12}
\\[4pt]
&\iff
\Tr\!\left(
    \rho'\,
    \Phi_{\Delta\eta}^{\dagger}\!\left(\sum_{c\in\mathcal{E}}\Pi_{c}\right)
\right)
\leq
\Tr\!\left(
    \rho'\,
    \sum_{c\in\mathcal{E}}\Pi_{c}
\right)
\qquad \forall\, \rho'
\label{eq:C2}
\\[4pt]
&\implies
\Tr\!\left(
    \rho \sum_{c\in\mathcal{E}} \Gamma_{c}^{\vec{\eta}'}
\right)
\leq
\Tr\!\left(
    \rho \sum_{c\in\mathcal{E}} \Gamma_{c}^{\vec{\eta}}
\right)
\qquad \forall\, \rho
\label{eq:C3}
\\[4pt]
&\iff
\sum_{c\in\mathcal{E}}\Gamma_{c}^{\vec{\eta}'}
    \leq
    \sum_{c\in\mathcal{E}}\Gamma_{c}^{\vec{\eta}} .\label{eq:C4}
\end{align}

The implication in \cref{eq:C3} follows from the fact that, to establish the monotonicity relation for the POVM elements \( \Gamma_c^{\vec{\eta}} \), it is sufficient that \cref{eq:C2} holds for some state \( \rho' \). In other words, once the ordering is preserved at the level of the projectors \( \Pi_c \), the corresponding ordering for the POVM elements \( \Gamma_c \) follows for all states \( \rho \).
The lossy channel \( \Phi_{\Delta\eta} \) has the following Kraus form:
    \begin{align}
    \begin{split}
    &\Phi_{\Delta\eta}(\rho{'}) = \sum_{n=0}^{\infty}K_n\rho'K_n^{\dagger}\\
    & K_n = \idd_{b_1 ... b_{i-1}}\otimes\sum_{m=n}^{\infty}\sqrt{{m\choose n}\left(1-\Delta\eta\right)^{n} \left(\Delta\eta\right)^{m-n}} \ket{m-n}_{b^{'}_l}\bra{m}_{b_i} \otimes \idd_{b_{i+1} ... b_k}.
     \end{split}
    \end{align}
Let us consider the structure the projector $\sum_{c\in\mathcal{A}}\Pi_{c}$. Consider any click pattern \( c \in \mathcal{A} \) such that the \( i^{\text{th}} \) entry of \( c \) is \( 1 \) (i.e., \( m_i = 1 \)).  
By assumption, the pattern \( c' \), obtained from \( c \) by changing only the \( i^{\text{th}} \) entry from \( 1 \) to \( 0 \) (i.e., \( m_i = 0 \)), also lies in \( \mathcal{A} \).  
In other words, we can naturally pair each such pattern \( c \) with its corresponding pattern \( c' \). Thus, from the construction of the projector in \cref{eq:C10},
\begin{align}
    \Pi_c+\Pi_{c'} = G_{b_1..b_{i-1}} \otimes \idd_{b_i} \otimes  H_{b_1..b_{i-1}}
\end{align}
for some positive operator $ G_{b_1..b_{i-1}}$ and $H_{b_1..b_{i-1}}$.
Since the channel $\Phi_{\Delta\eta}$ 
\begin{align}
    \Phi^{\dagger}_{\Delta\eta}\left(\Pi_c+\Pi_{c'}\right) = \Pi_c+\Pi_{c'}
\end{align}
For the rest of the click patterns $c''$ in $\mathcal{A}$, we know
\begin{align}
   \Pi_{c''} = F_{b_1..b_{i-1}} \otimes \kb{0}_{b_i} \otimes  J_{b_1..b_{i-1}}
\end{align}
for some positive operators $ F_{b_1..b_{i-1}}$ and $J_{b_1..b_{i-1}}$. Then,
\begin{align}
    \Phi^{\dagger}_{\Delta\eta}\left(\Pi_{c''}\right) = 
F_{b_1..b_{i-1}} \otimes  \left(\sum_{n=0}^{\infty}\left(1-\Delta\eta\right)^n\kb{n}_{b_i}\right) \otimes  J_{b_1..b_{i-1}}\geq \Pi_{c''}.
\end{align}
Thus, we conclude that $\sum_{c\in\mathcal{A}}\Pi_{c}\leq\Phi_{\Delta\eta}^{\dagger}\left(\sum_{c\in\mathcal{A}}\Pi_{c}\right)$. By \cref{eq:C12,eq:C2,eq:C3,eq:C4}, we conclude the result.
\end{proof}
\end{lemma}
\begin{corollary}\label{cor:lmmono}
    Let \( \Gamma^{\vec{\eta}}_{\mathrm{mc}} \) be the multi-click POVM element in the loss-only case, and let  
    \( \Gamma^{\etacom}_{\mathrm{mc}} \) be the multi-click POVM element obtained by increasing the loss of each detector to the most lossy one.  
    Then,
    \begin{align}
        \lambda_{\min}\bigl( \Pign \Gamma^{\etacom}_{\mathrm{mc}} \Pign \bigr)
        \;\leq\;
        \lambda_{\min}\bigl( \Pign \Gamma^{\vec{\eta}}_{\mathrm{mc}} \Pign \bigr).
    \end{align}
\end{corollary}

\begin{proof}
It is easy to verify that the single-click and no-click patterns forming the set \( \mathcal{A} \) satisfy this condition.  
For the \( i^{\text{th}} \) detector, the only pattern in which that detector clicks is the single–\( i^{\text{th}} \)-click pattern.  
Changing the \( i^{\text{th}} \) entry of this pattern to a no-click yields the global no-click pattern, which also lies in \( \mathcal{A} \).  
Thus, the multi-click POVM element is monotonic under increasing loss in any detector. In order to show the eigenvalue inequality for the multi-click POVM elements,
    we increase the loss of each detector to \( \etacom \) one at a time.  
    By~\cref{lemma:mcmonoton}, each such replacement induces an operator inequality of the form
    \(
        \Gamma^{\vec{\eta}'}_{\mathrm{mc}}
        \leq
        \Gamma^{\vec{\eta}}_{\mathrm{mc}},
    \)
    where \( \vec{\eta}\;' \) is obtained by replacing one more detector efficiency with \( \etacom \).
    Chaining these inequalities yields
    \begin{align}
        \Gamma^{\etacom}_{\mathrm{mc}}
        \;\leq\;
        \Gamma^{\vec{\eta}}_{\mathrm{mc}}.
    \end{align}

    Let \( \rho^{*} \) be a density operator that achieves the minimum in  
    \( \lambda_{\min}\!\bigl( \Pign \Gamma^{\vec{\eta}}_{\mathrm{mc}} \Pign \bigr) \).  
    Using the operator inequality above, we obtain
    \begin{align}
        \lambda_{\min}\bigl( \Pign \Gamma^{\etacom}_{\mathrm{mc}} \Pign \bigr)
        &\leq
        \Tr\!\left( \Pign \Gamma^{\etacom}_{\mathrm{mc}} \Pign \, \rho^{*} \right)
        \leq
        \Tr\!\left( \Pign \Gamma^{\vec{\eta}}_{\mathrm{mc}} \Pign \, \rho^{*} \right)
        =
        \lambda_{\min}\bigl( \Pign \Gamma^{\vec{\eta}}_{\mathrm{mc}} \Pign \bigr),
    \end{align}
    completing the proof.
\end{proof}
Now the task reduces to computing the minimum eigenvalues of multi-click POVM elements with common loss in the $(>\pn)$-photon subspace.
Since the common losses commute with the linear optics, without loss of generality, we can view that the input state to our passive detectors first goes through a beam splitter with splitting ratio $\etacom$, followed by a lossless passive detection setup. 

Intuitively, the beam splitter (or equivalently the lossy channel) maps an $N$-photon state to a state with $L$ photons with probability ${N\choose L}(1-\etacom)^{N-L}(\etacom)^{L}$. Then, we can simply relate this probability to the minimum eigenvalue for multi-click POVM elements for the lossless detection setup to obtain the desired bound. We use the following lemma to formalize our intuition.

\begin{lemma}
    Let $\rhoN$ be a normalized state with $N$ photons. Let $\Pi_N$ be the projector onto the $N$ photon subspcae ($\Pi_N\rhoN\Pi_N=\rhoN$). Let $\Phi_{\eta}$ be a lossy channel with transmittance $\eta$. Then, the output state $\Phi_{\eta}(\rhoN)$ has the following form:
    \begin{align}
        \begin{split}
            \Phi_{\eta}(\rhoN) = \bigoplus_{L=0}^{N}p_L\rho^{(L)},
        \end{split}
    \end{align}
where $p_L = {N\choose L}(1-\eta)^{N-L}(\eta)^{L}$ and $\rho^{(L)}$ is some normalized state with $L$ photons across all modes.
\end{lemma} 
\begin{proof}
Given that the lossless detection setup consists of threshold detectors, we can assume $\Phi_{\eta}\left(\rhoN\right)$ is block diagonal in photon number:
    \begin{align}
        \begin{split}
            \Phi_{\eta}(\rhoN) = \bigoplus_{L=0}^{N}p_L\rho^{(L)},
        \end{split}
    \end{align}
where $p_L\geq0$ is some probability distribution. Specifically,
\begin{align}
\begin{split}
    p_L &= \Trace\left(  \Phi_{\eta}(\rhoN)\Pi_L\right)\\
    \Pi_L &=  \sum\limits_{\substack{m_1+...+m_k=L}} \ket{m_1}\bra{m_1}_{a_1}\otimes...\otimes\ket{m_k}\bra{m_k}_{a_k}.
   \end{split} 
\end{align}
Because of the projector $\Pi_L$, to compute $p_L$, we only care about the diagonal terms $\ket{m_i}\bra{m_i}_{a_i}$ in $ \Phi_{\eta}(\rhoN)$.
A state with $N$ photons across $k$ modes can be written in the following form:
\begin{align}
    \rhoN = \sum\limits_{\substack{m_1+...+m_k=N\\n_1+...n_k=N}} \alpha_{m_1...m_k n_1...n_k}\ket{m_1}\bra{n_1}_{a_1}\otimes...\otimes\ket{m_k}\bra{n_k}_{a_k}
\end{align}
The common loss channel acts identically on all modes. Moreover, the off-diagonal components remain off-diagonal after the channel:
\begin{align}
\Phi_{\eta}\!\left(\ket{m_i}\bra{n_i}_{a_i}\right)
= \sum_{l_i=0}^{\min\{m_i,n_i\}}
\sqrt{{m_i \choose l_i}{n_i \choose l_i}}\,
\eta^{\frac{m_i+n_i}{2}}\,(1-\eta)^{l_i}\,
\ket{m_i-l_i}\bra{n_i-l_i}.
\end{align}
Therefore, for the purpose of computing \( p_L \), we may, without loss of generality, assume that \( \rho_N \) has nonzero entries only on the diagonal:
\begin{align}
    \rhoN
    = \sum_{\substack{m_1 + \cdots + m_k = N}}
    \alpha_{m_1 \ldots m_k}\,
    \ket{m_1}\bra{m_1}_{a_1}
    \otimes \cdots \otimes
    \ket{m_k}\bra{m_k}_{a_k}.
\end{align}
The output state after the channel is
\begin{align}
\begin{split}
    \Phi_{\eta}(\rhoN)
    = \sum_{\substack{m_1 + \cdots + m_k = N}}
    \alpha_{m_1 \ldots m_k}\,
    \left(\sum_{l_1=0}^{m_1}{m_1 \choose l_1}\eta^{m_1}(1-\eta)^{l_1}\ket{m_1-l_1}\bra{m_1-l_1}_{a_1}\right)\otimes ...\\...\otimes\left(\sum_{l_k=0}^{m_k}{m_k \choose l_k}\eta^{m_k}(1-\eta)^{l_k}\ket{m_k-l_k}\bra{m_k-l_k}_{a_k}\right)
    \end{split}
\end{align}
Then,
\begin{align}
\begin{split}
    p_L 
    &= \Tr\!\left( \Phi_{\eta}(\rho_N)\, \Pi_L \right) \\[3pt]
    &= \sum_{\substack{m_1 + \cdots + m_k = N}}
        \alpha_{m_1 \ldots m_k}\,
        \eta^{N}(1-\eta)^{N-L}
        \sum_{\substack{l_1+\cdots+l_k = N-L}}
        {m_1 \choose l_1}\cdots{m_k \choose l_k}
        \\[2pt]
    &= \eta^N(1-\eta)^{N-L}\biggl(\sum_{\substack{m_1 + \cdots + m_k = N}}\alpha_{m_1 \ldots m_k}\biggr) {N\choose {N-L}}\\
        &={N\choose {N-L}}\eta^N(1-\eta)^{N-L},
\end{split}
\end{align}
where the third equality follows from applying Vandermonde's identity, and the third equality follows from the fact that 
\( \Tr(\rhoN) = \sum_{m_1 + \cdots + m_k = N} \alpha_{m_1 \ldots m_k} = 1 \), 
since \( \rhoN \) is a density operator.
\end{proof}
Therefore, we can now compute 
\begin{align}
\begin{split}
    \lm(\Pign \;{\Gamma^{\etacom}_{\mathrm{mc}}}\; \Pign) &\geq \min\limits_{N>\pn}\{\lm(\Pi_N \;{\Gamma^{\etacom}_{\mathrm{mc}}}\; \Pi_N)\}\\
    &=\min\limits_{N>\pn}\{\min\limits_{\rhoN}\Trace(\rhoN\Pi_N \;{\Gamma^{\etacom}_{\mathrm{mc}}}\; \Pi_N)\}\\
     &=\min\limits_{N>\pn}\{\min\limits_{\rhoN}\Trace(\Phi_{\etacom}(\rhoN)\;{\Gamma^{\text{perfect}}_{\mathrm{mc}}}\;)\}\\
     &=\min\limits_{N>\pn}\left\{\sum_{L=0}^{N}p_L\lm\left(\Pi_L\Gamma^{\text{perfect}}_{\mathrm{mc}}\Pi_L\right)\right\},
     \end{split}
\end{align}
where $\Gamma^{\text{perfect}}_{\mathrm{mc}}$ is the multi-click POVM for the detection setup without loss. This is easy to compute. In fact, Ref.~\cite[Theorem 1]{lars} provides a way to compute this quantity for arbitrary linear optical setups. For our case, $\lm\left(\Pi_L\Gamma^{\text{perfect}}_{\mathrm{mc}}\Pi_L\right) = 1-s^{L}-(1-s)^{L}$ for $L>1$. Thus, we can compute a lower bound on $\lm(\Pign {\Gamma}_{\mathrm{mc}} \Pign)$ for the case $M=1$ and imperfection parameters are in the ranges \( \eta_i \in [\eta^l_i,\eta^u_i] \) 
and \( \pxb \in [s^* - \theta, s^* + \theta] \):
\begin{align}\label{eq:lmformula}
    \begin{split}
        \lm(\Pigone {\Gamma}_{\mathrm{mc}} \Pigone)&\geq 
         \etamin^2\left(1-{s}^2-(1-s)^2\right)\\
         &= \etamin^2 2s(1-s)\\
       &\geq\etamin^2 2\bar{s}(1-\bar{s}),
    \end{split}
\end{align}
where \( \eta_{\min} \) is the minimum efficiency across all detectors, and 
\( \bar{s} = \tfrac{1}{2} + \max\!\left\{ \lvert s^{*} - \theta - \tfrac{1}{2} \rvert,\; \lvert s^{*} + \theta - \tfrac{1}{2} \rvert \right\} \), 
which corresponds to the value of \( s \) in the allowed range that is farthest from \( \tfrac{1}{2} \).
\section{\texorpdfstring{Bounding $a,\delta$}{}}\label{appendix:adelta}
In this appendix, we show how to choose $a$ and bound $\delta$ in the case $\pn=1$, which quantifies the efficiency mismatch between different bases.
We restate the definition of $\delta$ from \cref{deltabound}, with $\pn=1$:
\begin{align}\label{eq:D0}
\delta:=\infnormlong{\sqrt{\ftxone} \GtXneqone \sqrt{\ftxone} - a\sqrt{\ftzone} \GtXneqone \sqrt{\ftzone}},
\end{align}
where $a>0$ can be freely picked.
In particular, we upper bound $\delta$ when the imperfections are partially characterised; that is, the beam splitting ratio $s$, the efficiencies $\eta_i$ and dark count rates $d_i$ are in some ranges (\( \eta_i \in [\eta^l_i,\eta^u_i] \), 
\( d_i \in [d^l_i,d^u_i] \), 
and \( \pxb \in [s^* - \theta, s^* + \theta] \)).
Computing $\delta$ directly from the definition is a challenging task. Thus, we first simplify the expression to make the task easier. We list the simplifications to provide an overview of the rest of the appendix before explaining the details.
\begin{enumerate}
    \item 
    We use the triangle inequality to remove the dependence of $\delta$ on the common term $\GtXneqone$ and obtain a bound in terms of $\sqrt{\ftxone}$ and $\sqrt{\ftzone}$ only:
\begin{align}\label{eq:D1}
    \delta \leq (1+\sqrt{a})\infnormlong{\sqrt{\ftxone}-\sqrt{a}\sqrt{\ftzone}}.
\end{align}
    \item For ease of computation \footnote{We choose to further simplify the bound by removing the square root structure of $\ftx$ and $\ftz$. One can explicit compute the bound \cref{eq:D1} as it might be tighter than the current approach in this work. }, we reduce the bound in \cref{eq:D1} so that it depends only on ${\ftxone}$ and ${\ftzone}$ as
    \begin{align}\label{eq:D2}
         \delta \leq \frac{(1+\sqrt{a})}{2}\infnorm{\ftxone-a\ftzone}\;\times\max\biggl\{\frac{1}{\sqrt{\lm(\ftxone)}},\frac{1}{\sqrt{a\lm(\ftzone)}}\biggr\}.
    \end{align}
    \item The bound now depends on the joint POVM elements. We reduce the bound in \cref{eq:D2} to a simpler expression that depends on Bob's POVM elements and Alice's basis choice probability ($\pza,\pxa$) as follows
    \begin{align}\label{eq:D3}
        \delta \leq &\frac{(1+\sqrt{a})}{2}\; \pxa\infnormlong{\sqrt{\POVMBcon}^+}^2\infnormlong{\POVMBX-a\frac{\pza}{\pxa}\POVMBZ}\notag \\&\times\max\left\{\frac{1}{\pxa \lm\biggl(\sqrt{\POVMBcon}^+\POVMBX \sqrt{\POVMBcon}^+\biggr)},\frac{1}{a\pza \lm\biggl(\sqrt{\POVMBcon}^+\POVMBZ \sqrt{\POVMBcon}^+\biggr)}\right\}.
    \end{align}
    We define the coarse-graining of Bob's POVM elements $\POVMBcon,\POVMBX,\POVMBZ$ later in this section. 
    \item We then bound the terms in \cref{eq:D3} separately:
    \begin{itemize}
        \item $\infnormlong{\POVMBX-a\frac{\pza}{\pxa}\POVMBZ}$
        \item $\infnorm{\sqrt{\POVMBcon}^+}^2$
        \item The minimum eigenvalues $\lm\left(\sqrt{\POVMBcon}^+(\POVMBZ/\POVMBX)\sqrt{\POVMBcon}^+\right)$
    \end{itemize}
\end{enumerate}
\subsection{\texorpdfstring{Derivations of \cref{eq:D1} to \cref{eq:D3}}{}}
We first use the triangle inequality and the submultiplicativity of the infinity-norm to remove the dependence of $\delta$ on the common term $\GtXneqone$ and derive \cref{eq:D1} as follows
\begin{align}
    \delta =&\infnormlong{\sqrt{\ftxone} \GtXneqone\sqrt{\ftxone} - a\sqrt{\ftxone} \GtXneqone\sqrt{\ftzone}}\label{eq:D4} \\
    =& \infnormlong{\sqrt{\ftxone}\GtXneqone\sqrt{\ftxone} - \sqrt{a}\sqrt{\ftxone} \GtXneqone\sqrt{\ftzone} +\sqrt{a} \sqrt{\ftxone} \GtXneqone\sqrt{\ftzone} - a\sqrt{\ftzone} \GtXneqone\sqrt{\ftzone}}\label{eq:D5}\\
    \leq& \infnormlong{\sqrt{\ftxone}\GtXneqone}\infnormlong{\sqrt{\ftxone}-\sqrt{a}\sqrt{\ftzone}}+\sqrt{a}\infnormlong{\sqrt{\ftxone}-\sqrt{a}\sqrt{\ftzone}}\infnormlong{\GtXneqone\sqrt{\ftzone}}\label{eq:D6}\\
    \leq& (1+\sqrt{a})\infnormlong{\sqrt{\ftxone}-\sqrt{a}\sqrt{\ftzone}}\label{eq:D7},
\end{align}
where \cref{eq:D7} follows from the fact that 
$\infnormlong{\sqrt{\ftxone}\GtXneqone} \leq 1$ and $\infnormlong{\GtXneqone\sqrt{\ftzone}}\leq 1$.
\cref{eq:D2} follows directly from \cite[Eq.(X.46)]{bhatia2013matrix}:
\begin{align}\label{eq:D8}
\begin{split}
    \infnorm{A^r-B^r}\leq rt^{r-1}\infnorm{A-B}\\
    t = \min\{\lm(A),\lm(B)\},
    \end{split}
\end{align}
and identifying 
\begin{equation}
    A \leftrightarrow \ftxone,\;\; B \leftrightarrow \ftzone,\;\; r \leftrightarrow \half.
\end{equation}

The last simplifying step comes from explicitly describing the joint POVM elements defined in \cref{eq:definePOVMA,eq:definejoint,eq:definefx}, and defining convenient notation to represent the bound as follows. First, the joint POVM elements are given as
   \begin{align}
        \begin{split}
        & \GoneXeq = \pxa\kb{+}\otimes\Gamma^{(B,1)}_{(X,0)} + \pxa\kb{-}\otimes\Gamma^{(B,1)}_{(X,1)}\\
        & \GoneXneq = \pxa\kb{+}\otimes\Gamma^{(B,1)}_{(X,1)} + \pxa\kb{-}\otimes\Gamma^{(B,1)}_{(X,0)}\\
        & \GoneZeq = \pza\kb{0}\otimes\Gamma^{(B,1)}_{(Z,0)} + \pza\kb{1}\otimes\Gamma^{(B,1)}_{(Z,1)}\\
        & \GoneZneq = \pza\kb{0}\otimes\Gamma^{(B,1)}_{(Z,1)} + \pza\kb{1}\otimes\Gamma^{(B,1)}_{(Z,0)}\\[10pt]
        &\fxone = \GoneXeq + \GoneXneq = \idd_A\otimes\pxa\big(\Gamma^{(B,0)}_{(X,0)}+\Gamma^{(B,1)}_{(X,1)}\big)\\
        &\fzone = \GoneZeq + \GoneZneq=\idd_A\otimes\pza\big(\Gamma^{(B,0)}_{(Z,0)}+\Gamma^{(B,1)}_{(Z,1)}\big)\\
        & \Fonescone = \fxone + \fzone = \idd_A\otimes \biggl(\pxa\big(\Gamma^{(B,0)}_{(X,0)}+\Gamma^{(B,1)}_{(X,1)}\big)+\pza\big(\Gamma^{(B,0)}_{(Z,0)}+\Gamma^{(B,1)}_{(Z,1)}\big)\biggr),
        \end{split}
    \end{align}
where $\Gamma^{(B,1)}_{(X/Z,0/1)}$ corresponds to Bob's POVM element $\Gamma^{B}_{(X/Z,0/1)}$ in the single-photon subspace.
For notational simplicity, define
\begin{align}\label{eq:definePOVMBcon}
    \begin{split}
        &\POVMBX:=\Gamma^{(B,1)}_{(X,0)}+\Gamma^{(B,1)}_{(X,1)}\\
        &\POVMBZ:=\Gamma^{(B,1)}_{(Z,0)}+\Gamma^{(B,1)}_{(Z,1)}\\
        &\POVMBcon:=\pxa\POVMBX+\pza\POVMBZ.
    \end{split}
\end{align}
 With this simplified notation, we write
\begin{align}
    \begin{split}
        \ftxone &= \sqrt{\idd_A\otimes\POVMBcon}^+ \big(\idd_A\otimes\pxa\POVMBX\big)\sqrt{\idd_A\otimes\POVMBcon}^+ \\
        &=\idd_A\otimes \pxa \sqrt{\POVMBcon}^+\POVMBX \sqrt{\POVMBcon}^+.
    \end{split}
\end{align}
Similarly, we simply $\ftzone$ to
\begin{align}
     \ftzone =\idd_A\otimes \pza \sqrt{\POVMBcon}^+\POVMBZ \sqrt{\POVMBcon}^+.
\end{align}
Thus, we simplify $\infnorm{\ftxone-a\ftzone}$ from \cref{eq:D2}:
\begin{align}\label{eq:D10}
    \begin{split}
        \infnorm{\ftxone-a\ftzone} &= \infnormlong{\idd_A\otimes \pxa \sqrt{\POVMBcon}^+\POVMBX \sqrt{\POVMBcon}^+-a\idd_A\otimes \pza \sqrt{\POVMBcon}^+\POVMBZ \sqrt{\POVMBcon}^+}\\
        &=\infnormlong{\pxa \sqrt{\POVMBcon}^+\POVMBX \sqrt{\POVMBcon}^+-a\pza \sqrt{\POVMBcon}^+\POVMBZ \sqrt{\POVMBcon}^+}\\
        &=\pxa\infnormlong{\sqrt{\POVMBcon}^+\big(\POVMBX-a\frac{\pza}{\pxa}\POVMBZ\big)\sqrt{\POVMBcon}^+}\\
        &\leq\pxa\infnormlong{\sqrt{\POVMBcon}^+}^2\infnormlong{\POVMBX-a\frac{\pza}{\pxa}\POVMBZ},
    \end{split}
\end{align}
where we used the submultiplicativity of the infinity-norm to obtain the last inequality. Furthermore, we can also simplify the expressions for $\frac{1}{\sqrt{\lm(\ftxone)}}$ and $\frac{1}{\sqrt{a\lm(\ftzone)}}$ as
\begin{align}\label{eq:D11}
    \begin{split}
        \lm(\ftxone) &= \lm\biggl(\idd_A\otimes \pxa \sqrt{\POVMBcon}^+\POVMBX \sqrt{\POVMBcon}^+\biggr)\\
        & =\pxa \lm\biggl(\sqrt{\POVMBcon}^+\POVMBX \sqrt{\POVMBcon}^+\biggr),
    \end{split}
\end{align}
and similarly
\begin{align}\label{eq:D12}
    \lm(a\ftzone) = a\pza \lm\biggl(\sqrt{\POVMBcon}^+\POVMBZ \sqrt{\POVMBcon}^+\biggr).
\end{align}
Combining \cref{eq:D2,eq:D10,eq:D11,eq:D12} gives us \cref{eq:D3}.

\subsection{Bounding the individual terms}
In the remaining subsection, we show the bounds on 
\begin{itemize}[topsep=2pt, itemsep=2pt, parsep=2pt]
    \item \( \left\| \POVMBX - a \frac{\pza}{\pxa} \POVMBZ \right\|_{\infty} \)
    \item \( \left\| \sqrt{\POVMBcon}^{+} \right\|_{\infty} \)
    \item \( \lambda_{\min}\left( \sqrt{\POVMBcon}^{+} \left( {\POVMBX}/{\POVMBZ} \right) \sqrt{\POVMBcon}^{+} \right) \)
\end{itemize}
separately. The computations in this subsection depend on the explicit model of the POVM elements. Thus, before deriving the requisite bounds, we define some notation and explicitly state Bob's POVM elements. Let us denote
\begin{align}
    \begin{split}
       \etazzero = \eta_1\quad
        \etazone = \eta_2\quad
         \etaxzero = \eta_3\quad
          \etaxone = \eta_4\quad\\
          \dzzero = d_1\quad
          \dzone = d_2\quad
          \dxzero = d_3\quad
          \dxone = d_4\quad
    \end{split}
\end{align}
so that it is easier to track the efficiencies and dark count rates for different detectors. Recall that the basis-choice beam splitting ratio is $\pxb$.
Then, Bob's POVM elements are the following:
\begin{align}\label{eq:writedownPOVMB}
    \begin{split}
    &\POVMBZ = \\
     &(\pzb)\biggl((1-\etazzero)(\dzzero+\dzone-2\dzzero\dzone)(1-\dxzero)(1-\dxone)
     \\&\quad\quad\quad\;\;+{\etazzero(1-\dzone)(1-\dxzero)(1-\dxone)}\biggr)\kb{0}\\
     +&(\pzb)\biggl((1-\etazone)(\dzzero+\dzone-2\dzzero\dzone)(1-\dxzero)(1-\dxone)
     \\&\quad\quad\quad\;\;+{\etazone(1-\dzzero)(1-\dxzero)(1-\dxone)}\biggr)\kb{1}\\
     +&\pxb\biggl((1-\etaxzero)(\dzzero+\dzone-2\dzzero\dzone)(1-\dxzero)(1-\dxone)\biggr)\kb{+}\\
     +&\pxb\biggl((1-\etaxone)(\dzzero+\dzone-2\dzzero\dzone)(1-\dxzero)(1-\dxone)\biggr)\kb{-}\\[10pt]  
    &\POVMBX =\\
    &(\pzb)\biggl((1-\etazzero)(1-\dzzero)(1-\dzone)(\dxzero+\dxone-2\dxzero\dxone)\biggr)\kb{0}\\
    +&(\pzb)\biggl((1-\etazone)(1-\dzzero)(1-\dzone)(\dxzero+\dxone-2\dxzero\dxone)\biggr)\kb{1}\\
    +&\pxb\biggl((1-\etaxzero)(1-\dzzero)(1-\dzone)(\dxzero+\dxone-2\dxzero\dxone)
    \\&\quad+{\etaxzero(1-\dzzero)(1-\dzone)(1-\dxone)}\biggr)\kb{+}\\
    +&\pxb\biggl((1-\etaxone)(1-\dzzero)(1-\dzone)(\dxzero+\dxone-2\dxzero\dxone)
    \\&\quad+{\etaxone(1-\dzzero)(1-\dzone)(1-\dxzero)}\biggr)\kb{-}
    \end{split}
\end{align}
We remind the reader that for our partially characterised detectors, we have
\( \eta_{(i,j)} \in [\eta^l_{(i,j)},\eta^u_{(i,j)}] \), 
\( d_{(i,j)}\in [d^l_{(i,j)},d^u_{(i,j)}] \), 
and \( \pxb \in [s^* - \theta, s^* + \theta] \).
We define
\begin{align}\label{eq:D9}
    \begin{split}
    &\etazmin:=\min\limits_{\substack{j\in\{0,1\}}}\{\eta_{(Z,j)}\};\;
    \etaxmin:=\min\limits_{\substack{j\in\{0,1\}}}\{\eta_{(X,j)}\};\;
   \\
   &\etazmax:=\max\limits_{\substack{j\in\{0,1\}}}\{\eta_{(Z,j)}\};\;
   \etaxmax:=\min\limits_{\substack{j\in\{0,1\}}}\{\eta_{(X,j)}\};\;
  \\
    &\etamin:=\min\limits_{\substack{i\in\{X,Z\} \\ j\in\{0,1\}}}\{\eta_{(i,j)}\};\;
    \etamax:=\max\limits_{\substack{i\in\{X,Z\} \\ j\in\{0,1\}}}\{\eta_{(i,j)}\};\\
    &\dmin:=\min\limits_{\substack{i\in\{X,Z\} \\ j\in\{0,1\}}}\{d_{(i,j)}\};\;
    \dmax:=\max\limits_{\substack{i\in\{X,Z\} \\ j\in\{0,1\}}}\{d_{(i,j)}\},
    \end{split}
\end{align}
where the above minimization or maximization is performed over the ranges \( [\eta^l_{(i,j)},\eta^u_{(i,j)}],
 [d^l_{(i,j)},d^u_{(i,j)}] \). We further define the following real numbers $l_Z$,$l_X$,$u_Z$ and $u_X$ such that they satisfy the following operator inequalities
 \begin{align}\label{eq:definelz}
    \begin{split}
    l_Z\idd_2:=& \biggl[\biggl((1-s^*-\theta)(1-\etazmax)+(s^*-\theta)(1-\etaxmax)\biggr)(2\dmin-2\dmax^2)(1-\dmax)^2\\&+(1-s^*-\theta)\etazmin(1-\dmax)^3\biggr]\idd_2 \\ \leq& \;\POVMBZ\\
   l_X\idd_2:=& \biggl[\biggl((1-s^*-\theta)(1-\etazmax)+(s^*-\theta)(1-\etaxmax)\biggr)(2\dmin-2\dmax^2)(1-\dmax)^2\\&+(s^*-\theta)\etaxmin(1-\dmax)^3\biggr]\idd_2
    \\ \leq & \;\POVMBX\\
    u_Z\idd_2:=& \biggl[\biggl((1-s^*+\theta)(1-\etazmin)+(s^*+\theta)(1-\etaxmin)\biggr)(2\dmax-2\dmin^2)(1-\dmin)^2\\&+(1-s^*+\theta)\etazmax(1-\dmin)^3\biggr]\idd_2 \\\geq&\;\POVMBZ\\
   u_X\idd_2:=& \biggl[\biggl((1-s^*+\theta)(1-\etazmin)+(s^*+\theta)(1-\etaxmin)\biggr)(2\dmax-2\dmin^2)(1-\dmin)^2\\&+(s^*+\theta)\etaxmax(1-\dmin)^3\biggr]\idd_2 \\ \geq&\;\POVMBX,
    \end{split}
\end{align}  
which can be seen straightforwardly from the definitions in \cref{eq:writedownPOVMB} and \cref{eq:D9}.

\subsubsection{\texorpdfstring{Bounding $\bm{\infnormlong{\POVMBX - a\frac{\pza}{\pxa}\POVMBZ}}$}{}}
To bound $\bm{\infnormlong{\POVMBX - a\frac{\pza}{\pxa}\POVMBZ}}$,
we need to fix our choice of $a$ for which we make the heuristic choice 
\begin{align}\label{eq:definea}
    a = \frac{\pxa\pxb}{\pza(\pzb)}.
\end{align}
The description of Bob's POVM elements given in \cref{eq:writedownPOVMB} has many terms. Thus, we simplify the calculations by using the triangle inequality as
\begin{align}\label{eq:definedelta1}
    \begin{split}
         &\infnorm{\POVMBX - \frac{\pxb}{(\pzb)}\POVMBZ} \leq \Delta_1 +\Delta_2\eqqcolon\zeta,
    \end{split}
\end{align}
where $\Delta_1$ and $\Delta_2$ are
\begin{align}
        \Delta_1 &=\pxb\norm{
        \begin{aligned}
        {\etaxzero(1-\dzzero)(1-\dzone)(1-\dxone)\kb{+}+\etaxone(1-\dzzero)(1-\dzone)(1-\dxzero)\kb{-}}\\
    -(\etazzero(1-\dxzero)(1-\dxone)(1-\dzone)\kb{0}+\etazone(1-\dxzero)(1-\dxone)(1-\dzzero)\kb{1})
    \end{aligned}}_\infty\\[10pt]
    \Delta_2 &= \norm{\begin{aligned}
    &(\pzb)\biggl((1-\etazzero)(1-\dzzero)(1-\dzone)(\dxzero+\dxone-2\dxzero\dxone)\biggr)\kb{0}\\+&(\pzb)\biggl((1-\etazone)(1-\dzzero)(1-\dzone)(\dxzero+\dxone-2\dxzero\dxone)\biggr)\kb{1}\\ -& \pxb\biggl((1-\etazzero)(\dzzero+\dzone-2\dzzero\dzone)(1-\dxzero)(1-\dxone))\biggr)\kb{0}\\ -& \pxb\biggl((1-\etazone)(\dzzero+\dzone-2\dzzero\dzone)(1-\dxzero)(1-\dxone)\biggr)\kb{1}\\ +&\pxb\biggl((1-\etaxzero)(1-\dzzero)(1-\dzone)(\dxzero+\dxone-2\dxzero\dxone)\biggr)\kb{+}\\ +&
    \pxb\biggl((1-\etaxone)(1-\dzzero)(1-\dzone)(\dxzero+\dxone-2\dxzero\dxone)\biggr)\kb{-}\\ -&
\frac{{\pxb}^2}{\pzb}\biggl((1-\etaxzero)(\dzzero+\dzone-2\dzzero\dzone)(1-\dxzero)(1-\dxone)\biggr)\kb{+}\\-& 
\frac{{\pxb}^2}{\pzb}\biggl((1-\etaxone)(\dzzero+\dzone-2\dzzero\dzone)(1-\dxzero)(1-\dxone)\biggr)\kb{-}
    \end{aligned}}_\infty.
\end{align}

We bound $\Delta_1$ by adding and subtracting $\etamax(1-\dmin)^3\idd_2$ (in other words, using $\infnorm{A-B} = \infnorm{A-P+P-B} \leq \infnorm{A-P} + \infnorm{P-B}$) which results in
\begin{align}
     \begin{split}
    \Delta_1 &\leq \pxb\biggl(2\etamax(1-\dmin)^3-(\etaxmin+\etazmin)(1-\dmax)^3\biggr)\\
    &\leq (s^*+\theta)\biggl(2\etamax(1-\dmin)^3-(\etaxmin+\etazmin)(1-\dmax)^3\biggr).
     \end{split}
\end{align}
A routine, but tedious calculation with multiple uses of the triangle inequality gives us the following bound on $\Delta_2$,
\begin{align}
     \begin{split}
    \Delta_2 \leq& \biggl( \pxb(1-\etamin)+\frac{{\pxb}^2}{\pzb}(1-\etamin) \biggr) \biggl( \dzzero+\dzone-2\dzzero\dzone \biggr) (1-\dxzero)(1-\dxone)\\
    &+\biggl((\pzb)(1-\etamin)+\pxb(1-\etamin)\biggr)(1-\dzzero)(1-\dzone)\biggl(\dxzero+\dxone-2\dxzero\dxone\biggr)\\
    \leq& \frac{1}{\pzb}(1-\etamin)\biggl(2\dmax-2\dmin^2 \biggr) (1-\dmin)^2\\
    \leq& \frac{1}{1-s^*-\theta}(1-\etamin)\biggl(2\dmax-2\dmin^2 \biggr) (1-\dmin)^2.
     \end{split}
\end{align}
\subsubsection{\texorpdfstring{Bounding $\bm{\infnormlong{\sqrt{\POVMBcon}^+}}^2$}{}}
To bound $\bm{\infnormlong{\sqrt{\POVMBcon}^+}}^2$, we 
recall from \cref{eq:definePOVMBcon} that 
\begin{align}
    {\POVMBcon} =&  {\pxa \POVMBX+\pza \POVMBZ}.
\end{align} 
We use the spectral decomposition of $\POVMBcon$ to obtain
\begin{align}\label{eq:reducetolm}
    \infnorm{\sqrt{\POVMBcon}^+} = \frac{1}{\min\{\sqrt{\lambda(\POVMBcon)}\}},
\end{align}
where $\lambda(\POVMBcon)$ denotes the set of positive eigenvalues of $\POVMBcon$. Hence, the problem reduces to finding the smallest non-zero eigenvalue of $\POVMBcon$.
\begin{lemma}\label{lmlemma}
    Let $A$ and $B$ be two finite dimensional matrices.
    If $A\leq B$, $\lm(A)\leq\lm(B)$, where $\lm(\cdot)$ is the minimal eigenvalue.
\end{lemma}
\begin{proof}
    $\lm(A) = \min\limits_{\rho}\Trace{(A\rho)} \leq \Trace{(A\rho')} \leq \Trace{(B\rho')}
    = \min\limits_{\rho}\Trace{(B\rho)} = \lm(B)$, where $\rho'$ achieves the minimum eigenvalue of B.
\end{proof}
Recall from \cref{eq:definelz} that $l_Z\idd_2 \leq \POVMBZ$ and  $l_X\idd_2 \leq \POVMBX$. Thus,
\begin{align}
    \begin{split}
    {\POVMBcon}=\pxa\POVMBX+\pza\POVMBZ \geq \; (\pxa l_X+\pza l_Z)\idd_2 \implies& \lm(\POVMBcon) \geq \pxa l_X+\pza l_Z
    \end{split}
\end{align}
Thus, by \cref{eq:reducetolm}
\begin{align}\label{eq:definek}
 \infnorm{\sqrt{\POVMBcon}^+}^2 \leq k^2 :=\begin{cases}
     (\pxa l_X+\pza l_Z)^{-1} &\text{ if } \pxa l_X+\pza l_Z>0\\
      \infty &\text{ if } \pxa l_X+\pza l_Z\leq0\\
 \end{cases}
\end{align}
Note that in most realistic cases, $\pxa l_X+\pza l_Z>0$.

\subsubsection{\texorpdfstring{Bounding $\bm{\lm(\sqrt{\POVMBcon}^+  (\POVMBX/\POVMBZ) \sqrt{ \POVMBcon}^+)}$}{Text}}
First, we check if $l_Z + l_X= 0$. If this is the case simply use the trivial lower bound,\\ $\lm(\sqrt{\POVMBcon}^+  (\POVMBX/\POVMBZ) \sqrt{ \POVMBcon}^+) \geq 0$. If $l_Z + l_X >0$, then we have $\lm(\POVMBcon) > 0$. This is the typical case, and we bound $\lm(\sqrt{\POVMBcon}^+  (\POVMBX/\POVMBZ) \sqrt{ \POVMBcon}^+)$ using the following method.

Recall from \cref{eq:definelz} that
\begin{align}
    l_Z\idd_2\leq \POVMBZ \quad\text{and}\quad  l_X\idd_2\leq \POVMBX.
\end{align}
Therefore,
\begin{align}
    \begin{split}
   \sqrt{\POVMBcon}^+\POVMBZ \sqrt{ \POVMBcon}^+\geq \sqrt{\POVMBcon}^+ &l_Z\idd_2 \sqrt{ \POVMBcon}^+=l_Z{\POVMBcon}^+\\
   \sqrt{\POVMBcon}^+\POVMBX \sqrt{ \POVMBcon}^+\geq \sqrt{\POVMBcon}^+ &l_X\idd_2 \sqrt{ \POVMBcon}^+=l_X{\POVMBcon}^+
\end{split}
\end{align}
where we use the fact that $ \sqrt{\POVMBcon}^+  (A-B) \sqrt{ \POVMBcon}^+ \geq 0$ if  $A-B \geq 0$.
Thus, applying \cref{lmlemma} directly gives
\begin{align}
    \begin{split}
            \lm( \sqrt{\POVMBcon}^+\POVMBZ \sqrt{ \POVMBcon}^+) \geq l_Z\lm({\POVMBcon}^+) =& l_Z\frac{1}{\max\{\lambda(\POVMBcon)\}} = l_Z\infnorm{\POVMBcon}^{-1}\\
    \lm( \sqrt{\POVMBcon}^+\POVMBX \sqrt{ \POVMBcon}^+) \geq l_X\lm({\POVMBcon}^+) =& l_X\frac{1}{\max\{\lambda(\POVMBcon)\}} = l_X\infnorm{\POVMBcon}^{-1}.
    \end{split}
\end{align}
We now use the following lemma to bound $\infnorm{\POVMBcon}^{-1}$.
\begin{lemma}Let $A$ and $B$ be finite dimensional matrices.
    If $A \geq B$, then $\infnorm{A}\geq\infnorm{B}$
    \begin{proof}
        $\infnorm{A} = \max\limits_{\rho}\Trace(\rho A) \geq \Trace(\rho^* A) \geq \Trace(\rho^* B) =  \max\limits_{\rho}\Trace(\rho B)=\infnorm{B}$ where $\rho^*$ achieves the maximum for $\Trace(\rho B)$.
    \end{proof}
\end{lemma}
Recall \cref{eq:definelz} that
\begin{align}
   \POVMBZ \leq u_Z\idd_2 \quad  \POVMBX \leq u_X\idd_2.
\end{align}
Thus,
\begin{align}
    \begin{split}
   &\POVMBcon \leq (\pxa u_X + \pza u_Z) \idd_2\\
    &\infnorm{\POVMBcon } \leq \pxa u_X + \pza u_Z \\\implies &\infnorm{\POVMBcon }^{-1} \geq (\pxa u_X + \pza u_Z )^{-1}.
    \end{split}
\end{align}
Again, typically, $\pxa u_X + \pza u_Z\geq0$. If $\pxa u_X + \pza u_Z\leq0$, we obtain a trivial bound.
Thus,
\begin{align}
    \begin{split}
    \lm( \sqrt{\POVMBcon}^+\POVMBZ \sqrt{ \POVMBcon}^+) \geq& l_Z\infnorm{\POVMBcon}^{-1}\geq l_Z (\pxa u_X + \pza u_Z)^{-1}\\
    \lm( \sqrt{\POVMBcon}^+\POVMBX \sqrt{ \POVMBcon}^+) \geq& l_X\infnorm{\POVMBcon}^{-1}
    \geq l_X(\pxa u_X + \pza u_Z)^{-1}.
    \end{split}
\end{align}

\subsection{Summary of bounds}

Finally, we summarize the computations in this appendix.
As stated in \cref{eq:definea}, we choose $a$ to be 
\begin{align}
    a = \frac{\pxa\pxb}{\pza(\pzb)}.
\end{align}
to compute a bound on $\delta$:
\begin{align}
   \delta \leq (1+\sqrt{a})\frac{\pxa\zeta k^2}{2}    \max\biggl\{\sqrt{\frac{\pxa u_X + \pza u_Z}{{\pxa}l_Z}},\sqrt{\frac{\pxa u_X + \pza u_Z}{a{\pza}l_X}}\biggr\}.
\end{align}

$\pxa$ and $\pza$ are the probabilities that Alice chooses the $X$ and $Z$ basis, respectively. $l_Z$, $l_X$, $u_Z$ and $u_Z$ are defined in \cref{eq:definelz}. $k^2$ is defined in \cref{eq:definek}. $\zeta$ is defined in \cref{eq:definedelta1}.
We do not know $s$ exactly. Thus, we need to replace $a$ with its upper and lower bounds:
\begin{align}
      & \frac{\pxa(s^*-\theta)}{\pza(1-s^*+\theta)} \leq a \leq \frac{\pxa(s^*+\theta)}{\pza(1-s^*-\theta)}.
\end{align}
Note that we must replace \( a \) with its bounds not only in the computation of \( \delta \) but also in the phase-error rate bound in~\cref{NumberofErrorBound}.
Thus,
\begin{align}
   \delta \leq \left(1+\sqrt{\frac{\pxa(s^*+\theta)}{\pza(1-s^*-\theta)}}\right)\frac{\pxa \zeta k^2}{2}    \max\left\{\sqrt{\frac{\pxa u_X + \pza u_Z}{{\pxa}l_Z}},\sqrt{\frac{\pxa u_X + \pza u_Z}{\frac{\pxa(s^*-\theta)}{\pza(1-s^*+\theta)}{\pza}l_X}}\right\}.
\end{align}
\section{\texorpdfstring{Bounding $q_Z$}{}}\label{appendix:qz}
Recall from \cref{sec:zerobound} that we upper bound the number of 0-photon key rounds, $\nkzero$. This bound is obtained by upper bounding the infinity norm of the POVM elements \( \POVMZeq^{(0)} + \POVMZneq^{(0)} \), as defined in \cref{eq:defineqz}, where \( q_Z \) quantifies the fraction of  0-photon key rounds. In this section, we explicitly upper bound \( q_Z \) in terms of the detector imperfection parameters for passive BB84 detection setups.

The POVM elements in the 0-photon subspace corresponding to the $Z$-basis detection outcomes can be written as
\begin{align}
    \begin{split}
    \gmzne^{(0)}+\gmze^{(0)} &=\pza (1-\dxzero)(1-\dxone)\biggl(\dzzero+\dzone-2\dzzero\dzone\biggr)\idd_A\otimes\kb{\text{vac}}\\
     &\leq \pza(1-\dmin)^2\biggl(2\dmax-2\dmin^2\biggr)\idd_A\otimes\kb{\text{vac}}
    \end{split}
  \end{align}
Thus,
\begin{equation}
    q_Z\leq \pza(1-\dmin)^2\biggl(2\dmax-2\dmin^2\biggr)
\end{equation}
\section{Proof for Correlated Imperfect Detectors}\label{appendix:commute}

 We prove the statement \cref{NumberofErrorBound2}, \cref{step2multibound2} and \cref{step3zerobound2} in this appendix.
 The estimation protocol used for this analysis is in \cref{fig:correlated_equivalent}.
\begin{enumerate}[label=\circletext{\arabic*}, leftmargin=*]
\item They perform the filtering process. For the \( i^{\text{th}} \) round, based on the outcomes of the previous \( l_c \) rounds, 
they decide whether to keep or reject the current round. 
If the round is kept, they apply the filter measurement. 
If it is rejected, they perform the full measurement (using the correlated POVM). 
In either case, they determine whether the round results in a click or a no-click outcome. 
This information is then used to decide whether to reject the next rounds. 
They repeat this process for all \( n \) rounds. 

After they complete the filtering process, they perform a QND measurement with $\{\Pi_0,\Pi_1,\Pi_{>1}\}$ on those kept and clicked rounds. They complete the measurement in 0-photon subspace with $\{\vec{F}^{(0,\text{keep})}\}$. 
\( \nkzerocor \) is the number of key rounds from the 0-photon subspace that are kept.
We prove the following in \cref{sec:zerobound2}:
    \begin{align}
        \text{Pr}\biggl(\nkzerocor \geq n(q_Z+2\gmzero(n,\epzero))\biggr) \leq \epzero^2\;,
    \end{align}
    where $\gmzero(n,\epzero) = \sqrt{\frac{-\ln(\epzero)}{n}}$.
    \item They then complete the multi-click measurement $\{ F_{\mathrm{mc}}^{(>1,\text{keep})},\idd- F_{\mathrm{mc}}^{(>1,\text{keep})}\}$ on the (clicked, kept, $>1$-photon) rounds $\ngNbcor$. The proof is in \cref{sec:multibound2}:
    \begin{align}
        &\text{Pr}\biggl(\ngNbcor > \frac{-2\log(\eppntb)}{4{\lm}^2}+\frac{\nmctbcor}{\lm}+\frac{\sqrt{(-\log(\eppntb))(-\log(\eppntb)+4\lm\nmctbcor)}}{2{\lm}^2}\biggr) \leq \eppntb^2\;.
    \end{align}
They then complete all remaining measurements on \( (>1) \)-photon signals using some POVM \( \{\vec{G}^{(>1,\text{keep})}\} \). The specific structure of $\{\vec{G}^{(>1,\text{keep})}\}$ is not relevant to the present proof.
    \item They then perform a coarse-grained measurement
    $\{F^{(1,\text{keep})}_{\mathrm{sc}},F^{(1,\text{keep})}_{\mathrm{mc}},F^{(1,\text{keep})}_{\text{basis mismatch}}\}$ on single-photon signals. Note that due to the filter, we no longer have the no-click POVM element here.
    Recall from \cref{eq:definefxcor} and \cref{eq:defineftx},
    \begin{align}
        \begin{split}
        &\fx^{(1)} = \GoneXeq + \GoneXneq; \;\;  \fz^{(1)} = \GoneZeq + \GoneZneq;\\
        &F^{(1)}_{sc} =  \fx^{(1)}+\fz^{(1)}
        \end{split}
    \end{align} 
    Then,
  \begin{align}\label{eq:definefxcor}
    \begin{split}
         F^{(1,\text{keep})}_{\mathrm{sc}} = &\sqrt{\Fc^{(1)}}^+ \big(\GoneXeq + \GoneXneq + \GoneZeq + \GoneZneq\big)\sqrt{\Fc^{(1)}}^+\\
        &=  \sqrt{\Fc^{(1)}}^+ \big(\fx^{(1)} + \fz^{(1)} \big)\sqrt{\Fc^{(1)}}^+\\
        &= \sqrt{\Fc^{(1)}}^+ F^{(1)}_{sc} \sqrt{\Fc^{(1)}}^+
    \end{split}
    \end{align}
    from \cref{lemma:twostep}.
    \item They complete the measurements for single-photon signals. In particular, they first perform another coarse-grained measurement \( \{\ftxcor, \ftzcor\} \), where
    \begin{align}\label{eq:defineftxcor}
        \begin{split}
        &\ftxcor  =\sqrt{\FonescCor}^+\sqrt{\Fcm}^+ \fx \sqrt{\Fcm}^+\sqrt{\FonescCor}^+\\
        &\ftzcor  =\sqrt{\FonescCor}^+\sqrt{\Fcm}^+ \fz \sqrt{\Fcm}^+\sqrt{\FonescCor}^++\idd-\Pi_{\FonescCor}
         \end{split}
    \end{align}
    from \cref{lemma:twostep}.
    Then, upon the outcome \( \ftxcor \), they complete the measurement using \( \{\GtXeqcor, \GtXneqcor\} \). We denote \( \Nxonebcor \) as the number of erroneous outcomes corresponding to \( \GtXneqcor \) from $X$-basis measurements for single-photon signals.
    \begin{remark}
         For proving security using EUR statement, we apply EUR on the state with the result $\ftzcor$ after this stage conditioned on event $\Omega(\errzobs,\Nxobscor,\nxcor,\nkcor,\nmccor,\nkonecor)$, namely \\${\rho_{A^{\nkonecor}B^{\nkonecor}Z_A^{\tilde{n}_{({K,0})}}Z_A^{\tilde{n}_{({K,>1})}}E^nC}}_{\condi(\errzobs,\Nxobscor,\nxcor,\nkcor,\nmccor,\nkonecor)}$.
    \end{remark}
    \item They complete the $Z$-basis measurement for single-photon signals using 
\( \{\GtZeqcor, \GtZneqcor\} \). 
The number of phase errors, \( \Nkonebcor \), is defined as the number of 
\( \GtXneqcor \) outcomes if they had instead completed the measurement using 
\( \{\GtXeqcor, \GtXneqcor\} \). 
\begin{align}
    &\text{Pr}\biggl(
        \Nkonebcor \geq 
        \frac{
            \Nxonebcor 
            + \sqrt{-2 \log(\epazua^2) (\nxonebcor + \nkonebcor)} 
            + \delta (\nxonebcor + \nkonebcor)
        }{a} \notag \\
    &\hspace{30pt}
    + \sqrt{-2 \log(\epazub^2) (\nxonebcor + \nkonebcor)}
    \biggr) \notag \\
    &\leq \epazua^2 + \epazub^2 ,
\end{align}
where \( \nxonebcor \) and \( \nkonebcor \) denote the number of signals measured by 
$X$- and $Z$-basis measurements, respectively, for single-photon signals. The proof is in \cref{sec:NumberofErrorBound2}:
\end{enumerate}
\subsection{\texorpdfstring{Upper bound on the number of phase errors~(\cref{NumberofErrorBound2})}{}}\label{sec:NumberofErrorBound2}
We can use the same analyses in the memoryless case in \cref{sec:phase error estimation} to bound the number of phase errors.
Consider the state $\filteredstate$ after the filtering process. Alice and Bob apply the following measurement on this global state.
\begin{align}
    \begin{split}
    \{&\sqrt{\Fc}^+\Gamma_{X,=}\sqrt{\Fc}^+,\;\sqrt{\Fc}^+\Gamma_{X,\neq}\sqrt{\Fc}^+,\;\sqrt{\Fc}^+\Gamma_{Z,=}\sqrt{\Fc}^+,\;\sqrt{\Fc}^+\Gamma_{Z,\neq}\sqrt{\Fc}^+,\\&\sqrt{\Fc}^+\Gamma_{\mathrm{mc}}\sqrt{\Fc}^+,\;\sqrt{\Fc}^+\Gamma_{\text{other}}\sqrt{\Fc}^+\},
     \end{split}
\end{align}
which has a block-diagonal structure in photon number. Furthermore, we construct the estimation protocol in \cref{fig:correlated_equivalent} with the same break down as in \cref{lemma:NumberofErrorBound}. Therefore, we can directly apply \cref{lemma:NumberofErrorBound} on\\ $\filteredstate$ to obtain a bound on the number of phase errors $\Nkonebcor$. 

The remaining task is to pick $a$ and compute $\delta$. One can explicitly compute \( a \) and \( \delta \) from \( \ftxcor \) and \( \ftzcor \) by writing down POVM elements following steps from \cref{lemma:twostep}. We can also utilize properties of these two POVM elements to simplify the \( a \) and \( \delta \) computation.
We now show that \( \ftxcor \) and \( \ftzcor \) are the same as $\ftx$ and $\ftz$ in the memoryless case if specific POVM elements commute. Recall from \cref{eq:defineftx}, \cref{eq:definefxcor}, and \cref{eq:defineftxcor} when we construct the coarse-grained POVM for single-click, basis-matched POVM, and basis choice POVM:
\begin{align}
    \begin{split}
  &\ftxone  \overset{\cref{eq:defineftx}}{=} \sqrt{\FonescCoruncor}^+ \POVMABXcon^{(1)} \sqrt{\FonescCoruncor}^+;\\
    &\FonescCor \overset{\cref{eq:definefxcor}}{=}  \sqrt{ \sqrt{\Fcm}^+ \FonescCoruncor \sqrt{\Fcm}^+}^+;\\[5pt]
    &\ftxcor  \overset{\cref{eq:defineftxcor}}{=}\sqrt{\FonescCor}^+\sqrt{\Fcm}^+ \POVMABXcon^{(1)} \sqrt{\Fcm}^+\sqrt{\FonescCor}^+\\
    &\;\;\;\;\;\;\;\;= \sqrt{ \sqrt{\Fcm}^+ \FonescCoruncor \sqrt{\Fcm}^+}^+\sqrt{\Fcm}^+ \POVMABXcon^{(1)} \sqrt{\Fcm}^+ \sqrt{ \sqrt{\Fcm}^+ \FonescCoruncor \sqrt{\Fcm}^+}^+\;\;,
    \end{split}
\end{align}
where $F^+$ is the pseudo-inverse of $F$.
Notice that if \( \Fcm \) and \( \FonescCoruncor \) commute, then \( \FonescCoruncor \) and \( \sqrt{\Fcm}^+ \) also commute. Under the assumption that  \( \Fcm \) and \( \FonescCoruncor \) commute, $\ftxcor=\ftxone$ as shown below:
\begin{align}
    \begin{split}
    \ftxcor  &= \sqrt{ \sqrt{\Fcm}^+ \FonescCoruncor \sqrt{\Fcm}^+}^+\sqrt{\Fcm}^+ \POVMABXcon^{(1)} \sqrt{\Fcm}^+ \sqrt{ \sqrt{\Fcm}^+ \FonescCoruncor \sqrt{\Fcm}^+}^+\\
     &=\sqrt{ {\Fcm}^+ \FonescCoruncor}^+\sqrt{\Fcm}^+ \POVMABXcon^{(1)} \sqrt{ {\Fcm}^+ \FonescCoruncor}^+\sqrt{\Fcm}^+ \\
    &=\sqrt{ {\Fcm}^+\Fcm \FonescCoruncor}^+ \POVMABXcon^{(1)} \sqrt{ {\Fcm}^+\Fcm \FonescCoruncor}^+ \\
    &=\sqrt{ \Pi_{\Fcm} \FonescCoruncor}^+ \POVMABXcon^{(1)} \sqrt{  \Pi_{\Fcm}\FonescCoruncor}^+ \\
    &=\sqrt{\FonescCoruncor}^+ \POVMABXcon^{(1)} \sqrt{\FonescCoruncor}^+ \\
    &=\ftxone\;\;,
    \end{split}
\end{align}
where $\Pi_{\Fcm}$ is the projector on the support of $\Fcm$.  

Note that if dark counts can be modelled as a post-processing map applied to the outcomes obtained without dark counts, then
\begin{align}
    \begin{split}
    \FonescCoruncor &= \mathscr{P}_{{\mathrm{sc}} \rightarrow {\mathrm{sc}}}\,{\FonescCoruncor}^* 
    + \mathscr{P}_{{\mathrm{nc}} \rightarrow {\mathrm{sc}}}\,{\Gamma_{{\mathrm{nc}}}^{(1)}}^*, \\
    \Fcm &= \FonescCoruncor 
    + \mathscr{P}_{{\mathrm{sc}} \rightarrow {\mathrm{mc}}}\,{\FonescCoruncor}^* 
    + \left( \mathscr{P}_{{\mathrm{nc}} \rightarrow {\mathrm{sc}}} 
    + \mathscr{P}_{{\mathrm{nc}} \rightarrow {\mathrm{mc}}} \right)\,{\Gamma_{\mathrm{nc}}^{(1)}}^*,
    \end{split}
\end{align}
where \( \mathscr{P}_{a \rightarrow b} \) denotes the probability of mapping outcome \( a \) (single click, multi-click, no-click) to outcome \( b \). 
The elements \( {\FonescCoruncor}^* \) and \( {\Gamma_{\mathrm{nc}}^{(1)}}^* \) are the POVM elements in the absence of dark counts. 
Moreover, without dark counts, only single-click and no-click outcomes occur for single-photon signals:
\begin{align}
    {\FonescCoruncor}^* + {\Gamma_{\mathrm{nc}}^{(1)}}^* = \Pi_1 \;
    \implies \quad \left[ {\FonescCoruncor}^*, {\Gamma_{\mathrm{nc}}^{(1)}}^* \right] = 0.
\end{align}
Thus,
\begin{align}
    \begin{split}
    [\Fcm,\FonescCoruncor] &= [\mathscr{P}_{{\mathrm{nc} \rightarrow \mathrm{mc}}}{\FonescCoruncor}^* +(\mathscr{P}_{{\mathrm{nc} \rightarrow \mathrm{sc}}}+\mathscr{P}_{{\mathrm{nc} \rightarrow \mathrm{mc}}}){\Gamma_{\mathrm{nc}}^{(1)}}^*,\mathscr{P}_{{\mathrm{sc} \rightarrow \mathrm{sc}}}{\FonescCoruncor}^*+\mathscr{P}_{{\mathrm{nc} \rightarrow \mathrm{sc}}}{\Gamma_{\mathrm{nc}}^{(1)}}^*]\\
    &= \mathscr{P}_{{\mathrm{sc} \rightarrow \mathrm{mc}}} \mathscr{P}_{{\mathrm{nc} \rightarrow \mathrm{sc}}} [{\FonescCoruncor}^*,{\Gamma_{\mathrm{nc}}^{(1)}}^*] + (\mathscr{P}_{{\mathrm{nc} \rightarrow sc}}+\mathscr{P}_{{nc \rightarrow \mathrm{mc}}})\mathscr{P}_{{\mathrm{sc} \rightarrow \mathrm{sc}}}[{\Gamma_{\mathrm{nc}}^{(1)}}^*,{\FonescCoruncor}^*]\\
    &=0\;\;.
    \end{split}
\end{align}

Thus, one can use the same \( a \) and \( \delta \) from the memoryless case calculated for \( \pn = 1 \) in \cref{appendix:adelta}. 
Using \cref{lemma:NumberofErrorBound} and summing over all $n_{(\text{keep,click})}$, we obtain the following bound:
  \begin{align}
    \begin{split}
        &\text{Pr}\biggl(\Nkonebcor \geq \frac{\Nxonebcor+\sqrt{-2\ln(\epazua^2)(\nxonebcor+\nkonebcor)}+\delta(\nxonebcor+\nkonebcor)}{a}\notag\\+&\sqrt{-2\ln(\epazub^2)(\nxonebcor+\nkonebcor)}\biggr)\\
        &\leq \epazua^2 + \epazub^2.
    \end{split}    
    \end{align}
\subsection{\texorpdfstring{Upper bound on the number of ($>1$)-photon signals (\cref{step2multibound2})}{}}\label{sec:multibound2}
We can use the same analyses in the memoryless case in \cref{sec:phase error estimation} to bound the number of phase errors.
Consider the state $\filteredstate$ after the filtering process. Alice and Bob apply the following measurement on this global state.
\begin{align}
    \begin{split}
    \{&\sqrt{\Fc}^+\Gamma_{X,=}\sqrt{\Fc}^+,\;\sqrt{\Fc}^+\Gamma_{X,\neq}\sqrt{\Fc}^+,\;\sqrt{\Fc}^+\Gamma_{Z,=}\sqrt{\Fc}^+,\;\sqrt{\Fc}^+\Gamma_{Z,\neq}\sqrt{\Fc}^+,\\&\sqrt{\Fc}^+\Gamma_{\mathrm{mc}}\sqrt{\Fc}^+,\;\sqrt{\Fc}^+\Gamma_{\text{other}}\sqrt{\Fc}^+\},
     \end{split}
\end{align}
which has a block-diagonal structure in photon number. Thus, we use \cref{lemma:multibound} to bound the number of multi-photon kept signals $\ngNbcor$.
The missing step is finding the minimum eigenvalue of the multi-click POVM element $F_{mc}^{(>1),\text{keep}}$.
Note the following relation:
\begin{align}
    \begin{split}
    F_{mc}^{(>1,\text{keep})} &= \sqrt{\Fcgm}^+\Gamma_{\mathrm{mc}}^{(>1)}\sqrt{\Fcgm}^+ + P\\ 
                    &\geq \lm\sqrt{\Fcgm}^+\Pi_{>1}\sqrt{\Fcgm}^++P\\
                    & = \lm\Pi_{>1}{\Fcgm}^++P\;\;,
    \end{split}
\end{align} where we can pick any $P$ satisfying $0\leq P \leq \idd$.
We pick $P = \lm\Pi_{{\Fcgm}^{\perp}}$, where $\Pi_{{\Fcgm}^{\perp}}$ is the projector onto the orthogonal
complement of the support of $\Fcgm$, then we obtain the following:
\begin{align}
      F_{mc}^{(>1,\text{keep})} \geq \lm\Pi_{>1}\;\;.
\end{align}
Thus, by applying \cref{lemma:multibound}, replacing $\nmctbcor$ by $\nmcbcor$ (using monotonicity of the bound) and summing over all possible $n_{(\text{keep,click})}$:
\begin{align}
        &\text{Pr}\biggl(\ngNbcor > \frac{\nmcbcor}{\lm}+\frac{\sqrt{(-\log(\eppntb))(-\log(\eppntb)+4\lm\nmcbcor)}}{2{\lm}^2}+\biggl(\frac{-2\log(\eppntb)}{4{\lm}^2}\biggr)\biggr) \leq \eppntb^2\;.
\end{align}
\subsection{\texorpdfstring{Upper bound on the number of 0-photon signals in key generation rounds (\cref{step3zerobound2})}{}}\label{sec:zerobound2}

We begin by explaining why the method in~\cref{sec:zerobound}, which works in the memoryless case, does not apply here. The POVM elements that arise from~\cref{lemma:twostep} after applying the click/no-click filter in the 0-photon subspace are no longer small. As a result, we cannot use \cref{lemma:zerobound} to obtain a meaningful upper bound.

Therefore, we to estimate the key rounds in the 0-photon subspace before the filtering using Azuma's inequality.

\begin{restatable}{lemma}{azuma}\label{lemma:azuma}[Azuma's inequality]
    Let $\tX_1...\tX_n$ be $n$ random variables. $\rvx$ takes value $a \in \{a_1,...,a_j\}$. Denote $\Xhis{i-1} = \textbf{X}_1 ... \textbf{X}_{i-1}$ as a sequence of random variables. Define random variables
    \begin{align}
        &\rvy := \begin{cases}
            1, &\text{if } \rvx = a_l\\
            0, &\text{otherwise}
        \end{cases} \:\:\:\:,
    \end{align}
where $a_l$ is the outcome for which we want to estimate the number of occurrences. Given $\epsilon>0$, we have the following relations:
\begin{align}
    \begin{split}
    &\text{Pr}\biggl(\sum_{i=1}^n\rvy - \sum_{i=1}^n\rvE{\rvy|\Xhis{i-1}}\leq -\sqrt{-2\ln(\epsilon^2)n}\biggr)\leq \epsilon^2\\
    &\text{Pr}\biggl(\sum_{i=1}^n\rvy - \sum_{i=1}^n\rvE{\rvy|\Xhis{i-1}}\geq \sqrt{-2\ln(\epsilon^2)n}\biggr) \leq \epsilon^2,
    \end{split}
\end{align}
where the random variable $\rvE{\textbf{A}|\textbf{B}}$ is the conditional expected value of $\textbf{A}$ given $\textbf{B}$. 
\end{restatable}
\begin{remark}
    Note that conditional expectation is a random variable. Let $\textbf{A}$ and $\textbf{B}$ be two random variables. One can interpret $\rvE{\textbf{A}|\textbf{B}}$ as a random variable taking the value $\text{E}(\textbf{A}|\textbf{B}=b)$, which is the expected value of \textbf{A} given \textbf{B}=b, with probability $\text{Pr}(\textbf{B}=\text{b})$.
\end{remark}

Let us define the random variable \( \rvx \) as the outcome of the filtering process, the QND measurement, and the 0-photon measurement for all \( n \) rounds, \( 1 \leq i \leq n \). That is, the outcome at stage \circletext{2} in \cref{fig:correlated_equivalent}. 
Let \( \textbf{K}_i \) be the indicator variable that equals 1 if the result is a $Z$-basis click, 
Alice and Bob's bases match, and Bob decides to keep the round. 
\begin{align}
     \textbf{K}_i = \begin{cases}
        1,\;\; \rvx = ({K},0,\text{keep})\\
        0, \;\;\text{otherwise}
    \end{cases}
\end{align}
Using \cref{lemma:azuma}, we have the following bound on $\nkzero$:
\begin{align}
    &\text{Pr}\biggl(\sum_{i=1}^n  \textbf{K}_i - \sum_{i=1}^n\rvE{\textbf{K}_i|\Xhis{i-1}}\geq \sqrt{-2\ln(\epzero^2)n}\biggr) \leq \epzero^2\;.
\end{align}
We can relate the conditional expectation to $q_Z$.:
\begin{align}
    \begin{split}
    \rvE{\textbf{K}_i|\Xhis{i-1}} &= \textbf{Pr}(\textbf{K}_i|\Xhis{i-1})\\
        &\leq \Trace\biggl(\big(F_{(Z,=)}^{(0,\text{keep})}+F_{(Z,\neq)}^{(0,\text{keep})}\big)\;\Pi_0\sqrt{\Fc}\rho_{|\Xhis{i-1}}\sqrt{\Fc}\Pi_0\biggr)\\
        &= \Trace\biggl({\sqrt{\Fc^{(0)}}^+}\big(\GzeroZeq+\GzeroZneq\big){\sqrt{\Fc^{(0)}}^+}\;\Pi_0\sqrt{\Fc}\rho_{|\Xhis{i-1}}\sqrt{\Fc}\Pi_0\biggr)\\
        &= \Trace\biggl({\sqrt{\Fc^{(0)}}^+}\big(\GzeroZeq+\GzeroZneq\big){\sqrt{\Fc^{(0)}}^+}\;\sqrt{\Fc^{(0)}}\rho_{|\Xhis{i-1}}\sqrt{\Fc^{(0)}}\biggr)\\
        &= \Trace\biggl(\Pi_0\big(\GzeroZeq+\GzeroZneq\big)\Pi_0\rho_{|\Xhis{i-1}}\biggr)\\
        &= \Trace\biggl(\big(\GzeroZeq+\GzeroZneq\big)\rho_{|\Xhis{i-1}}\biggr)\\
        &\leq q_Z\;,
    \end{split}
\end{align}
where $\rho_{|\Xhis{i-1}}$ is the state of $i^{\text{th}}$ round conditioned on the previous outcomes $\Xhis{i-1}$. The first inequality comes from the fact that Alice and Bob may reject this round.
Thus, we obtain a bound on $\nkzerocor$:
\begin{align}
    &\text{Pr}\biggl(\nkzerocor \geq n\biggl(q_Z+2\gmzero(n,\epzero)\biggr)\biggr) \leq \epzero^2,
\end{align}
where $\gmzero(n,\epzero) = \sqrt{\frac{-\ln(\epzero)}{n}}$. We pick up a factor of 2 in the finite-size correction term compared to the memoryless case.
\section{Recipe for computing key rates}
We provide a recipe for computing key rates both in the memoryless case and the correlated case with $\pn=1$.
\subsection{The memoryless case}\label{appendix:recipememoryless}
We begin with the joint POVM for Alice and Bob \( \{\POVMXeq,\; \POVMXneq,\; \POVMZeq,\; \POVMZneq,\; \Gamma_{{\mathrm{mc}}},\; \Gamma_{\mathrm{nc}},\; \Gamma_{\text{basis mismatch}}\} \). Here, \( X \) and \( Z \) denote the basis choices of Alice and Bob, respectively, while \( = \) and \( \neq \) indicate whether the outcome is an error. \( \Gamma_{\mathrm{mc}} \) corresponds to the multi-click POVM. $\Gamma_{\mathrm{nc}}$ is the no-click POVM. We do not use $\Gamma_{\text{basis mismatch}}$ in this work.
We assume that the POVM elements are block-diagonal in photon number. Thus, for each POVM element, we can write
\begin{align}
    \Gamma_i = \bigoplus_{\pnm=0}^{\infty}\Gamma^{(\pnm)}_i,
\end{align}
where $\Gamma^{(\pnm)}_i$ lives in the $\pnm$-photon subspace.
\begin{enumerate}
    \item Compute 
    \begin{align}
               \lm \leq \lm(\Pigone\Gmc\Pigone)
    \end{align}
    We provide a way to compute $\lm$ in \cref{eq:lmformula}.
    \item Compute
    \begin{align}
        q_Z \geq \infnorm{\POVMZeq^{(0)} + \POVMZneq^{(0)}}
    \end{align}
    \item Consider a new POVM $\{\ftxone,\ftzone\}$, where
    \begin{align}
        \begin{split}
            &F_{\mathrm{sc}}^{(1)} = \GoneXeq+ \GoneXneq + \GoneZeq+ \GoneZneq\\
            &\ftxone = \sqrt{F_{\mathrm{sc}}^{(1)}}^+\big(\GoneXeq+ \GoneXneq\big)\sqrt{F_{\mathrm{sc}}^{(1)}}^+\\
             &\ftzone = \sqrt{F_{\mathrm{sc}}^{(1)}}^+\big(\GoneZeq+ \GoneZneq\big)\sqrt{F_{\mathrm{sc}}^{(1)}}^+
        \end{split}
    \end{align}
    \item Consider another POVM $\{\GtXeqone,\GtXneqone\}$
       \begin{align}
        \begin{split}
            \GtXneqone = \sqrt{\ftxone}^+\sqrt{F_{sc}^{(1)}}^+\POVMXneq^{(1)}\sqrt{F_{sc}^{(1)}}^+\sqrt{\ftxone}^+
        \end{split}
    \end{align}
    \item Pick any convenient $a$ such that $\delta$ is minimize, where
    \begin{equation}
         \delta = \infnormlong{\sqrt{\ftxone} \GtXeqone \sqrt{\ftxone} - a\sqrt{\ftzone} \GtXneqone \sqrt{\ftzone}}
    \end{equation}
    We provide a way to pick $a$ and compute $q_Z$ and $\delta$ in \cref{appendix:adelta} when imperfection parameters ($\eta_i,d_i$) are not fully known.
    \item Use key rate formula \cref{eq:keyratedecoy} to compute key rate.
\end{enumerate}
\subsection{The correlated case}\label{appendix:recipecorrelated}
Similar to the memoryless case, we begin with the joint POVM for Alice and Bob. If the previous $l_c$ rounds are no-click outcomes, the joint POVM is the uncorrelated one: \( \{\POVMXeq,\; \POVMXneq,\; \POVMZeq,\; \POVMZneq,\; \Gamma_{\mathrm{mc}},\; \Gamma_{\mathrm{nc}},\; \Gamma_{\text{basis mismatch}}\} \). 
\begin{enumerate}
    \item The first and second steps are the same as the memoryless case in \cref{appendix:recipememoryless}.
    \item Consider a new POVM $\{\ftxcor,\ftzcor\}$  
    \begin{align}
        \begin{split}
    &\Fc^{(1)} = \idd_2-\Gamma_{\mathrm{nc}}^{(1)}\\
    &F_{\mathrm{sc}}^{(1)} = \GoneXeq+ \GoneXneq + \GoneZeq+ \GoneZneq\\
 &\FonescCor =  \sqrt{ \sqrt{\Fcm}^+ \FonescCoruncor \sqrt{\Fcm}^+}^+\\
    &\ftxcor  =\sqrt{\FonescCor}^+\sqrt{\Fcm}^+ \POVMABXcon^{(1)} \sqrt{\Fcm}^+\sqrt{\FonescCor}^+\\
    &\ftzcor  =\sqrt{\FonescCor}^+\sqrt{\Fcm}^+ \POVMABZcon^{(1)} \sqrt{\Fcm}^+\sqrt{\FonescCor}^+
        \end{split}
    \end{align}
    \item Consider another POVM $\{\GtXeqcor,\GtXneqcor\}$
       \begin{align}
        \begin{split}
            \GtXneqcor = \sqrt{\ftxcor}^+\sqrt{\FonescCor}^+\sqrt{\Fc^{(1)}}^+\POVMXneq^{(1)}\sqrt{\Fc^{(1)}}^+\sqrt{\FonescCor}^+\sqrt{\ftxcor}^+
        \end{split}
    \end{align}
      \item Pick any convenient $a$ such that $\delta$ is minimize, where
    \begin{equation}
         \delta = \infnormlong{\sqrt{\ftxcor} \GtXneqcor \sqrt{\ftxcor} - a\sqrt{\ftzcor} \GtXneqcor \sqrt{\ftzcor}}
    \end{equation}
    We show in \cref{sec:NumberofErrorBound2} that if one models dark counts as a post-processing map on the outcomes without dark counts, one can use the same $a$ and $\delta$ from the memoryless case.
    \item Use key rate formula \cref{eq:keyratecorrelated} to compute key rate.
\end{enumerate}

\section{Technical Statements}\label{technical}
We denote $S_{\circ}(Q)$ to be the set of normalized states on $Q$, where $Q$ is a quantum register.
\azuma*
\begin{proof}
    We begin by constructing the martingale we will use in this proof:
    \begin{align}
        &\textbf{M}_0 := 0 \:\:\:\:\:\:\:\:
        \rvm := \textbf{M}_{i-1} + \rvy - \rvE{\rvy|\Xhis{i-1}}.
    \end{align}
It is not hard to show that $\rvm$ is a martingale with respect to the sequence $\textbf{X}_1 ... \textbf{X}_{i-1}$:
\begin{align}
     \begin{split}
    \text{E}(\rvm|\xhis{i-1}) = &\text{E}(\textbf{M}_{i-1}|\xhis{i-1}) + \text{E}(\rvy|\xhis{i-1}) - \text{E}(\rvE{\rvy|\Xhis{i-1}}|\xhis{i-1})\\
    = & m_{i-1} + \text{E}(\rvy|\xhis{i-1}) - \text{E}(\rvE{\rvy|\xhis{i-1}})\\
    = & m_{i-1} + \text{E}(\rvy|\xhis{i-1}) - \text{E}({\rvy|\xhis{i-1}})\\
    = & m_{i-1} \\
    \implies \text{E}(\rvm|\Xhis{i-1})& = \textbf{M}_{i-1}.
     \end{split}
\end{align}
 Note that $|\rvm-\textbf{M}_{i-1}|\leq1$. Thus, we can apply Azuma's inequality \cite{azuma} on $\rvmn$. It gives the following bounds.
\begin{align}
    \begin{split}
    &\text{Pr}(\sum_{i=1}^n\rvy - \sum_{i=1}^n\rvE{\rvy|\Xhis{i-1}}\leq -\sqrt{-2\ln(\epsilon^2)n})\leq \epsilon^2\\
    &\text{Pr}(\sum_{i=1}^n\rvy - \sum_{i=1}^n\rvE{\rvy|\Xhis{i-1}}\geq \sqrt{-2\ln(\epsilon^2)n}) \leq \epsilon^2
     \end{split}
\end{align}
\end{proof}
\operatorazuma*
\begin{proof}\label{proof:operatorazuma}
     Let $\tX_1...\tX_{n}$ be $n$ random variables where $\rvx$ records the outcomes (a, b, other) of measurement $\{\Gamma_{\text{a}},\Gamma_{\text{b}},...\}$. We define the following random variables:
     \begin{align}
    &\rvy := \begin{cases}
        1, &\text{if } \rvx = {\text{a}}\\
        0, &\text{otherwise}
    \end{cases}, \:\:\:\:
    \rvyt := \begin{cases}
        1, &\text{if } \rvx = {\text{b}}\\
        0, &\text{otherwise}
    \end{cases}
\end{align}
Notice that by construction, $\bm{n_{\textbf{a}}} = \sum_{i=1}^n\rvy$ and $\bm{n_{\textbf{b}}} = \sum_{i=1}^n\rvyt$.
By \cref{lemma:azuma}, we have the following relations:
\begin{align}
    &\text{Pr}\biggl(\bd{n_{\textbf{a}}} - \sum_{i=1}^{n}\rvE{\rvy|\Xhis{i-1}}\leq -\sqrt{-2\ln(\epazua^2)n}\biggr)\leq \epazua^2\label{proof:resultazuma1}\\
    &\text{Pr}\biggl(\bd{n_{\textbf{b}}} - \sum_{i=1}^{n}\rvE{\rvyt|\Xhis{i-1}}\geq \sqrt{-2\ln(\epazub^2)n}\biggr)\leq \epazub^2\label{proof:resultazuma2}
\end{align}
We can then obtain the following operator inequality since:
\begin{align}
      \begin{split}
    \infnorm{\Gamma_{\text{a}}-a\Gamma_{\text{b}}}\leq\delta
    \implies \Gamma_{\text{a}} \geq a\Gamma_{\text{b}} - \delta\idd\label{proof:operator}
      \end{split}
\end{align}
By construction, since $\rvy$ and $\rvyt$ are binary, we can connect the conditional expectations to probabilities:
\begin{align}
        \rvE{\rvy|\Xhis{i-1}} = \text{Pr}(\rvy|\Xhis{i-1}), \;\;\; \rvE{\rvyt|\Xhis{i-1}} = \text{Pr}(\rvyt|\Xhis{i-1})
\end{align}
We connect these probabilities through operator inequality~\cref{proof:operator}:
\begin{align}
      \begin{split}
        \text{Pr}(\rvy|\Xhis{i-1}) =& \Trace(\Gamma_{\text{a}}\rho_{|\Xhis{i-1}})\\
     \geq& \Trace(a\Gamma_{\text{b}} \rho_{|\Xhis{i-1}}) - \delta\\
     =& a\text{Pr}(\rvyt|\Xhis{i-1}) - \delta \;\;\;\;\; \forall \rho_{|\Xhis{i-1}},
      \end{split}
\end{align}
where $\rho_{|\Xhis{i-1}}$ is the state at the $i^{\text{th}}$ round conditioned on previous outcomes. Thus,
\begin{equation}
    \sum_{n}\rvE{\rvy|\Xhis{i-1}} \geq a\sum_{n}\rvE{\rvyt|\Xhis{i-1}}-\delta n
\end{equation}
Combine with~\cref{proof:resultazuma1}~\cref{proof:resultazuma2} we have the following bound:
\begin{align}
      \begin{split}
    &\text{Pr}\biggl(\bd{n_{\textbf{a}}} - a\sum_{i=1}^{n}\rvE{\rvyt|\Xhis{i-1}}-\delta n\leq -\sqrt{-2\ln(\epazua^2)n}\biggr)\leq \epazua^2\\
    &\text{Pr}\biggl(\bd{n_{\textbf{b}}} - \sum_{i=1}^{n}\rvE{\rvyt|\Xhis{i-1}}\geq  \sqrt{-2\ln(\epazub^2)n}\biggr) \leq \epazub^2
      \end{split}
\end{align}
Using the union bound, we combine these two inequalities:
\begin{equation}
   \text{Pr}\biggl(\bd{n_{\textbf{b}}} \geq \frac{\bd{n_{\textbf{a}}}+\sqrt{-2\ln(\epazua^2)n}+\delta n}{a}+\sqrt{-2\ln(\epazub^2)n}\biggr)\leq \epazua^2 + \epazub^2  
\end{equation}
\end{proof}

\section{Combining Bounds}\label{appendix:combine bounds}
\subsection{Combining bounds for memoryless case}
In this section, we show how we combine bounds together using the union bound to obtain \cref{eq:mainresult} in the single-photon source. 

From \cref{step3zerobound}, \cref{step3multibound} and \cref{eq:uncorrelated nk}, we have the following:
\begin{align}
    \begin{split}
        &\text{Pr}(\nkoneb < \nkb - \nkzero-\ngNb)=0\\
        &\text{Pr}\biggl(\nkzero \geq n\biggl(q_Z+\gmzero(n,\epzero)\biggr)\biggr) \leq \epzero^2\\
        &\text{Pr}\biggl(\ngNb > \frac{-2\log\eppntb}{4{\lm}^2}+\frac{\nmcb}{\lm}+\frac{\sqrt{(-\log\eppntb)(-\log\eppntb+4\lm\nmcb)}}{2{\lm}^2}\biggr) \leq \eppntb^2.
    \end{split}
\end{align}
Then,
\begin{align}
    \begin{split}
    &\text{Pr}\biggl(\nkoneb \leq \K(\nkb,\nmcb)\biggr)\\[15pt]
    =&\text{Pr}\biggl(\nkoneb \leq \nkb - n\biggl(q_Z+\gmzero(n,\epzero)\biggr)-\frac{-2\log\eppntb}{4{\lm}^2}-\frac{\nmcb}{\lm}-\frac{\sqrt{(-\log\eppntb)(-\log\eppntb+4\lm\nmcb)}}{2{\lm}^2}\biggr)\\[15pt]
    =&\text{Pr}\biggl(\nkoneb \leq \nkb - n\biggl(q_Z+\gmzero(n,\epzero)\biggr)-\frac{-2\log\eppntb}{4{\lm}^2}-\frac{\nmcb}{\lm}-\frac{\sqrt{(-\log\eppntb)(-\log\eppntb+4\lm\nmcb)}}{2{\lm}^2}\;\land\;\\
    & \nkzero \geq n\biggl(q_Z+\gmzero(n,\epzero)\biggr)\biggr) \\
    &+\text{Pr}\biggl(\nkoneb \leq \nkb - n\biggl(q_Z+\gmzero(n,\epzero)\biggr)-\frac{-2\log\eppntb}{4{\lm}^2}-\frac{\nmcb}{\lm}-\frac{\sqrt{(-\log\eppntb)(-\log\eppntb+4\lm\nmcb)}}{2{\lm}^2}\;\land\;\\
    & \nkzero < n\biggl(q_Z+\gmzero(n,\epzero)\biggr)\biggr)\\[15pt]
    \leq& \text{Pr}\biggl(\nkzero \geq n\biggl(q_Z+\gmzero(n,\epzero)\biggr)\biggr) \\
    &+\text{Pr}\biggl(\nkoneb \leq \nkb - \nkzero-\frac{-2\log\eppntb}{4{\lm}^2}-\frac{\nmcb}{\lm}-\frac{\sqrt{(-\log\eppntb)(-\log\eppntb+4\lm\nmcb)}}{2{\lm}^2}\biggr)\\[15pt]
    \leq&\epzero^2 + \text{Pr}\biggl(\nkoneb \leq \nkb - \nkzero-\frac{-2\log\eppntb}{4{\lm}^2}-\frac{\nmcb}{\lm}-\frac{\sqrt{(-\log\eppntb)(-\log\eppntb+4\lm\nmcb)}}{2{\lm}^2}\;\land\;\\
    & \ngNb > \frac{-2\log\eppntb}{4{\lm}^2}+\frac{\nmcb}{\lm}+\frac{\sqrt{(-\log\eppntb)(-\log\eppntb+4\lm\nmcb)}}{2{\lm}^2}\biggr) \\
    &+ \text{Pr}\biggl(\nkoneb \leq \nkb - \nkzero-\frac{-2\log\eppntb}{4{\lm}^2}-\frac{\nmcb}{\lm}-\frac{\sqrt{(-\log\eppntb)(-\log\eppntb+4\lm\nmcb)}}{2{\lm}^2}\;\land\;\\
    & \ngNb \leq \frac{-2\log\eppntb}{4{\lm}^2}+\frac{\nmcb}{\lm}+\frac{\sqrt{(-\log\eppntb)(-\log\eppntb+4\lm\nmcb)}}{2{\lm}^2}\biggr) \\[15pt]
    \leq& \epzero^2 + \text{Pr}\biggl(\ngNb > \frac{-2\log\eppntb}{4{\lm}^2}+\frac{\nmcb}{\lm}+\frac{\sqrt{(-\log\eppntb)(-\log\eppntb+4\lm\nmcb)}}{2{\lm}^2}\biggr)\\
    &+\text{Pr}(\nkoneb < \nkb - \nkzero-\ngNb)\\[15pt]
    \leq& \epzero^2 + \eppntb^2.
    \end{split}
\end{align}
From \cref{eq:numberoferror}, we have:
\begin{align}
    \begin{split}
    &\text{Pr}\biggl(\errks \geq \frac{\Nxobsb}{a\nkoneb}+\sqrt{\frac{-2\ln(\epazua^2)}{a}\biggl(\frac{\nxb}{\nkoneb^2}+\frac{1}{\nkoneb}\biggr)}\\&\phantom{\text{Pr}\biggl(\errks\geq}+\frac{\delta}{a}(\frac{\nxb}{\nkoneb}+1)+\sqrt{{-2\ln(\epazub^2)}\biggl(\frac{\nxb}{\nkoneb^2}+\frac{1}{\nkoneb}\biggr)}\biggr)\\
    \leq& \epazua^2 + \epazub^2.
\end{split}
\end{align}
Thus,
\begin{align}
    \begin{split}
    \text{Pr}&\biggl(\errks \geq \frac{\Nxobsb}{a\nkoneb}+\sqrt{\frac{-2\ln(\epazua^2)}{a}\biggl(\frac{\nxb}{\nkoneb^2}+\frac{1}{\nkoneb}\biggr)}+\frac{\delta}{a}(\frac{\nxb}{\nkoneb}+1)\\&+\sqrt{{-2\ln(\epazub^2)}\biggl(\frac{\nxb}{\nkoneb^2}+\frac{1}{\nkoneb}\biggr)}\;\lor \; \nkoneb \leq \K(\nkb,\nmcb)\biggr)\\
    \leq& \epazua^2 + \epazub^2+\epzero^2+\eppntb^2.
\end{split}
\end{align}
This implies
\begin{align}
    \begin{split}
    \text{Pr}&\biggl(\errks < \frac{\Nxobsb}{a\nkoneb}+\sqrt{\frac{-2\ln(\epazua^2)}{a}\biggl(\frac{\nxb}{\nkoneb^2}+\frac{1}{\nkoneb}\biggr)}+\frac{\delta}{a}(\frac{\nxb}{\nkoneb}+1)\\&+\sqrt{{-2\ln(\epazub^2)}\biggl(\frac{\nxb}{\nkoneb^2}+\frac{1}{\nkoneb}\biggr)}\;\land \; \nkoneb > \K(\nkb,\nmcb)\biggr)\\
    \leq& 1-(\epazua^2 + \epazub^2+\epzero^2+\eppntb^2).
\end{split}
\end{align}
Using the monotonicity of the bound on the phase error rate, we have
\begin{align}
    \begin{split}
    \text{Pr}&\biggl(\errks <  \frac{\Nxobsb}{a \K(\nkb,\nmcb)}+\sqrt{\frac{-2\ln(\epazua^2)}{a}\biggl(\frac{\nxb}{\K(\nkb,\nmcb)^2}+\frac{1}{\K(\nkb,\nmcb)}\biggr)}\\
    &+\frac{\delta}{a}\biggl(\frac{\nxb}{\K(\nkb,\nmcb)}+1\biggr)+\sqrt{{-2\ln(\epazub^2)}\biggl(\frac{\nxb}{\K(\nkb,\nmcb)^2}+\frac{1}{\K(\nkb,\nmcb)}\biggr)}\\
    & \land \; \nkoneb > \K(\nkb,\nmcb)\biggr)\\
    \leq& 1-(\epazua^2 + \epazub^2+\epzero^2+\eppntb^2).
\end{split}
\end{align}
Thus at the end we obtain the bound 
\begin{align}
    \begin{split}
    \text{Pr}&\biggl(\errks \geq  \bimp(\nxb,\nkb,\Nxobsb,\nmcb)\lor \; \nkoneb \leq \K(\nkb,\nmcb)\biggr)\\
    \leq& \epazua^2 + \epazub^2+\epzero^2+\eppntb^2.
\end{split}
\end{align}
\subsection{Combining bounds for source imperfections}
For source imperfections, we have the following two bounds:
\begin{align}
    \begin{split}
    & \Prdep\left({\bd{{e}_{{{X,1}}}^{\textbf{key}}}} > \Edep\big(\bm{\vec{\tilde{n}}_{{X},1}},\nksb\big)\right) \leq \epssrc^2+\epsATb^2+\epsATc^2\\
     &\;\;\;\text{Pr}\bigg((\bm{\vec{\tilde{n}}_{{X},1}},\nksb)\notin  \mathcal{S}(\bm{\vec{\tilde{n}}_{{X}}},\nkb)\bigg) \leq \epsilon^2_{\text{bounds}}.
    \end{split}
\end{align}
Then, 
\begin{align}
    \begin{split}
        &\text{Pr}\bigg({\bd{{e}_{{{X,1}}}^{\textbf{key}}}} > \Edep\big(\bm{\vec{\tilde{n}}_{{X},1}},\nksb\big)\; \lor \;(\bm{\vec{\tilde{n}}_{{X},1}},\nksb)\notin  \mathcal{S}(\bm{\vec{\tilde{n}}_{{X}}},\nkb)\bigg) \leq  \epssrc^2+\epsATb^2+\epsATc^2+\epsilon^2_{\text{bounds}}\\
         &\text{Pr}\bigg({\bd{{e}_{{{X,1}}}^{\textbf{key}}}} \leq \Edep\big(\bm{\vec{\tilde{n}}_{{X},1}},\nksb\big)\; \land \;(\bm{\vec{\tilde{n}}_{{X},1}},\nksb)\in  \mathcal{S}(\bm{\vec{\tilde{n}}_{{X}}},\nkb)\bigg) \geq 1- \epssrc^2-\epsATb^2-\epsATc^2-\epsilon^2_{\text{bounds}}
    \end{split}
\end{align}
since $\Pr(A\lor B)\leq\Pr(A)+\Pr(B)$ and $\Pr(\neg(A\lor B)) = \Pr(\neg A \land \neg B)$.
Thus, we have the following:
\begin{align}
    \begin{split}
    &\Prdep\left({\bd{{e}_{{{X,1}}}^{\textbf{key}}}} \leq \max\limits_{({\vec{\tilde{n}}_{{X},1}},\nks)\in \mathcal{S}(\bm{\vec{\tilde{n}}_{{X}}},\nkb)}\;\Edep\big({\vec{\tilde{n}}_{{X},1}},\nks\big) \land \nksb >\K(\nkb,\nmcb)\right) \\ &\geq 1-\epssrc^2-\epsATb^2-\epsATc^2 - \epsilon^2_{\text{bounds}}\\[10pt]
    &\Prdep\left({\bd{{e}_{{{X,1}}}^{\textbf{key}}}} > \max\limits_{({\vec{\tilde{n}}_{{X},1}},\nks)\in \mathcal{S}(\bm{\vec{\tilde{n}}_{{X}}},\nkb)}\;\Edep\big({\vec{\tilde{n}}_{{X},1}},\nks\big)\lor\nksb \leq \K(\nkb,\nmcb)\right) \\& \leq \epssrc^2+\epsATb^2+\epsATc^2 + \epsilon^2_{\text{bounds}}.
\end{split}
\end{align}
\end{document}